\documentclass{article}
\usepackage{amssymb,graphicx,amsmath,array,verbatim,mathdots,mathrsfs,amsthm}
\usepackage{youngtab}
\usepackage[all]{xy}
\usepackage{authblk}
\usepackage{cite}
\usepackage{hyperref}

\newcommand {\kahler}{K\"ahler }
\newcommand{\nn}{\nonumber\\}
\def\Lbox{\fbox{\phantom{a}}}
\def\T{{\rm T}}
\def\wt{\widetilde}

\newcommand{\beq}{\begin{eqnarray}}
\newcommand{\eeq}{\end{eqnarray}}
\newcommand{\p}{\partial}
\newcommand{\NF}{N_{\rm F}}

\newcommand{\hs}[1]{\hspace{#1 mm}}
\newcommand{\bpm}{\begin{pmatrix}}
\newcommand{\epm}{\end{pmatrix}}
\newcommand{\Z}{\mathbb{Z}}
\newcommand{\R}{\mathbb{R}}
\newcommand{\C}{\mathbb{C}}
\newcommand{\tr}{{\rm Tr}}
\newcommand{\D}{\mathcal D}

\newcommand{\ba}{\left(\begin{array}}
\newcommand{\ea}{\end{array} \right)}

\renewcommand{\thefootnote}{\arabic{footnote}}

\makeatletter
\@addtoreset{equation}{section}

\theoremstyle{plain}
\newtheorem{thm}{Theorem}[section]
\newtheorem{lmm}{Lemma}[section]
\theoremstyle{definition}

\makeatother

\usepackage[usenames]{color}

\makeatletter
\@addtoreset{equation}{section}

\renewcommand{\thefootnote}{\fnsymbol{footnote}}
\makeatother

\begin{document}
\thispagestyle{empty}
\title{{\LARGE Moduli Spaces of Instantons in Flag Manifold Sigma Models} \\
{\Large - Vortices in Quiver Gauge Theories -}\footnote{
This paper is dedicated to Prof. Norisuke Sakai, who passed away in June 2022.
He contributed to the development of the moduli matrix formalism, which forms the basis of this paper, and built it together with us.
}}

\author[1,3]{Toshiaki Fujimori\footnote{toshiaki.fujimori018@gmail.com}}
\author[2,3,4]{Muneto Nitta\footnote{nitta@phys-h.keio.ac.jp}}
\author[3]{Keisuke Ohashi\footnote{keisuke084@gmail.com}}
\affil[1]{Department of Fundamental Education, Dokkyo Medical University, 880 Kitakobayashi, Mibu, Shimotsuga, Tochigi 321-0293, Japan
}
\affil[2]{Department of Physics, Keio University, Hiyoshi 4-1-1, Yokohama, Kanagawa 223-8521, Japan}
\affil[3]{Research and Education Center for Natural Sciences, Keio University, Hiyoshi 4-1-1, Yokohama, Kanagawa 223-8521, Japan}
\affil[4]{
International Institute for Sustainability with Knotted Chiral Meta Matter(SKCM$^2$), Hiroshima University, 1-3-2 Kagamiyama, Higashi-Hiroshima, Hiroshima 739-8511, Japan
}

\maketitle

\setcounter{page}{1}
\setcounter{footnote}{0}
\renewcommand{\thefootnote}{\arabic{footnote}}
\begin{abstract}
In this paper, we discuss lumps (sigma model instantons)
in flag manifold sigma models. 
In particular, we focus on the moduli space of BPS lumps 
in general K\"ahler flag manifold sigma models.
Such a K\"ahler flag manifold, which takes the form 
$\frac{U(n_1+\cdots+ n_{L+1})}{U(n_1) \times \cdots \times U(n_{L+1})}$, 
can be realized as a vacuum moduli space of a $U(N_1) \times \cdots \times U(N_L)$ quiver 
gauged linear sigma model. 
When the gauge coupling constants are finite, 
the gauged linear sigma model admits BPS vortex configurations,
which reduce to BPS lumps in the low energy effective sigma model
in the large gauge coupling limit. 
We derive an ADHM-like quotient construction of the moduli space of BPS vortices and lumps 
by generalizing the quotient construction in $U(N)$ gauge theories by Hanany and Tong. 
As an application, we check the dualities of the 2d models 
by computing the vortex partition functions using the quotient construction. 

\end{abstract}
\newpage
\tableofcontents
\section{Introduction}
Since their discovery, nonlinear sigma models (NL$\sigma$Ms) have been studied extensively in diverse subjects, including high energy physics and condensed matter physics.
In high energy physics, NL$\sigma$Ms in two dimensions 
share many non-perturbative properties with gauge theories in four dimensions, 
such as asymptotic freedom, dynamical mass gap, confinement and instantons \cite{DAdda:1978vbw,Witten:1978bc,Polyakov:1975yp}, 
and thus they are investigated as toy models of gauge theories in four dimensions.
NL$\sigma$Ms are defined by a map from spacetime to target spaces. 
Among possible target spaces, 
the ${\mathbb C}P^{N-1}$ model with 
the complex projective space 
${\mathbb C}P^{N-1} \simeq SU(N)/[SU(N-1)\times U(1)]$ as the target space has been most considered 
\cite{Eichenherr:1978qa,Golo:1978de,Cremmer:1978bh,DAdda:1978vbw,Witten:1978bc}
together with the $O(N)$ model with $S^{N-1} \simeq O(N)/O(N-1)$ target space.
In particular, the ${\mathbb C}P^{N-1}$ model has instanton solutions, 
which play a central role in the non-perturbative dynamics of the model. 
The ${\mathbb C}P^{N-1}$ model appears as the effective theory of a single 
non-Abelian vortex in supersymmetric $U(N)$ gauge theories \cite{Hanany:2003hp,Auzzi:2003fs,Eto:2005yh,Eto:2006cx,Tong:2005un,Eto:2006pg,Shifman:2007ce,Shifman:2009zz}, 
dense QCD at high density \cite{Balachandran:2005ev,Nakano:2007dr,Eto:2009kg,Eto:2009bh,Eto:2009tr,Eto:2013hoa}, 
and two-Higgs doublets models \cite{Dvali:1993sg,Eto:2018hhg,Eto:2018tnk}.
In the recent development of the resurgence theory,
the ${\mathbb C}P^{N-1}$ model on ${\mathbb R}^1 \times S^1$ with 
a twisted boundary condition along $S^1$ 
has been extensively discussed, where a single ${\mathbb C}P^{N-1}$ instanton is decomposed into $N$ fractional instantons 
with induced domain wall charges that sum to zero \cite{Eto:2004rz,Eto:2006mz}. 
Then, a pair of fractional instanton and anti-instanton called a bion may play an 
essential role in the resurgence theory 
\cite{Dunne:2012ae,Dunne:2012zk,Misumi:2014jua,Misumi:2014bsa,Fujimori:2016ljw,Fujimori:2017oab,Fujimori:2018kqp,Misumi:2019upg,Fujimori:2020zka}.
Self-consistent non-homogeneous solutions of the ${\mathbb C}P^{N-1}$ model 
were discussed in the large-$N$ limit in infinite space 
\cite{Nitta:2017uog}
and a finite interval
\cite{Bolognesi:2016zjp,Betti:2017zcm,Bolognesi:2018njt,
Flachi:2017xat,Flachi:2019jus}. 
In condensed matter physics, the ${\mathbb C}P^{N-1}$ model appears in 
spin chains \cite{Haldane:1982rj,Affleck:1988nt},
deconfined criticality \cite{Senthil:2003eed,PhysRevB.70.144407,Nogueira:2013oza}, 
$SU(N)$ Heisenberg models \cite{Beard:2004jr}
and ultracold atomic gases \cite{Zohar:2015hwa,Laflamme:2015wma}.

Recently, yet another class of target spaces, 
flag manifolds, have attracted great attention from both high energy and 
condensed matter physics \cite{Affleck:2021ypq}.
The flag manifold sigma models are 
NL$\sigma$Ms whose target space is the generalized flag manifold ${\cal F}_{n_1n_2\cdots n_{L+1}}$, 
which is a homogeneous space $G/H$ of the form
\begin{equation}
{\cal F}_{n_1n_2\cdots n_{L+1}} ~\equiv~ G/H 
~\cong~ \frac{U(n_1+n_2+\cdots+ n_{L+1})}{U(n_1) \times U(n_2) \times \cdots \times U(n_{L+1})}
~\cong~ \frac{SU(n_1+n_2+\cdots+ n_{L+1})}{S[U(n_1) \times U(n_2) \times \cdots \times U(n_{L+1})]} .
\label{eq:flagmanifold}
\end{equation}
The flag manifold sigma models appear in various fields of physics as low-energy effective theories 
\cite{Bykov:2014efa,Bykov:2015pka,Bykov:2019jbz,Bykov:2019vkf,Hongo:2018rpy,Tanizaki:2018xto,Ohmori:2018qza,
PhysRevA.93.021606,Amari:2017qnb,Amari:2018gbq,Wamer:2020inf,Kobayashi:2021qfj}:
spin chains \cite{Wamer:2020inf,Bykov:2014efa},
flag manifold sigma model on ${\mathbb R}\times S^1$ \cite{Hongo:2018rpy},
anomaly and topological $\theta$ term \cite{Tanizaki:2018xto,Kobayashi:2021qfj},
world-sheet theories of composite non-Abelian vortices \cite{Eto:2010aj,Ireson:2019gtn}, 
and a non-Abelian vortex lattice \cite{Kobayashi:2013axa}.
As in other sigma models, the flag manifold sigma models admit 
topologically non-trivial configurations \cite{PhysRevA.93.021606,Amari:2017qnb,Amari:2018gbq}.
In particular, there exist sigma model lumps (also called sigma model instantons in two dimensions)
characterized by the second homotopy class $\pi_2(G/H)$ of the flag manifolds
\begin{eqnarray}
\pi_2 (G/H)
%=\pi_2 \left(\frac{SU(n_1+n_2+\cdots+n_{k+1})} {S[U(n_1)\times U(n_2)\times \cdots\times U(n_{k+1})]}\right) 
=\pi_1(H)=\pi_1 \left(S[U(n_1)\times U(n_2)\times \cdots\times U(n_{L+1})]\right) ={\mathbb Z}^L.
\end{eqnarray}
In the case of $L=1$, the target space is a Grassmaniann for which lumps have been studied in Refs.~\cite{Shifman:2006kd,Eto:2007yv}. 
For $L>1$, various properties of lumps have been elucidated in \cite{PhysRevA.93.021606,Amari:2017qnb}. 
In the previous works, many authors have focused on the symmetric points in the space of sigma model coupling constants (decay constants), such as the $\Z_3$ symmetric point in the $SU(3)/U(1)^2$ sigma model \cite{PhysRevA.93.021606,Amari:2017qnb}. 

There is another special subspace in the parameter space related to supersymmetric versions of the flag manifold sigma models \cite{Bando:1983ab,Bando:1984cc,Bando:1984fn,Itoh:1985ha,Itoh:1985jz,Nitta:2003dv}.
When the coupling constants  satisfy a certain relation,
the target space becomes a K\"ahler manifold \cite{Itoh:1985ha,Itoh:1985jz} for which the model can be made supersymmetric \cite{Zumino:1979et}.
In such K\"ahler sigma models, 
sigma model lumps are Bogomol'nyi-Prasad-Sommerfield (BPS) objects,
whose moduli spaces, in general, have rich structures 
due to the property that no static force is exerted among BPS objects. 
In this paper, we study the moduli space of BPS lumps in the flag manifold sigma models. 

A convenient way to describe NL$\sigma$Ms, 
particularly K\"ahler sigma models, is to use gauged linear sigma models (GL$\sigma$Ms)
whose moduli space of vacua gives the target space \cite{Donagi:2007hi}.
When the gauge coupling constants are finite, 
the GL$\sigma$Ms admit semi-local vortex solutions \cite{Vachaspati:1991dz,Achucarro:1999it,Shifman:2006kd,Eto:2007yv}
characterized by the fundamental group $\pi_1$ of the spontaneously broken gauge group. 
Since they reduce to the sigma model lumps in the large gauge coupling limit 
(or low-energy limit), 
the moduli space of BPS vortices is equivalent to that of BPS lumps 
except for the small lump singularities which are resolved 
by the finite gauge coupling constants.  
In the case of $L=1$ (and the case of local vortices), 
the moduli space of BPS vortices is conjectured in terms of a D-brane configuration in string theory \cite{Hanany:2003hp}, 
which is described by half of the Atiyah-Drinfeld-Hitchin-Mannin(ADHM) construction for Yang-Mills instantons.
This half-ADHM formalism was shown to coincide \cite{Eto:2005yh,Eto:2006cx} with one obtained in a purely field-theoretic way called 
the moduli matrix approach \cite{Isozumi:2004vg,Eto:2005yh,Eto:2006cx,Eto:2006pg} and has been used to analyze the structure of the vortex moduli spaces \cite{Eto:2006cx,Eto:2007yv,Eto:2010aj,Hanany:2014hia}.

In this paper, we consider GL$\sigma$Ms that realize 
the \kahler coset manifolds with arbitrary complex structures as its target manifolds: 
we formulate the flag manifold sigma models by quiver gauge theories \cite{Donagi:2007hi}.
We then construct BPS vortices (lumps, instantons), obtain their moduli space through the moduli matrix approach, 
and reformulate it from the viewpoint of the ADHM-like construction.
As applications of the half ADHM moduli space, 
we compute vortex partition functions
and use them to check the Seiberg-like duality in two dimensions.

This paper is organized as follows.
In Sec.~\ref{sec:GLSM}, we formulate the flag manifold sigma models by quiver GL$\sigma$Ms.
In Sec.~\ref{sec:vortices}, we construct BPS vortices and NL$\sigma$M instantons in the flag manifold sigma models.
In Sec.~\ref{sec:kahlerquotient}, the half-ADHM quotient construction of the moduli space of BPS vortices is formulated in the quiver GL$\sigma$Ms and the flag manifold NL$\sigma$Ms, 
and in Sec.~\ref{sec:lumpduality}, the moduli space of sigma model instantons is discussed.
In Sec.~\ref{sec:partition}, we calculate the vortex partition functions and check the Seiberg-like duality.
Sec.~\ref{sec:summary} is devoted to summary and discussion. 
In Appendix \ref{appendix:Riemann}, we clarify the relation between \kahler and Riemannian flag manifolds.
In Appendix \ref{appendix:uniqueness}, we give comments on the proposition on the existence of the BPS solutions addressed in the main text and on the non-existence of other solutions.
Appendix \ref{sec:ZPsiPatches} summarizes coordinate patches of half-ADHM data.
In Appendix \ref{sec:non-singular}, we give a condition of non-singular instantons.
Appendices \ref{sec:ZPsiPatches} and \ref{sec:non-singular} focus on the case of L = 1,
which forms the foundation of general cases with $L>1$. We give explicit proofs of the theorems related to 
the equivalence of the moduli spaces of the moduli matrix and the half-ADHM data, discussed previously in \cite{Eto:2005yh, Eto:2007yv,Eto:2006pg},
in a more comprehensive manner for the sake of self-containment.
In Appendix \ref{appendix:embedding}, we give embeddings of the moduli matrix and the half-ADHM data in the case of $L=1$ to those in the general cases. 
Appendix \ref{sec:brane} describes a D-brane configuration in string theory that provides a quotient construction of the moduli space of BPS vortices and flag manifold sigma model instantons.
In Appendix \ref{sec:tW}, a Lagrange multiplier and its vanishing theorem are described.
In Appendix \ref{sec:TorusActions}, we summarize the torus action on the \kahler quotient corresponding to the vortex moduli space.
In Appendix \ref{appendix:VP}, we derive the integration formula for the vortex partition function.

%\newpage

%If we consider  a map form $\mathbb R^2 \cup \{\infty\} =S^2$ to the flag manifold $ {\cal M}=G/H$,  
%we can defines topological charges 
%\begin{eqnarray}
%\pi_2 (G/H)
%%=\pi_2 \left(\frac{SU(n_1+n_2+\cdots+n_{k+1})} {S[U(n_1)\times U(n_2)\times \cdots\times U(n_{k+1})]}\right) 
%=\pi_1(H)=\pi_1 \left(S[U(n_1)\times U(n_2)\times \cdots\times U(n_{L+1})]\right) ={\mathbb Z}^L.
%\end{eqnarray}
%which are concretely  given by,  with  $1\le i\le L+1$,
%\begin{eqnarray}
%m_i \equiv\frac{i}{2\pi} \int_{\mathbb R^2} \tr [\tau_i\,  d U \wedge d U^\dagger ]
%=\frac{1}{2\pi i} \oint_{S^1} \tr [\tau_i\,  d U  U^\dagger ]\quad  \in \mathbb Z.
%\end{eqnarray}
%Note that one of these charges must be not independent and actually the total charge vanishes as,  using $\sum_i \tau_i ={\bf 1}$, 
%\begin{eqnarray}
%\sum_{i=1}^{k+1} m_i &=&\frac{1}{2\pi i} \oint_{S^1} \tr [d U \wedge  U^{-1}]
%=\frac{1}{2\pi i} \oint_{S^1} d \ln \det U =0,
%\end{eqnarray}
%where we used  that $\ln \det U$ is single-valued since $U\in G=SU(n_1+\cdots+n_{k+1})$. 

\section{Quiver gauge theories and flag manifold sigma models} \label{sec:GLSM}
%\subsection{Strong-coupling limits deriving  nonlinear-sigma models}
In this section, we present the gauged linear sigma model (GL$\sigma$M) description of the flag manifold sigma model.

\subsection{Flag manifolds}
Before describing the GL$\sigma$M for flag manifolds, we first recapitulate the basics of flag manifolds. 
Let $\mathcal V$ be an $N$-dimensional complex linear space and 
$(\mathcal V_0, \mathcal V_1, \cdots , \mathcal V_L, \mathcal V_{L+1})$ be a flag, i.e. a sequence of vector spaces such that 
\beq
\{0\} = \mathcal V_0 \subset {\cal V}_1 \subset \cdots \subset {\cal V}_L \subset {\cal V}_{L+1} = \mathcal V,
\eeq
where $\mathcal V_i~(i=0, 1, \cdots,L+1)$ are linear subspaces with ${\rm dim}_\C \, \mathcal V_i = N_i$ satisfying 
\beq
0=N_0<N_1<\cdots<N_L<N_{L+1}=N.
\eeq
A flag manifold is the space of possible configurations of the flag
\begin{eqnarray}
{\cal F}_{n_1 n_2 \cdots n_{L+1}} \equiv \Big \{ ({\cal V}_0, {\cal V}_1,\cdots, {\cal V}_{L+1}) \, \Big| \,
{\cal V}_i : \mbox{vector space}, \ \{0\} = \mathcal V_0 \subset {\cal V}_1 \subset \cdots \subset {\cal V}_L \subset {\cal V}_{L+1} = \mathcal V \Big\}.
\label{eq:def_flag}
\end{eqnarray}
In this paper, we label flag manifolds 
by a sequence of integers $(n_1,n_2,\cdots,n_L)$ defined by 
\begin{eqnarray}
n_i \equiv {\rm dim}_\C \, \mathcal W_i = N_i - N_{i-1}, \hs{10} 
\bigg(N_i = {\rm dim}_\C \, \mathcal V_i = \sum_{j=1}^i n_j \bigg),
\end{eqnarray}
where $\mathcal W_i$ is the orthogonal complements 
of $\mathcal V_{i-1}$ in $\mathcal V_i$ $(\mathcal V_i = \mathcal V_{i-1} \oplus \mathcal W_i)$. 
A point in the flag manifold $\mathcal F_{n_1 n_2 \cdots n_{L+1}}$ can be specified by a set of matrices $(\xi_1,\xi_2,\cdots,\xi_L)$, where $\xi_i$ is an $N_i$-by-$N$ matrix 
whose rows form a basis of $\mathcal V_i$
\beq
\xi_i = \left(\boldsymbol v_{i}^{(1)} , \, \boldsymbol v_{i}^{(2)} , \cdots, \, \boldsymbol v_{i}^{(N_i)} \right)^{\rm T} \hs{10} \mbox{$\big\{ \boldsymbol v_i^{(a)} \big\} $: basis of $\mathcal V_i$}.
\eeq
Since $\mathcal V_i$ is a linear subspace of $\mathcal V_{i+1}$, the basis vectors of $\mathcal V_i$ can be expressed as linear combinations of those of $\mathcal V_{i+1}$. 
Hence, there exist a $N_i$-by-$N_{i+1}$ matrix $q_i$ such that
\beq
\xi_i = q_i \xi_{i+1},~~(\exists \,q_i : \mbox{full rank $N_i$-by-$N_{i+1}$ matrix}), \hs{5} \xi_{L+1} = \mathbf 1_{L+1}.
\label{eq:xi_cond}
\eeq
Note that this condition implies that $\xi_i$ can be written as $\xi_i = q_i q_{i+1} \cdots q_{L}$. 
Two different sets of matrices $(\xi_1,\xi_2,\cdots,\xi_L)$ and $(\xi_1',\xi_2',\cdots,\xi_L')$ corresponds to the same flag if they are related by a change of basis of $( \mathcal V_1, \mathcal V_2, \cdots , \mathcal V_{L})$, i.e. by a $GL(N_1,\C) \times GL(N_2,\C) \times \cdots \times GL(N_L,\C)$ transformation
\beq
\xi_i' = V_i \xi_i  ~~~ \Longleftrightarrow ~~~
(\xi_1,\xi_2,\cdots,\xi_{L}) \sim (\xi_1',\xi_2',\cdots,\xi_{L}'),
\label{eq:xi_equiv}
\eeq
where $V_i \in GL(N_i,\C) ~ (i=1,2,\cdots,L)$.
Therefore, the flag manifold \eqref{eq:def_flag} 
can be identified with 
the space of the equivalence class \eqref{eq:xi_equiv}
satisfying the condition \eqref{eq:xi_cond}
\beq
\mathcal F_{n_1 n_2 \cdots n_{L+1}} = \left\{ \, (\xi_1,\xi_2,\cdots,\xi_L) \ \bigg| \ \begin{array}{l} \xi_i : \mbox{full rank $N_i$-by-$N$ matrix} \\ \xi_i = q_i \xi_{i+1}, \hs{3} \exists \, q_i : \mbox{full rank $N_i$-by-$N_{i+1}$ matrix}  \end{array} \right\} / \sim \,.
\eeq
Since $\xi_i = q_i q_{i+1} \cdots q_{L}$, 
the flag manifold \eqref{eq:def_flag} can also be regarded as 
the space of the equivalence classes of the matrices 
$(q_1,q_2,\cdots,q_L)$
\beq
\mathcal F_{n_1 n_2 \cdots n_{L+1}} = \left\{ \, (q_1,q_2,\cdots,q_L) \ \Big| \ q_i : \mbox{full rank $N_i$-by-$N_{i+1}$ matrix} \right\} / \sim \,,
\label{eq:flag_def_q}
\eeq
where the equivalence relation for $(q_1,q_2,\cdots,q_L)$ is given by
\beq
q_i' = V_i q_i V_{i+1}^{-1}  ~~~ \Longleftrightarrow ~~~
(q_1,q_2,\cdots,q_L) \sim (q_1',q_2',\cdots,q_L'),
\label{eq:q_equiv}
\eeq
with $V_i \in GL(N_i,\C)~ (i=1,2,\cdots,L)$ and $V_{L+1} = \mathbf 1_N$. 

We can show that the flag manifold is a homogeneous space given by the coset space
\beq
{\cal F}_{n_1 n_2 \cdots n_{L+1}} \cong \frac{U(N)}{U(n_1) \times U(n_2) \times \cdots \times U(n_{L+1})}.
\label{eq:flag_coset}
\eeq
To show this, let us note that  
any set of full rank matrices $(q_1,q_2,\cdots,q_L)$
can be rewritten by using the equivalence relation \eqref{eq:q_equiv} as
\begin{eqnarray}
q_i = \left\{ \begin{array}{ll} q_{i}^o & \mbox{for $1 \le i \le L-1$} \\ 
q_{L}^o \, U & \mbox{for $i = L$} \end{array} \right. ~~~\mbox{with}~~~q_i^o \equiv ( {\bf 1}_{N_i},{\bf 0}_{N_i\times n_{i+1}}),
\label{eq:o_unitary}
\end{eqnarray}
where $U$ is an element of $U(N)$ and 
$q_{i}^o$ are matrices corresponding to the standard flag $(\mathcal V_0^o, \mathcal V_1^o,\cdots,\mathcal V_{L+1}^o)$, i.e.
the flag consisting of the vector space $\mathcal V_i^o$ spanned by the first $i$ fundamental unit vectors.
This indicates that any flag is related to the standard flag by a $U(N)$ transformation. 
For a given flag, the corresponding unitary matrix $U$ is not unique since the flag is invariant under $U(n_1) \times \cdots \times U(n_{L+1})$ transformations, i.e. 
\beq
(q_1^o,\cdots,q_{L-1}^o,q_L^o U) \sim (q_1^o,\cdots,q_{L-1}^o,q_L^o U' U), ~~~\mbox{with}~~~ 
U' = 
{\renewcommand{\arraystretch}{0.9}
{\setlength{\arraycolsep}{0.8mm} 
\ba{ccc} U_1 & & \\ & \ddots & \\ & & U_L \ea}}, \hs{5} U_i \in U(n_i).
\eeq
The unitary matrices $U$ and $U'U$ give the same flag, 
and hence the flag manifold is given by the coset space \eqref{eq:flag_coset}. 

The denominator of the coset space \eqref{eq:flag_coset} 
implies that if $(n_1',n_2',\cdots,n_{L+1}')$ is 
a permutation of $(n_1,n_2,\cdots,n_{L+1})$, 
the flag manifolds ${\cal F}_{n_1 n_2 \cdots n_{L+1}}$ 
and ${\cal F}_{n_1' n_2' \cdots n_{L+1}'}$ are identical 
as a homogeneous space\footnote{
For a permutation $\sigma : (n_1,\cdots,n_{L+1}) \mapsto (n_1',\cdots,n_{L+1}') = (n_{\sigma(1)},\cdots,n_{\sigma(L+1)})$, one can define a diffeomorphism ${\cal F}_{n_1 \cdots n_{L+1}} \to {\cal F}_{n_1' \cdots n_{L+1}'}$ as $(\mathcal V_0,\mathcal V_1, \cdots , \mathcal V_{L+1}) \mapsto (\mathcal V_0',\mathcal V_1', \cdots , \mathcal V_{L+1}')$ with $\mathcal V_i = \mathcal W_{1} \oplus \cdots \oplus \mathcal W_{i}$ and $\mathcal V_i' = \mathcal W_{\sigma(1)} \oplus \cdots \oplus \mathcal W_{\sigma(i)}$.
}. 
However, in general, 
they have different complex structures and
hence they are distinct as complex manifolds.
To make the complex structure manifest, 
let us rewrite an arbitrary set of full rank matrices $(q_1,\cdots,q_L)$ by using the equivalence relation \eqref{eq:q_equiv} as
\beq
q_i = \left\{ \begin{array}{ll} q_{i}^o & \mbox{for $1 \le i \le L-1$} \\ 
q_{L}^o \, \mathcal G  & \mbox{for $i = L$} \end{array} \right., \hs{10} \mathcal G \in GL(N,\C). 
\label{eq:o_GL}
\eeq
In this case, the isotropy group of $(q_1^o,\cdots,q_L^o)$ is the parabolic subgroup $\hat H(n_1,\cdots,n_{L+1}) \subset GL(N,\C)$, i.e. the subgroup whose elements are matrices of the form  
\beq
\hat h = 
{\renewcommand{\arraystretch}{0.8}
{\setlength{\arraycolsep}{0.9mm} 
\ba{cccc}
h_1 & {\bf 0}&\cdots &{\bf 0}\\
\star &h_2&\ddots & \vdots \\
\vdots &\ddots  & \ddots &{\bf 0} \\
\star & \cdots &\star & h_{L+1}
\ea}}, \quad \quad 
\mbox{with} \quad \quad 
\begin{array}{l} 
h_i : \mbox{element of $GL(n_i,\C)$} \\  
\star : \mbox{complex block matrix}
\label{eq:parabolic_element}
\end{array}.
\eeq
Since the matrices $\mathcal G$ and $\hat h \mathcal G$ give the same flag
\beq
(q_1^o,\cdots,q_{L-1}^o,q_L^o \mathcal G) \sim (q_1^o,\cdots,q_{L-1}^o,q_L^o \hat h  \mathcal G),
\eeq
the flag manifold can also be written as the coset space 
\beq
{\cal F}_{n_1,\cdots,n_{L+1}} \cong \, GL(N,\C) / \hat  H(n_1,\cdots,n_{L+1}).
\label{eq:coset_complex}
\eeq
In general, for different ordering of the integers $(n_1,\cdots,n_L)$ and $(n'_1,\cdots,n'_L)$, the parabolic subgroups are not isomorphic to each other and hence give different complex manifolds. 
The only exception is the case with $(n_1',n_2',\cdots,n_{L+1}') = (n_{L+1},\cdots,n_2,n_1)$, 
for which the map between flags $(\mathcal V_0, \mathcal V_1,\cdots,\mathcal V_{L+1}) \mapsto (\mathcal V_0', \mathcal V_1', \cdots, \mathcal V_{L+1}') = ( \mathcal V_{L+1}^\perp , \cdots , \mathcal V_1^\perp, \mathcal V_0^\perp)$
 defines a biholomorphic map between ${\cal F}_{n_1 n_2 \cdots n_{L+1}}$ and ${\cal F}_{n_{L+1} \cdots n_2 n_1}$(see Sec.\,\ref{subsec:duality}). 
Correspondingly, there exists a duality between 
the GL$\sigma$Ms for ${\cal F}_{n_1 n_2 \cdots n_{L+1}}$ and ${\cal F}_{n_{L+1} \cdots n_2 n_1}$.

The holomorphic coordinates of 
${\cal F}_{n_1,\cdots,n_{L+1}}$ 
are the coordinates parameterizing the coset defined by ${\cal G} \sim \hat h {\cal G}$ with $\hat h \in \hat H(n_1,\cdots,n_{L+1})$. 
For example, in the neighborhood of ${\cal G} = \mathbf 1$, which corresponds to the standard flag $(q_1^o,\cdots,q_L^o)$,
we can decompose the matrix ${\cal G}$ as
\beq
{\cal G} = {\cal L} \, {\cal U}, \hs{10}
{\cal U} = 
{\renewcommand{\arraystretch}{0.9}
{\setlength{\arraycolsep}{1.1mm}
\ba{cccc} 
\mathbf 1_{n_1} & \varphi_{12} & \cdots & \varphi_{1,L+1} \\
\mathbf 0 & \mathbf 1_{n_2} & \ddots & \vdots \\
\vdots & \ddots & \ddots & \varphi_{L,L+1} \\
\mathbf 0 & \cdots & \mathbf 0 & \mathbf 1_{n_{L+1}}
\ea}},
\label{eq:inhomogeneous}
\eeq
where ${\cal L}$ is an element of the parabolic subgroup $\hat H(n_1,\cdots,n_{L+1})$ (lower-triangular block matrix)
and $\mathcal U$ is an upper-unitriangular block matrix 
whose blocks $\varphi_{ij}~(1\leq i < j \leq L+1)$ are
$n_i$-by-$n_j$ complex matrices. 
The entries of $\varphi_{ij}$ parameterizes the coset space and hence they can be regarded as the holomorphic coordinates
in this coordinate patch. 
For this matrix $\mathcal G$, 
the set of matrices $(q_1, \cdots, q_L) = (q_1^o,\cdots,q_L^o \mathcal G)$ can be rewritten by using the equivalence relation \eqref{eq:q_equiv} as
\beq
(q_1^o,\cdots,q_L^o \mathcal G) \sim (\mathcal U_1^{-1} q_1^o \,\mathcal U_2, \cdots, \mathcal U_L^{-1} q_L^o \, \mathcal U) 
\hs{3} \mbox{with} \hs{3}
\mathcal U_i^{-1} q_i^o \, \mathcal U_{i+1} = 
{\renewcommand{\arraystretch}{0.8}
{\setlength{\arraycolsep}{1.2mm}
\ba{ccc|c} 
\mathbf 1_{n_1} & & & \varphi_{1,i+1}' \\ 
& \ddots & & \vdots \\
& & \mathbf 1_{n_i} & \varphi_{i,i+1}' \ea, 
}}
\label{eq:q_coordinates}
\eeq
where $\mathcal U_i$ are the first $N_i$-by-$N_i$ submatrices of $\mathcal U$ and the $n_i$-by-$n_j$ block $\varphi_{ij}' = \varphi_{ij} + \mathcal O(\varphi^2)$ are certain polynomials of $\varphi$'s. 
In general, we can find a representative 
in each class $[q_1,\cdots,q_L]$ 
such that the matrices $(q_1,\cdots,q_L)$ are 
holomorphic in $\varphi$'s 
in each coordinate patch. 
%(Modified by K.O
%\beq
%q_i = {\cal U}_i^{-1} q_i^o {\cal U}_{i+1}
%=({\bf 1}_{N_i}, \star ), \quad 
%\xi_i ={\cal U}_i^{-1}\xi_i^o {\cal U}\quad 
%{\rm with~} 
%{\cal U}_i = 
%{\renewcommand{\arraystretch}{0.9}
%{\setlength{\arraycolsep}{1.1mm}
%\ba{cccc} 
%\mathbf 1_{n_1} & \varphi_{12} & \cdots & \varphi_{1,i} \\
%\mathbf 0 & \mathbf 1_{n_2} & \ddots & \vdots \\
%\vdots & \ddots & \ddots & \varphi_{i-1,i} \\
%\mathbf 0 & \cdots & \mathbf 0 & \mathbf 1_{n_{i}}
%\ea}}
%\eeq
%Here "$\star$" indicates a $N_i$-by-$n_{i+1}$ matrix of which elements are regarded as independent complex coordinates related to those in $\cal U$. 
%For instance, 
%\beq
%q_1=({\bf 1}_{n_1}, \varphi_{12}), \quad
%q_2=\left(
%\begin{array}{ccc}
%   {\bf 1}_{n_1}  &  {\bf 0} & \varphi_{13}-\varphi_{12}\varphi_{23}\\
%   {\bf 0}  & {\bf 1}_{n_2} & \varphi_{23}
%\end{array}\right),\dots.
%\eeq

Although the decomposition \eqref{eq:inhomogeneous} is not always possible, there exists at least one element of the symmetric group $\sigma : (1,\cdots,N) \mapsto (\sigma(1),\cdots,\sigma(N))$ such that 
\beq
\mathcal G = \mathcal L_\sigma \, \mathcal U_\sigma \, P_\sigma, \hs{10} P_\sigma \in S_N,
\eeq
where $\mathcal L_\sigma \in H(n_1,\cdots,n_{L+1})$, 
$\mathcal U_\sigma$ is an upper-triangular block matrix 
and $P_\sigma$ is the permutation matrix 
corresponding to the element of the symmetric group $\sigma$.  
For a generic $\mathcal G$, the element $\sigma$ is not unique and hence there are several ways to decompose $\mathcal G$
\beq
\mathcal G = \mathcal L_\sigma \, \mathcal U_\sigma \, P_{\sigma} = \mathcal L_{\sigma'} \, \mathcal U_{\sigma'} \, P_{\sigma'} = \cdots. 
\eeq
The relation between $\mathcal U_{\sigma}(\varphi_\sigma)$ and $\mathcal U_{\sigma'}(\varphi_{\sigma'})$ gives the coordinate transformation $\varphi_{\sigma} \leftrightarrow \varphi_{\sigma'}$ between the patches specified by $\sigma$ and $\sigma'$. 
The ``origin"of each patch 
$\mathcal U_{\sigma}=\mathbf 1_N$ $(\varphi=0)$
corresponds to the flag obtained 
by the permuting the basis of the standard flag $(\mathcal V_0^o,\mathcal V_1^o,\cdots,\mathcal V_L^o, \mathcal V)$ by $\sigma$.
Since the permutations of the basis within $\mathcal W_i^o$ (the orthogonal complements 
of $\mathcal V_{i-1}$ in $\mathcal V_i$) 
do not change the standard flag, it is invariant under the subgroup $S_{n_1} \times \cdots \times S_{n_{L+1}} \subset S_N$.  
Hence the number of the ``origins", which is also
the number of coordinate patches requiered to cover the whole manifold, is $N!/(n_1!n_2!\cdots n_{L+1}!)$.\footnote{
The ``origins" correspond to the fixed points of a torus action $U(1)^N \subset U(N)$ and their number  is given by Euler characteristic of the flag manifold $N!/(n_1!n_2!\cdots n_{L+1}!)$.}
%({\bf discussion:}  $\cal G$ is always decomposed to 
%$LUP$ with an appropriate permutation matrix $P$. 
%  Different $P$ corresponds to a different coordinate patch and thus a transition function given as 
%${\cal U}'=\hat h {\cal U} P$.
%Therefore, Euler chracteristic $=$ the minimal number of patches needed $=$ the number of effective permutations on the flag $=N!/(n_1!n_2!\cdots n_{L+1}!)$.)

Let us see the simplest example of $L=1$. 
In this case, the flag manifold is identified with the set of planes in a vector space, i.e. the Grassmaniann 
\beq
\mathcal F_{n_1 n_2} = \left\{ \mathcal V : \mbox{vector space in $\C^N$} \ | \ {\rm dim}_\C \, \mathcal V = M \right\} = G(M,N),
\eeq
with $M=n_1,~N=n_1+n_2$.
In particular, $\mathcal F_{n_1=1,n_2=1} = \C P^1$ for $n_1=n_2=1$. To see how $q_1$ is parametrized by the holomorphic coordinate $\phi$, let us consider the decomposition \eqref{eq:inhomogeneous} for $GL(2,\C)$. 
Any matrix $\mathcal G \in GL(2,\C)$ can be decomposed into at least one of the forms 
\begin{alignat}{3}
\bullet ~~ \mathcal G &= \ba{cc} A & B \\ C & D \ea = \mathcal L \, \mathcal U, &\hs{10}
\mathcal L &= \ba{cc} a & 0 \\ c & d \ea, &\hs{5}
\mathcal U &= \ba{cc} 1 & \phi \\ 0 & 1 \ea, \label{eq:G_decompose1} \\
\bullet ~~  \mathcal G &= \ba{cc} A & B \\ C & D \ea = \mathcal L' \, \mathcal U' \, P, &\hs{10}
\mathcal L' &= \ba{cc} a' & 0 \\ c' & d' \ea, &\hs{5}
\mathcal U' &= \ba{cc} 1 & \phi'  \\ 0 & 1 \ea, \hs{5} 
P = \ba{cc} 0 & 1 \\ -1 & 0 \ea, 
\label{eq:G_decompose2}
\end{alignat}
where
\beq
a = A, ~~~~ 
c = C, ~~~~ 
d = \frac{A D - B C}{A}, \hs{10}
a' = B, ~~~~ c' = D, ~~~~ d' = \frac{AD-BC}{B},
\eeq
and $\phi$ and $\phi'$ are inhomogeneous coordinates of $\C P^1$
\beq
\phi = \frac{B}{A}, \hs{10}
\phi' = - \frac{A}{B}.
\eeq
The decomposed forms \eqref{eq:G_decompose1} and \eqref{eq:G_decompose2} exist 
except for the matrices with 
$A=0$ and $B=0$, respectively.
Multiplying these decomposed forms of $\mathcal G$ and $q_1^o=(1,0)$, 
we find two different forms of $q_1$, each of which is parametrized by the holomorphic coordinate on the respective coordinate patch
\beq
q_1 \ = \ q_1^o \, \mathcal G \ \sim \ ( 1, 0 ) \ \mathcal U 
\ = \ ( 1 , \phi )
~~~~~~\mbox{or}~~~~~~
q_1 \ = \ q_1^o \, \mathcal G \ \sim \ ( 1, 0 ) \ \mathcal U' P 
\ = \ ( - \phi' , 1 ).
\eeq
Similarly, using the decomposition of the matrix $\cal G$, we can obtain holomorphic parametrizations of $q_i$ also for general $L$.

%In the later subsection, we will discuss duality of the sigma models. In the case of the Grassmaniann $\mathcal F_{n_1 n_2} = G(M,N)$, it is identical to 
%\beq
%\mathcal F_{n_2 n_1} = \left\{ \mathcal W : \mbox{vector space in $\C^N$} \ | \ {\rm dim}_\C \, \mathcal W = N-M \right\} =  G(N-M,N) , 
%\eeq
%since any plane $\mathcal V \in \C^N$ can also be specified by its orthogonal complement $\mathcal W = \mathcal V^{\perp}$.The GL$\sigma$Ms corresponding to $G(M,N)$ and $G(N-M,N)$ are $U(M)$ and $U(N-M)$ gauge theories with $N$ fundamentals, respectively. This is the simplest example of the duality of flag manifolds and the corresponding GL$\sigma$Ms. 

\subsection{GL$\sigma$M for flag manifolds}
In this subsection, we review 
the gauged linear sigma models (GL$\sigma$Ms) corresponding to the flag manifold sigma models. 

As shown in Appendix \ref{appendix:Riemann}, 
the flag manifold $\mathcal F_{n_1 n_2 \cdots n_{L+1}}$ becomes 
a K\"ahler manifold in an $L$ dimensional subspace of the $L(L+1)/2$ dimensional parameter space of Riemann metric on $\mathcal F_{n_1 n_2 \cdots n_{L+1}}$. 
In such a subspace, the flag manifold sigma model 
can be described by 
$U(N_1) \times U(N_2) \times \cdots \times U(N_L)$ 
GL$\sigma$M specified by the quiver diagram \cite{Donagi:2007hi}
\begin{eqnarray}
\xymatrix@M=0pt{
\stackrel{N_1}\bigcirc \ar@{>}[r]^{q_1}&
\stackrel{N_2}\bigcirc \ar@{>}[r]^-{q_2}&\bigcirc \cdots\cdots
 \bigcirc\ar@{>}[r]&
\stackrel{N_{L}}\bigcirc \ar@{>}[r]^{q_L}&
\stackrel{N_{L+1}}\Lbox} 
\hs{10}
(N_1 < N_2 < \cdots < N_L < N_{L+1} = N),
\label{eq:quiver}
\end{eqnarray}
where the $i$-th node corresponds to the $U(N_i)$ gauge group, the $i$-th arrow denotes a
bifundamental field of $U(N_i) \times U(N_{i+1})$
and the last box stands for 
the $U(N_{L+1})=U(N)$ global (flavor) symmetry\footnote{
The overall $U(1)$ of the global symmetry is unphysical since it can be absorbed into the gauge group $U(N_1) \times \cdots \times U(N_L)$.}.
The Lagrangian is written in terms of
$L$ bifundamental scalar fields 
$Q_i$ ($N_i$-by-$N_{i+1}$ matrix, $i=1,\cdots,L$), 
auxiliary $U(N_i)$ gauge fields $A_\mu^i~(i=1,\cdots,L)$ and Lagrange multipliers
$D^i$ ($N_i$-by-$N_i$ matrix, $i=1,\cdots,L$) 
in the adjoint representation of $U(N_i)$
\beq
\mathcal L_0 \ = \ \sum_{i=1}^L \tr \bigg[ (\D_\mu Q_i)(\D^\mu Q_i)^\dagger + D^i \left( Q_i Q_i^\dagger - Q_{i-1}^\dagger Q_{i-1} - r_i \mathbf 1_{N_i} \right) \bigg], 
\label{eq:L_0}
\eeq
where $Q_0 = 0$ and $r_i~(i=1,\cdots,L)$ are 
positive constants
parametrizing the K\"ahler metric. 
The gauge group acts on the bifundamental field $Q_i$ as
\beq
Q_i \rightarrow U_i Q_i U_{i+1}^\dagger, \hs{10} 
U_i \in U(N_i),~~~U_{i+1} \in U(N_{i+1}),
\eeq
The covariant derivatives are defined as
\beq
\D_{\mu} Q_i  = \p_{\mu} Q_i + i (A_\mu^i Q_i - Q_i A_\mu^{i+1}), 
\hs{10}
\D_{\mu} Q_{L} = \p_{\mu} Q_L + i A_\mu^{L} Q_{L}.
\eeq

To see that this GL$\sigma$M describes 
the flag manifold sigma model, 
we need to eliminate $(A_\mu^i,D^i)$ by solving 
their equations of motion. 
The variations of the action 
with respect to the Lagrange multipliers $D^i$ give the constraints 
\beq
Q_i Q_i^\dagger - Q_{i-1}^\dagger Q_{i-1} = r_i \mathbf 1_{N_i}. 
\label{eq:Dterm_0}
\eeq
To solve these constraints, it is convenient to write $Q_i$ as
\beq
Q_i = S_{i}^{-1} q_i S_{i+1}, \hs{10} (S_{L+1}= \mathbf 1_{N}), 
\label{eq:Q_Sq}
\eeq
where $q_{i}~(i=1,\cdots,L)$ are $N_i$-by-$N_{i+1}$ matrices of complex scalar fields 
and $S_{i}~(i=1,\cdots,L)$ are elements of 
the complexified gauge group $GL(N_i,\C)$\footnote{
The matrices $S_i$ can be regarded as (the lowest components of) the auxiliary vector superfields $S_i = e^{-V_i}$
in the supersymmetric version of our system, 
\begin{eqnarray}
 S=\int d^4 x \int d^4 \theta \sum_{i=1}^L \tr\left[ e^{-2V_i}q_i e^{2V_{i+1}} q_i^\dagger + 2 r_i V_i  \right].% + \sum_{i=1}^L S_{V_i},  
 \label{eq:action}
\end{eqnarray}
where $q_i$ are chiral superfields, $r_i$ are called Fayet-Iliopoulos parameters in this context.
}. 
Then, the constraints \eqref{eq:Dterm_0} can be rewritten as 
\beq
q_i \Omega_{i+1} q_i^\dagger \Omega_i^{-1} - \Omega_i q_{i-1}^\dagger \Omega_{i-1}^{-1} q_{i-1} = r_i \mathbf 1_{N_i}, 
\hs{5} \mbox{with} \hs{5}
\Omega_i = S_i S_i^\dagger.
\label{eq:Dterm}
\eeq
These equation can be uniquely solved for $\Omega_i \in GL(N_i \C)$ as long as the $q_i$ are full rank matrices. 
Once we obtain the solution $\Omega_i$ 
for a given set of matrices $(q_1,\cdots,q_L)$, 
we can determine $S_i$ up to gauge transformations.
Note that the expression \eqref{eq:Q_Sq} in terms of $q_i$ and $S_i$ is redundant since the scalar fields $Q_i$
do not change under the complexified gauge transformation 
\beq
q_i \rightarrow V_i q_i V_{i+1}^{-1}, \hs{10}
S_i \rightarrow V_i S_i, \hs{10}
(\Omega_i \rightarrow V_i \Omega_i V_i^\dagger), 
\label{eq:redundancy}
\eeq  
where $V_i$ are arbitrary elements of $GL(N_i,\C)$ and $V_{L+1} = 1$. 
Since $S_i$ are unique (up to gauge transformation) for a given set of matrices $(q_1,\cdots,q_L)$, 
the moduli space of vacua (the set of solutions of \eqref{eq:Dterm} modulo gauge transformations) is given by
\beq
\mathcal M_{\rm vac} = \left\{ \, (q_1,q_2,\cdots,q_L) \ \Big| \ q_i : \mbox{full rank $N_i$-by-$N_{i+1}$ matrix} \right\} / \sim \,,
\eeq
where $\sim$ denotes the equivalence relation $q_i \sim V_i q_i V_{i+1}^{-1}~(i=1,\cdots,L)$. 
This is nothing but one of the representations 
of the flag manifold \eqref{eq:flag_def_q} and 
hence the moduli space of vacua is isomorphic to 
$\mathcal F_{n_1 n_2 \cdots n_{L+1}}$
\beq
\mathcal M_{\rm vac} = \mathcal F_{n_1 n_2 \cdots n_{L+1}}.
\eeq

The general solution of Eq.\,\eqref{eq:Dterm} 
can be obtained as follows. 
As we have mentioned in Eq.\,\eqref{eq:o_unitary} , 
any set of full rank matrices $(q_1,\cdots,q_L)$ 
can be rewritten, 
by using the equivalence relation \eqref{eq:redundancy}, 
into a unitary transform of the standard flag 
$(q_1^o,\cdots,q_{L-1}^o,q_L^o U)$ with $U\in U(N)$. 
Since Eq.\,\eqref{eq:Dterm} is invariant under the $U(N)$ global symmetry,  
the solution $\Omega_i$ to Eq.\,\eqref{eq:Dterm} 
for $(q_1,\cdots,q_L)=(q_1^o,\cdots,q_{L-1}^o,q_L^o U)$ 
is given by the solution $\Omega_i^o$ for the standard flag $(q_1^o,\cdots,q_{L-1}^o,q_L^o)$
\begin{eqnarray}
\Omega_i = \Omega_{i}^o \equiv {\rm diag}\left( \frac1{a_{1i}}{\bf 1}_{n_1}, \frac1{a_{2i}}{\bf 1}_{n_2},\cdots, \frac1{a_{ii}} {\bf 1}_{n_i}\right) ~~~ \mbox{with} ~~~ a_{ij}=\prod_{l=j}^L \left(\sum_{m=i}^l r_m \right) > 0.
\label{eq:psol} 
\label{eq:general_sol}
\end{eqnarray}
From these solution $\Omega_i^o=S_i^o (S_i^o)^\dagger$, we obtain solution to Eq.\,\eqref{eq:Dterm_0} through Eq.\,\eqref{eq:Q_Sq} as
\begin{align}
Q_i = \left\{ 
\begin{array}{ll} \left( {\cal Q}_i^o, {\bf 0}_{N_i \times n_{i+1}} \right) & {\rm for~} i<L \\
\left( {\cal Q}_L^o, {\bf 0}_{N_L \times n_{L+1}} \right)U & {\rm for~} i=L
\end{array} \right., \quad 
{\cal Q}_i^o = {\rm diag}( b_{1i}{\bf 1}_{n_1}, b_{2i}{\bf 1}_{n_2},\cdots, b_{ii} {\bf 1}_{n_i}), \quad
b_{ij}= \Big( \sum_{m=i}^j r_m \Big)^{\frac{1}{2}},
\label{eq:Qo}
\end{align}
up to $U(N_1)\times U(N_2)\times \dots \times U(N_L)$ gauge transformations.
Although this is 
the general solution of the constraint \eqref{eq:Dterm}, 
the complex structure of 
${\cal F}_{n_1,\cdots,n_{L+1}}$ 
is not manifest in this form of the general solution. 
To describe the K\"ahler flag manifold sigma model, 
it is convenient to make the complex structure manifest 
by parametrizing the matrices $(q_1,\cdots,q_L)$
with holomorphic coordinates. 
To this end, let us rewrite $(q_1,\cdots,q_L)$ 
by using the equivalence relation \eqref{eq:redundancy}
into the form $(q_1^o,\cdots,q_{L-1}^o,q_L^o \mathcal G)$ 
given in \eqref{eq:o_GL}.
For $\mathcal G \in GL(N,\C)$, 
we can find a pair of matrices $(\hat h, U)$ such that\footnote{
For a given $\mathcal G \in GL(N,\C)$, the pair $(\hat h, U)$ is unique up to $U(n_1) \times \cdots \times U(n_{L+1})$ transformations $(\hat h, U)\rightarrow (\hat h {U'}^\dagger , U' U)$ with $U' \in U(n_1) \times \cdots \times U(n_{L+1})$.} 
\begin{eqnarray}
{\cal G}= \hat h \, U \hs{3} \mbox{with} \hs{3} 
\hat h \in \hat H(n_1,\cdots,n_{L+1}) \hs{3} \mbox{and} \hs{3} U \in U(N),
\label{eq:iwasawa}
\end{eqnarray}
where $\hat H(n_1,\cdots,n_{L+1}) \subset GL(N,\C)$ 
is the parabolic subgroup given in \eqref{eq:parabolic_element}. 
%By using $\mathcal G \in GL(N,\C)$, 
%the general solution of the constraint \eqref{eq:Dterm}
%can be written as 
%\begin{eqnarray}
%q_i = \left\{ \begin{array}{ll} q_{i}^o & \mbox{for $1 \le i \le L-1$} \\ 
%q_{L}^o \, {\cal G} & \mbox{for $i = L$} \end{array} \right., 
%\hs{10}
%\Omega_i =\hat h_i \Omega_i^o \hat h_i^\dagger, 
%\label{BKMUgauge}
%\end{eqnarray}
Noting that 
\beq
\hat h_L^{-1} q_L^o \, \mathcal G =  q_L^o U, \hs{5} 
\hat h_i^{-1} q_i^o \, \hat h_{i+1} = q_i^o
\hs{3} \mbox{with} \hs{3}
\hat h_i = 
{\renewcommand{\arraystretch}{0.8}
{\setlength{\arraycolsep}{1.0mm} 
\ba{cccc}
h_1 & {\bf 0}&\cdots &{\bf 0}\\
\star &h_2&\ddots & \vdots \\
\vdots &\ddots  & \ddots &{\bf 0} \\
\star & \cdots &\star & h_{i}
\ea}} \in \hat H(n_1,\cdots,n_i),
\label{eq:parabolic}
\eeq
we can rewrite the general solution of the constraint \eqref{eq:Dterm} by using the equivalence relation \eqref{eq:redundancy} as 
\begin{alignat}{3}
(q_1,\cdots,q_L) &= (q_1^o,\cdots,q_L^o U), \hs{10} &
\Omega_i &= \Omega_i^o \\
&\downarrow &
&\downarrow \notag \\
(q_1,\cdots,q_L) &= (q_1^o,\cdots,q_L^o \mathcal G), & 
\Omega_i &= \hat h_i \Omega_i^o \hat h_i^{\dagger},
\label{BKMUgauge}
\end{alignat}
where two forms of the solution are related
by \eqref{eq:redundancy} with $V_i = \hat h_i$.
In this form of the solution, 
the matrices $\{ q_i \}$ 
are parametrized by the holomorphic coordinates 
$\phi^\alpha~(\alpha = 1, \cdots, {\rm dim}_{\C} \, {\cal F}_{n_1,\cdots,n_{L+1}})$,  
which are entries of the block matrices 
$\varphi_{ij}~(1 \leq i < j \leq L+1)$
given in Eq.\,\eqref{eq:inhomogeneous}. 

Next, let us write down the Lagrangian of the NL$\sigma$M 
in terms of the complex coordinates 
$(\phi^\alpha, {\bar \phi}^{{\bar \beta}})$
by regarding them as scalar fields depending on the spacetime coordinates. 
The auxiliary gauge fields $A_\mu^i$ can be eliminated 
by solving their equations of motion 
\beq
i \Big[ Q_i (\D_\mu Q_i)^\dagger - \D_\mu Q_i Q_i^\dagger + Q_{i-1}^\dagger \D_\mu Q_{i-1} - (\D_\mu Q_{i-1})^\dagger Q_{i-1} \Big] = 0.
\eeq
These equations can be solved as
\beq
A_\mu^i = - i S_i^{-1} \left( \p_\mu - \p_\mu \phi^\alpha \frac{\p}{\p \phi^\alpha} \Omega_i \Omega_i^{-1} \right) S_i = - i \p_\mu \bar \phi^{\bar \alpha} S_i^{-1} \frac{\p S_i}{\p \bar \phi^{\bar \alpha}}+i  \p_\mu \phi^\alpha \frac{\p S_i^\dagger}{\p   \phi^\alpha} {S_i^\dagger}^{-1}. \label{eq:Amu}
\eeq
Substituting into the original action \eqref{eq:L_0}, 
we obtain the NL$\sigma$M in terms of the complex coordinates $(\phi^\alpha, {\bar \phi}^{{\bar \beta}})$
\beq
\mathcal L_0 = - g_{\alpha {\bar \beta}} \, \p_\mu \phi^{\alpha} \overline{\p^\mu \phi^\beta} ,
\eeq
where the K\"ahler metric $g_{\alpha {\bar \beta}}$ is given by the formula
\beq
g_{\alpha {\bar \beta}} \ = \ \sum_{i=1}^L r_i \frac{\p}{\p \bar \phi^\beta} \tr \left( \Omega_i^{-1} \frac{\p q_i}{\p \phi^\alpha} \Omega_{i+1} q_i^\dagger \right) 
\ = \ \frac{\p^2}{\p \phi^\alpha \p {\bar \phi^{\bar \beta}}} \sum_{i=1}^L r_i \log \det \Omega_i
\ = \ \frac{\p^2 K}{\p \phi^\alpha \p {\bar \phi^{\bar \beta}}}.
\eeq
This form of the K\"ahler metric implies that 
the K\"ahler potential $K$ takes the form
\beq
K = \sum_{i=1}^L r_i \log \det \Omega_i. 
\eeq
Using the solution of the constraint \eqref{BKMUgauge}, 
we find that
\begin{eqnarray}
\log \det \Omega_i= \log |\det \hat h_i \,|^2 + \log \det \Omega_i^o .
\end{eqnarray}
Although this is the general formula for the K\"ahler potential for 
${\cal F}_{n_1,\cdots,n_{L+1}}$, it is more convenient to express $|\det \hat h_i|^2$ in terms of holomorphic quantities.
Let us consider 
\begin{eqnarray}
\xi_i \equiv q_iq_{i+1}\cdots q_{L-1}q_{L}, 
\label{eq:xi}
\end{eqnarray}
which takes the form $\xi_i = q_i^o \cdots q_L^o {\cal G}$ for the set of matrices $(q_1,\cdots,q_L)=(q_1^o,\cdots,q_L^o \mathcal G)$. 
As we have seen in Eq.\,\eqref{eq:q_coordinates}, 
the matrices $(q_1,\cdots,q_L)$ 
are holomorphically parametrized by the coordinates and hence $\xi_i$ are also holomorphic. 
By using the decomposition \eqref{eq:iwasawa}, the relations $q_j^o \hat h_{j+1}=\hat h_j q_j^o$ 
and $q_i^o {q_i^o}^\dagger = \mathbf 1_{N_i}$, 
we can show that
\beq
\xi_i \xi_i^\dagger = \hat h_i \hat h_i^\dagger. 
\eeq
Thus, we find that the K\"ahler potential is given by
\begin{eqnarray}
K = \sum_{i=1}^L r_i \log \det(\xi_i \xi_i^\dagger),
\label{eq:kahler_target}
\end{eqnarray}
where we have neglected the unphysical constant term $\log \det \Omega_i^o$. 
This expression coincides with the  K\"ahler potential 
constructed in Refs.\,\cite{Bando:1983ab,Bando:1984fn,Bando:1984cc,Itoh:1985ha}.
Note that this formula is applicable 
for any gauge choice other than Eq.\,(\ref{BKMUgauge}) 
since this K\"ahler potential is invariant under 
the complexified gauge transformations (\ref{eq:redundancy}) up to a \kahler transformation.\footnote{
For any gauge choice, we can confirm that $\log \det \Omega_i = \log \det (\xi_i \xi_i^\dagger) + const$
by using the explicit form of the solution of 
(\ref{eq:Dterm}) 
\begin{eqnarray}
\Omega_i
= \frac1{a_{ii}} \xi_i \xi_i^\dagger + \sum_{j=1}^{i-1}\left(
 \frac1{a_{ji}}-\frac1{a_{j+1,i}}\right)  \xi_i\xi_j^\dagger (\xi_j\xi_j^\dagger)^{-1} \xi_j \xi_i^\dagger. \notag
\end{eqnarray}}

\paragraph{Example of K\"ahler potential \\}
Let us see an explicit example of the K\"ahler potential in the case of $L=2$, $n_1=n_2=n_3=1$. 
Let us introduce inhomogeneous complex coordinates 
$(\phi_{12},\phi_{13},\phi_{23})$ 
of the target manifold $G^{\mathbb C}/\hat H = \mathcal F_{1,1,1}$. 
They are contained in the matrix $\mathcal G$ 
in Eq.\,\eqref{eq:inhomogeneous} as
\begin{eqnarray}
{\cal G}= \mathcal L \left( 
\begin{array}{ccc} 
1 & \phi_{12} &\phi_{13} \\ 
0 & 1 & \phi_{23}  \\ 
0 & 0 & 1 
\end{array} \right) \in GL(3,\C),
\end{eqnarray}
where $\mathcal L \in \hat H$ is a lower-triangular matrix. 
Using $\xi_i = q_i^o \cdots q_L^o {\cal G}$, we obtain the holomorphic parametrization of $\xi_i$ as
\beq
\xi_1 = (1,0,0) \, \mathcal G \sim (1,\phi_{12},\phi_{13}), \hs{10}
\xi_2 = \ba{ccc} 1 & 0 & 0 \\ 0 & 1 & 0 \ea \mathcal G \sim 
\ba{ccc} 1 & 0 & \phi_{13} - \phi_{12} \phi_{23} \\ 0 & 1 & \phi_{23} \ea,
\eeq
where we have used the equivalence relation $\xi_i \sim V_i \, \xi_i$ 
with $V_1 \in GL(1,\C)$ and $V_2 \in GL(2,\C)$. 
Inserting these expressions into \eqref{eq:kahler_target}, 
we obtain 
\beq
K = r_1 \log (1+|\phi_{12}|^2+|\phi_{13}|^2) + r_2 \log ( 1 + |\phi_{23}|^2 + |\phi_{13}-\phi_{12}\phi_{23}|^2). 
\eeq
In this way, we can obtain the explicit forms of the K\"ahler potentials for the flag manifolds. 

\paragraph{Coefficients of the beta functions \\}
In two dimensions, the target space metric flows 
under the renormalization group flow
\beq
\mu \frac{\p}{\p \mu} g_{\alpha \bar \beta} = \frac{1}{2\pi} R_{\alpha \bar \beta} + \mathcal O(1/r_i),
\eeq
where $\mu$ is the renormalization scale
and $R_{\alpha \bar \beta}$ is the Ricci curvature. 
For a K\"ahler flag manifold,   
this renormalization group equation for the metric
can be rewritten into that for the parameters $r_i$
\beq
\mu \frac{\p}{\p \mu} r_i = \frac{N_{i+1} - N_{i-1}}{2\pi} + \mathcal O(1/r_i).
\eeq
At the one-loop order, the solutions can be written as 
\beq
r_i = \frac{\beta_i}{2 \pi} \log \frac{\mu}{\Lambda_i} + \cdots, \hs{10} 
\beta_{i} \equiv N_{i+1}-N_{i-1}.
\label{eq:beta_function}
\eeq
where $\Lambda_i$ are dynamically generated scale parameters. 
Since $\beta_{i} > 0$, 
the sigma model coupling constants $1/r_i$ 
become small for $\mu \rightarrow \infty$ and 
hence the system is asymptotically free. 
As we will see below, 
the coefficients $\beta_i$ are also related to 
the dimension of the moduli space of vortices. 

\subsection{Duality at classical level}\label{subsec:duality}
In general, a flag manifold ${\cal F}_{n_{\sigma(1)}, \cdots,n_{\sigma(L+1)}}$ obtained by permuting the integers 
$(n_1,\cdots,n_{L+1}) \rightarrow (n_{\sigma(1)}, \cdots, n_{\sigma(L+1)})$ has a different complex structure 
from that of ${\cal F}_{n_1,\cdots,n_{L+1}}$. 
However, 
when $(n_{\sigma(1)}, \cdots, n_{\sigma(L+1)}) = (n_{L+1},\dots,n_2,n_{1})$ 
the two manifold have an identical complex structure, 
i.e. ${\cal F}_{n_1,n_2,\dots,n_{L+1}}$ and ${\cal F}_{n_{L+1},\dots,n_2,n_{1}}$ are identical as a complex manifold
\begin{align}
{\cal F}_{n_1,n_2,\dots,n_{L+1}} = {\cal F}_{n_{L+1},\dots,n_2,n_{1}}.
\label{eq:dualF}
\end{align}
This equivalence is explicitly given by the biholomorphic map
given in terms of the matrix $\mathcal G$
\begin{eqnarray}
{\cal G} \in GL(N,\mathbb C) \quad \mapsto \quad  {\cal G}_{\rm dual}= R \, ({\cal G}^{\rm T})^{-1}R^\dagger
 \in GL(N,\mathbb C) 
\label{eq:dualG}
\end{eqnarray}
where the matrix $R$ is defined as
\begin{eqnarray}
R= 
{\renewcommand{\arraystretch}{1.0}
{\setlength{\arraycolsep}{0.1mm} 
\left(\begin{array}{cccc}
{\bf 0} & \cdots & ~{\bf 0} & ~~ {\bf 1}_{n_{L+1}} \\
\vdots& \iddots & ~\iddots & {\bf 0}\\
{\bf 0} & ~ {\bf 1}_{n_2}& ~\iddots & \vdots \\
~~{\bf 1}_{n_1} & {\bf 0} & ~\cdots &{\bf 0}
\end{array}\right)}} 
\in U(N). 
\end{eqnarray}
This transformation reduces to the map between 
the equivalence class 
$\mathcal G \sim \hat h \mathcal G$ with $\hat h \in \hat H(n_1,n_2,\cdots,n_{L+1})$ and
$\mathcal G_{\rm dual} \sim \hat h' \mathcal G_{\rm dual}$
with $\hat h' = R \, (\hat h^{\rm T})^{-1} R^\dagger$. 
We can show that matrices of the form 
$ R \, (\hat h^{\rm T})^{-1} R^\dagger$ 
are elements of $\hat H(n_{L+1},\dots,n_2,n_1)$
\begin{align}
\hat h' ~\in~ 
R \, \hat H(n_1,n_2,\dots,n_{L+1})^{{\rm T}-1} R^\dagger 
~ \cong ~ \hat H(n_{L+1},\dots,n_2,n_1) .
\end{align}
Therefore, the transformation \eqref{eq:dualG} gives a one to one map between the flag manifolds
$GL(N,\C)/\hat H(n_1,\cdots,n_{L+1})$ and
$GL(N,\C)/\hat H(n_{L+1},\cdots,n_1)$.
Suppose $\mathcal G$ can be decomposed into $\mathcal L$ and $\mathcal U$ as given in Eq.\,\eqref{eq:inhomogeneous} ($\mathcal G = \mathcal L \, \mathcal U$).
Then, $\mathcal G_{\rm dual}$ can be decomposed into $\mathcal L_{\rm dual} \in \hat H(n_{L+1},\dots,n_2,n_1)$ and 
an upper block-triangular matrix $\mathcal U_{\rm dual}$ as 
\beq
\mathcal G_{\rm dual} = \mathcal L_{\rm dual} \, \mathcal U_{\rm dual}
\eeq
where $\mathcal L_{\rm dual}$ and $\mathcal U_{\rm dual}$ are related to $\mathcal L$ and $\mathcal U$ as 
\beq
\mathcal L_{\rm dual} \equiv R \, (\mathcal L^{\rm T})^{-1} R^\dagger, \hs{10}
\mathcal U_{\rm dual} \equiv R \, (\mathcal U^{\rm T})^{-1} R^\dagger. 
\label{eq:dual_U}
\eeq
Since the complex coordinates are contained in the matrices $\mathcal U$ and $\mathcal U_{\rm dual}$ (see Eq.\,\eqref{eq:inhomogeneous}),  
the relation between $\mathcal U$ and $\mathcal U_{\rm dual}$ gives an explicit holomorphic coordinate transformation between the complex coordinates of $\mathcal F_{n_1, \dots, n_{L+1}}$ and $\mathcal F_{n_{L+1},\cdots,n_1}$.
 
Correspondingly, by replacing the ranks of gauge groups
and the FI parameters as 
\begin{eqnarray}
N_i \to N_i^{\rm dual }=N-N_{L+1-i}, \hs{5} 
r_i \to r_i^{\rm dual}=r_{L+1-i},
\label{eq:dualFIpara}
\end{eqnarray}
we obtain a dual GL$\sigma$M and an effective NL$\sigma$M 
whose K\"ahler potential is identical to the original one up to a K\"ahler transformation
\begin{align}
\sum_i^L r_i^{\rm dual}\ln \det \xi_i^{\rm dual} (\xi_i^{\rm dual})^\dagger = 
\sum_i^L r_i \ln \det \xi_i \xi_i^\dagger
+{\rm \hbox{K\"ahler trf}.},
\end{align}
with $\xi_i^{\rm dual} =({\bf 1}_{N_i^{\rm dual}},{\bf 0} ) \, {\cal G}_{\rm dual}$\footnote{
Here we have used the following identities for the determinants of $\xi_i \xi_i^\dagger$ and $\det \xi_i^{\rm dual} (\xi_i^{\rm dual})^\dagger$
\begin{align}
\det \xi_i \xi_i^\dagger= \det \hat h_i \hat h_i^\dagger =\prod_{j=1}^i |\det  h_j |^2,\quad
\det \xi_i^{\rm dual} (\xi_i^{\rm dual})^\dagger = \prod_{j=L-i+2}^{L+1}  |\det  h_j^{{\rm T}-1}|^2, \notag
\end{align}
where $h_i$ are the matrices given in Eq.(\ref{eq:iwasawa}). Taking determinants of the 
both sides of Eq.(\ref{eq:iwasawa}), we find that
\begin{align}
\sum_{i=1}^{L+1}\ln \det h_i= \ln  \det  ({\cal G}) \notag 
\end{align}
which is holomorphic and thus, can be removed using a K\"ahler transformation.
}. 
This shows that two GL$\sigma$Ms are equivalent at the classical level.
In Sec.\,\ref{sec:partition}, we will check this duality at the quantum level by comparing the vortex partition functions. 

As an example, let us consider the $L=1$ case. 
In the case of $U(M)$ gauge theory with $N$ fundamentals, 
the moduli space of vacua $\mathcal M_{\rm vac}$ is the Grassmaniann 
\beq
\mathcal F_{n_1 n_2} = \left\{ \mathcal V : \mbox{vector space in $\C^N$} \ | \ {\rm dim}_\C \, \mathcal V = M \right\} = G(M,N),
\eeq
with $M=n_1,~N=n_1+n_2$.
The dual theory is $U(N-M)$ gauge theory with $N$ fundamentals 
and its $\mathcal M_{\rm vac}$ is given by
\beq
\mathcal F_{n_2 n_1} = \left\{ \mathcal W : \mbox{vector space in $\C^N$} \ | \ {\rm dim}_\C \, \mathcal W = N-M \right\} = G(N-M,N) , 
\eeq
These spaces are identical since any plane $\mathcal V \in \C^N$ can also be specified by its orthogonal complement $\mathcal W = \mathcal V^{\perp}$. 
Let us see the explicit coordinate transformation between these spaces. Let $\varphi$ and $\tilde \varphi$ be $M$-by-$N-M$ and $N-M$-by-$M$ matrices whose entries are inhomogeneous coordinates of the Grassmaniann $G(M,N)$ and $G(N-M,M)$, respectively. 
They are contained in the matrices 
\beq
\mathcal U = \ba{cc} \mathbf 1_{n_1} & \varphi ~~ \\ \mathbf 0 & \mathbf 1_{n_2} \ea, \hs{10}
\mathcal U_{\rm dual} = \ba{cc} \mathbf 1_{n_2} & \tilde \varphi ~~ \\ \mathbf 0 & \mathbf 1_{n_1} \ea.
\eeq
From the duality relation $\mathcal U_{\rm dual} = R \, (\mathcal U^{\rm T} )^{-1} R^\dagger$, we can read off the coordinate transformation between $\varphi$ and $\tilde \varphi$ as
\beq
\tilde \varphi = -\varphi^{\rm T}. 
\eeq
This is the simplest example of the duality of the flag manifold sigma models and the corresponding GL$\sigma$Ms. 

\subsection{Sigma model instantons} 
For any K\"ahler manifold $\mathcal M$, 
the non-linear sigma model with target space $\mathcal M$ 
admits BPS instanton solutions. 
They are given by holomorphic maps $\p_{\bar z} \phi^i =0$, which saturates the lower bound of the action
\beq
\int d^2 x \, \mathcal L \ = \ \int d^2 x \, g_{i \bar j} \p_M \phi^i \p^M \bar \phi^j \ = \ 4 \int d^2 x \, |\!| \p_{\bar z} \phi^i |\!|^2 + \int_{\R^2} i g_{i \bar j} d\phi^i \wedge d \bar \phi^j \ \geq  \ \int_{\R^2} i g_{i \bar j} d\phi^i \wedge d \bar \phi^j,
\eeq
where $z=x_1+ix_2$ and $|\!| \p_{\bar z} \phi^i |\!|^2 = g_{i \bar j} \p_z \phi^i \p_{\bar z} \bar \phi^j$ is the norm of $\p_{\bar z} \phi^i$ with respect to the K\"ahler metric $g_{i \bar j}$. 
The lower bound is given by the topological charge obtained by integrating the pullback of the K\"ahler form $\phi^\ast(\omega) = \frac{i}{2} g_{i \bar j} d\phi^i \wedge d \bar \phi^j$. 
Once we fix the configuration at the spatial infinity to a point on the target space, $\phi$ can be viewed as a map $\phi : \R^2 \cup \{\infty\} = S^2 \rightarrow \mathcal M$ 
and hence the instanton configurations are classified by $\pi_2(\mathcal M)$. 

In the case of the flag manifold, 
the topological charge is given by
\beq
\int_{\R^2} i g_{i \bar j} d\phi^i \wedge d \bar \phi^j = 
i \sum_{i=1}^L r_i \int_{\R^2} \p \bar \p \log \det \xi_i(z) \xi_i(z)^\dagger = - \frac{i}{2} \sum_{i=1}^L r_i \oint_{S^1_\infty} ( dz \p_z  - d \bar z \p_{\bar z} ) \log \det \xi_i(z) \xi_i(z)^\dagger,
\eeq
where we have used the explicit form of the K\"ahler form obtained from the K\"ahler potential \eqref{eq:kahler_target}.
Assuming that the asymptotic form of $\det \xi_i(z) \xi_i(z)^\dagger$ for large $|z|$ is given by
\beq
\det \xi_i(z) \xi_i(z)^\dagger = |z|^{2k_i} + \cdots, \hs{10} (|z| \rightarrow \infty),
\eeq
we can determine the topological charge as
\beq
\int_{\R^2} i g_{i \bar j} d\phi^i \wedge d \bar \phi^j = 2 \pi \sum_{i=1}^L r_i k_i,
\eeq
where $(k_1,\cdots,k_L) \in \Z^L = \pi_2(\mathcal F_{n_1, \dots, n_{L+1}})$ are topological numbers.
The space of instanton solutions satisfying the boundary condition with fixed topological numbers is called 
the moduli space of sigma model instantons. 
As is well known, 
there exist small instanton singularities 
in the moduli space of sigma model instantons. 
Such singularities can be resolved by introducing 
the kinetic terms for the gauge fields in the GL$\sigma$M. 
In the next section, 
we discuss vortex solutions which 
can be viewed as resolved the sigma model instantons in the GL$\sigma$M.

%%%%%%%%%%%%%%%%%%%%%%%%%%%%%%%%%%%%%%%%%%%%%%%
%%%%%%%%%%%%%%%%%%%%%%%%%%%%%%%%%%%%%%%%%%%
\section{BPS vortices in GL$\sigma$M}\label{sec:vortices}
\def\mH{\xi}
\def\mD{\mathfrak D}
\def\mJ{\mathfrak J}
\def\mV{\mathcal V}
\def\mU{\mathcal U}
\def\mq{q}
%%%%%%%%%%%%%%%%%%%%%%%%

\subsection{BPS equations and moduli matrices}
In this section, we discuss BPS vortices 
in the framework of the GL$\sigma$M 
with finite gauge coupling constants $g_i$
\begin{eqnarray}
{\cal L} = \sum_{i=1}^L
\tr\left[ \frac1{2g^2_i} F_{\mu\nu}^i F^{i \mu\nu} +{\cal D}_\mu Q_i {\cal D}^\mu Q_i^\dagger - \frac{1}{g_i^2} D_i^2 + D_i \left(Q_i Q_i^\dagger-Q_{i-1}^\dagger Q_{i-1}-r_i {\bf 1}_{N_i}\right) \right],  
\end{eqnarray}
where $F_{\mu \nu}^i = \p_\mu A_\nu^i - \p_\nu A_\mu^i + i [A_\mu^i , A_\nu^i]$ are the field strength for the $i$-th gauge field and $g_i$ are gauge coupling constants.
One can go back to the original GL$\sigma$M \eqref{eq:L_0} 
by taking the $g_i \rightarrow \infty$ limit.
In the vacua of this system, 
the gauge symmetry $U(N_1) \times \cdots \times U(N_L)$ is spontaneously broken 
and hence this model admits BPS vortex configurations, 
satisfying the boundary conditions
\begin{align}
\lim_{|z|\to \infty }Q_i = U_i(\theta)^\dagger 
\left( {\cal Q}_i^o, {\bf 0}\right) U_{i+1}(\theta), \hs{7} 
\lim_{|z| \to \infty } A_\mu^i dx^\mu = - iU_i(\theta)^\dagger d U_i(\theta),
\label{eq:boundary_cond}
\end{align}
where ${\cal Q}_i^o$ is the constant square matrices defined in Eq.\,\eqref{eq:Qo} and $U_i(\theta) \in U(N_i)$ are nontrivial elements of the gauge group depending on $\theta = \arg z$.
Each $U_i(\theta)$ carries the topological charges of 
$\pi_1(U(N_i)) = \mathbb Z$, 
which is related to the magnetic flux of $i$-th overall $U(1)$ factor
\begin{eqnarray}
k_i \equiv - \frac1{2\pi} \int d^2x \, \tr F_{12}^i = \frac{i}{2\pi} \int_0^{2\pi} d \theta \, \tr \left[ U_i(\theta)^\dagger \p_\theta U(\theta) \right] \in \mathbb Z \quad {\rm for~} 1\le i\le L. 
\label{eq:charge}
\end{eqnarray}
In the large gauge coupling limit $g_i \rightarrow \infty$, 
these vortex solutions reduce to the sigma model instantons 
(or singular configurations). 
Eliminating the auxiliary field $D_i$ 
by solving their equations of motion
\beq
D_i = \frac{g_i^2}{2} \left(Q_i Q_i^\dagger-Q_{i-1}^\dagger Q_{i-1}-r_i {\bf 1}_{N_i}\right),
\eeq
we can rewrite the Lagrangian as
\begin{eqnarray}
\int d^2x \, {\cal L} \ = \ \sum_{i=1}^L \int d^2x \, \tr \left[\frac{g_i^2}4 {\cal E}_i^2+ 4 \tilde{\cal E}_i \tilde{\cal E}_i^\dagger - r_i F_{12}^i \right] 
\ \ge \ 2\pi \sum_{i=1}^L r_i k_i
\end{eqnarray}
where we have defined 
\beq
\tilde{\cal E}_i \hs{-2} & \equiv & \hs{-2} {\cal D}_{\bar z} Q^i,  \\
{\cal E}_i \hs{-2} & \equiv & \hs{-2}
Q_i Q_i^\dagger-Q_{i-1}^\dagger Q_{i-1}-r_i {\bf 1}_{N_i} -\frac{2}{g_i^2} F_{12}^i
\eeq
with ${\cal D}_{\bar z} Q^i \equiv \frac12 ({\cal D}_1+i {\cal D}_2) Q^i$.
For a fixed set of topological charges $(k_1,\cdots,k_L)$, 
the action is minimized when ${\cal E}_i=\tilde{\cal E}_i=0$, 
i.e. the following equations are satisfied
\beq
0 \hs{-2} &=& \hs{-2} {\cal D}_{\bar z} Q^i, \label{eq:BPSQ} \\
0 \hs{-2} &=& \hs{-2} Q_i Q_i^\dagger-Q_{i-1}^\dagger Q_{i-1}-r_i {\bf 1}_{N_i} -\frac{2}{g_i^2} F_{12}^i. 
\label{eq:BPSD}
\eeq
These equations are called the BPS equations for vortices. 

Solutions to these equations describe
configurations of vortices, that is, 
squeezed magnetic fluxes in the Higgs phase.
The vortices in this system are classified 
as {\it local vortices} or {\it semi-local vortices} depending on their asymptotic behaviors 
at the spatial infinity. 
The local vortex exhibits 
an exponentially dumping behavior and 
is obtained by, roughly speaking, embedding  well-known Abrikosov-Nielsen-Olesen vortex into the matrix elements.
The semi-local vortex has a (decreasing) 
power-law behavior due to the tails of the massless Nambu-Goldstone fields parameterizing the moduli space of vacua $\mathcal F_{n_1, \dots, n_{L+1}}$.
The size of a semi-local vortex is one of moduli parameters of vortices and hence it can become arbitrarily large without changing the energy. 
However, it has a minimum size of order ${\mathcal O}(1/g\sqrt{r})$. 
In the small size limit, the semi-local vortex reduces to the local type.

\paragraph{The master equation and the boundary condition \\}
Let us rewrite the BPS equations into a convenient form.
The first set of BPS equations (\ref{eq:BPSQ}) 
can be solved as
\begin{eqnarray}
Q_i= S_i^{-1} \mq_i(z) \, S_{i+1}, \hs{10}
A_{\bar z} \equiv \frac{1}{2} (A_1^i+iA_2^i) = - i S_i^{-1} \partial_{\bar z} S_i.\label{eq:QAtoS}
\end{eqnarray}
where $S_i$ are elements of $GL(N_i,\C)~(S_{L+1}={\bf 1})$
and $q_i(z)$ are $N_i$-by-$N_{i+1}$ matrices
whose entries are arbitrary polynomials of $z=x_1+i x_2$.\footnote{Although these entries are arbitrary entire functions in general, we can assume that they are polynomials without loss of generality 
as shown in Appendix \ref{subsec:patch_L=1}.} 
Note that the description in terms of $q_i$ and $S_i$ 
is redundant since the following transformation does not change the original fields $Q_i$ and $A_i$:
\beq
q_i(z) \rightarrow V_i(z) q_i(z) V_{i+1}^{-1}(z), \hs{10}
S_i(z,\bar z) \rightarrow V_i(z) S_i(z,\bar z), \hs{10}
\mbox{with} ~~~ V_i(z) \in GL(N_i,\mathbb C).
\label{eq:V-transf}
\eeq
This is a complexified gauge transformation depending on the holomorphic coodinate $z=x_1+i x_2$. 
We call this transformation the {\it V-transformation}.
To study the moduli space of vortices, 
it is convenient to define $N_i$-by-$N$ matrices $\mH_i(z)$ as
\begin{eqnarray}
\mH_i(z) \equiv \mq_i(z) \mq_{i+1}(z)\cdots \mq_{L}(z), \hs{10}
{\rm for~} 1\le i\le L.
\end{eqnarray}
For a given set of matrices $(q_1(z),\cdots,q_L(z))$, 
the matrices $S_i(z,\bar z)$ are determined 
(up to gauge transformations)
by solving the second BPS equation \eqref{eq:BPSD}, 
which can be rewritten in terms of 
$\Omega_i = S_i S_i^\dagger \in GL(N_i,\mathbb C)$ as
\begin{eqnarray}
q_i \Omega_{i+1} q_i^\dagger \Omega_i^{-1} - \Omega_i q_{i-1}^\dagger \Omega_{i-1}^{-1} q_{i-1} = r_i {\bf 1}_{N_i}
- \frac{4}{g_i^2} \partial_{\bar z} \left( \p_z \Omega_i \Omega_i^{-1} \right), \hs{10} (q_0=0).
\label{eq:mastereq}
\end{eqnarray}
We call this set of equations {\it the master equation} for vortices. 
Once we determine $\Omega_i$ by solving the master equation, we can obtain the original fields $Q_i$ and $A_\mu$ satisfying the BPS equation for vorties. 
The boundary conditions for $\Omega_i$ can be determined as follows.
Without loss of generality, 
we can fix the vacuum at the spatial infinity 
to the point on $\mathcal F_{n_1, \dots, n_{L+1}}$ 
corresponding to the standard flag 
${\mathcal V}_0^o \subset {\mathcal V}_1^o \subset \cdots \subset {\mathcal V}_N^o \subset {\mathcal V}$ as
\beq
\xi_i(z)\quad ``\to" \quad \xi_i^o = q_i^o q_{i+1}^o \cdots q_L^o=( \mathbf 1_{N_i}, \mathbf 0_{N_i,N-N_i} ),   \quad {\rm for~} |z| \to \infty,
\eeq 
up to the redundancy of the $V$-transformation \eqref{eq:V-transf}. 
To precisely describe what the limit ``$\to $" means, let us decompose the $N_i$-by-$N$ matrix $\xi_i(z)$ 
into a $N_i$-by-$N_i$ matrix $\mD(z)$ and a $N_i$-by-$(N-N_i)$ matrix $\widetilde \mD(z)$
as $\xi_i(z)=(\mD_i(z), \widetilde {\mD}_i(z))$.
This decomposition corresponds to the orthogonal decomposition of the vector space $\mathcal V = \mathcal V_i^o \oplus \mathcal V_i^{o\perp}$
\beq
\mH_i(z)=(\mD_i(z), \widetilde {\mD}_i(z)) ~~~ 
\leftrightarrow ~~~ \mD_i(z) \ \equiv \ \xi_i(z) \, {\xi_i^o}^\dagger, \hs{5}
\widetilde {\mD}_i(z) \ \equiv \ \xi_i(z) \, {\xi_i^{o\perp}}^\dagger,
\eeq
\label{eq:def_mD}
where 
$\xi_i^{o\perp} = (\mathbf 0_{N-N_i,N_i}, \mathbf 1_{N-N_i})$.   
%Due to the redundancy of $\{ q_i, S_i\}$, we can take different gauges for  different spatial domains. 
%In the vicinity of the spatial infinity, $D_\infty,   ( \{ \infty\} \subset D_\infty \subset S^2) $,  it is convenient to 
%take a singular gauge analog  $\{ q_i^{\rm bd}, S_i^{\rm bd}\}$.
Using the $N_i$-by-$N_i$ matrix $\mD_i(z)$, 
we set the boundary conditions for $\xi_i(z)$ as 
\begin{align}
\lim_{|z|\to\infty} \{\mD_i(z)^{-1} \xi_i(z)\} = \xi_i^o. 
\label{eq:BC_xi}
\end{align}
Correspondingly, the boundary coniditions for $\{q_i, S_i\}$ are given by
\begin{align}
\lim_{|z|\to\infty}\left\{\mD_i(z)^{-1} q_i(z) \, \mD_{i+1}(z)\right\}
= q_i^{\rm o}, \hs{10}
\lim_{|z|\to \infty} \left\{  \mD_i(z)^{-1}S_i(z,\bar z) \right\}
=S_i^{\rm o} U_i(\theta),  
%\left(  \lim_{|z|\to \infty} \Omega_i^{\rm bd}(z,\bar z)=\Omega_i^{\rm o}\right)  
\label{eq:qSboundaryCond}
\end{align}
where $\mD_{L+1}(z) = \mathbf 1_N$,
$U_i(\theta)\in U(N_i)$ and 
$S_i^o$ is a matrix such that $\Omega_i^o = S_i^o (S_i^o)^\dagger$ is the diagonal matrix defined in Eq.\,\eqref{eq:general_sol}. 
These boundary conditions correspond to those for the original quantities Eq.\,\eqref{eq:boundary_cond}. 
%\footnote{
%Since the gauge fields $A_\mu^i$ behave asymptotically as
%\begin{align}
%A_r^i \sim 0,\quad A_\theta^i\sim  - iU_i(\theta)^\dagger \partial_\theta U_i(\theta), 
%\end{align}
%the vorticities $k_i$ can be calculated in a different, but more familiar way as
%\begin{align}
%k_i =-\oint  d\theta \tr [A_\theta^i]=i \oint d\theta \partial_\theta \log \det U_i(\theta).
%\end{align}
%}
Note that the transformations with $\mD_i(z)$ in Eq.\eqref{eq:qSboundaryCond}
can be regraded as local $V$-transformations analogous to singular gauge transformations, which are regular in the asymptotic region $|z| \rightarrow \infty$. 
Substituting these settings to the master equation, 
we find that the asymptotic behavior of $\Omega_i$ is given by
\begin{align}
\mD_i(z)^{-1} \, \Omega_i(z\,\bar z) \, \mD_i(z)^{\dagger-1} \ = \ \Omega_i^o+ {\cal O}(|z|^{-2}) \quad {\rm for} \quad |z| \to \infty. 
\label{eq:asymptoticS}
\end{align}
These are the boundary conditions for the master equations \eqref{eq:mastereq}. 

In terms of these matrices satisfying the boundary conditions given above, 
the vortex numbers $k_i$ defined in \eqref{eq:charge} are given by
\begin{eqnarray}
k_i=\frac1{4\pi i}\oint (dz \partial_z - d{\bar z} \partial_{\bar z}) \, \log | \det  \mD_i(z) |^2, 
\label{eq:chargeH}
\end{eqnarray}
where we have used the following formula for 
the magnetic flux and its asymptotic behavior
\begin{eqnarray}
\tr F_{12}^i \,
=-2 \p_{\bar z}\p_z \log \det \Omega_i,\qquad \log \det \Omega_i= \log |\det \mD(z)|^2+{\rm const.}+{\cal O}(|z|^{-2}) 
 ~~~ \mbox{for} ~~~ |z| \rightarrow \infty. 
\label{eq:flux_asymptotic}
\end{eqnarray}
Eq.\,\eqref{eq:chargeH} implies that 
the matrices $(q_1(z),\cdots,q_L(z))$ must be chosen such that $\det\mD_i(z)$ has $k_i$ zeros in the topological sector with vortex numbers $(k_1,\cdots,k_L)$.\footnote{
Using the $V$-transformation,   $\det\mD_i(z)$ can be set to be a monic polynomial of degree $k_i$.}
The zeros of $\det \mD_i(z)$ can be regarded as the positions of vortices.
To see this, let us consider the $p$-th moment of 
the $i$-th magnetic fluxes. 
If all vortices are well separated from each other, the magnetic flux $F_{12}^i$ is symmetrically localized around $k_i$ distinct points $\{z_{(i,\alpha)}| \alpha=1,2,\dots,k_i\}$.
In such a case, the $p$-th moment is given by
\begin{align}
\langle z^p \rangle_i \equiv -\frac1{2\pi} \int d^2x \, z^p \, \tr[F_{12}^i] =
\sum_{\alpha=1}^{k_i} (z_{(i,\alpha)})^p,
\label{eq:moment_position}
\end{align}  
On the other hand, using \eqref{eq:flux_asymptotic}, we can calculate $\langle z^p \rangle_i$ as
\begin{align}
\langle z^p \rangle_i = 
\frac1{2\pi i} \lim_{R\to \infty }\oint_{|z|=R} dz \, z^p \partial_z \log \det \mD_i(z) =
\sum_{\alpha=1}^{k_i} (w_{(i,\alpha)})^p,
\label{eq:moment_position_2}
\end{align}
where $\{w_{(i,\alpha)}| \alpha=1,2,\dots,k_i\}$ are zeros of $\det \mD_i(z)$.
From \eqref{eq:moment_position} and \eqref{eq:moment_position_2}, the vortex positions 
$\{ z_{i,\alpha} \}$ can be identified with the zeros 
$\{ w_{i,\alpha} \}$ of $\det \mD_i(z)$\footnote{
Since $\mD_i(z)$ can be reconstructed as $\det \mD_i(z)=z^{k_i}\exp(-\sum_{p=1}^\infty \langle z^p \rangle_i /(p z^p) )$,  
each zero $w_{(i,\alpha)}$ can always be uniquely read from any configuration of the magnetic flux.}.
%\footnote{For instance, in the limit $g_i \rightarrow \infty$, the magnetic flux of a local vortex configuration consists of superposition of $k_i$ delta functions localized at the zeros of $\det\mD_i(z)$.}
Extending this identification, we adopt $\{w_{(i,\alpha)}\}$ as the definition of the vortex positions 
even if several vortices are in close proximity and their flux profiles are overlapping. 
 
\paragraph{Uniqueness and existence of solution \\}
Under the boundary conditions given above, 
one can prove the uniqueness of the solution $\{\Omega_i\}$ to the set of the master equations \eqref{eq:mastereq} with a given set of $\{q_i(z)\}$ (see Appendix \ref{appendix:uniqueness} for the proof).
In the following, we only assume that 
\begin{enumerate}
\item[]
{\it there exists a solution for the set of the master equations \eqref{eq:mastereq}  
with a given set of $q_i=q_i(z)$ such that 
$ \Omega_i=\Omega_i(z,\bar z)$ is an element of $GL(N_i,\mathbb C)$ everywhere 
and all entries of $\Omega_i$ are smooth functions of $z$ and $\bar z$.}
\end{enumerate}
At least, this assumptions is true in the large coupling limit $g_i^2\to \infty$
as will be explained later.
The set of the master equations is a generalization of the so-called Taubes equation, where this assumption has been shown to be true \cite{Taubes:1979tm}. 

\paragraph{The moduli matrices and the moduli space of vortices \\}
If the above assumption for 
the existence of the solution is true, 
there is a one-to-one correspondence between 
the moduli space of vortices and the set of the equivalence classes defined by the $V$-transformation
\beq
q_i(z) \sim V_i(z) \, q_i(z) \, V_{i+1}^{-1}(z).
\label{eq:V-transf_q}
\eeq 
Hence the matrices $q_i(z)$ are called {\it the moduli matrices}. 
As shown in Appendix \ref{sec:ZPsiPatches}, 
all the entries of the matrices $q_i(z)$ and $V_i(z)$ can be assumeed to be polynomials.
From Eqs.\,\eqref{eq:BC_xi} and \eqref{eq:chargeH}, 
we find that the boundary conditions for vortex configurations 
carrying the topological charges $\{k_i\}$ are expressed in terms of $\mH_i(z)=(\mD_i(z), \widetilde {\mD}_i(z))$ as
\begin{eqnarray}
\det \mD_i(z) ={\cal O}(z^{k_i}), \quad
\mD_i(z)^{-1} \widetilde \mD_i(z)={\cal O}(z^{-1}).
\label{eq:patches}
\end{eqnarray}
Roughly speaking, 
the first and second conditions in \eqref{eq:patches} 
specify the vortex numbers $\{k_i\}$
and the vacuum at the spacial infinity, respectively. 
For fixed vortex numbers and boundary conditions, 
we define the moduli space of vortices 
as the space of equivalence classes of the matrices 
$(q_1(z),\dots,q_L(z))$ satisfying the condition \eqref{eq:patches}
\begin{align}
{\cal M}_{{\rm vtx}}{}^{n_1,n_2,\dots,n_{L+1}}_{\, k_1,k_2,\cdots,k_L} = \left\{ (q_1(z),q_2(z),\dots,q_L(z)) \, \Big| \, q_i(z) : \mbox{conditions} \,\eqref{eq:patches} \right\}/\sim,
\label{eq:vtx_moduli}
\end{align}
where $\sim$ denotes the equivalence relation \eqref{eq:V-transf_q}. 
We can determine the dimension of the moduli space 
by counting the number of zero modes satisfying 
the linearized BPS equations around a BPS configuration. 
As discussed in \cite{Hanany:2003hp} and \cite{Eto_2009}, 
the number of zero modes can be determined by the index theorem as
%\paragraph{moduli space of vortices and the index theorem \\}
%It is expected that for an arbitrary set of $\{q_i(z) \,|\, i=1,2,\dots,L\}$, Eq.\,\eqref{eq:mastereq} always has a unique solution, at least for finite gauge couplings. For instance, roughly speaking, in the area farther than the typical distance represented by $1/g_i\sqrt{r_i}$ from cores of vortices, the second term in the r.h.s.of Eq.\eqref{eq:mastereq} can be omitted and solutions are well approximated by those given by Eq.\eqref{eq:explicitSolution}. Therefore, under this assumption, the moduli space of vortices is expressed as
\begin{align}
 {\rm dim}_{\mathbb C} \, {\cal M}_{{\rm vtx}}{}^{n_1,n_2,\dots,n_{L+1}}_{\, k_1,k_2,\cdots,k_L} 
\ = \ - \sum_{i=1}^L \frac{N_{i+1}-N_{i-1}}{2\pi} \int d^2x \, \tr F_{12}^i 
\ = \ \sum_{i=1}^L (N_{i+1}-N_{i-1}) k_i
\ = \ \sum_{i=1}^L \beta_{i} k_i,
\label{eq:dim_M}
\end{align}
where $\beta_i = N_{i+1}-N_{i-1}$ is 
the first coefficient of the beta function \eqref{eq:beta_function}.

The moduli space of BPS vortices is endowed with a natural complex structure such that the variables $\phi^A~(A=1,\cdots,\sum_{i=1}^L \beta_{i} k_i)$ holomorphically parametrizing the equivalence class of the $V$-transformation \eqref{eq:V-transf_q} are the complex coordinates of the vortex moduli space. Furthermore, the vortex moduli space is equipped with a K\"ahler metric that determines classical dynamics of the vortices. As shown in Appendix \ref{sec:Kahlermetric}, 
the K\"ahler metric on the vortex moduli space is given by the formula
\begin{align}
g_{A\bar B} = 
\int d^2x \, 
\frac{\p}{\p \bar \phi^{\bar B}}\sum_{i=1}^L \tr \left[ \Omega_{i}^{-1} 
\frac{\partial q_i}{\partial \phi^A} \Omega_{i+1}q_i^\dagger \right]_{\Omega=\Omega^{\rm sol}}.
%\label{eq:Kmetricformula}
\end{align}

\paragraph{Local and semi-local vortices \\}
The first condition in \eqref{eq:patches} implies that $\xi_i(z)$ are full rank matrices at a generic point $z \in \C$. 
Such a set of full rank matrices $(\xi_1(z),\cdots,\xi_L(z))$ at a point $z$ specifies a flag ${\mathcal V}_0 \subset {\mathcal V}_1 \subset \cdots \subset {\mathcal V}_L \subset {\mathcal V}$
and hence a point in the flag manifold $\mathcal F_{n_1,\cdots,n_{L+1}}$. 
Therefore, if $\xi_i(z)$ are full rank matrices $({\rm rank}(\xi_i(z))=N_i)$ everywhere on $\C$, that is,  
\beq
\det \xi_i(z) \xi_i(z)^\dagger \not=0,~~\forall z \in \C,~~(i=1,\cdots,L),
\label{eq:cond_rank}
\eeq 
the set of matrices $(\xi_1(z),\cdots,\xi_L(z))$ gives a holomorphic map $\C \to \mathcal F_{n_1,\cdots,n_{L+1}}$. 
In such a case, we can solve the equation \eqref{eq:mastereq} in the large gauge coupling limit $g_i \rightarrow \infty$ 
by promoting the vacuum solution satisfying \eqref{eq:Dterm}
into a $z$-dependent configuration
\begin{eqnarray}
\Omega_i
= \frac1{a_{ii}} \xi_i(z) \xi_i(z)^\dagger + \sum_{j=1}^{i-1}\left(
 \frac1{a_{ji}}-\frac1{a_{j+1,i}}\right)  \xi_i(z) \xi_j(z)^\dagger (\xi_j(z) \xi_j(z)^\dagger)^{-1} \xi_j(z) \xi_i(z)^\dagger.
 \label{eq:sol_omega_large}
\end{eqnarray} 
Comparing physical quantities such as energy density, 
we can confirm that the vortex configuration reduces to the instanton solution specified by the same set of matrices $(\xi_1,\cdots,\xi_L)$
\footnote{
From these matrices $\xi_i(z)$,
we can always construct an instanton solution 
in terms of the inhomogeneous coordinates 
by comparing $\xi_i(z)$ and ${\cal U}={\cal U}(z)$ 
in Eq.\eqref{eq:inhomogeneous} at each $z$ as 
\begin{align}
\xi_i(z)= \mathcal L_i(z) \, \xi_i^o  \, {\cal U}(z)
% = q_i^oq_{i+1}^o\cdots q_L^o \, \mathcal L_i(z) \, {\cal U}(z),
\notag
\end{align} 
where $\mathcal L_i(z)$ is an appropriate element of the parabolic subgroup $\hat H(n_1,\cdots,n_i)$ except for some singular points corresponding to the zeros of $\det \mathcal L_i(z) = \det \mD_i(z) ={\cal O}(z^{k_i})$. \notag
}. 
Therefore, the moduli space of sigma model instantons is given by restricting the vortex moduli space \eqref{eq:vtx_moduli} with the additional condition \eqref{eq:cond_rank}
\begin{align}
    {\cal M}_{{\rm inst}}{}^{n_1,n_2,\dots,n_{L+1}}_{~k_1,k_2,\cdots,k_L}
    &=\left\{(q_1(z),q_2(z),\dots,q_L(z)) \, \Big| \, q_i(z):{\rm Eq.\,}\eqref{eq:patches}, \, {\rm Eq.\,} \eqref{eq:cond_rank} \right\}/\sim.
\end{align}
The points removed by the condition \eqref{eq:cond_rank} are configurations 
with matrices $\xi_i(z)$ whose rank becomes smaller $({\rm rank}(\xi_i(z))<N_i)$ at some points on $\C$.
Such configuration cannot be viewed as a holomprphic map since $\xi_i$ must be full rank matrices on the flag manifold. 
As we approach the removed points on the instanton moduli space, the sizes of some insntantons become infinitesimally small and hence such points are called the small instanton singularities.
Although instanton configurations are singular at such points, 
the corresponding vortex solutions are regular 
as long as the gauge coupling constants $g_i$ are finite. 
Instead of the singular instantons, 
regular vortices with size $\mathcal O(1/g \sqrt{r})$ are located at the points where ${\rm rank}(\xi_i) < N_i$ when the gauge coupling constants are finite. 
Such vortices are the so-called local vortices
whereas vortices corresponding to regular instantons are called the semi-local vortices. 
A set of matrices $\{\xi_i\}$ with $\wt \mD_i = 0$ corresponds to a configuration in which all the vortices are of local type. The corresponding modui subspace is called the local vortex moduli space
while we call
the subspace with no local vortex the semi-local vortex moduli space. 
As shown in Appendix \ref{sec:inst-sol}, there is a one-to-one correspondence between the moduli spaces of semi-local vortices and sigma model instantons
\beq
 {\cal M}_{{\rm semi}}{}^{n_1,n_2,\dots,n_{L+1}}_{~k_1,k_2,\cdots,k_L} 
 = 
 {\cal M}_{{\rm inst}}{}^{n_1,n_2,\dots,n_{L+1}}_{~k_1,k_2,\cdots,k_L}. 
\eeq
The semi-local vortex moduli space can be obtained from the local one by turning on $\wt \mD_i$,
which corresponds to the fibration on the local vortex moduli space.
In the large gauge coupling limit $g_i \rightarrow \infty$, the local vortex moduli space shrinks to the small instanton singularity. 

For a dual pair of GL$\sigma$Ms, 
the moduli space of sigma model instantons 
are identical since the corresponding NL$\sigma$Ms agree in the $g_i \rightarrow \infty$ limit. 
On the other hand, as we will see below, 
the votex moduli spaces can be viewed as different regularizations of the instanton moduli space
\begin{align}
\begin{matrix}
{\cal M}_{{\rm inst}}{}^{n_1,n_2,\dots,n_{L+1}}_{~k_1,k_2,\cdots,k_L} & = & {\cal M}_{{\rm inst}}{}^{n_{L+1},\dots,n_2,n_1}_{~k_L,\cdots,k_2,k_1} \\ \cap & & \cap \\  {\cal M}_{{\rm vtx}}{}^{n_1,n_2,\dots,n_{L+1}}_{~k_1,k_2,\cdots,k_L} & \not = & 
{\cal M}_{{\rm vtx}}{}^{n_{L+1},\dots,n_2,n_1}_{~k_L,\cdots,k_2,k_1}.
\end{matrix}
\end{align}
In Sec.\ref{sec:lumpduality}, we will check the duality of instanton solutions by presenting explicit biholomorphic map on the moduli parameters. 
The small-instanton singularities are regularized in different ways in the dual pair of GL$\sigma$Ms, i.e. both models have distinct local vortex moduli spaces. 
Nevertheless, as we will check in Sec.\,\ref{sec:partition},
the vortex partition functions computed using the information on the local vortex moduli spaces are in perfect agreement. 

\subsection%[Example 1: Review of L=1 case]
{Example 1: review of $L=1$ case}
\label{subsec:ExampleMM_L=1}
As the simplest example, 
let us review the case with 
$L=1$ \cite{Eto:2005yh}
where the target manifold of the NL$\sigma$M is the complex Grassmaniann $G(N,n)$ $(n_1=n,\,n_2=N-n)$. 
Here we omit the index $i~(i=1,\cdots,L)$ 
since there is only one gauge group factor $U(n)$ in this case. 

For a $k$-vortex configuration, 
the matrices $\mD(z)$ and $\tilde \mD(z)$ are $n$-by-$n$ and $n$-by-$(N-n)$ matrices satisfying 
\begin{eqnarray}
\det \mD(z)={\cal O}(z^{k}), \hs{5} 
\mD(z)^{-1} \widetilde \mD(z)={\cal O}(z^{-1}).
\label{eq:DDcond}
\end{eqnarray}
Two pairs of matrices $(\mD(z),\,\widetilde \mD(z))$ and 
$(\mD(z)',\,\widetilde \mD(z)')$ are 
equivalent if there is a matrix $V(z) \in GL(N,\C)$
such that
\begin{eqnarray}
(\mD(z),\,\widetilde \mD(z)) = (V(z) \mD(z)',\, V(z) \widetilde \mD(z)')
 ~~~ \Longleftrightarrow ~~~
(\mD(z),\wt \mD(z)) \sim (\mD(z)',\wt \mD(z)').
\label{eq:Vtr} 
\end{eqnarray}
Therefore, the moduli space of $k$-vortex configurations is given by
\beq
{\mathcal M_{\rm vtx \,}}_{k}^{n,N-n} \cong \left\{ (\mD(z), \wt \mD(z) ) \, \Big| \det \mD(z) = \mathcal O(z^k)\,,\, \mD(z)^{-1} \wt \mD(z) = \mathcal O(z^{-1}) \right\} / \sim.
\label{eq:moduli_v_L=1}
\eeq
This moduli space can be parametrized in the following way (see Appendix \ref{subsec:patch_L=1} for more details).
Let $\lambda=(l_1,l_2,\cdots,l_n)$ be a set of non-negative integers 
such that $k=\sum_{b=1}^n l_b$. 
By using the $V$-transformation \eqref{eq:Vtr}, 
a generic matrix $\mD(z)$ with $\det \mD(z)={\cal O}(z^k)$
can be transformed into the following form  
%\begin{equation}
%(\mD(z))^a{}_{b}= \delta^a{}_{b} z^{l_a} - \sum_{m=1}^{l_b}T^a{}_{b,m} z^{m-1} 
%\label{eq:Dpatch}
%\end{equation}
\beq
\mD(z) = 
\mD_\lambda(z) \equiv 
{\renewcommand{\arraystretch}{0.8}
{\setlength{\arraycolsep}{0.8mm} 
\ba{ccc}
z^{l_1} & & \\
& \ddots & \\
& & z^{l_n}
\ea 
+ 
\ba{cccc}
P^{11} & \cdots & P^{1n} \\
\vdots & \ddots & \vdots \\
P^{n1} & \cdots & P^{nn} 
\ea}}, \hs{5}
P^{ab}(z) =\sum_{m=1}^{l_b} T^{ab}_{m} \, z^{m-1}. 
\label{eq:Dpatch}
\eeq
%({\bf K.O.s' comment:}
%With a given set $(l_1,l_2,\dots, l_n)$, the above is a general form of a matrix $\mD(z)$ satisfying \begin{align}
%\mD(z) =\left[{\bf 1}_n+{\cal O}(z^{-1})\right]
%\times {\rm diag.}(z^{l_1},z^{l_2},\dots,z^{l_n}).
%\end{align} 
%Therefore, if $V=\mD^{(1)}(\mD^{(2)})^{-1}$ is regular %everywhere for two matrices $\mD^{(1)},\mD^{(2)}$ in this %gauge, then we can conclude $V=1$.)
For each ``gauge choice", 
i.e. the choice of $(l_1,\cdots,l_n)$, 
all the degrees of freedom of 
the $V$-transformation \eqref{eq:Vtr} is fixed.
Therefore, the coefficients $T^{ab}_{m}$
can be regarded as part of the complex coordinates of the moduli space of vortices
in this coordinate patch, 
which we call the $(l_1,l_2,\cdots,l_n)$-patch.
The other coordinates parameterize the degrees of freedom contained in the matrix $\wt \mD(z)$ obeying the condition $\mD(z)^{-1} \widetilde \mD(z)={\cal O}(z^{-1})$. 
To extract such degrees of freedom, 
let us consider $n$-component column vectors $j_i(z)$ 
with polynomial entries satisfying  
\begin{eqnarray}
\mD(z)^{-1} j_i(z) = {\cal O}(z^{-1}) ~~~ (z \rightarrow \infty).
\label{eq:cond_j}
\end{eqnarray}
We can show that there exist $k$ linearly independent solutions 
$j_i(z)~(i=1,\cdots,k)$ satisfying this condition.
Let $\mJ(z)$ be an $n$-by-$k$ matrix whose columns form a basis of the solutions to \eqref{eq:cond_j}, that is,
$\mJ(z)=(j_1(z),\cdots,j_k(z))$.
For instance, in the $(l_1,l_2,\cdots,l_n)$-patch (\ref{eq:Dpatch}), 
$\mJ(z)$ can be chosen as 
\beq
\mJ_\lambda = \ba{cccc} \boldsymbol{\mJ}_1 & \boldsymbol{\mJ}_2 & \cdots & \boldsymbol{\mJ}_n \ea, \hs{10}
(\boldsymbol{\mJ}_\alpha)^a{}_p = \frac{ \mD^a{}_\alpha}{z^p}   \Big|_{\rm reg} = \delta^a{}_{\alpha} \, z^{l_\alpha-p} + \sum_{m=1}^{l_\alpha-p} T^{a\alpha}_m z^{m-1},
\label{eq:Jpatch}
\eeq
where $\boldsymbol{\mJ}_\alpha$ are $n$-by-$l_\alpha$ block matrices and $|_{\rm reg}$ stands for the regular part. 
Since $\wt \mD(z)$ satisfies the condition $\mD(z)^{-1} \widetilde \mD(z)={\cal O}(z^{-1})$, 
it can be written as linear combinations of $j_i(z)$, that is, the matrix $\wt \mD(z)$ in the $(l_1,l_2,\cdots,l_n)$-patch can be written as
\beq
\wt \mD(z) = \wt \mD_\lambda(z) \equiv \mJ_\lambda \wt \Psi_\lambda, 
\eeq
where $\wt \Psi$ is a $k$-by-$(N-n)$ matrix. 
The components of $\wt \Psi$ parameterize the degrees of freedom of $\wt \mD$ 
and hence can be regarded as the remaining coordinates of the moduli space of vortices.

One can check that the number of the coordinates $T^{ab}_{m}~(a,b=1,\cdots,n,~m=1,\cdots,l_b)$ and $\wt \Psi_{ic}~(i=1,\cdots,k,~c=1,\cdots,N-n)$ agrees with the dimension of the moduli space ${\rm dim}_{\C} \,{\mathcal M_{\rm vtx \,}}_{k}^{n,N-n} = k N$ 
obtained through the analysis of the index theorem \cite{Hanany:2003hp} (see Eq.\,\eqref{eq:dim_M}). 
There are $(k+n-1)!/(k! (n-1)!)$ patches and 
the transition functions can be read off from 
the $V$-transformation 
between two different fixed forms 
$(\mD_{\lambda'},\wt \mD_{\lambda'}) = (V \, \mD_{\lambda} , V \, \wt \mD_{\lambda})$ (see the example below).

Next, let us read off the sigma model instanton solutions 
in the large coupling limit from $\xi=(\mD(z),\wt \mD(z))$.  
Let $\varphi$ is the $n$-by-$(N-n)$ matrix 
which appears in the matrix $\mathcal G$ in Eq.\,\eqref{eq:inhomogeneous} as
\begin{eqnarray}
{\cal G} = \mathcal L 
\left( 
\begin{array}{cc} 
{\bf 1}_n & \varphi ~~ \\ {\bf 0} & {\bf 1}_{N-n}
\end{array} \right) \in GL(N,\mathbb C),
\end{eqnarray}
where $\mathcal L$ is a lower-triangular block matrix. 
The components of $\varphi$ can be regarded as the inhomogeneous coordinates of Grassmannian $G(N,n)$.
From the relation 
$(\mD(z),\wt \mD(z)) \sim (\mathbf 1_n, \mathbf 0) \, \mathcal G$, 
the matrix $\varphi(z)$ can be read off as
\begin{eqnarray}
\varphi(z) = \mD(z)^{-1} \widetilde \mD(z).
\end{eqnarray}
This implies that the inhomogeneous coordinate
are rational functions of $z$. 
Note that only semilocal vortices can be observed as 
sigma model instantons. 
To see this, let us assume that 
if $l$ vortices in $(\mD(z),\wt \mD(z))$ are of local type. 
Then, as shown in Appendix \ref{sec:LSdecomposition}, 
the matrices $\mD(z)$ and $\wt \mD(z)$ have a common factor $\mD_{\rm lc}(z)$ representing local vortices 
\begin{align}
(\mD(z),\wt \mD(z))= \mD_{\rm lc}(z)  (\mD_{\rm sm}(z),\wt \mD_{\rm sm}(z)),\quad 
\det \mD_{\rm lc}(z)={\cal O}(z^l),
\end{align}
where $\xi_{\rm sm}(z)=(\mD_{\rm sm}(z),\wt \mD_{\rm sm}(z))$ satisfies the condition \eqref{eq:cond_rank}. 
From this factorized form, we find that the local vortices do not appear in the sigma model instantons in the large coupling limit:
\begin{align}
\varphi(z) = \mD(z)^{-1} \widetilde \mD(z)= \mD_{\rm sm}(z)^{-1} \widetilde \mD_{\rm sm}(z).
\end{align}
%A local vortex can be observed only as a squeezed magnetic flux in the Higgs phase. Since a typical size of a vortex is $1/g_1\sqrt{r_1}$, which is an inverse of the gauge-boson mass,
%it has a singular configuration in the large coupling limit.
%

\paragraph{Abelian case \\}
For $n=1$, the gauge group is $U(1)$ and 
$\mD(z)$ is a monic polynomial of $z$
\begin{eqnarray}
{\mD(z)}= \prod_{\alpha}(z-z_\alpha)^{d_\alpha}, 
\hs{5} \sum_\alpha d_\alpha =k,
\end{eqnarray}
where $d_\alpha$ denotes the multiplicity of the vortex sitting at $z=z_\alpha$. 
The condition $\mD(z)^{-1} \widetilde \mD(z)={\cal O}(z^{-1})$ for $\widetilde \mD(z)=(\wt \mD_2(z), \wt \mD_3(z), \cdots, \wt \mD_N(z))$ can be solved 
by setting $\wt \mD_a(z)$ to be polynomials of degree $k-1$.  
Then the sigma model instanton solution takes the form
\begin{eqnarray}
\frac{\wt \mD_a(z)}{\mD(z)} = \sum_\alpha \sum_{p=1}^{d_\alpha} \frac{c_{a,\alpha,p}}{(z-z_\alpha)^p}, \hs{5} {\rm for~} 2\le a \le N.
\end{eqnarray}
Since the number of the coefficients is $k$ for each flavor $a$, 
this solution has $k N$ moduli parameters. 
For separated vortex configuration ($d_\alpha=1$ for all $\alpha$), 
each vortex has $N-1$ parameters 
in addition to its position modulus. 
For example, for $N=2$ and $d_\alpha=1$, 
the absolute values of the coefficient $c_{2,\alpha,1}$ determines the size of the vortex at $z=z_\alpha$, 
and hence $c_{2,\alpha,1}~(\alpha=1,\cdots,k)$ are called {\it size moduli}. 
The phase of $c_{2,\alpha,1}$ corresponds to
the Nambu-Goldstone mode of the $H=U(1)$ global symmetry that is preserved by the vacuum but broken by the vortex. 
The configurations with $c_{a,\alpha,p=d_{\alpha}}=0$ correspond to a small-instanton singularities. 

\paragraph{Non-Abelian case \\}
In the case of $n>1$, 
vortices possesses another type of moduli parameters that the Abelian vortices do not have. 
Since $\mD(z)$ has a smaller rank at the center of each vortex
($z=z_\alpha$, $\det \mD(z_\alpha)=0$), 
there exists an $n$-column vector $\psi$ satisfying\footnote{A set of this type of relations for all vortices will be summarized in the equation \eqref{eq:hADHM1} in the next section.}
\begin{align}
\mD(z_\alpha) \, \psi_\alpha =0. 
\end{align}
Each vector $\psi_\alpha$, 
defined up to a normalization factor, specifies a point on 
$\mathbb CP^{n-1}=U(n)/ [U(n-1)\times U(1)]$.  
These degrees of freedom correspond 
to the Nambu-Goldstone zero mode 
of the $U(n)$ color-flavor locked symmetry  broken due to existence of a vortex. 
Even when a vortex has a vanishing size modulus, 
this type of {\it orientational moduli} survives and thus describes internal degrees of freedom of the local vortex.

Next, let us demonstrate how to describe the moduli space of vortices with the simplest example of non-Abelian vortices in the case of $N=3, n=2, k=1$ ($n_1=2,n_2=1$).
There are two coordinate patches $(l_1,l_2) = (1,0)$ and $(0,1)$ 
for which the ``gauge-fixed forms" of the matrix $\mD$ in Eq.\,\eqref{eq:Dpatch} are 
respectively given by
\begin{eqnarray}
\mD_{(1,0)}(z)=
\left( \begin{array}{cc} z-a& 0  \\ -b& 1
\end{array} \right), \hs{10} 
\mD_{(0,1)}(z)=
\left( \begin{array}{cc} 1& -\tilde b \\ 0 & z - \tilde a
\end{array} \right). \label{eq:Dk2}
\end{eqnarray}
The conditions for $\mJ_{(1,0)}(z)=(J_1,J_2)^\T$ and 
$\mJ_{(0,1)}(z)=(J'_1, J'_2)^\T$ are given by
\begin{eqnarray}
\mD_{(1,0)}^{-1} \, \mJ_{(1,0)}
=\frac{J_1(z)}{z-a} \left( \begin{array}{c} 1 \\ b \end{array}\right) + J_2(z) \left( \begin{array}{c} 0\\ 1 \end{array}\right)  ={\cal O}(z^{-1}). \\
\mD_{(0,1)}^{-1} \, \mJ _{(0,1)}
=\frac{J'_2(z)}{z-a} \left( \begin{array}{c} \tilde b \\ 1 \end{array}\right) +  J'_1(z) \left( \begin{array}{c} 1 \\ 0 \end{array}\right)  ={\cal O}(z^{-1}).
\end{eqnarray}
By using the solutions to these conditions 
$\mJ_{(1,0)} = (1,0)^T$ and 
$\mJ_{(0,1)} = (0,1)^T$, 
the matrices $\wt \mD_{(1,0)}$ and $\wt \mD_{(0,1)}$ can be written as 
\beq
\wt \mD_{(1,0)} = c \, \mJ_{(1,0)} = \ba{c} c \\ 0 \ea, \hs{5}
\wt \mD_{(0,1)} = \tilde c \, \mJ_{(0,1)} = \ba{c} 0 \\ \tilde c \ea,
\label{eq:tDk2}
\eeq
where $c$ and $\tilde c$ are constants. 
The parameters $(a, b, c)$ and $(\tilde a, \tilde b, \tilde c)$ are 
the coordinates of the moduli space 
in the $(l_1,l_2)$-patches with $(l_1,l_2)=(1,0)$ and $(l_1,l_2)=(0,1)$, respectively. 
They are related by the coordinate transformation
\beq
(\tilde a, \tilde b, \tilde c) = ( a, b^{-1}, c b), 
\label{eq:coordinate_transformation}
\eeq
which can be determined from 
the regularity condition for the $V$-transformation
between the two gauge-fixed forms of $\mD$
\begin{eqnarray}
V \, \mD_{(1,0)}=\mD_{(0,1)} ~~ \Longrightarrow ~~
V(z)= \mD_{(0,1)} (\mD_{(1,0)})^{-1} = \left( \begin{array}{cc} \frac{1-b \tilde b}{z-a} & -\tilde b  \\ b \frac{z-\tilde a}{z-a} & z-\tilde a
\end{array} \right),
\end{eqnarray}
and the relation $V \, \widetilde \mD_{(1,0)} = \widetilde \mD_{(0,1)}$.
The parameter $a$ and $c$ are 
the position and size moduli of the vortex, respectively. 
The parameter $b$ is the inhomogeneous coordinate of 
the orientational moduli $\mathbb CP^1$.
%\footnote{
%In this case, the whole moduli space is 
%given by a weighted projective space 
%$\mathbb CP^3(0,1,1,-1)$, where
%\begin{eqnarray}
%(a, 1, b, c) \sim (a, \lambda, \lambda b, \lambda^{-1}c)= (a, \tilde b, 1, \tilde c),\quad {\rm with~} \lambda =\tilde b =b^{-1}.
%\notag
%\end{eqnarray}
%}
The sigma model solution 
$\varphi_{12}(z) = \mD^{-1} \widetilde \mD$ 
in the $(1,0)$- and $(0,1)$-patches take the forms
\begin{equation}
\varphi_{12}^{(1,0)}(z) =\frac{c}{z-a}\left( \begin{array}{c} 1 \\ b \end{array}\right), \hs{10}
\varphi_{12}^{(0,1)}(z) =\frac{\tilde c}{z-a}\left( \begin{array}{c}  \tilde b \\ 1 \end{array}\right).
\label{eq:instanton_L=1}
\end{equation}
By using the coordinate transformation \eqref{eq:coordinate_transformation}, 
one can confirm that these are identical  
$\varphi_{12}^{(1,0)}(z) = \varphi_{12}^{(0,1)}(z)$ 
on the overlap of the coordinate patches.

\paragraph{Duality and moduli spaces \\}
In the case of $N=3, n=2, k=1$ ($n_1=2,n_2=1$) discussed above, the total vortex moduli space is given by
\begin{align}
{\cal M}_{\rm vtx \,}{\,}_{k=1}^{n_1=2,n_2=1} = \mathbb C \times  \mathcal O_{\C P^1}(-1),
\end{align}
where the first factor corresponds to the vortex position 
$a \in \C$ and 
the second factor is the total space of the line bundle $\mathcal O_{\C P^1}(-1)$ over $\C P^1$ 
\begin{align}
\mathcal O_{\C P^1}(-1) = \left \{ (b_1,b_2,c) \, \big| \, (b_1,b_2)\in \mathbb C^2 \backslash \{(0,0)\} , \, c \in \mathbb C \right\}/ \sim \hs{3} \mbox{with} \hs{3} 
(b_1,b_2, c) \sim ( \lambda b_1,\lambda b_2,\lambda^{-1} c ).
\end{align}
By removing the subspace $c=0$ 
corresponding to the small-instanton singularity, 
this space reduces to the moduli space of sigma model instanton
\begin{align}
{\cal M}_{\rm inst}{\,}_{k=1}^{n_1=2,n_2=1} = \mathbb C\times \left \{(b_1,b_2,c) \, \big| \, (b_1,b_2) \in \mathbb C^2 \backslash \{(0,0)\},~c \in \mathbb C \backslash \{0\} \right\} / \sim \ = \mathbb C \times \left( \mathbb C^2 \backslash \{(0,0)\} \right).
\end{align}
In the dual theory, 
which is an Abelian theory with $N=3,n=1$ ($n_1=1,n_2=2$), 
the vortex and instanton moduli spaces are respectively given by
\begin{align}
{\cal M}_{\rm vtx \,}{\,}^{n_1=1,n_2=2}_{k=1} = \mathbb C \times \mathbb C^2, \qquad 
{\cal M}_{\rm inst}{\,}^{n_1=1,n_2=2}_{k=1} = \mathbb C \times \big( \mathbb C^2 \backslash \{(0,0)\} \big). 
\end{align}
We can confirm that 
even though the instanton moduli spaces are identical, 
the vortex moduli spaces are different in the dual theories 
\beq
\begin{matrix}
{\cal M}_{{\rm inst}}{\,}^{n_1=2,n_2=1}_{k=1} & = & {\cal M}_{{\rm inst}}{\,}^{n_1=1,n_2=2}_{k=1} \\ \cap & & \cap \\  {\cal M}_{{\rm vtx}}{\,}^{n_1=2,n_2=1}_{k=1}& \not = & 
{\cal M}_{{\rm vtx}}{\,}^{n_1=1,n_2=2}_{k=1}.
\end{matrix}
\eeq
In each theory, the small-instanton singularity is regularized by
replacing it with the local vortex moduli space in each theory
\begin{align}
  {\cal M}_{\rm vtx \,}{\,}^{n_1=2,n_2=1}_{k=1} \backslash 
  {\cal M}_{\rm inst}{\,}^{n_1=2,n_2=1}_{k=1} 
  = \mathbb C  \times \mathbb CP^1, \qquad 
 {\cal M}_{\rm vtx \,}{\,}^{n_1=1,n_2=2}_{k=1} \backslash 
 {\cal M}_{\rm inst}{\,}^{n_1=1,n_2=2}_{k=1} =\mathbb C. 
\end{align}
In other words, 
the singularity is blown up in the case of $n_1=1,\,n_2=2$ whereas 
$\C \times \{(0,0)\}$ is added
along the singularity in the case of $n_1=2,\,n_2=1$. 

%%%%%%%%%%%%%%%%%%%%%%%%%%%%%%%%%%%%%%%%%%%%%%%%%%%%%%%%%%%%%%
\subsection%[Example 2: L=2 case]
{Example 2: $L=2$ case}
\label{subsec:example_L=2}
Next, we consider the case with $L=2$. 
As we have seen in the previous example with $L=1$, 
all the information on the moduli space of vortices is 
contained in the matrix $\mH = (\mD, \widetilde \mD)$
obeying the constraints $\det \mD(z) = \mathcal O(z^k)$ and $\mD(z)^{-1} \wt \mD(z) = \mathcal O(z^{-1})$. 
For $L>1$, the matrices $\mH_i = (\mD_i, \widetilde \mD_i)$
must satisfy additional constraints 
since they are composite quantities 
obtained from $q_i~(L=1,\cdots,N)$.
For example, in the case of $L=2$, $\mH_1$ and $\mH_2$ 
are related as $\mH_1(z)=q_1(z) \mH_2(z)$ 
and hence they are not independent. 
To read off the information on the vortex moduli space, 
we need to clarify the constraints for $\mH_i(z)$. 

For simplicity, let us focus on the case of 
vortex numbers $(k_1,k_2)=(1,1)$
in the model with $n_1=n_2=n_3=1$   
as the simplest nontrivial example\footnote{
Another simple example is the case where
the matrices $\mH_i$ can be obtained 
by embedding that of the $L=1$ case. 
A general discussion on the embedding of the $L=1$ case is given in Appendix \ref{appendix:embedding}.}. 
%The asymptotic behavior of the magnetic fluxes \eqref{eq:flux_asymptotic}
%implies that the numbers of vortices $(k_1,k_2)$ can be 
%read off from the asymptotic behavior of $\det (\mH_i \mH_i^\dagger)$
%\beq
%\det (\mH_i \mH_i^\dagger) \rightarrow |z|^{2k_i}. 
%\eeq
It follows from the conditions 
$\det \mD_i = \mathcal O(z^{k_i})$ and $\mD_i^{-1} \wt \mD_i = \mathcal O(z^{-1})$ that 
the matrices $\mH_1(z)$ and $\mH_2(z)$ 
for $(k_1,k_2)=(1,1)$ take the forms
\beq
\mH_1(z)=(z-a', c', c''), \hs{10} 
\mH_2(z)=(\mD_{(1,0)}(z),\widetilde \mD_{(1,0)}(z)) = \ba{ccc} z-a & 0 & c \\ -b & 1 & 0 \ea,
\label{eq:L=2_(1,0)}
\eeq
where we have fixed $\mH_2(z)=(\mD_2,\widetilde \mD_2)$ 
so that it takes the $(l_1,l_2)=(1,0)$ form
given in Eq.\,(\ref{eq:Dk2}) and Eq.\,(\ref{eq:tDk2}). 
These matrices must be related as 
$\mH_1(z)=q_1(z) \mH_2(z)$ 
with a certain non-singular 1-by-2 matrix $q_1(z)$. 
The regularity of $q_1(z)$ requires that the parameters are related as
\begin{eqnarray}
a' = a+b c' \hs{10}
c''=c.
\label{eq:para_relation1}
\end{eqnarray}
In the $(0,1)$-patch, $\mH_1(z)$ and $\mH_2(z)=(\mD_{(0,1)}(z),\widetilde \mD_{(0,1)}(z))$ take the forms
\beq
\mH_1(z)=(z-a', c', c''), \hs{10}
\mH_2(z)=(\mD_{(0,1)}(z),\widetilde \mD_{(0,1)}(z)) = 
\ba{ccc} 1 & - \tilde b & 0 \\ 0 & z-\tilde a & \tilde c \ea,
\label{eq:L=2_(0,1)}
\eeq
where the parameters $(\tilde a, \tilde b, \tilde c)$ are 
related to $(a,b,c)$ in the same way as the case of $L=1$ 
\eqref{eq:coordinate_transformation}. 
The regularity of $q_1(z)$ requires that
\begin{eqnarray}
\quad c'=\tilde b\,  (a'-\tilde a), \quad c''=\tilde b\, \tilde c.
\label{eq:coordinate_transformation2}
\end{eqnarray}
This relation is consistent with \eqref{eq:para_relation1} and \eqref{eq:coordinate_transformation}. 
The constraints on  \eqref{eq:L=2_(1,0)} and \eqref{eq:L=2_(0,1)} with \eqref{eq:coordinate_transformation} implies that the moduli space is given by
\beq
{{\cal M}_{\rm vtx\,}}^{n_1=1,n_2=1,n_3=1}_{k_1=1,k_2=1} = \C \times (\mathcal O_{\C P^1}(-1) \oplus \mathcal O_{\C P^1}(-1)),
\eeq 
where $\C$ is the center of mass position parametrized by $a+a'$ and $\mathcal O_{\C P^1}(-1) \oplus \mathcal O_{\C P^1}(-1)$ is the space given by
\begin{align}
\mathcal O_{\C P^1}(-1) \oplus \mathcal O_{\C P^1}(-1) = \left \{ (b_1,b_2,c_1,c_2) \, \big| \, (b_1,b_2) \in \mathbb C^2 \backslash \{(0,0)\} , \, (c_1,c_2) \in \mathbb C^2 \right\}/ \sim,
\end{align}
where the equivalence relation $\sim$ is defined as 
\beq
(b_1,b_2, c_1,c_2) \sim ( \lambda b_1,\lambda b_2,\lambda^{-1} c_1,\lambda^{-1} c_2 ) ~~~ \mbox{with} ~~~ \lambda \in \C^\ast.
\eeq
The coordinates in the $(1,0)$ and $(0,1)$ patches are related to the $\C^\ast$ invariants as
\begin{alignat}{3}
b &= \frac{b_1}{b_2}, \hs{5}& c &= b_2 c_2, \hs{3} & c' &= b_2 c_1, \\
\tilde b &= \frac{b_2}{b_1}, \hs{5}& \tilde c &= b_1 c_2, \hs{3}& a'-a &= b_1c_1. 
\end{alignat}
Here, a complex parameter $b (\tilde b)$ parametrizing   
$\mathbb CP^1$ appears like in the non-Abelian case given in Eq.\eqref{eq:Dk2}.
Its argument, $\arg (b)$, is again the Nambu-Goldstone zero mode due to a broken $U(1)$-symmetry,
but unlike in the non-Abelian case, 
its absolute value, $|b|$ is no longer a Nambu-Goldstone mode. We can observe that 
the configurations of the magnetic fluxes depend on $|b|$, as seen in Fig.\ref{fig:fluxes}.
\begin{figure}[ht]
 \centering
{\includegraphics[keepaspectratio,scale=0.9]
      {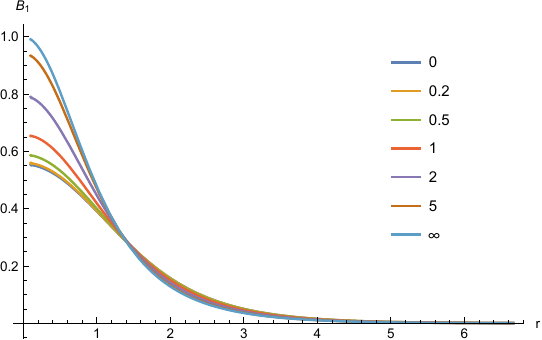}}\quad
{\includegraphics[keepaspectratio,scale=0.9]
      {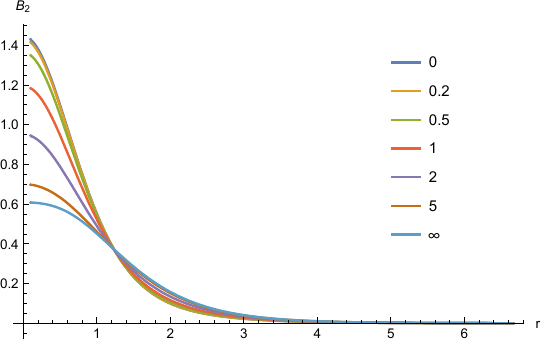}}
 \caption{Profiles of the magnetic fields, $B_{i=1,2}\equiv -\tr[F_{12}^{i=1,2}]$ for local vortices in the $L=2$ model with setting $r_1=r_2=g_1=g_2=1$. 
 Here we use the moduli matrix in Eq.\eqref{eq:L=2_(1,0)} with setting $a=c=c'=0$ and
varying values of $b$ as $b=0,0.2,0.5,1,2,5,\infty$.}
 \label{fig:fluxes}
\end{figure}

Next, let us consider the corresponding sigma model 
instanton solution. 
Let us introduce inhomogeneous complex coordinates 
$(\phi_{12},\phi_{13},\phi_{23})$ 
of the target manifold $G^{\mathbb C}/\hat H = \mathcal F_{1,1,1}$. 
They are contained in the matrix $\mathcal G$ 
in Eq.\,\eqref{eq:inhomogeneous} as
\begin{eqnarray}
{\cal G}= \mathcal L \left( 
\begin{array}{ccc} 
1 & \phi_{12} &\phi_{13} \\ 
0 & 1 & \phi_{23}  \\ 
0 & 0 & 1 
\end{array} \right) \in G^{\mathbb C},
\end{eqnarray}
where $\mathcal L$ is a lower-triangular matrix. 
In this parametrization, 
$\mD_i^{-1} \widetilde \mD_i$ are given by
\begin{eqnarray}
\mD_1^{-1} \widetilde \mD_1 = (\phi_{12},\phi_{13}), \hs{10}
\mD_2^{-1} \widetilde \mD_2 = \left( \begin{array}{c} \phi_{13}' \\ \phi_{23} \end{array} \right),  \label{eq:DD2phi}
\end{eqnarray}
where $\phi_{13}'$ is related to $\phi_{13}$ as\footnote{
The flag manifold ${\cal F}_{n_1,n_2,n_3}$ can be realized as two orthogonal flags, $\xi_1 \in G(n_1,N)$ and $\tilde\xi\in G(n_3,N)$. Eq.\eqref{eq:phirel} can be regard as the orthogonality condition $\xi_1\tilde \xi^{\rm T} =0$ with identifying
$\xi_1 = ({\bf 1},\phi_{12},\phi_{13})$ and $\tilde\xi=(-{\phi_{13}'}^{\rm T},-\phi_{23}^{\rm T},{\bf 1})$.
}
\begin{eqnarray}
\phi_{13}'=\phi_{13}-\phi_{12}\phi_{23}.  \label{eq:phirel}
\end{eqnarray}
If we ignore the constraints $\xi_i = q_i \xi_{i+1}$, 
each $\xi_i$ can be regarded as 
a flag for the Grasmannian $G(N_i,N)$.
In this case, $\xi_1$ and $\xi_2$ specify points on 
$G(1,3) = \C P^2$ and $G(2,3) = \C P^2$, respectively. 
In this case, both $(\phi_{12},\phi_{13})$ and $(\phi_{23},\phi_{13}')$ are the inhomogenrous coordinates of 
$\mathbb CP^2$, and hence the sigma model instanton solutions with $k_1=1$ and $k_2=1$ are generally takes the forms (see \eqref{eq:instanton_L=1})
\begin{align}
 (\phi_{12},\phi_{13})=\frac1{z-a'}(B,A'),\quad
 (\phi_{23},\phi_{13}')=\frac1{z-a}(C,A). \label{eq:inst11solution}
\end{align}
The additional condition \eqref{eq:phirel} 
gives rise to the following constaraint on the moduli parameters $a',A',B,a,A,C$
\begin{align}
AD=BC, \quad  A=A', \quad {\rm with} \quad D \equiv a'-a.    
\end{align}
The space given by the condition $AD=BC$ has singularities 
but they are covered by the small-instanton singularities located at $(B,A')=(0,0)$ and $(C,A)=(0,0).$\footnote{
There are two type of singularity:  the former is algebraic singularity where the tangent space is ill-defined, and the latter is ``physical" singularity where the NL$\sigma$M breaks down.}
These singularities can be simultaneously 
removed by requiring $A \not =0$, 
and hence the moduli space of instantons is given by 
\begin{align}
{{\cal M}_{\rm inst\,}}^{n_1=1,n_2=1,n_3=1}_{k_1=1,k_2=1} 
= \mathbb C \times \mathbb C^2 \times (\mathbb C\backslash \{0\}).
\end{align}
The vortex configuration given in \eqref{eq:L=2_(1,0)} and \eqref{eq:L=2_(0,1)} is mapped to 
$(\phi_{12},\phi_{23},\phi_{13},\phi_{13}')$ 
through relation \eqref{eq:DD2phi}, 
from which the parameters 
in \eqref{eq:inst11solution} can be read as, 
%This singularity can be resolved by considering browing up along a center $A=C=0$, as 
\begin{align}
(B,A')=b_2(c_1,c_2),\quad (C,A)=c_2(b_1,b_2) \quad {\rm and~} a'-a=b_1 c_1,
\end{align} 
with $[b_1:b_2:c_1:c_2] \in \mathcal O_{\C P^1}(-1) \oplus \mathcal O_{\C P^1}(-1)$.
The above mapping can be regarded 
as a blowup of the space $AD=BC$ along the center $A=C=0$.
Therefore, this resolution of the singularity defines an inclusion map between them as
\begin{align}
 {{\cal M}_{\rm inst\,}}^{n_1=1,n_2=1,n_3=1}_{k_1=1,k_2=1} = \C \times \C^2 \times (\C \backslash \{(0)\})
~ \hookrightarrow ~ {{\cal M}_{\rm vtx\,}}^{n_1=1,n_2=1,n_3=1}_{k_1=1,k_2=1} = \C \times (\mathcal O_{\C P^1}(-1) \oplus \mathcal O_{\C P^1}(-1)).
\end{align}
Removing the small instanton 
singularities at $b_2=0$ and $c_2=0$, 
we obtain the condition for 
the moduli space of the instanton 
\begin{align}
{{\cal M}_{\rm vtx\,}}^{n_1=1,n_2=1,n_3=1}_{k_1=1,k_2=1} ~~ \underset{b_2\not=0,~c_2 \not =0}{\longrightarrow} ~~ {{\cal M}_{\rm inst\,}}^{n_1=1,n_2=1,n_3=1}_{k_1=1,k_2=1}= \mathbb C \times \mathbb C^2 \times (\mathbb C\backslash \{0\}).
\end{align}

Next, let us see how the moduli spaces are 
related under the duality map. 
In the $n_1=n_2=n_3=1$ case, 
the duality theory is identical 
but the inhomogeneous coordinates $(\phi_{12},\phi_{23},\phi_{13})$ and $\phi_{13}'=\phi_{13}-\phi_{12} \phi_{23}$ are swapped as (see Eq.\,\eqref{eq:dual_U})
\begin{align}
(\phi_{12}, \phi_{23}, \phi_{13}, \phi_{13}')_{\rm dual} 
= - (\phi_{23}, \phi_{12},\phi_{13}', \phi_{13}). 
\end{align}
From this relation, 
we can read off 
the duality transformation 
for the moduli parameters as
\begin{align}
(b_1, b_2, c_1, c_2)_{\rm dual} =(c_1,c_2,-b_1,-b_2).
\end{align}
This map is well-defined on 
${{\cal M}_{\rm inst\,}}^{n_1=1,n_2=1,n_3=1}_{k_1=1,k_2=1}$ 
but ill-defined on 
${{\cal M}_{\rm vtx\,}}^{n_1=1,n_2=1,n_3=1}_{k_1=1,k_2=1}$
since the point $(c_1,c_2)_{\rm dual}=(0,0)$ of the vortex moduli space is mapped to the forbidden point $(b_1,b_2)=(0,0)$ in the original theory. This indicates that there are vortex configurations that do not have corresponding configuration in the dual theory. 

This asymmetry of the vortex moduli space becomes manifest if we focus on its subspaces containing local vortices. 
The vortex described by $\xi_1 \, (\xi_2)$ 
becomes a local vortex in the limit of $b_2 \to 0~(c_2\to 0)$
 Interestingly, there exist two subspaces ($\mathbb C^2$ and $\mathbb C \times \C P^1$ shown in the bottom row of Fig.\,\ref{fig:moduli_space}) where both of the two vortices becomes local ones.
%\begin{alignat}{4}
%& & &\hs{-15} {{\cal M}_{\rm vtx\,}}^{n_1=1,n_2=1,n_3=1}_{k_1=1,k_2=1} & & & & \\ & & & \rotatebox{90}{=} & & & & \nn & & \mathbb C \times \mathbb (\mathcal O(-1) & \oplus \mathcal O(-1)) & & & \nn b_2 \to 0 & \swarrow & & & \searrow & \,c_2 \to 0 \nn \mathbb C^3 & & & & & \mathbb C \times \mathcal O(-1) \nn c_2\to 0 ~& \rotatebox{135}{$\subset$} & & & \swarrow & \, b_2 \to 0 & \searrow \, c_1 \to 0~(a'\to a) \nn & & \mathbb C^2 \ni ( & a,a') & & & \hs{-3} \mathbb C \times \mathbb CP^1~~~ \notag
%\end{alignat}
\begin{figure}[ht]
 \centering
\includegraphics[width=120mm]{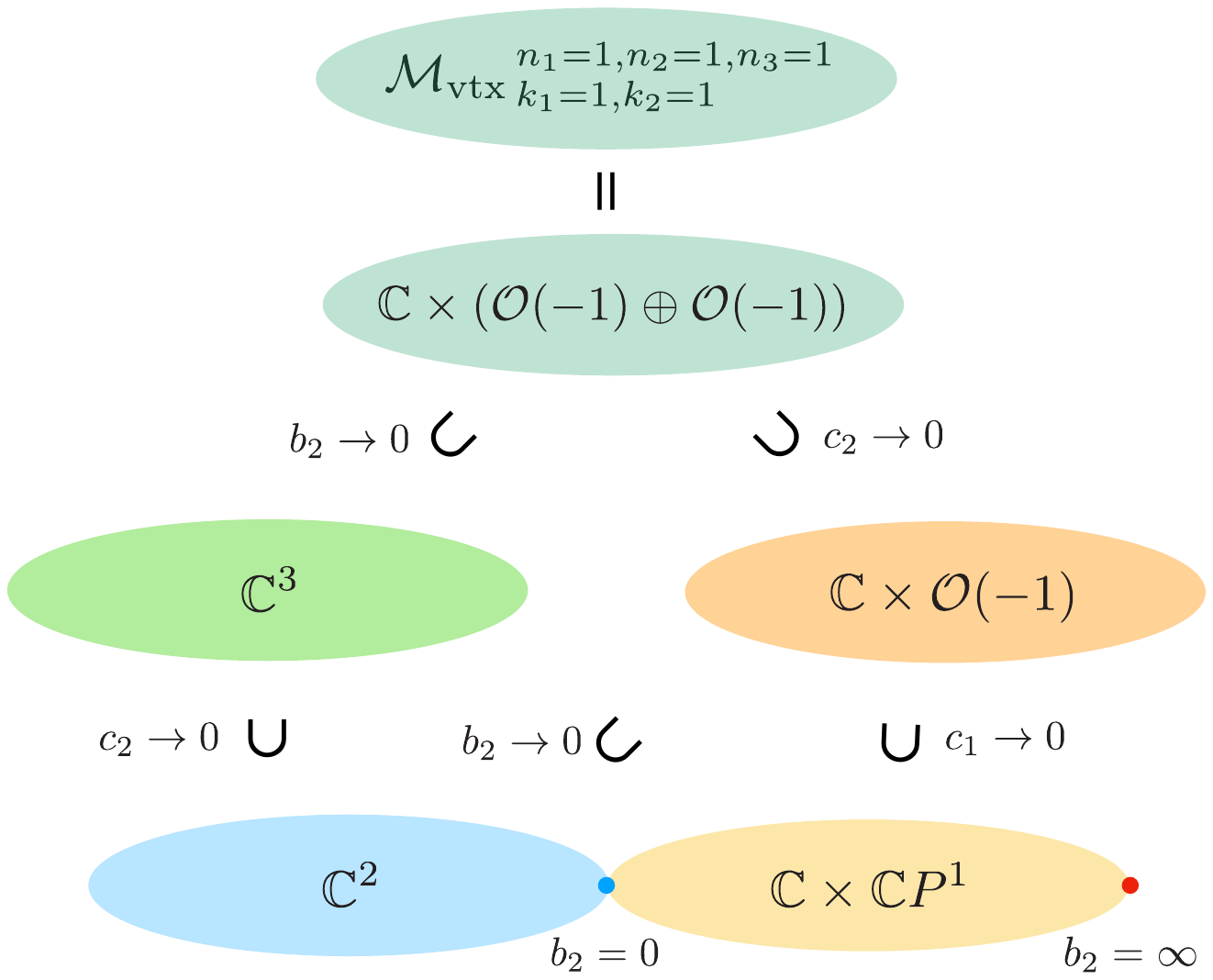}
\caption{The structure of moduli space of vortices with $n_1=n_2=n_3=1$, $k_1=k_2=1$.}
\label{fig:moduli_space}
\end{figure}
%
%The limit $b=0$ means  coincidence of vortices $a'=a$ 
%\begin{eqnarray}
%\phi_{12}(z)=\frac{c'}{z-a},\quad \phi_{13}(z)=\frac{c}{z-a},\quad \phi_{23}(z)=0.
%\end{eqnarray}
%In the limit of $\tilde b=0$,  we find 
%\begin{eqnarray}
%\phi_{12}(z)=0,\quad \phi_{13}(z)=0,\quad \phi_{23}(z)=\frac{\tilde c}{z-a}.
%\end{eqnarray}
%which means a small lamp singularity.

\paragraph{Fixed points of torus action \\}
The two subspaces $\mathbb C^2$ and $\mathbb C \times \C P^1$ shown in the bottom row of Fig.\,\ref{fig:moduli_space} 
contain two fixed points of a torus action which will be discussed in Appendix \ref{sec:TorusActions}. 
They are also viewed as the origins of the 
It is convenient to characterize these 
fixed points by Young tableaux as
\begin{align}
\bigg( \, \lower3mm\hbox{\yng(1,1)} \,, {\bf 1} \bigg) : & \quad q_1=(1,0),  \quad q_2 =
    \left(\begin{array}{ccc}
         z& 0&0 \\
         0& 1 &0
    \end{array} \right) \in C_2, \\
    \Big( \lower1mm\hbox{\yng(1)}\,,\lower1mm\hbox{\yng(1)} \Big):&\quad
    q_1=(z,0), \quad 
    q_2= \left(\begin{array}{ccc}
         1& 0&0 \\
         0& z &0
    \end{array} \right) \in C_1 \cap C_2.
\end{align}
The height $d$ of each young diagram indicates  
a composite state of $d$ different types of vortices. 
For $n_1=n_2=n_3=1$, the general fixed point and the corresponding set of Young diagrams is given by 
\begin{align}
  \bigg( \, \lower3mm\hbox{\young({1}{{\scriptstyle{\cdots}}}{l},{}{}{}{1}{{\scriptstyle{\cdots}}}{m})} \, ,\, \lower3mm \hbox{\young({1}{{\scriptstyle{\cdots}}}{n})} \,
\bigg):
\quad
    q_1=(z^{m},0), \quad 
    q_2= \left(\begin{array}{ccc}
         z^l& 0&0 \\
         0& z^n &0
    \end{array} \right),
\end{align}
with $(k_1,k_2)=(l+m,\,l+n)$.
In Sec.\,\ref{subsec:torus_action}, we will see the way to classify the fixed points in terms Young tableaux.

\section{Quotient construction of vortex moduli space}\label{sec:kahlerquotient}
In this section, we discuss a quotient construction of the vortex moduli space. We show that the vortex moduli space \eqref{eq:vtx_moduli} can be identified with a quotient of a vector space of matrices by a complex Lie group. 

\subsection%[L=1 case]
{$L=1$ case and half-ADHM mapping relation} 
Let us first review 
the quotient construction in the $L=1$ case \cite{Eto:2005yh,Eto:2006cx}. 
We show that the vortex moduli space \eqref{eq:moduli_v_L=1}
is given by the $GL(N,\C)$ quotient of the vector space of matrices $\{ Z, \Psi,\wt \Psi\}$ 
\begin{eqnarray}
{\cal M}_{\rm vtx \,}{\,}^{n, N-n}_{k} \, \cong \, \left\{ (Z,\Psi,\wt \Psi) \, \Big| \, \{Z,\Psi\} \ \mbox{on which $GL(k,\mathbb C)$ action is free} \right\}/GL(k,\mathbb C),
\end{eqnarray}
where $Z$ $k$-by-$k$ matrix, $\Psi$ is a $n$-by-$k$ matrix and $\wt \Psi$ is a $k$-by-$(N-n)$ matrix on which $GL(k,\C)$ acts
\beq
Z \rightarrow g^{-1} Z g, \hs{5} 
\Psi \rightarrow \Psi g, \hs{5}
\wt \Psi \rightarrow g^{-1} \wt \Psi,
\eeq
They are related to the moduli matrix 
$\mH(z)=q(z)=(\mD,\widetilde \mD)$ through the relations, which we call 
{\it the half-ADHM mapping relation}
\begin{eqnarray}
\mD(z)\Psi =\mJ(z)(z{\bf 1}_k -Z),
\hs{10}
\wt \mD(z) = \mJ(z) \wt \Psi,
\label{eq:hADHM}
\end{eqnarray} 
where $\mJ(z)=(j_1(z),\cdots,j_k(z))$ is the $n$-by-$k$ matrix defined in section \ref{subsec:ExampleMM_L=1}, which is characterized by the property 
\beq
\mD(z)^{-1} \mJ(z) = \mathcal O(z^{-1}) ~~~(z \rightarrow \infty).
\label{eq:J_property}
\eeq
The $GL(k,\C)$ transformation acts on $\mJ(z)$ as
\beq
\mJ(z) \rightarrow \mJ(z) \, g ~~~~ (g \in GL(k,\C)), 
\eeq
and can be regarded as the change of basis $\{j_1(z),\cdots,j_k(z) \}$ of the solutions of \eqref{eq:J_property}.

From the matrices $\mD$ and $\mJ$, 
the matrices $(Z,\Psi)$ can be obtained 
through the first equation in \eqref{eq:hADHM}. 
The existence of such constant matrices $(Z,\Psi)$ 
can be shown by using the following decomposition algorithm. 
By using $\mD$ and $\mJ$, 
any column vector $\vec f(z)$ 
with arbitrary polynomial entries can be decomposed as
 \begin{eqnarray}
%\forall \vec f(z),\, \exists \vec g(z),\, \exists \vec v: \quad 
\vec f(z) = \mD(z) \vec g(z) + \mJ(z) \boldsymbol{v}, \label{eq:PolynomialDiv}
\end{eqnarray}
with a column vector $\vec g(z)$ 
with polynomial entries and 
a constant vector $\boldsymbol{v} \in \C^k$. 
Note that for a given $\vec f(z)$, 
the column vectors $\vec g(z)$ and $\boldsymbol{v}$ are unique 
since the columns of $\mD(z)$ and $\mJ(z)$ are independent
in the sense that
\begin{eqnarray}
\vec 0 =\mD(z) \vec g(z)+ \mJ(z) \boldsymbol{v} 
\quad \Leftrightarrow \quad 
\vec g(z) =\vec 0, \quad 
\boldsymbol{v} = \boldsymbol{0}. 
\end{eqnarray}

Applying the decomposition \eqref{eq:PolynomialDiv}
to each column of $z \mJ(z)$, 
we can show that there exist 
a $n$-by-$k$ matrix $\Psi$ and 
a $k$-by-$k$ matrix $Z$ such that
\begin{eqnarray}
 z\mJ(z) = \mD(z) \Psi+\mJ(z) Z, \qquad \hbox{or equivalently } \qquad \mD(z) \Psi = \mJ(z)(z{\bf 1}_k -Z).
 \label{eq:hADHM1}
\end{eqnarray}
Note that $\Psi = \mD(z)^{-1} \mJ(z)(z{\bf 1}_k -Z)$ must be a constant matrix since it is regular everywhere and $\mD(z)^{-1} (z \mJ(z)) = \mathcal O(1)$ in the limit $z \rightarrow \infty$. 
Similarly, by applying the decomposition \eqref{eq:PolynomialDiv}
to each column of $\wt \mD(z)$, 
we obtain the $k$ by $(N-n)$ matrix $\wt \Psi$ as 
\begin{eqnarray}
\wt \mD(z) = \mJ(z) \wt \Psi.
\label{eq:hADHM2}
\end{eqnarray} 
Note that $\wt \mD(z)$ has no term proportional to $\mD(z)$ 
since $\wt \mD(z)$ satisfies the condition $\mD(z)^{-1}\wt \mD(z)={\cal O}(z^{-1})$. 

As we have seen, the set of matrices $\{Z, \Psi, \wt \Psi\}$ 
can be extracted from the moduli matrix 
$\mH(z) = (\mD(z), \wt \mD(z))$. 
However, $\{Z, \Psi, \wt \Psi\}$ is not unique 
for a given $\mH(z) = (\mD(z), \wt \mD(z))$ 
due to the degrees of freedom 
of the change of the basis $\{ \vec j_a(z)\}$. 
Thus, for a given $\mH(z) = (\mD(z), \wt \mD(z))$, 
we obtain a unique equivalence class of matrices defined by
\begin{eqnarray}
\{Z,\Psi,\wt \Psi\} \sim \{ g^{-1} Z g \, , \, \Psi g \, , \,  g^{-1} \wt \Psi\} \quad {\rm with}~~ g \in GL(k,\mathbb C). \label{eq:GLk}
\end{eqnarray}
We can show that 
this $GL(k,\mathbb C)$ action is free 
on the part of the data $\{Z, \Psi \}$ obtained from $\mD(z)$;
that is, for any infinitesimal $GL(k,\C)$ action 
$\delta_X Z = [Z,X]$, $\delta_X \Psi = \Psi X$ with $X \in {\mathfrak gl}(k,\C)$, 
\beq
\delta_X Z=0, ~~~ \delta_X \Psi = 0 ~~~ \Rightarrow ~~~ X=0.
\eeq
As shown in Appendix \ref{sec:ADHMdata}, this condition on the data is equivalent to
the following statement for a vector $\vec v$:
\beq
\forall z\in \mathbb C,\quad \Psi (z{\bf 1}-Z)^{-1} \vec v=0\quad \Rightarrow \quad \vec v=0.\label{eq:free_cond}
\eeq
Since the relation \eqref{eq:hADHM1} is rewritten as
\beq
\Psi(z \mathbf 1 - Z)^{-1} = \mD(z)^{-1} \mJ(z),
\eeq
the above $GL(k,\mathbb C)$-free condition is 
satisfied when the $k$ columns of 
$\mJ(z)$ are linearly independent. 
Since this is true by construction, 
the infinitesimal action of $GL(k,\C)$ 
on $\{Z,\Psi\}$ is free.

Through the half-ADHM mapping relations  (\ref{eq:hADHM}), 
we can show that there exists a one-to-one map 
(see Appendix \ref{sec:ZPsiPatches})
between the moduli matrix $(\mD(z), \wt \mD(z))$ 
and the half-ADHM data $\{Z, \Psi, \wt \Psi \}$ 
in each coordinate patch given in Eq.\,(\ref{eq:patches})\footnote{
We can confirm that the number of the degrees of freedom 
of the equivalence class \eqref{eq:GLk} coincides with that of the moduli matrices $\# T^a{}_{b,m} =\# Z+\# \Psi- \# g=kn$.
See Appendix \ref{sec:ZPsiPatches} for more details.}.
Thus, we find that the moduli space of 
BPS vortices turnes out to be given by 
\begin{eqnarray}
{\cal M}_{\rm vtx \,}{\,}_k^{n,N-n} \hs{-3} &\cong& \hs{-2} \left\{ (\mD(z), \wt \mD(z) ) \, \Big| \det \mD(z) = \mathcal O(z^k)\,,\, \mD(z)^{-1} \wt \mD(z) = \mathcal O(z^{-1}) ) \right\} / \{\mbox{$V$-transf. in Eq.\,(\ref{eq:Vtr})}\} \\
\hs{-3} &\cong& \hs{-2} \left\{ (Z,\Psi,\wt \Psi) \, \Big| \, \{Z,\Psi\} \ \mbox{on which $GL(k,\mathbb C)$ action is free} \right\}/GL(k,\mathbb C).
\end{eqnarray}

Indeed, we can show that the matrices $\{Z, \Psi\}$ have all the $V(z)$ invariant information contained in $\mD(z)$ 
from the fact that 
all the invariants under the $V$-transformation 
consisting of $\mD(z)$ and $\mJ(z)$ can be
expressed in terms of $\{Z, \Psi\}$ as\footnote{Although any minor determinants of the matrix $(\mD(z), \mJ(z))$ are invariants, they are related to $\det \mD(z)$ and $\mD(z)^{-1}\mJ(z)$ through the Pl\"ucker relations. }
\begin{eqnarray}
\det \mD(z)=\det (z{\bf 1}_k-Z), \quad {\rm and~} \quad  \mD(z)^{-1}\mJ(z) = \Psi (z{\bf 1}_k-Z)^{-1}.
\end{eqnarray}    
% $\mD(z)^{-1}$ can be rewritten to
%\begin{eqnarray}
%\mD(z)^{-1}=P^\mD+\mD(z)^{-1}\mJ(z) P^\mJ =P^\mD+\Psi (z{\bf 1}-Z)^{-1} P^\mJ 
%\end{eqnarray}
Note that the second relation obeys from (\ref{eq:hADHM})
and the first one can be derived as follows. 
By applying the decomposition (\ref{eq:PolynomialDiv}) 
to each column of the unit matrix, we obtain
\begin{eqnarray}
{\bf 1}_n = \mD(z) P^\mD+\mJ(z) P^\mJ, \label{eq:PP}
\end{eqnarray}  
where $P^\mD$ and $P^\mJ$ are $n$-by-$n$ and $k$-by-$n$ constant matrices, respectively.
Since this equation is not invariant 
under the $V$-transformation, 
$P^\mJ$ and $P^\mD$ depend on the choice of the coordinate patch.\footnote{ 
For the $(l_1,l_2,\cdots, l_n)$ patch  in Eq.(\ref{eq:Dpatch})  and 
Eq.(\ref{eq:Jpatch}),  
 these matrices can be easily obtained as
 \begin{eqnarray}
(P^\mD)^a{}_b=\delta^a_b\, \delta^{l_a}_0,\quad (P^\mJ)^{(a,p)}{}_b=\delta^a_b (1-\delta^{l_a}_0) \delta^p_{l_a},
\end{eqnarray}
% since $\mD(z)$ and $\mJ(z)$ have the following column vectors
%\begin{eqnarray}
%(\mD(z))^a{}_b=\delta^a{}_{b} \quad {\rm for~} b \in \{b| 1\le b\le n,  l_b=0\},\quad 
%(\vec j (z)))^a= \delta^a_{b}  \quad {\rm for~}  p=l_b {\rm ~and~}  b\in  \{b| 1\le b\le n,  l_b\not=0\}
%\end{eqnarray}
%and a set of them forms an unit matrix of order $n$.
%Therefore  Eq.(\ref{eq:PP})  indicates that  a certain set of $n$ columns within  $((P^\mD)^\T, (P^\mJ)^\T)$ forms  
%an unit matrix of order $n$ and all the other columns vanish, and that is,  
which are independent of any moduli parameters in $\mD(z)$.
}
%after appropriate fixing of the V-transformation a squared matrix $P^\mD$ of order $n$ is also constant 
%since all elements of $\mD(z)^{-1}$ are  of ${\cal O}(1)$ at most.\footnote{
%This equation is not invariant under the $V$-transformation.
%With a patch given in Eq.(\ref{eq:Dpatch})  and 
%Eq.(\ref{eq:Jpatch}),  these matrices can be easily obtained since $\mD(z)$ and $\mJ(z)$ have the following column vectors
%\begin{eqnarray}
%(\mD(z))^a{}_b=\delta^a{}_{b} \quad {\rm for~} b \in \{b| 1\le b\le n,  l_b=0\},\quad 
%(\vec j (z)))^a= \delta^a_{b}  \quad {\rm for~}  p=l_b {\rm ~and~}  b\in  \{b| 1\le b\le n,  l_b\not=0\}
%\end{eqnarray}
%and a set of them forms an unit matrix of order $n$.
% } 
 Using the half-ADHM mapping relation \eqref{eq:hADHM} and (\ref{eq:PP}),  one can show that
\begin{eqnarray}
\left( \begin{array}{cc}
\mD(z)& \mJ(z) \\ {\bf 0} & {\bf 1}_k
\end{array} \right)
=\left( \begin{array}{cc}
{\bf 1}_n & {\bf 0} \\ P^\mJ & z{\bf 1}-Z
\end{array} \right){\cal N}^{-1},\quad {\rm with~~}{\cal N}\equiv \left( \begin{array}{cc}
P^\mD & -\Psi \\ P^\mJ & z{\bf 1}-Z
\end{array} \right).
\end{eqnarray}
By taking the determinant of the both sides 
and counting their degrees, 
we conclude this polynomial $\det {\cal N}$ is ${\cal O}(1)$, 
that is, $\det {\cal N}=1$ 
when $\det \mD(z)$ is chosen to be a monic polynomial.
Thus we find that $\det \mD(z) = \det (z \mathbf 1 - Z)$. 

%\paragraph{Example: $n_1=2$, $n_2=1$, $k=1$}
%Notice that there is coincidence as $\# T^a{}_{b,m} =kn =\# Z+\# \Psi- \# \mU$. 
%
%The above feature means $GL(k,\mathbb C)$ can completely be fixed by   taking an appropriate coordinate patch and one can find 
%$(k+n-1)!/(k! (n-1)!)$ patches again.

%%%%%%%%%%%%%%%%%%%%%%%%%%%%%%%%%%%%%%%%%%%%%%%%%%%%%%%%%%%%%%%%%%
\subsection%[Quotient construction for general L]
{Quotient construction for general $L$}
\def\mY{\Upsilon}
\def\mW{W}
\def\supW{{\cal W}}

For the case of general $L$, 
we can define $L$ copies of the matrices (and relations) 
defined in the previous subsection 
by attaching an index $i=1,\cdots,L$. 
For example, we can define $N_i$-by-$k_i$ matrix $\mJ_i(z)$ 
with polynomial entries by solving the condition 
\beq
\mD_i(z)^{-1} \mJ_i(z) = \mathcal O(z^{-1}). 
\eeq
Then, we can obtain matrices $\{Z_i, \Psi_i, \wt \Psi_i\}$ from 
the $N_i$-by-$N$ matrix $\mH_i(z) = (\mD_i(z),\wt \mD_i(z))$ 
thought the realtions
\beq
\mD_i(z) \Psi_i =\mJ_i(z) (z{\bf 1}_{k_i}-Z_i), \hs{10}
\wt \mD_i(z) =  \mJ_i(z) \wt \Psi_i,
\label{eq:h-ADHM_L}
\eeq
where $\{Z_i, \Psi_i, \wt \Psi_i\}$ are $k_i$-by-$k_i$, $N_i$-by-$k_i$ and $k_i$-by-$(N-N_i)$ matrices, respectively. 
Conversely, from a given set of constant matrices 
$\{Z_i, \Psi_i, \wt \Psi_i| 1\le i\le L\}$, 
we can obtain the matrices $\{\mH_i(z)| 1\le i \le L\}$
up to $V$-transformations. 
However, such matrices $\{\mH_i(z)\}$ do not necessarily 
satisfy the constraint that there must be everywhere non-singular $N_i$-by-$N_{i+1}$ matrices $\mq_i(z)$ such that 
\beq
\mq_i(z) \mH_{i+1}(z) = \mH_i(z). 
\label{eq:relation_H}
\eeq 
To guarantee that these constraints are satisfied, 
the matrices $\{Z_i, \Psi_i, \wt \Psi_i\}$ 
must satisfy some constraints.
To write down the constraints, 
let us decompose the matrices $\Psi_i$ and $\wt \Psi_i$ as
\beq
\Psi_{i+1} = \ba{c} \mY_{i+1}' \\ \mY_{i+1} \ea, \hs{10}
 \wt \Psi_i=(\wt \mY_i, \, \wt \mY_i' ), 
\label{eq:def_upsilon}
\eeq
where 
$\mY_{i+1}'$ is an $N_i$-by-$k_{i+1}$ matrix, 
$\mY_{i+1}$ is an $n_{i+1}$-by-$k_{i+1}$ matrix,
$\wt \mY_i$ is a $k_i$-by-$(N_{i+1}-N_i)$ matrix and
$\wt \mY_i'$ is a $k_i$-by-$(N-N_{i+1})$ matrix.
Then, the relation \eqref{eq:relation_H} can be rewritten as
\begin{eqnarray}
\wt \mY_i'&=&\mW_i \,\wt \Psi_{i+1}, \label{eq:const1} \\
 \mY_{i+1}' &=& \Psi_i \mW_i, \label{eq:const2} \\
 \wt \mY_i \mY_{i+1}&=&Z_i\mW_i-\mW_i Z_{i+1}, \label{eq:const3}
\end{eqnarray}
where $W_i$ is a $k_i$-by-$k_{i+1}$ matrix such that
\begin{eqnarray}
 \mq_i(z) \mJ_{i+1}(z) = \mJ_i(z) \mW_i.   
\end{eqnarray}
The constraints \eqref{eq:const1}-\eqref{eq:const3} can be
derived as follows.

\paragraph{Constraints on half-ADHM data \\}
Let us first show the half-ADHM data 
$\{Z_i, \Psi_i, \wt \Psi_i\}$ 
obtained through \eqref{eq:h-ADHM_L} satisfy
\eqref{eq:const1}-\eqref{eq:const3}. 
Then, we show that the matrices $\{ \mH_i(z) \}$
obtained from the half-ADHM data obeying the constraints  
\eqref{eq:const1}-\eqref{eq:const3} 
satisfy the relation \eqref{eq:relation_H} 
with suitable matrices $q_i(z)$. 

\paragraph{\eqref{eq:relation_H} to \eqref{eq:const1}-\eqref{eq:const3} \\}
We can rewrite the relation \eqref{eq:relation_H} 
for the matrices $\mH_i(z)=(\mD_i(z),\wt \mD_i(z)) = (\mD_i(z),\mJ_i(z) \wt \Psi_i)$ as
\begin{eqnarray}
\mq_i(z) \mD_{i+1}(z)&=&\left( \mD_i(z), \, \mJ_i(z) \wt \mY_i \right), \label{eq:rel1}\\ 
\mq_i(z)\mJ_{i+1}(z)\wt \Psi_{i+1}&=&\mJ_i(z) \wt \mY_i'. \label{eq:rel2}
\end{eqnarray}
It follows from the first equation
that $\mD_i(z)^{-1} \left[ \mq_i(z) \mJ_{i+1}(z) \right] = \mathcal O(z^{-1})$
and hence there exist a $k_i$-by-$k_{i+1}$ matrix $W_i$ such that\footnote{
Since
\begin{eqnarray}
\mD_i^{-1}\left( \mq_i \mJ_{i+1} \right) \ = \ \mD_i^{-1} \left( \mq_i \mD_{i+1} \right) \left(\mD_{i+1}^{-1}\mJ_{i+1}\right)
\ = \ \left({\bf 1}_{N_i},\,  {\mD}_i^{-1}\mJ_i \wt \mY_i  \right)  \left(\mD_{i+1}^{-1}\mJ_{i+1} \right) \ = \ {\cal O}(z^{-1}),
\notag
\end{eqnarray}
the columns of the matrix $\mq_i(z) \mJ_{i+1}(z)$ can be written as linear combinations of the columns of $\mJ_i(z)$, and hence 
there exist a matrix $W_i$ such that $\mq_i(z) \mJ_{i+1}(z) = \mJ_i(z) W_i$.
}
\begin{eqnarray}
\mq_i(z) \mJ_{i+1}(z) = \mJ_i(z) \mW_i.   \label{eq:W}
\end{eqnarray}
Then, by substituting this relation into (\ref{eq:rel2})
we find that \eqref{eq:const1} is satsified
\begin{eqnarray}
\mJ_i(z) \mW_i \,\wt\Psi_{i+1} = \mJ_i(z)\wt \mY_i' \quad \Longleftrightarrow \quad 
\wt \mY_i'=\mW_i \,\wt \Psi_{i+1}, 
\end{eqnarray}
where we have used the fact that 
the columns of $\mJ_i(z)$ are linearly independent. 
We can see that \eqref{eq:const2} and \eqref{eq:const3}
are satisfied as follows. 
By multiplying the both sides of Eq.(\ref{eq:rel1}) 
by $\Psi_{i+1} = (\mY_{i+1}', \mY_{i+1})^T$ from the right, 
we obtain
\beq
q_i(z) \mD_{i+1}(z) \Psi_{i+1} = \mD_i(z) \mY_{i+1}' + \mJ_i(z) \wt \mY_i \mY_{i+1} . 
\label{eq:eq30}
\eeq
The left hand side can be rewritten 
by using the half-ADHM mapping relation and (\ref{eq:W}) as 
\begin{eqnarray}
\mq_i \mD_{i+1} \Psi_{i+1} = \mq_i \mJ_{i+1} (z{\bf 1}_{k_{i+1}}-Z_{i+1}) = \mJ_{i} \mW_i (z{\bf 1}_{k_{i+1}}-Z_{i+1}) = \mD_i \Psi_i W_i + \mJ_{i} (Z_i \mW_i - \mW_i Z_{i+1}). 
\end{eqnarray}
Comparing this with the right hand side of \eqref{eq:eq30} 
and using the linear independence of $(\mD_i(z),\mJ_i(z))$, 
we find that\eqref{eq:const2} and \eqref{eq:const3} are satisfied. 

\paragraph{\eqref{eq:const1}-\eqref{eq:const3} to \eqref{eq:relation_H} \\}
Next, let us show that if the half-ADHM data satisfies 
the constraints \eqref{eq:const1}-\eqref{eq:const3}, 
the corresponding $\mH_i(z)$ related through \eqref{eq:h-ADHM_L} satisfy the relation \eqref{eq:relation_H}
(or equivalently \eqref{eq:rel1} and \eqref{eq:rel2}), 
with suitable matrices $q_i(z)$. 
Although formally the condition \eqref{eq:rel1} is always satisfied 
by adopting $q_i(z) = (\mD_i(z),\mJ_i(z) \wt \mY_i') \mD_{i+1}(z)^{-1}$, 
such matrices $q_i(z)$ may not be suitable 
since they can have some poles. 
We can show that 
$q_i(z) = (\mD_i(z),\mJ_i(z) \wt \mY_i') \mD_{i+1}(z)^{-1}$ 
do not have any pole 
if \eqref{eq:const1}-\eqref{eq:const3} are satisfied.  
To show this, let us rewrite $\mD_{i+1}^{-1}$ as
\beq
\mD_{i+1}^{-1} = P^\mD_{i+1} + \mD_{i+1}^{-1} \mJ_{i+1} P_{i+1}^\mJ = P^\mD_{i+1} + \Psi_{i+1} (z \mathbf 1_{k_{i+1}} - Z_{i+1})^{-1} P_{i+1}^\mJ, 
\eeq
where $P^\mD_{i+1}$ and $P_{i+1}^\mJ$ are matrices defined by
$\mathbf 1_{N_{i+1}} = \mD_{i+1}(z) P^\mD_{i+1} + \mJ_{i+1}(z) P_{i+1}^\mJ$. 
Then, $q_i= (\mD_i,\mJ_i \wt \mY_i')\mD_{i+1}^{-1}$ can be rewritten as
\beq
q_i(z) ~=~ ( \mD_i, \, \mJ_i \wt \mY_i ) \mD_{i+1}^{-1} ~=~ ( \mD_i, \, \mJ_i \wt \mY_i ) P^\mD_{i+1} + ( \mD_i, \, \mJ_i \wt \mY_i ) \Psi_{i+1} (z \mathbf 1_{k_{i+1}} - Z_{i+1})^{-1} P_{i+1}^\mJ. 
\eeq
Obviously the first term has no pole and 
the regularity of the second term can be shown by rewriting 
\beq
( \mD_i, \, \mJ_i \wt \mY_i ) \Psi_{i+1} (z \mathbf 1_{k_{i+1}} - Z_{i+1})^{-1} 
 &=& ( \mD_i \mY_{i+1}' +\mJ_i \wt \mY_i \mY_{i+1} ) (z \mathbf 1_{k_{i+1}} - Z_{i+1})^{-1} \notag \\
&=& ( \mD_i \Psi_i W_i + \mJ_i (Z_i W_i - W_i Z_{i+1}) ) (z \mathbf 1_{k_{i+1}} - Z_{i+1})^{-1} ~=~ \mJ_i W_i . \hs{5}
\eeq
Since this has no pole, $q_i(z) = ( \mD_i, \, \mJ_i \wt \mY_i ) P^\mD_{i+1} + \mJ_i W_i P_{i+1}^\mJ$ is regular
and hence \eqref{eq:rel1} is satisfied. 
Then we can show that \eqref{eq:rel2} is also satisfied as
\beq
q_i \mJ_{i+1} \wt \Psi_{i+1} = q_i \mD_{i+1} \mD_{i+1}^{-1} \mJ_{i+1} \wt \Psi_{i+1} = (\mD_i, \mJ_i \wt \mY_i) \Psi_{i+1}(z \mathbf 1_{k_{i+1}} - Z_{i+1})^{-1} \wt \Psi_{i+1} = \mJ_i W_i \wt \Psi_{i+1} = \mJ_i \wt \mY_i'. 
\eeq
Thus, we find that 
{\it no further constraints} other than 
\eqref{eq:const1}-\eqref{eq:const3} 
is needed on a data set $\{Z_i,\,\mY_i,\, \wt \mY_i,\,\mW_i\}$
to guarantee that $q_i(z)$ obtained through 
the half-ADHM mapping relation \eqref{eq:h-ADHM_L} 
satisfy the relation \eqref{eq:relation_H}.

The constraints \eqref{eq:const1} and \eqref{eq:const2} imply that $\mY_{i}'$ and $\wt \mY_{i}'$ are not independent but 
can be rewritten in terms of $\{W_i, \mY_i,\wt \mY_i\} $ as
\begin{eqnarray}
\Psi_i&=& \Big( \Psi_{i-1}\mW_{i-1} \ , \ \mY_{i} \Big)^T =
\Big( \mY_1 \mW_1 \mW_2 \cdots \mW_{i-1} \ , \ \mY_2 \mW_2 \cdots \mW_{i-1} \ , \ \cdots \ , \ \mY_{i-1}\mW_{i-1} \ , \ \mY_{i} \Big)^T, \label{eq:psi}
\\
\wt \Psi_i &=& \left( \wt \mY_i \ , \ \mW_i \wt \Psi_{i+1}  \right) \hs{6}
=\left( \wt \mY_i \ , \ \mW_i \wt \mY_{i+1} \ , \ \mW_i \mW_{i+1} \wt \mY_{i+2} \ , \ \cdots \ , \  \mW_i \mW_{i+1} \cdots \mW_{L-1} \wt \mY_{L} \right), \quad
\label{eq:psit}
\end{eqnarray}
with $\mY_1 \equiv \Psi_1$ and 
$\wt \mY_L \equiv \wt \Psi_L$. 
Therefore, all information describing vortex moduli are 
contained in the set of matrices 
$\{ Z_i,\mY_i,\wt \mY_i| 1\le i \le L\}$ and 
$\{ W_i| 1\le i\le L-1\}$ obeying the constraints
\begin{eqnarray}
 Z_i\mW_i-\mW_i Z_{i+1}=\wt \mY_i \mY_{i+1} \qquad 
 {\rm for~}  1\le i\le L-1
 \label{eq:constraint}.
\end{eqnarray}
Since $q_i(z)$ related to 
$\{Z_i,\,\mY_i,\, \wt \mY_i,\,\mW_i\}$ 
though the half-ADHM mapping relation 
do not change under the $GL(k_1, \mathbb C) \times \cdots \times GL(k_L, \mathbb C)$ transformation
\begin{eqnarray}
\{ Z_i, \,\mY_i,\,\wt \mY_i, \,\mW_i\} \simeq \{ \mU_i^{-1}Z_i\mU_i, \,\mY_i \mU_i,\, \mU_i^{-1}\wt \mY_i,\, 
\mU^{-1}_i\mW_i \mU_{i+1}\} \quad{\rm with~} \forall \, \mU_i \in GL(k_i, \mathbb C),
\end{eqnarray}
the quotient space of the data $\{Z_i,\,\mY_i,\, \wt \mY_i,\,\mW_i\}$ by $GL(k_1, \mathbb C) \times \cdots \times GL(k_L, \mathbb C)$ can be identified with the vortex moduli space. 
Note that the action of $GL(k_1, \mathbb C) \times \cdots \times GL(k_L, \mathbb C)$ must be free on $\Psi_i, Z_i$, that is,    
\begin{equation}
\exists\, \hbox{$k_i$-column vector~} \vec v_i: \mY_j \mW_j \mW_{j+1}\cdots \mW_{i-1} (z{\bf 1}_{k_i}-Z_i)^{-1}
 \vec v_i =0 ~ {\rm for~} 
1\le \forall \,j \le i, ~ \forall \, z \in \mathbb C ~
\Rightarrow ~ \vec v_i=0. \label{eq:GLkifree}
\end{equation}
Note that this set of conditions is a generalization of the condition \eqref{eq:free_cond} for $L=1$. 

In summary, the moduli space of vortices 
turns out to be given by the quotient
\begin{eqnarray}
{\cal M}_{{\rm vtx}}{}^{n_1,n_2,\dots,n_{L+1}}_{~k_1,k_2,\cdots,k_L}=\left\{ Z_i, \,\mY_i,\,\wt \mY_i,\, \mW_j \, \Big| \, %1\le i\le L,\,1\le j\le L-1,
{\rm constraints}(\ref{eq:constraint}),\,  {\rm condition}(\ref{eq:GLkifree}) \right\}
/GL(k_1, \mathbb C) \times \cdots \times GL(k_L, \mathbb C).
\label{eq:moduli_space_genearl}
\end{eqnarray}
The contents of this quotient are summarized in the quiver diagram Fig.\,\ref{fig:ADHM_quiver} and 
the corresponding gauged linear sigma model can also be obtained from the D-brane configuration for BPS vortices (see Appendix \ref{sec:brane}).

In Appendix \ref{sec:tW}, we prove that  
all possible singularities due to the constraints \eqref{eq:constraint} is removed by the condition \eqref{eq:GLkifree} and the moduli space defined above is smooth everywhere.

\begin{figure}[!t]
\begin{center}
\includegraphics[bb = 100 150 1920 500, width=170mm]{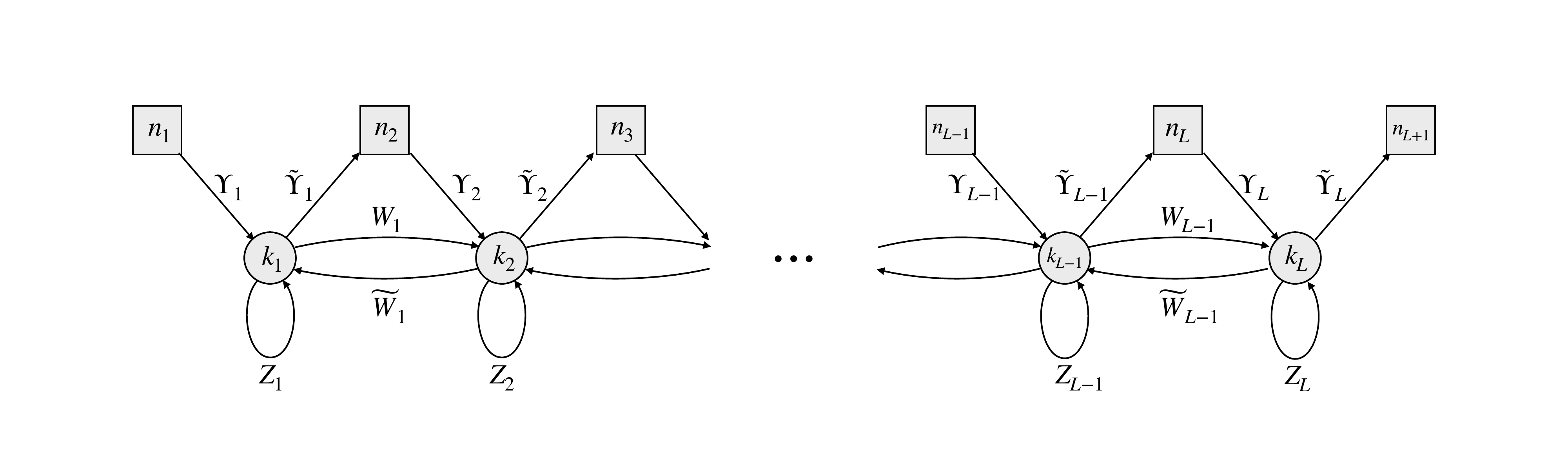}
\end{center}
\caption{Quiver diagram for vortex moduli. Here $\wt \mW_i$ represent Lagrange multipliers for the constraint (\ref{eq:constraint}), which will be explained in the next section.
\label{fig:ADHM_quiver}}
\end{figure}

\subsubsection
%[An example of L=2 : SU(3)/U(1)**2 sigma model]
{An example of $L=2$ : $SU(3)/U(1)^2$ sigma model}
Here, we illustrate the moduli space of vortices in the gauged linear sigma model corresponding to the $SU(3)/U(1)^2$ sigma model ($L=2, N_1=1, N_2=2, N_3=3, n_1=n_2=n_3=1$).
The model is the $U(1) \times U(2)$ gauge theory with an $SU(3)_F$ flavor symmetry. 
The matter content consists of two scalar fields $Q_1$ and $Q_2$ in the $(1,\bar{\mathbf 2},\mathbf 0)$ and $(0, \mathbf 2,\bar{\mathbf 3})$ of the $U(1) \times U(2) \times SU(3)_F$ symmetry, respectively. 
The topological sectors are labeled by two integers $(k_1,k_2)$
corresponding to the vortex numbers of $U(1) \times U(2)$ gauge group.   

\paragraph{Two coordinate patches in $(k_1,k_2)=(1,1)$ case \\}

First, let us consider the case of $(k_1,k_2)=(1,1)$.
As we have seen in subsection \ref{subsec:example_L=2}, 
there are two coordinate patches: 
\beq
\bullet \ \mbox{$(1,0)$-patch :} \hs{5} \mH_1(z)=(z-a', c', c), \hs{5} 
\mH_2(z) = \ba{ccc} z-a & 0 & c \\ -b & 1 & 0 \ea, \label{eq:L=2_(1,0)_2}\\
\bullet \ \mbox{$(0,1)$-patch :} \hs{5} 
\mH_1(z)=(z-a', c', c), \hs{5}
\mH_2(z) = 
\ba{ccc} 1 & - \tilde b & 0 \\ 0 & z-a & \tilde c \ea,
\label{eq:L=2_(0,1)_2}
\eeq
where the parameters are realted as
\beq
b = \tilde b^{-1}, \hs{5}
c' = \tilde b (a'-a), \hs{5}
c = \tilde b \tilde c. 
\eeq
We can move from the $(1,0)$-patch to the $(0,1)$-patch by using the $V$-transformation $\xi_i \rightarrow V_i \xi_i$ with
\beq
 V_1 = 1, \hs{5} V_2 = \ba{cc} 0 & -\tilde b \\ b & z-a \ea. 
\eeq

\noindent $\bullet$ $(1,0)$-patch \\
From the definition $\xi_i = (\mD_i,\wt \mD_i)$ and \eqref{eq:L=2_(1,0)_2}, 
we find that the matrices $\mD_i$ and $\wt \mD_i$ are givenby 
\beq
\mD_1 = z - a', \hs{5}
\wt \mD_1 = (c',c), \hs{10}
\mD_2 = \ba{cc} z - a & 0 \\ -b & 1 \ea, \hs{5}
\wt \mD_2 = \ba{c} c \\ 0 \ea. 
\eeq
From these matrices, the corresponding half-ADHM data can be read off as
\begin{alignat}{4}
&\mD_i^{-1} \mJ_i = \mathcal O(z^{-1}) & \hs{3} &\rightarrow & \hs{5} & \mJ_1 = 1, &\hs{5} \mJ_2 &=(1,0)^T \\
&\mD_i \Psi_i = \mJ_i (z - Z_i) & &\rightarrow & &Z_1 = a', &\hs{5}
Z_2 &= a, \hs{10} 
\Psi_1 = 1, \hs{3} \Psi_2 = (1, b)^T, \\
&\wt \mD_i = \mJ_i \wt \Psi_i & &\rightarrow & & \wt \Psi_1 = ( c', c ), &\hs{5} 
\wt \Psi_2 &= c, \\
&\Psi_i = (\mY_i', \mY_i)^T & &\rightarrow & &\mY_1 = 1, &\hs{5} \mY_2 &= b, \\
&\wt \Psi_i = (\wt \mY_i,\wt \mY_i') & &\rightarrow & &\wt \mY_1 = c', &\hs{5} \wt \mY_2 &= c, \\
&\wt \mY_1 \mY_2 = Z_1 W_1 - W_1 Z_2 & &\rightarrow & &W_1 = 1.
\end{alignat}  

\noindent $\bullet$ $(0,1)$-patch \\
From the definition $\xi_i = (\mD_i,\wt \mD_i)$ and \eqref{eq:L=2_(0,1)_2}, 
we find that the matrices $\mD_i$ and $\wt \mD_i$ are given by
\beq
\mD_1 = z-a', \hs{5} \wt \mD_1 = (c',c), \hs{5}
\mD_2= \ba{ccc} 1 & - \tilde b \\ 0 & z-\tilde a \ea, \hs{5}
\wt \mD_2 = \ba{c} 0 \\ \tilde c \ea. 
\eeq
From these matrices, the corresponding half-ADHM data can be read off as
\begin{alignat}{4}
&\mD_i^{-1} \mJ_i = \mathcal O(z^{-1}) & \hs{3} &\rightarrow & \hs{5} & \mJ_1 = 1, &\hs{5} \mJ_2 &=(0,1)^T \\
&\mD_i \Psi_i = \mJ_i (z - Z_i) & &\rightarrow & &Z_1 = a', &\hs{5}
Z_2 &= a, \hs{10} 
\Psi_1 = 1, \hs{3} \Psi_2 = (\tilde b, 1 )^T, \\
&\wt \mD_i = \mJ_i \wt \Psi_i & &\rightarrow & & \wt \Psi_1 = ( c', c ), &\hs{5} 
\wt \Psi_2 &= \tilde c, \\
&\Psi_i = (\mY_i', \mY_i)^T & &\rightarrow & &\mY_1 = 1, &\hs{5} \mY_2 &= 1, \\
&\wt \Psi_i = (\wt \mY_i,\wt \mY_i') & &\rightarrow & &\wt \mY_1 = c', &\hs{5} \wt \mY_2 &= \tilde c, \\
&\wt \mY_1 \mY_2 = Z_1 W_1 - W_1 Z_2 & &\rightarrow & &W_1 = \tilde b.
\end{alignat}  
%The matrices $\mJ_i$ are determined from the condition $\mD_i^{-1} \mJ_i = \mathcal O(z^{-1})$ as
%\beq
%\mJ_1 = 1, \hs{5} \mJ_2 = \ba{c} 0 \\ 1 \ea.
%\eeq
%Then, by solving the relations $\mD_i \Psi_i = \mJ_i %(z - Z_i),~\wt \mD_i = \mJ_i \wt \Psi_i$, we can determine $(Z_i,\Psi_i,\wt \Psi_i)$ as
%\beq
%Z_1 = a', \hs{5}
%Z_2 = \tilde a, \hs{5} 
%\Psi_1 = 1, \hs{3} \Psi_2 = \ba{c} \tilde b \\ 1 \ea, \hs{5}
%\wt \Psi_1 = ( c', c'' ), \hs{5} 
%\wt \Psi_2 = \tilde c. 
%\eeq
%From the definition $\Psi_i =(\mY_i', \mY_i)^T$, $\wt \Psi_i = (\wt \mY_i,\wt \mY_i')$, 
%we can determine $\mY_i$ and $\wt \mY_i$ as
%\beq
%\mY_1 = 1, \hs{5} \mY_2 = 1, \hs{5}
%\wt \mY_1 = c', \hs{5} \wt \mY_2 = \tilde c . 
%\eeq
%By solving the constraint $Z_1 W_1 - W_1 Z_2 = \wt \mY_1 \mY_2$, we can determine $W_1$ as
%\beq
%W_1 = \tilde b.
%\eeq
%From the half-ADHM mapping relation, the matrices $J_i$, $Z_i$ and $\Psi_i$ can be read off as
%\beq
%\mJ_1 = 1, \hs{5} Z_1 = z_2 + \tilde \alpha \tilde \beta, \hs{5}
%\Psi_1 = 1, \hs{5} \wt \Psi_1 = ( \tilde \alpha, \tilde \delta ),
%\eeq
%and
%\beq
%\mJ_2 = \ba{c} 1 \\ 0 \ea, \hs{5} Z_2 = z_2, \hs{5}
%\Psi_2 = \ba{c}1 \\ \tilde \beta \ea, \hs{5} \wt \Psi_2 = \tilde \delta.
%\eeq
%The matrices $\mY_i$ and $\wt \mY_i$ are given by
%\beq
%\mY_1 = 1, \hs{5} \wt \mY_1 = \tilde \alpha, \hs{10}
%\mY_2 = \tilde \beta, \hs{5} \wt \mY_2 = \tilde \delta. 
%\eeq
%From the constraint $Z_1 W_1 - W_1 Z_2 = \tilde \mY_1 \mY_2$, $W_1$ can be determined as
%\beq
%W_1 = 1. 
%\eeq

The coordinate transformation from the $(1,0)$-patch to $(0,1)$-patch is given by a group element $(g_1, g_2) \in GL(1,\C) \times GL(1,\C)$
\beq
\mY_i \rightarrow \mY_i g_i, \hs{5}
\wt \mY_i \rightarrow g_i^{-1} \wt \mY_i, \hs{5} \mbox{with} \hs{5} 
g_1 = 1, \hs{5} g_2 = \tilde b. 
\eeq

In Appendix \ref{appendix:embedding}, the half-ADHM data obtained by embedding from the $L=1$ case is discussed as another example.
%%%%%%%%%%%%%%%%%%%%%%%%%%%%%%%%%%%%%%%%%%%

\subsection{Coordinate patches on moduli space}
In this section, we discuss the coordinate patches of the moduli space of the half-ADHM data.
To define an analogue of the $(l_1,\cdots,l_k)$-patch for the $L=1$ case,
we first discuss the fixed point of a torus action that plays the role of the ``origin" in each coordinate patch. 

\subsubsection{Torus action and fixed points on vortex moduli space}\label{subsec:torus_action}
The torus action we discuss here is a combination of an Abelian subgroup of the flavor symmetry and the spatial rotation (see Appendix \ref{sec:TorusActions} for the details of the torus action). 
Its fixed point configurations can be viewed as the BPS vortex solutions 
in the presence of the omega background and the mass deformation. 
Such configurations must satisfy 
the following conditions
in addition to the vortex equations \eqref{eq:BPSQ} and \eqref{eq:BPSD}
\beq
 i \epsilon (z \D_z - \bar z \D_{\bar z}) Q_i + \Sigma_i Q_i - Q_i \Sigma_{i+1} = 0, \hs{5}
\mbox{with} \hs{5} \Sigma_{L+1} = - M,
\label{eq:inv_cond}
\eeq
where $\Sigma_i~(i=1,\cdots,L)$ are $SU(N_i)$ adjoint scalar fields\footnote{In 2d $\mathcal N = (2,2)$ supersymmetric models, $\Sigma_i$ can be interpreted as the adjoint scalar fileds in the vector multiplets and become auxiliary fields in the nonlinear sigma model limit.}, $\epsilon$ is the omega deformation parameter and 
$M={\rm diag}(m^1, \cdots, m^N)$ is the mass matrix.
A configuration satisfying \eqref{eq:inv_cond}  
is invariant under the infinitesimal spatial rotation and flavor rotation up to an infinitesimal gauge transformation $\Sigma_i$. 
For such a fixed point configuration, 
the magnetic fluxes take the diagonal forms
\beq
\frac{1}{2\pi} \int F_i ~=~ {\rm block\mbox{-}diag}(\boldsymbol l^{(i,1)},\cdots,\boldsymbol l^{(i,i)}) =
{\renewcommand{\arraystretch}{0.8}
{\setlength{\arraycolsep}{1.2mm}
\ba{ccc} \boldsymbol l^{(i,1)} & & \\ & \ddots & \\ & & \boldsymbol l^{(i,i)} \ea}} 
\hs{5} \mbox{with} \hs{5}
\boldsymbol l^{(i,j)} = 
{\renewcommand{\arraystretch}{0.6}
{\setlength{\arraycolsep}{0mm}
\ba{ccc} l^{(i,j,1)} & & \\ & \ddots & \\ & &~ l^{(i,j,n_j)} \ea}},
\eeq
where $\boldsymbol l^{(i,j)}$ denotes the $n_j$-by-$n_j$ diagonal block of 
the $SU(N_i)$ magnetic flux of the $i$-th gauge group. 
The labels $(i,j)$ and $(i,j,\alpha)~(i=1,\cdots,L,\,j=1,\cdots,i,\,\alpha=1,\cdots,n_j)$ specify the following subgroup of the gauge group:
\begin{itemize}
    \item $(i,j)$ \ \, :  $j$-th $U(n_j)$ subgroup of $i$-th gauge group $U(N_i) \supset U(n_1) \times \cdots \times U(n_i)$.
    \item $(i,j,\alpha)$ :  $\alpha$-th Cartan subgroup $U(1)$ of $j$-th $U(n_j)$ subgroup of $i$-th gauge group $U(N_i)$.
\end{itemize}
The magnetic fluxes $l^{(i,j,\alpha)}$ are also related to winding numbers of the scalar fields 
\beq
q_i = 
{\renewcommand{\arraystretch}{0.5}
{\setlength{\arraycolsep}{0.8mm}
\ba{ccc|ccc} z^{\boldsymbol \nu^{(i,1)}} & & & 0 & \cdots & 0 \\ & \ddots & & \vdots & \ddots & \vdots \\ & & z^{\boldsymbol \nu^{(i,i)}} \phantom{\Big|} & \phantom{\Big|} 0 & \cdots & 0 \ea}} 
\hs{5} \mbox{with} \hs{5} \boldsymbol \nu^{(i,j)} \equiv \boldsymbol l^{(i,j)} - \boldsymbol l^{(i+1,j)}, 
\eeq
where $\boldsymbol \nu^{(i,j)}$ are diagonal matrices of winding numbers. 
We can confirm that $q_i(z)$ is invariant under 
the torus action (the Cartan part of the spatial rotation and the flavor rotation) up to $V$-transformations
\begin{eqnarray}
q_i(z) = V_i \, q_i(e^{i\epsilon }z) \, V_{i+1}^{-1}, 
\hs{10} 
V_i = \exp (i \Sigma_i), \hs{5} 
V_{L+1}(z) = \exp (-i M). 
\end{eqnarray}
The element of the $V$-transformations are 
specified by the fixed point values of the adjoint scalar 
$\Sigma_i$, which take the forms 
\begin{eqnarray}
\Sigma_i = 
\mbox{block-diag}(
\boldsymbol \sigma^{(i,1)},\cdots, \boldsymbol \sigma^{(i,i)}), 
\hs{10}
\boldsymbol \sigma^{(i,j)} = 
{\rm diag} (\sigma^{(i,j,1)}, \cdots, \sigma^{(i,j,n_j)}),
\eeq
with the eigenvalues
\beq
\sigma^{(i,j,\alpha)} = - m^{(j,\alpha)} - l^{(i,j,\alpha)} \epsilon,
\eeq
where we have labeled the eigenvalues of the mass matrix as
\beq
\Sigma_{L+1} = M = 
\mbox{block-diag}(\boldsymbol m^1, \cdots, \boldsymbol m^{L+1} ), 
\hs{10} 
\boldsymbol m^j = 
{\rm diag}(m^{(j,1)}, \cdots, m^{(j,n_i)}).
\label{eq:mass_matrix}
\end{eqnarray}

Since the winding numbers $\boldsymbol \nu^{(i,j)} \equiv \boldsymbol l^{(i,j)} - \boldsymbol l^{(i+1,j)}$ of the scalar fields $q_i$ must be non-negative integers, 
the magnetic fluxes must satisfy $l^{(i,j,\alpha)} \geq l^{(i+1,j,\alpha)}$. 
Therefore, the fixed points are classified 
by a set of $N$ Young tableaux 
$Y^{(j,\alpha)}$ where $\alpha = 1, \cdots, n_j$ for each $j=1,\cdots,L$. 
The height of $Y^{(j,\alpha)}$ is $L-j+1$ and 
we denote the length of $i$-th row as $l^{(i+j-1,j,\alpha)}$, i.e. 
\beq
Y^{(j,\alpha)} = \left( l^{(j,j,\alpha)},l^{(j+1,j,\alpha)},\cdots,l^{(L,j,\alpha)} \right), \hs{10} 
l^{(j,j,\alpha)} \ge l^{(j+1,j,\alpha)} \ge \cdots \ge l^{(L,j,\alpha)} \ge 0.
\label{eq:Young_tab0}
\eeq
\begin{figure}
\centering
\fbox{
\includegraphics[width=100mm, bb = 0 0 900 600]{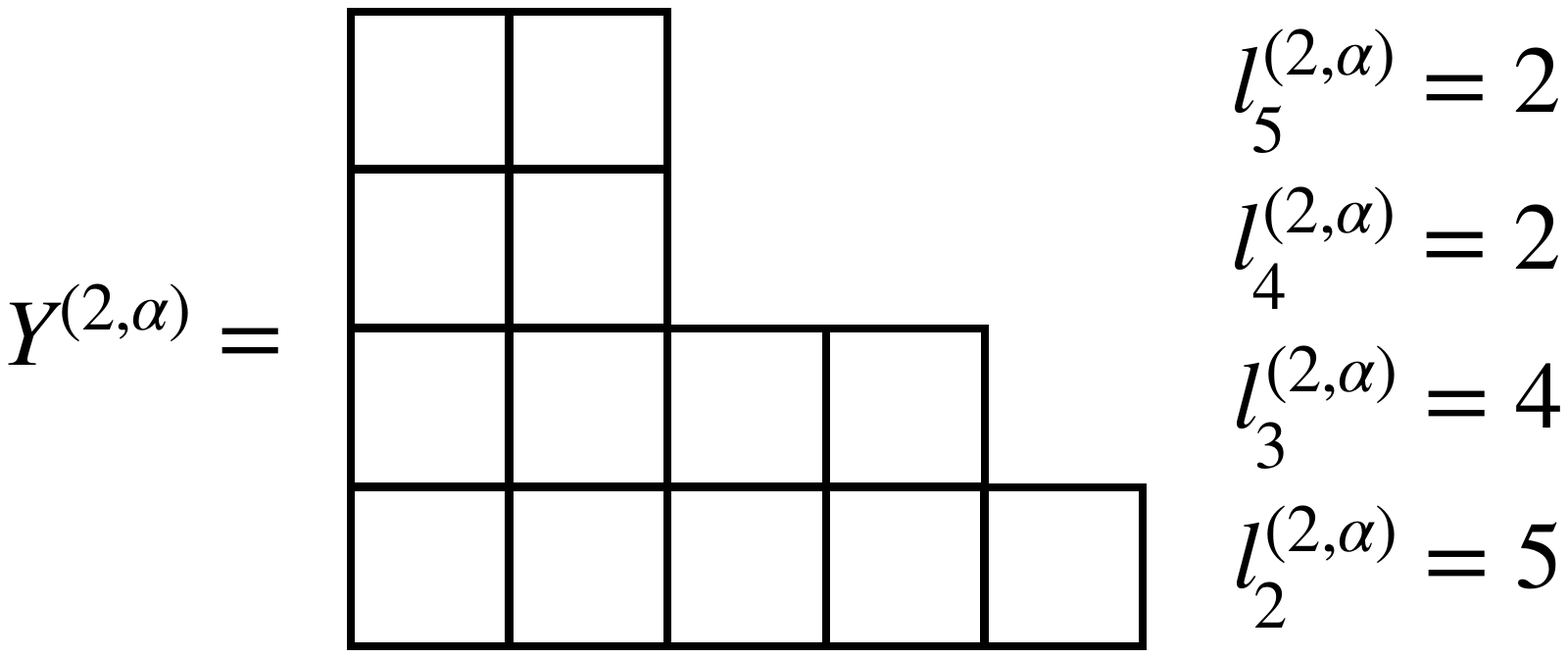}}
\caption{An example of Young tabuleax}\label{fig:Young_tab}
\end{figure}
See Fig.\,\ref{fig:Young_tab} for an example.

\paragraph{Half-ADHM data at fixed points \\ }
We can show that the invariant vortex data 
corresponding to the Young tableaux $Y^{(j,\alpha)}$ 
take the form
\beq
\mD_i = 
\mbox{block-diag}
(\boldsymbol \mD^{(i,1)},\cdots, \boldsymbol \mD^{(i,i)})
\hs{3} \mbox{with} \hs{3} 
\boldsymbol \mD^{(i,j)} = 
{\rm diag} (z^{l^{(i,j,1)}}, \cdots, z^{l^{(i,j,n_j)}}), \hs{3} \mbox{and} \hs{3} \wt \mD_i = 0.
\eeq
This implies that each diagonal component of $\boldsymbol \mD^{(i,j)}$ represents axially   symmetric Abelian vortices with flux $l_i^{(j,\alpha)}$
and hence all the half-ADHM data can be obtained by embedding those of Abelian vortices. 
For an axially symmetric Abelian vortex configuration $\mD = z^l$, 
the vortex data satisfying $\mD \Psi = \mJ (z \mathbf 1_{l} - Z)$ 
are given by (see Sec. \ref{subsec:patch_L=1})
\beq
\mJ(l) = ( z^{l-1} \,,\, z^{l-2} \,,\, \cdots \,,\, 1 ), \hs{10}
\Psi(l) = ( 1 \,,\, 0 \,,\, \cdots \,,\, 0 ), \hs{10}
Z(l) = \left. 
{\renewcommand{\arraystretch}{0.7}
{\setlength{\arraycolsep}{0.8mm}
\ba{c|ccc} 0 \phantom{|} & 1 & & \\ \vdots \phantom{|} & & \ddots & \\ 0 \phantom{|} & & & 1 \\ \hline 0 \phantom{|} & 0 & \cdots & 0 
\ea}} \ \right\} l.
\eeq
By embedding these matrices, 
we can construct the matrices satisfying 
$\mD_i \Psi_i = \mJ_i (z \mathbf 1_{k_i} - Z_i)$
as 
\beq
\!\!
\mJ_i = 
\mbox{block-diag} (\boldsymbol \mJ^{(i,1)}, \cdots, \boldsymbol \mJ^{(i,i)}), \hs{2}
\Psi_i = 
\mbox{block-diag}(\boldsymbol \Psi^{(i,1)}, \cdots, \boldsymbol \Psi^{(i,i)}), \hs{2}
Z_i = 
\mbox{block-diag} (\boldsymbol Z^{(i,1)}, \cdots, \boldsymbol Z^{(i,i)}),
\eeq
with
\beq
&\boldsymbol \mJ^{(i,j)} = 
\mbox{block-diag} \big(\mJ(l^{(i,j,1)}), \cdots, \mJ(l^{(i,j,n_j)})\big), \hs{5}
\boldsymbol \Psi^{(i,j)} = 
\mbox{block-diag} \big(\Psi(l^{(i,j,1)}),\cdots, \Psi(l^{(i,j,n_j)}) \big),& \\
&\boldsymbol Z^{(i,j)} = 
\mbox{block-diag}\big(Z(l^{(i,j,1)}), \cdots, Z(l^{(i,j,n_j)}) \big). &
\eeq
Note that $\wt \Psi_i =0$ since $\wt \mD_i = 0$
for the fixed point configurations. 
The matrices $\mY_i$ and $\wt \mY_i$ defined in \eqref{eq:def_upsilon} can be extracted from 
$\Psi_i$ and $\wt \Psi_i$ as
\beq
\mY_i = \ba{c|c} \mathbf 0_{n_i,k_{i-1}} & \boldsymbol \Psi^{(i,i)}  \ea, \hs{10}
\wt \mY_i = 0. 
\eeq
The matrix $W_i$ can be determined 
by solving the constraint 
$Z_i W_i - W_i Z_{i+1} = \wt \mY_i \mY_{i+1}$ as 
\beq
W_i = 
{\renewcommand{\arraystretch}{0.8}
{\setlength{\arraycolsep}{0.7mm}
\ba{ccc|ccc} \boldsymbol W^{(i,1)} & & & ~ \mathbf 0 & ~ \cdots & ~ \mathbf 0 \\ & \ddots & & ~ \vdots & ~ \ddots  & ~ \vdots \\ & & \boldsymbol W^{(i,i)} & ~ \mathbf 0 & ~ \cdots & ~ \mathbf 0 \ea
}}, \hs{5}
\boldsymbol W^{(i,j)} = 
{\renewcommand{\arraystretch}{0.6}
{\setlength{\arraycolsep}{0.1mm}
\ba{ccc} W\big(l^{(i,j,1)}, l^{(i+1,j,1)}\big) & & \\ & \ddots & \\ & & W\big(l^{(i,j,n_j)},l^{(i+1,j,n_j)}\big)
\ea}},
\eeq
where $W(l,l')$ is the matrix satisfying 
$Z(l) W(l,l') - W(l,l') Z(l') = 0$, which takes the form
\beq
W(l,l') = \ba{c} \mathbf 1_{l'} \\ \mathbf 0_{l-l',l'} \ea.
\eeq
Note that these half-ADHM data for the fixed points can also be obtained by solving the fixed point condition for the half-ADHM data.  
We can check these matrices are invariant under the torus action on the half-ADHM data (see Appendix \ref{Appendix:Torus action on Half-ADHM data}). 

%%%%%%%%%%%%%%%%%%%%%%%%%%%%%%%%%%%%%%%%%%%
\subsubsection{Coordinates around fixed points}
The coordinate patches around the fixed points discussed above can be obtained by considering fluctuations around the fixed point, eliminating the $GL(k_1,\C) \times \cdots \times GL(k_L,\C)$ gauge degrees of freedom and imposing the constraints \eqref{eq:constraint}.
After fixing the $GL(k_i,\mathbb C)$ transformations, 
we find the following non-zero components of the fluctuations
\begin{alignat}{4}
& (\delta \mY_i)^{\alpha}{}_{(i,j,\beta,p)} \quad & 
& {\rm for~} \alpha \in \{ \alpha \,|\, l^{(i,i,\alpha)}=0 \} \hs{10} & 
&  (\delta \wt \mY_i)^{(i,j,\alpha,p)}{}_{(i+1,\beta)} & 
& \\ 
& (\delta Z_i)^{(i,j,\alpha,p)}{}_{(i,k,\beta,q)} \quad & 
& {\rm with~} p=l^{(i,j,\alpha)}, \quad & 
& (\delta W_i)^{(i,j,\alpha,p)}{}_{(i+1,k,\beta,q)} \quad {\rm for~} & 
& 2 \le p \le l^{(i,j,\alpha)}.
\end{alignat}
%& (\delta \mY_i)^{\alpha}{}_{(i,j,\beta,p)} \quad & \hs{10}
%& {\rm for~} \alpha \in \{ \alpha \,|\, l^{(i,i,\alpha)}=0 \}, & \hs{10}
%& (\delta \wt  \mY_i)^{(j;a,p)_i}{}_{(b)_{i+1}}, \nn 
%& (\delta Z_i)^{(i,j,\alpha,p)}{}_{(i,k,\beta,q)} \quad &
%& {\rm with~} p=l^{(i,j,\alpha)}, \quad &
%& (\delta W_i)^{(i,j,\alpha,p)}{}_{(i+1,k,\beta,q) \quad {\rm for~} & 
%& 2 \le p \le l^{(i,j,\alpha)}.
Not all of these fluctuations independent 
since they are subject to the constraints (\ref{eq:constraint}). 
The total number of the components of (\ref{eq:constraint}) is given by
\begin{eqnarray}
%\sum_{i=1}^{L=1}\# \wt \mW_i \equiv 
d_{\rm c} = \sum_{i=1}^{L-1}k_i k_{i+1}. 
\end{eqnarray}
In Appendix \ref{sec:tW} we show that 
all components of the constraints \eqref{eq:constraint} are linearly independent of each other
for all points satisfying the condition \eqref{eq:GLkifree}.
Fixing $d_{\rm c}$ degrees of freedom 
by solving the constraints (\ref{eq:constraint}),
we can obtain the coordinates of the moduli space 
around this fixed point. 
%\footnote{
%\begin{eqnarray}
%&&\# \delta \mY_i+\# \delta Z_i +\# \delta \wt \mY_i +\# \delta \mW_j \nn
%&=&
%\sum_{i=1}^L \sum_{a=1}^{n_i} \delta^{l_{ii}^a}_0 k_i+\sum_{i=1}^L\sum_{j=1}^i \sum_{a=1}^{n_j}(1- \delta^{l_{ij}^a}_0) k_i+\sum_{i=1}^L k_i n_{i+1}
%+\sum_{i=1}^{L-1}\sum_{j=1}^i \sum_{a=1}^{n_j}(l_{ij}^a-(1-\delta^{l_{i+1,j}^a}_0))k_{i+1} \nn
%&=&\sum_{i=1}^L k_i(n_i+n_{i+1})+\sum_{i=1}^{L-1}k_ik_{i+1}.  
%\end{eqnarray}
%%where we used
%%\begin{eqnarray}
%%\sum_{i=1}^L\sum_{j=1}^{i-1} \sum_{a=1}^{n_j}(1- \delta^{l_{ij}^a}_0) k_i
%%-\sum_{i=1}^{L-1}\sum_{j=1}^i \sum_{a=1}^{n_j}(1- \delta^{l_{i+1,j}^a}_0)k_{i+1} 
%%=\sum_{i=1}^L\sum_{j=1}^{i-1} \sum_{a=1}^{n_j}(1- \delta^{l_{ij}^a}_0) k_i
%%-\sum_{i=2}^L\sum_{j=1}^{i-1} \sum_{a=1}^{n_j}(1- \delta^{l_{ij}^a}_0) k_i=0.
%%\end{eqnarray}
%}
%Actually, noting the constraints (\ref{eq:constraint}) as $0=f_A(\varphi) =f_A(\varphi_0)+g_{A x}(\varphi^x-\varphi_0^x)+{\cal O}((\varphi-\varphi_0)^2)$ for $1\le A\le d_{\rm c}$ 
% with  coordinates $\varphi^x$ around  arbitrary points $\varphi^x=\varphi^x_0$ satisfying the condition (\ref{eq:GLkifree}),
% a metric $g_{Ax}$ must have the maximal rank as  ${\rm rank}(g_{Ax})=d_{\rm c}$
%  as explained in the subsection \ref{sec:tW}.  Since the above coordinates keep the condition (\ref{eq:GLkifree}),  the constraints (\ref{eq:constraint})  can be, therefore, 
%solved with respect to appropriate $d_{\rm c}$ parameters within them, without causing any poles and brunch cuts,
%and that is, 
We can show that in the vicinity of the fixed point 
the linearized constraint can be solved without any singularity 
and hence a smooth coordinate patch can be constructed 
around each fixed point.
%The coordinate transformation 
%between two different patches is obtained 
%through the $GL(k_i,\mathbb C)$ transformations.
We can check that the complex dimension of 
the moduli space of vortices is given by
\begin{eqnarray}
 {\rm dim}_{\mathbb C} \, {\cal M}_{{\rm vtx}}{}^{n_1,n_2,\dots,n_{L+1}}_{~k_1,k_2,\cdots,k_L}
  &=&\sum_{i=1}^L (\# \mY_i+\# \wt \mY_i +\# Z_i -\# GL(k_i,\mathbb C))
 +\sum_{i=1}^{L-1} (\# \mW_i - \# \wt \mW_i )\nn
 &=& \sum_{i=1}^L(n_i k_i+k_i n_{i+1}+k_i^2-k_i^2)+\sum_{i=1}^{L-1} (k_i k_{i+1} -k_{i+1} k_i)\nn
 &=&\sum_{i=1}^Lk_i(n_i+n_{i+1}),
\end{eqnarray}
which agrees with a result given by the index theorem \eqref{eq:dim_M}.

%%%%%%%%%%%%%%%%%%%%%%%%%%%%%%%%%%%%%%%%%%%
\paragraph{Solutions of the constraints for separated vortices\\}
%\subsubsection{General solutions of the constraints for separated vortices}
Here we discuss solutions of the constraints (\ref{eq:constraint}). 
For a generic point on the moduli space, 
we can easily construct a solution in the following way.
Let us consider the case of separated vortices given by 
\begin{eqnarray}
z_{i,\alpha} \not = z_{i,\beta}\quad {\rm for~~} \alpha\not=\beta  \quad~~ {\rm with} ~~\quad 
\det(z{\bf 1}_{k_i}-Z_i) =\prod_{\alpha=1}^{k_i}(z-z_{i,\alpha})  % (Z_i)^\alpha{}_\beta =\delta^\alpha{}_\beta  \, z_{i,\alpha} 
 \quad  {\rm for~} 1 \le i \le L. 
\end{eqnarray}
In this case, the square matrices $\{Z_i \,|\, 1\le  i\le L\}$ 
can be diagonalized as 
$(Z_i)^\alpha{}_\beta =\delta^\alpha{}_\beta \, z_{i,\alpha}$.
In addition, if we assume that 
\begin{eqnarray}
z_{i,\alpha}\not=z_{i+1,\beta} \quad {\rm for} \quad 1\le \alpha \le k_i,\quad 1\le \beta \le k_{i+1} \quad {\rm and} \quad 1\le  i\le L-1,
\end{eqnarray}
we find the constraints are solved with respect to $\{W_i\}$ as
\begin{eqnarray}
(\mW_i)^\alpha{}_\beta= \frac{(\wt \mY_i \mY_{i+1})^{\alpha}{}_\beta}{z_{i,\alpha}-z_{i+1,\beta}}.
\end{eqnarray}
This result implies that all the components 
of the constraints (\ref{eq:constraint}) 
are independent and each of them has a solution.  
\subsection{Metric on the moduli space and K\"ahler quotient}\label{subsec:metric}
As we have seen that the vortex modui space is given by 
the $GL(k_1,\C) \times \cdots \times GL(k_L,\C)$ quotient \eqref{eq:moduli_space_genearl} of the matrices $(Z_i,\mY_i,\wt \mY_i,W_i)$ as a complex manifold. 
One may think that it is also possible to consider 
the corresponding $U(k_1) \times \cdots \times U(k_L)$  K\"ahler quotient by introducing an appropriate K\"ahler potential on the space of the matrices $(Z_i,\mY_i,\wt \mY_i,W_i)$. 
A natural choice of the K\"ahler potential would be
\beq
\mathcal K = \sum_{i=1}^L \tr \left[ Z_i Z_i^\dagger + \mY_i^\dagger \mY_i + \wt \mY_i \wt \mY_i^\dagger + W_i W_i^\dagger  \right],
\label{eq:naive_kahler}
\eeq
which gives a flat metric on the linear space of the matrices. 
In addition, we need to impose the constraint \eqref{eq:constraint}.
Following the standard procedure of the K\"ahler quotient construction, one can obtain a K\"ahler potential and metric on the moduli space. 
However, the K\"aher metric obtained in this way does not agree with the correct metric, shown in Appendix \ref{sec:Kahlermetric}, that describes the classical dynamics of the vortices. 
Nonetheless, the 2d $\mathcal N = (2,2)$ 
$U(k_1) \times \cdots \times U(k_L)$ gauge theory constructed based on the above K\"ahler potential $\mathcal K$ and the constraint \eqref{eq:constraint} captures some quantum aspects of the original $U(N_1) \times \cdots \times U(N_L)$ quiver gauge theory. 
In section \ref{sec:partition}, we compute the vortex partition function from the viewpoint of the quotient construction as an example of the quantities that do not depend on the detail of the K\"ahler potential.   

%%%%%%%%%%%%%%%%%%%%%%%%%%%%%%%%%%%%
\section{Sigma model instantons and duality} \label{sec:lumpduality}
\def\cA{{\cal A}} 
\def\cB{{\cal B}}

In this section, we discuss the sigma model solutions in the flag manifold sigma model.
We check that the duality of the sigma model (\ref{eq:dualFIpara}) defined by the relation (\ref{eq:dualG}) holds even on the moduli space of sigma model instantons, except for the instanton singularities.

\subsection%[Grassmannian case : L=1]
{Grassmannian case : $L=1$} 
 
In the Grassmannian case ($L=1$), 
the inhomogeneous coordinates $\phi$ is an $n$-by-$(N-n)$ matrix related to $\mH = (\mD, \wt \mD)$ as $\varphi = \mD^{-1} \wt \mD$.
Using the half-ADHM mapping relation Eq.\,\eqref{eq:hADHM}, 
we can show that the sigma model instanon solution $\varphi(z)$ corresponding to the half-ADHM data $\{Z,\Psi,\wt \Psi\}$ can be written as
\begin{eqnarray}
\varphi(z) ~=~ \Psi (z{\bf 1}_k-Z)^{-1} \wt \Psi. \label{eq:inst-sol}
\end{eqnarray}
For separated vortices, 
the marix $Z$ can be diagonalized as $(Z)^\alpha{}_\beta=\delta^\alpha{}_\beta\, z_\alpha$
with $z_\alpha\not=z_\beta~(\alpha \not =\beta)$
and hence the instanton solution takes the form 
\begin{eqnarray}
(\varphi(z))^a{}_b =\sum_{\alpha=1}^k \frac{(\Psi)^a{}_\alpha (\wt \Psi)^\alpha{}_b }{z-z_\alpha}. 
\end{eqnarray}
The column vectors $(\Psi)^a{}_\alpha$ and 
row vectors $(\wt \Psi)^\alpha{}_b$ 
for each $\alpha$ are respectively called 
the orientational moduli and size moduli 
of the vortex sitting at $z=z_\alpha$. 
Here we have partially fixed 
the $GL(k,\mathbb C)$ redundancy by 
diagonalizing the $k$-by-$k$ matrix $Z$. 
The remnant group which 
does not change the form of $Z$
is $(U(1)^\mathbb C)^k \subset GL(k,\C)$
and each $U(1)^\mathbb C \cong \C^\ast$ acts on 
the orientational moduli $\Psi^a{}_\alpha$. 
Due to the condition that the $GL(k,\C)$ action is free, $\Psi^a{}_\alpha$ cannot be a zero column vector for each $\alpha$ 
and hence the orientational moduli space of each vortex described by $\Psi^a{}_\alpha$ is a $\mathbb CP^{n-1} = (\C^n \backslash \{0\})/\sim$.   
On the contrary, the ``size" moduli $\wt \Psi^a{}_\alpha$ can be a zero vector, which corresponds to a local vortex when the gauge coupling constant is finite.
%and its typical real size is given by an inverse of a gauge boson mass, $1/g_1\sqrt{r_1}$. In the strong coupling limit, $g_1\to \infty$, such a local vortex becomes a singular lump configuration in the nonlinear sigma models. The points on the moduli space corresponding to such singular lumps are known as small-lump singularities of the instanton moduli space ${\cal M}_{\rm inst}$.

As shown in Eq.\,\eqref{eq:cond_rank} in Appendix \ref{sec:cond_semi},  
the condition that all vortices are of semi-local type can be rewritten in terms of the half-ADHM data as
\begin{align}
\exists \,\vec v \in \C^k ~\mbox{(row vector)}~ ~~\mbox{s.t.}~~
 \vec v \,(z {\mathbf 1}_k- Z)^{-1} \wt \Psi=0 ~~\mbox{for}~~ \forall z \in \mathbb C \quad \Rightarrow \quad \vec v=0. \label{eq:wtfree_cond}
\end{align}
This condition requires that the $GL(k,\mathbb C)$ acts freely not only on $\{Z,\Psi\}$, but also on $\{Z,\wt \Psi\}$. On the other hand, 
we can uniquely determine the half-ADHM data satisfying \eqref{eq:wtfree_cond} corresponding to any given sigma model instanton solution (See Appendix \ref{sec:inst-sol}).
Therefore the moduli space of  instanton solutions ${\cal M}_{\rm inst}$ are written in terms of the half-ADHM data as
\begin{align}
{\cal M}_{\rm inst} &\equiv \left\{ (Z,\Psi,\wt \Psi) \, \big| \,
\hbox{ $GL(k,\mathbb C)$ actions on $\{Z,\Psi\}, \{Z,\wt \Psi\}$ are free}  
\right \} / GL(k,\mathbb C). 
\end{align} 
This is a subspace of the total vortex moduli space ${\cal M}_{\rm vortex}$, for which the $GL(k,\mathbb C)$ free condition is imposed only $\{Z,\Psi\}$
\begin{align}
{\cal M}_{\rm vortex} &\equiv \left\{ (Z,\Psi,\wt \Psi) \, \big| \,
\hbox{ $GL(k,\mathbb C)$ action on $\{Z,\Psi\}$  is free}  
\right \} / GL(k,\mathbb C). 
\end{align} 
Note that, as we have discussed, 
there is a correspondence between 
sigma model instanton solutions 
in the dual pair of theories. 
The half-ADHM data 
$\{Z^{\rm dual},\Psi^{\rm dual}, \wt \Psi^{\rm dual}\}$ describing the dual 
sigma model instanton is given by
\begin{align}
\{Z^{\rm dual},\Psi^{\rm dual}, \wt \Psi^{\rm dual}\}=\{Z^{\rm T}, \wt \Psi^{\rm T},- \Psi^{\rm T}\},\label{eq:dual-origin}
\end{align}
up to $GL(k,\mathbb C)$ transformations. 
We can read off this relation from 
the duality transformation ${\cal U}^{\rm dual} =R \, {\cal U}^{\rm T-1} R^\dagger$ (see Eq.\eqref{eq:dual_U}), which maps a solution $\varphi(z)$ to a dual solution $\varphi^{\rm dual}(z)$ as
\begin{align}
\varphi^{\rm dual}(z)=-\varphi(z)^{\rm T}=-\wt \Psi^{\rm T} (z {\bf 1}_k-Z^{\rm T})^{-1} \Psi^{\rm T}.
\end{align}
The total moduli space of the dual vorties are given by 
\beq
{\cal M}_{\rm vortex}^{\rm dual}=
\left\{(Z^{\rm dual},\Psi^{\rm dual},\wt \Psi^{\rm dual}) \, \big| \, \hbox{ $GL(k,\mathbb C)$ action on $\{Z^{\rm dual},\Psi^{\rm dual}\}$ is free}  
\right \} / GL(k,\mathbb C).
%\simeq {\cal M}_{\rm vtx \,}{}^{N-n,n}_k.
\eeq
Here the $GL(k,\mathbb C)$-free condition on 
$\{Z^{\rm dual},\Psi^{\rm dual}\}$ is nothing but the condition \eqref{eq:wtfree_cond} through the relation \eqref{eq:dual-origin}. 
Therefore ${\cal M}_{\rm inst}$ is given as an 
intersection of the original vortex moduli space and the dual one as
\begin{align}
{\cal M}_{\rm inst} ={\cal M}_{\rm vortex}\cap {\cal M}^{\rm dual}_{\rm vortex}.
\end{align}

%
%{\bf (to be elaborated)} To remove the subspace containing all small-lump singularities form ${\cal M}_{\rm vortex}$ and obtain ${\cal M}_{\rm lump}$, 
%we have to impose that $GL(k,\mathbb C)$ acts freely not only on $\{Z,\Psi\}$ but also on $\{Z,\wt \Psi\}$ \cite{Eto:2007yv}
%\begin{eqnarray}
%{\cal M}_{\rm lump}={\cal M}_{\rm vortex}\cap {\cal M}^{\rm dual}_{\rm vortex},\quad {\rm with} \quad
%{\cal M}^{\rm dual}_{\rm vortex} = {\cal M}_{\rm vortex} \Big|_{\Psi \leftrightarrow \wt \Psi^{\rm T},  Z \leftrightarrow Z^{\rm T}}.
%\end{eqnarray}
\subsection%[General L]
{General $L$}

For general $L$, a well-defined sigma model instanton solution is 
given if and only if 
\begin{align}
\forall z:  \quad \det \mH_i(z) \mH_i(z)^\dagger \not=0,\quad i=1,2,\cdots,L.
\end{align}
Repeating the discussion in the case of $L=1$, we can show that the above condition for each $i$ is equivalent to the following two conditions on $\{Z_i,\Psi_i,\widetilde \Psi_i\}$ 
\begin{enumerate}
\item $GL(k_i,\mathbb C)$ action is free on $\{Z_i,\Psi_i\}$,
\item $GL(k_i,\mathbb C)$ action is free on $\{Z_i,\widetilde \Psi_i\}$.
\end{enumerate}
From the viewpoint of the original gauge theory, the first condition comes from the definition of the vortex moduli space and 
the second condition is imposed to avoid small-instanton singularities. 
On the other hand,
from the viewpoint of the dual gauge theory characterized by (\ref{eq:dualFIpara}),
the roles of the above two conditions are interchanged.

For general $L$, a solution is written in terms of
$n_i$-by-$n_j$ matrices $\varphi_{ij}(z)~(1\le i < j \le L+1)$ whose entries are inhomogeneous coordinates \eqref{eq:inhomogeneous} of $G^\mathbb C /\hat H$ 
\begin{eqnarray}
{\cal G}={\cal G}(\varphi_{ij}) \equiv \left( 
\begin{array}{ccccc}
{\bf 1}_{n_1} & \varphi_{12} & \varphi_{13}&\cdots&\varphi_{1,L+1}\\
{\bf 0} & {\bf 1}_{n_2}& \varphi_{23}\\
{\bf 0}&{\bf 0}& \ddots &\ddots& \vdots\\
\vdots&&\ddots&{\bf 1}_{n_{L}}&\varphi_{L,L+1}\\
{\bf 0}& \cdots &&{\bf 0}&{\bf 1}_{n_{L+1}}
\end{array}\right).
\label{eq:Gcal}
\end{eqnarray}
For a given half-ADHM data $\{Z_i,\mY_i, \wt \mY_i,\,\mW_j\}$,
the corresponding $\varphi_{ij}(z)$ are given by
\begin{eqnarray}
\varphi_{ij}(z)=\mY_i\left(z{\bf 1}-Z_i\right)^{-1} \mW_i\mW_{i+1}\cdots \mW_{j-2}\wt \mY_{j-1}\quad {\rm for~} i<j.
\label{eq:GenRmap}
\end{eqnarray}
This can be shown as follows. 
The matrices $\varphi_{ij}(z)~(i=i+1,\cdots,L+1)$ 
are contained in $\xi_i = \xi_i^p {\mathcal G}$
\begin{eqnarray}
\xi_i= \left( \begin{array}{ccc}
{\bf 1}_{N_{i-1}} & \cA_i & \wt \cA_i \\ {\bf 0} & {\bf 1}_{n_{i}} & \cB_i 
\end{array} \right) \sim 
\left( \begin{array}{ccc}
{\bf 1}_{N_l} & {\bf 0} & \wt \cA_i - \cA_i \cB_i \\ {\bf 0} & {\bf 1}_{n_{i+1}} & \cB_i
\end{array} \right)
\quad {\rm with} \quad \cB_i = (\varphi_{i,i+1},\varphi_{i,i+2},\cdots,\varphi_{i,L+1}). 
\label{eq:lthxi}
\end{eqnarray}
 It follows from 
 $\xi_i \sim (\mathbf 1_{N_{i+1}}, \mD_i^{-1} \mD_i) =(\mathbf 1_{N_{i+1}}, \Psi_{i}(z \mathbf 1 - Z_{i})^{-1} \wt \Psi_{i} )$ that 
\beq
\left( \begin{array}{c}
\wt \cA_i-\cA_i \cB_i  \\ \cB_i
\end{array} \right) 
= 
\left( \begin{array}{c}
\Psi_{i-1} W_{i-1} (z \mathbf 1 - Z_{i})^{-1} \wt \Psi_{i} \\
\mY_{i} (z \mathbf 1 - Z_{i})^{-1} \wt \Psi_{i}
\end{array} \right). 
\eeq
From the lower blocks of the both hand sides, we find that 
\beq
\cB_i = \mY_{i} (z{\bf 1}-Z_{i})^{-1} \wt \Psi_{i} = \left( \mY_{i} (z{\bf 1}-Z_{i})^{-1} \wt \mY_i \ , \ \cdots \ , \  \mY_{i} (z{\bf 1}-Z_{i})^{-1} \mW_i \mW_{i+1} \cdots \mW_{L-1} \wt \mY_{L} \right). 
\eeq
This indicates the relation \eqref{eq:GenRmap}.\footnote{
One can check the equation for the upper blocks is also satisfied. 
Since 
$\xi_{i-1} = p_{i-1} \xi_{i} \sim ({\bf 1}_{N_{i-1}}, \cA_i, \wt \cA_i)
= (\mathbf 1_{N_{i-1}}, \mD_{i-1}^{-1} \wt \mD_{i-1})
= (\mathbf 1_{N_{i-1}}, \Psi_{i-1} (z{\bf 1}-Z_{i-1})^{-1}\wt \Psi_{i-1})$, one finds that
\begin{eqnarray}
\cA_i = \Psi_{i-1} (z{\bf 1}-Z_{i-1})^{-1} \wt \mY_{i-1}, \hs{10}
\wt \cA_i = \Psi_{i-1} (z{\bf 1}-Z_{i-1})^{-1} W_{i-1} \wt \Psi_{i}. \notag
\end{eqnarray}
Then, one can show that
\beq
\wt \cA_i - \cA_i \cB_i &=& \wt \cA_i - \Psi_{i-1} (z{\bf 1}-Z_{i-1})^{-1} \wt \mY_{i-1} \mY_{i} (z{\bf 1}-Z_{i})^{-1} \wt \Psi_{i} \notag \\
&=& \wt \cA_i - \Psi_{i-1} (z{\bf 1}-Z_{i-1})^{-1} (Z_{i-1}W_{i-1}-W_{i-1}Z_i) (z{\bf 1}-Z_{i})^{-1} \wt \Psi_{i} ~=~ \Psi_{i-1} W_{i-1} (z{\bf 1}-Z_{i})^{-1} \wt \Psi_{i}. \notag
\eeq
}

Next, let us discuss 
how the half-ADHM data transform 
under the duality transformation. 
For given half-ADHM data $\{Z_i,\mY_i,\wt \mY_i,W_i\}$ 
which give a non-singular lump solution, 
there exist data $\{Z_i^{\rm dual},\mY_i^{\rm dual}, \wt \mY_i^{\rm dual}, W_i^{\rm dual} \}$ in the dual theory
which give the same lump solution. 
They are related as 
\begin{eqnarray}
\{Z_i^{\rm dual}, \, \mY_i^{\rm dual},\, \wt \mY^{\rm dual}_i,\,\mW_j^{\rm dual}\}
=\{Z_{L+1-i}^\T,\, \wt \mY_{L+1-i}^\T,\, -\mY_{L+1-i}^\T,\, \mW_{L-j}^\T  \}. \label{eq:duality}
\end{eqnarray}
This relation can be shown as follows. 
As shown in \eqref{eq:dualG}, for the matrix $\mathcal G \in G^\C$ in \eqref{eq:Gcal}, 
the corresponding matrix $\mathcal G_{\rm dual}$ in the dual theory is given by
\begin{eqnarray}
\quad {\cal G}_{\rm dual}=R ({\cal G}^{-1})^T R^\dagger \in G^\mathbb C \quad {\rm with~}\quad
R=\left(\begin{array}{cccc}
{\bf 0} & \cdots & {\bf 0}& {\bf 1}_{n_{L+1}}\\
\vdots& \iddots & \iddots & {\bf 0}\\
{\bf 0} &{\bf 1}_{n_2}&\iddots & \vdots \\
{\bf 1}_{n_1}&{\bf 0}& \cdots  &{\bf 0}
\end{array}\right) \in U(N).
\label{eq:g_dual5}
\end{eqnarray}
One can check that the inverse of ${\cal G}={\cal G}(\varphi_{ij}(z))$ takes the form ${\cal G}(\varphi^{\rm inv}_{ij}(z))$ with 
\begin{eqnarray}
\varphi_{ij}^{\rm inv}(z)=-\mY_i\mW_i\mW_{i+1}\cdots \mW_{j-2}\left(z{\bf 1}-Z_{j-1}\right)^{-1} \wt \mY_{j-1}\quad {\rm for~} i<j.
\end{eqnarray}
Here, one can confirm ${\cal G}(\varphi_{ij}(z)){\cal G}(\varphi^{\rm inv}_{ij}(z))={\bf 1}$ by using the following identity
\begin{eqnarray}
\sum_{k=i+1}^{j-1}\varphi_{ik}(z)\varphi_{kj}^{\rm inv}(z)
&=&-\sum_{k=i+1}^{j-1}\mY_i\left(z{\bf 1}-Z_i\right)^{-1} \mW_i\cdots \mW_{k-2}\wt \mY_{k-1}
\mY_k\mW_k\cdots \mW_{j-2}\left(z{\bf 1}-Z_{j-1}\right)^{-1} \wt \mY_{j-1}\nn
&=&-\sum_{k=i+1}^{j-1}\mY_i\left(z{\bf 1}-Z_i\right)^{-1} \mW_i\cdots \mW_{k-1}\left(z{\bf 1}-Z_{k}\right)W_k\cdots \mW_{j-2}\left(z{\bf 1}-Z_{j-1}\right)^{-1} \wt \mY_{j-1}\nn
&&{}+\sum_{k=i+1}^{j-1}\mY_i\left(z{\bf 1}-Z_i\right)^{-1} \mW_i\cdots \mW_{k-2}\left(z{\bf 1}-Z_{k-1}\right)W_{k-1}\cdots \mW_{j-2}\left(z{\bf 1}-Z_{j-1}\right)^{-1} \wt \mY_{j-1}\nn
&=&-\mY_i\left(z{\bf 1}-Z_i\right)^{-1} \mW_i\cdots \mW_{j-2} \wt \mY_{j-1}
+\mY_iW_{i}\cdots \mW_{j-2}\left(z{\bf 1}-Z_{j-1}\right)^{-1} \wt \mY_{j-1}\nn
&=&-\varphi_{ij}(z)-\varphi_{ij}^{\rm inv}(z).
\end{eqnarray}
%giving the duality  $G^\mathbb C/\hat H \mapsto G^\mathbb C/\hat H_{\rm dual} $
%where the equivalence relation is expressed as $\ {\cal G}_{\rm dual} \simeq  {\hat h}_{\rm dual} {\cal G}_{\rm dual} $ 
%with ${\hat h}_{\rm dual}\in \hat H_{\rm dual} =R ({\hat H}^\T)^{-1}R^\dagger $ and  the order  
%$(n_1,n_2,\cdots, n_{L+1})$ characterizing the complex structure is exchanged as 
% $(n_1^{\rm dual},n_2^{\rm dual},\cdots,n_{L+1}^{\rm dual})=(n_{L+1},n_L,\cdots, n_1)$.
By substituting the inverse ${\cal G}(\varphi^{\rm inv}_{ij}(z))$ into Eq.\,\eqref{eq:g_dual5}, 
we find that ${\cal G}_{\rm dual}$ takes the form ${\cal G}(\varphi_{ij}^{\rm dual}(z))$ with
\begin{eqnarray}
\varphi_{ij}^{\rm dual}(z)=-\wt \mY^\T_{L+1-i}(z{\bf 1} -Z_{L+1-i}^\T)^{-1}\mW_{L-i}^\T \cdots  \mW_{L+1-j}^\T \mY^\T_{L+2-j}, 
\end{eqnarray}
from which we can read the duality transformation 
\begin{eqnarray}
\{Z_i, \, \mY_i,\, \wt \mY_i,\,\mW_j\}\quad \mapsto \quad \{Z_i^{\rm dual}, \, \mY_i^{\rm dual},\, \wt \mY^{\rm dual}_i,\,\mW_j^{\rm dual}\}.
%=\{Z_{L+1-i}^\T,\, \wt \mY_{L+1-i}^\T,\, -\mY_{L+1-i}^\T,\, \mW_{L-j}^\T \}.
\label{eq:dual_hADHM}
\end{eqnarray}

Note that the $GL(k_i,\mathbb C)$ free condition on $Z_i,\wt \Psi_i$ for each $i$ 
\begin{equation}
\exists\, \hbox{$k_i$-column vector~} \vec v: \wt \mY_j^\T \mW_{j-1}^\T\mW_{j-2}^\T\cdots \mW_{i}^\T (z{\mathbf 1}_{k_i}-Z_i^\T)^{-1} \vec v =0  \quad {\rm for~} 
i \le \forall \,j \le L,  \quad \forall \, z\in \mathbb C 
~ \Rightarrow ~ \vec v=0, 
\label{eq:cond_dual}
\end{equation}
is equivalent to the $GL(k_i,\mathbb C)$ free condition on $Z_{L+1-i}^{\rm dual},\Psi^{\rm dual}_{L+1-i}$.
By denoting the conditions (\ref{eq:GLkifree}) and \eqref{eq:cond_dual} by ${\cal C}_i$ and $\wt {\cal C}_i$, respectively,
the condition for lump configuration without any small-lump singularity can be written as
\begin{eqnarray}
\bigwedge_i\left({\cal C}_i \wedge \wt {\cal C}_i\right)= \left(\bigwedge_i {\cal C}_i\right)\wedge  \left(\bigwedge_i \wt {\cal C}_i\right)
\end{eqnarray}
%\begin{eqnarray}
%{\cal M}^{\rm dual}_{{\rm vtx}}{}^{n_1,n_2,\dots,n_{L+1}}_{~k_1,k_2,\cdots,k_L} =
%{\cal M}_{{\rm vtx}}{}^{n_{L+1},\dots,n_2,n_1}_{~k_L,\cdots,k_2,k_1}
%\end{eqnarray}
and thus the moduli space for lump configurations is obtained as 
\begin{eqnarray}
{\cal M}_{{\rm lump}}{}^{n_1,n_2,\dots,n_{L+1}}_{~k_1,k_2,\cdots,k_L} =
{\cal M}_{{\rm vtx}}{}^{n_1,n_2,\dots,n_{L+1}}_{~k_1,k_2,\cdots,k_L} \cap {\cal M}_{{\rm vtx}}{}^{n_{L+1},\dots,n_2,n_1}_{~k_L,\cdots,k_2,k_1}.
\end{eqnarray}

%%%%%%%%%%%%%%%%%%%%%%%%%%%%%%%%%%%%%%%%%%%%%
\section{Vortex partition function from  K\"ahler quotient}\label{sec:partition}
\subsection{Vortex effective action}
In this section, we consider the vortex partition functions in the quiver GL$\sigma$Ms. 
In the $L=1$ case, the vortex partition functions have been calculated from the viewpoint of the half-ADHM formalism in \cite{Yoshida:2011au, Bonelli:2011fq}.
As an application of the half-ADHM formalism for general $L$, we compute the vortex partition function and check the duality between GL$\sigma$Ms \cite{Hori:2006dk} as was done in \cite{Benini:2012ui} for the $L=1$ case. 

In three dimensions, the dynamics of vortices can be 
described by the quantum mechanical GL$\sigma$M 
specified by the quiver diagram \eqref{eq:quiver}. 
Let us introduce chemical potentials for the conserved charges 
in the vortex effective theory
\beq
\mathcal Z_{{\footnotesize \mbox{hADHM-QM}}}(m_a,\epsilon,\mu_f) = \tr \left[ e^{-\beta ( H + i m_a q^a + i \epsilon J + i \mu_f F )} \right],
\label{eq:Z_3d}
\eeq
where 
$q^a~(a=1,\cdots,N)$ are the Cartan part of the flavor charge, $J$ is the angular momentum operator,
$F$ is the Fermion number operator
and 
$(m_a,\epsilon,\mu_f)$ are the (imaginary) chemical potentials 
for the corresponding charges.
It is well known that $Z$ at $\mu_f = \pi/\beta$, 
which we consider in the floowing, 
is exactly calculable through the supersymmetric localization method. Although it is possible to calculate $Z$ in three dimensions, for simplicity, we focus on the 2d limit $\beta \rightarrow 0$ in the following.

The vortex effective theory 
in the 2d quiver GL$\sigma$M is described 
by the 0d half-ADHM GL$\sigma$M specified by the quiver diagram Fig.\,\ref{fig:ADHM_quiver}. 
Using its action $S_{\rm eff}$, 
we can write down the integral expression 
for the 2d limit of the partition function \eqref{eq:Z_3d}
\beq
\mathcal Z_{k_1,k_2,\cdots,k_L}^{n_1,n_2,\cdots,n_{L+1}} = e^{-\sum_{i=1}^L 2\pi r_i k_i} \int d\mu \, \exp \left( - S_{\rm eff} \right),
\label{eq:vortex_partition_funciton_k}
\eeq
where $d\mu$ stands for the measure for all the degrees of freedom of the 0d half-ADHM GL$\sigma$M. 

Let us focus on the case in which 
the original model has 2d $\mathcal N=(2,2)$ supersymmetry. 
Since the vortices preserve the half of supersymmetry, 
the effective theory possesses two (real) supercharges. 
The supermultiplets in the vortex effective theory are 
chiral multiplets 
$(\varphi_I,\psi_I)$ whose scalar components are $\varphi_I \in (Z_i,\mY_i,\wt \mY_i,W_i)$, 
$U(k_i) \times U(k_{i+1})$ anti-bifundamental Fermi multiplets $(\tilde \psi_i, \wt W_i)$
and $U(k_i)$ adjoint gauge multiplets $(\Phi_i,\lambda_i,D_i)$. 
Their supersymmetry transformations are given by
\begin{alignat}{4}
\delta \varphi_I &= \varepsilon \psi_I, & 
\delta \tilde \psi_i &= \varepsilon \wt W_i, & 
\delta \Phi_i &= \delta D_i = 0, & 
\delta \lambda_i &= \varepsilon D_i  
\\
\delta \psi_I &= \bar \varepsilon \Delta \varphi_I, \hs{10} &
\delta \wt W_i &= \bar \varepsilon \Delta \tilde \psi_i, \hs{10} & 
\delta \bar \Phi_i &= i( \bar \varepsilon \lambda_i - \varepsilon \bar \lambda_i), \hs{10} & 
\delta \bar \lambda_i &= \bar \varepsilon D_i,
\end{alignat}
where $\Delta \varphi_I$ deenote the infinitesimal transformation of the $U (k_1)\times \cdots \times U(k_L)$ gauge, spatial and flavor rotations
\beq
&\Delta Z_i = [\Phi_i,Z_i] + \epsilon Z_i , \quad
\Delta \mY_i = M_i \mY_i - \mY_i \Phi_i, \quad
\Delta \wt \mY_i = \Phi_i \wt \mY_i - \wt \mY_i M_{i+1} + \epsilon \wt \mY_i,& \phantom{\bigg|} \\
&\Delta W_i = \Phi_i W_i - W_i \Phi_{i+1}, \quad 
\Delta \wt W_i = \Phi_{i+1} \wt W_i - \wt W_i \Phi_i - \epsilon \wt W_i,& \phantom{\bigg|}
\eeq
where $M_i={\rm diag}(m^{(i,1)},\cdots,m^{(i,n_i)})$ and $\epsilon$ are the parameters
corresponding to 
the twisted masses and 
the omega deformation parameter
in the original 2d $\mathcal N=(2,2)$ model.
If we adopt the naive K\"ahler potential \eqref{eq:naive_kahler}, 
the explicit form effective action is given by
\beq
S_{\rm eff} &=& \sum_{i=1}^L \tr \Bigg[ \big|\!\big| \Delta Z_i \big|\!\big|^2 + \big|\!\big| \Delta \mY_i \big|\!\big|^2 + \big|\!\big| \Delta \wt \mY_i \big|\!\big|^2 + \big|\!\big| \delta W_i \big|\!\big|^2 + \wt W_i \wt W_i^\dagger + \left\{ \wt W_i ( Z_i\mW_i-\mW_i Z_{i+1} - \wt \mY_i \mY_{i+1}) + (c.c) \right\} \Bigg] \notag \\
&+& \sum_{i=1}^L \tr\left[ D_i \left( [Z_i,Z_i^\dagger] - \mY_i^\dagger \mY_i + \wt \mY_i \wt \mY_i^\dagger + W_i W_i^\dagger - W_{i-1}^\dagger W_{i-1} + \frac{4\pi}{g_i^2} \right) \right] + (\mbox{fermionic terms}), \phantom{\bigg[}
\label{eq:vortex_effective_action}
\eeq
where $\tr \big|\!\big| \Delta A_I \big|\!\big|^2$ 
stands for the norms of the infinitesimal transformations 
\beq
\tr \big|\!\big| \Delta Z_i \big|\!\big|^2 = \tr \left[\Delta Z_i (\Delta Z_I)^\dagger + \Delta Z_i^\dagger (\Delta Z_i^\dagger)^\dagger) \right], ~~~ etc.
\eeq
The FI parameters $\frac{4\pi}{g^2}$ are chosen so that
the areas of the two cycles in the vortex moduli space 
agree with that calculated from the 2d perspective.  
Eliminating the auxiliary fields $\wt W_i$, 
we obtain a potential whose minimization condition gives the constraints \eqref{eq:constraint}
\beq
\wt W_i = Z_i W_i - W_i Z_{i+1} - \wt \mY_i \mY_{i+1} = 0.
\eeq
The variations with respect to $D_i$ give the $D$-term constraints, whose set of solution modulo gauge transformations agrees with the vortex moduli space \eqref{eq:moduli_space_genearl}. 
In the presence of the twisted masses and the omega deformation parameters, the conditions $\big|\!\big| \delta A_I \big|\!\big|^2=0$ 
are satisfied only at the fixed points of the torus action 
discussed in subsection \ref{subsec:torus_action}. 
When we apply the supersymmetric localization method, 
the integral \eqref{eq:vortex_partition_funciton_k} 
localizes to those fixed points.

\subsection{Contour integral for vortex partition function}
Although this explicit effective action \eqref{eq:vortex_effective_action} 
does not give a correct moduli space metric, 
it gives the correct vortex partition function 
since the deformation of the K\"ahler potential is a $Q$-exact deformation, which do not change the vortex partition function thanks to the supersymmetric localization. 

%After eliminating the gauge degrees of freedom and %solving the constraints, 
%the effective action is described by 
%the bosonic moduli parameters $\varphi^A$ and 
%the partner fermionic moduli $\psi^A$.
%In the presence of the omega-background $\epsilon$ and 
%twisted masses $m^a~(a=1,\cdots,N)$, 
%the contribution to the partition function 
%from each topological sector takes the form
%\beq
 %\mathcal Z_{k_1,k_2,\cdots,k_L}^{n_1,n_2,\cdots,n_{L+1}} = e^{-\sum_{i=1}^L 2\pi r_i k_i}  \int_{\mathcal M} d^n \varphi \, d^n \psi \, d^n \bar \varphi \, d^n \bar \psi \, \exp \left( - S_{\rm eff} \right), 
%\label{eq:def_Z}
%\eeq
%where $S_{\rm eff}$ is 
%the effective potential induced on the moduli space,
%which is determined by the holomorphic Killing vector %$\xi^A$ 
%for the torus action of the spatial rotation and the %flavor rotation
%\beq
%S_{\rm eff} = g_{A \bar B} \left( \xi^A \bar \xi^B + %\nabla_C \xi^A \psi^C \bar \psi^B \right). 
%\eeq
It is well known that the integral \eqref{eq:vortex_partition_funciton_k} can be evaluated 
by using the localization formula, 
which relates the partition function 
to the weights of the torus action at the fixed points
\beq
\mathcal Z_{k_1,k_2,\cdots,k_L}^{n_1,n_2,\cdots,n_{L+1}} \ = \ \left( \, \prod_{i=1}^L \Lambda_i^{\beta_i k_i} \right) \sum_{s \in \mathfrak S} \frac{1}{\det \mathcal M_s}, 
\label{eq:PF_wegiht} 
\eeq
where $\mathfrak S$ is the set of the fixed points and 
$\mathcal M_s$ is the generator of the torus action at the fixed point $s$ discussed in subsection \ref{subsec:torus_action}. 
Explicitly, it can be read off from the contour integral (see Appendix \ref{appendix:VP})
\beq
\mathcal Z_{k_1,k_2,\cdots,k_L}^{n_1,n_2,\cdots,n_{L+1}} \ = \ \prod_{i=1}^L \left[ \frac{1}{k_i!} \, \Lambda_i ^{\beta_{0i} k_i} \prod_{r=1}^{k_i}\oint_{C_i^+} \frac{d \phi_i^r}{2\pi i \epsilon} \right] 
 \left[\prod_{i=1}^L {\cal Z}_i^{\mY \wt \mY} {\cal Z}_i^{Z \Phi} \prod_{i=1}^{L-1} {\cal Z}_i^{W \wt W} \right],
 \label{eq:topopartition}
\eeq
where ${\cal Z}_i^{\mY\wt \mY}$, ${\cal Z}_i^{Z \Phi}$ and $ {\cal Z}_i^{W \wt W}$ are given by
\begin{eqnarray}
{\cal Z}_i^{\mY \wt \mY} &\equiv&
\prod_{r=1}^{k_i} \left[\prod_{\alpha=1}^{n_i} \frac{1}{\phi_i^r-m^{(i,\alpha)}}\prod_{\beta=1}^{n_{i+1}}\frac{1}{m^{(i+1,\beta)}-\phi_i^r-\epsilon} \right], \label{eq:Z_UUt} \\
{\cal Z}_i^{Z \Phi} &\equiv& \prod_{r=1}^{k_i} \prod_{s=1}^{k_i} \hs{-1} {\phantom{\bigg|}}' \frac{\phi_i^r-\phi_i^s}{\phi_i^r-\phi_i^s-\epsilon}, \\
{\cal Z}_i^{W \wt W} &\equiv& 
\prod_{r=1}^{k_i} \prod_{s=1}^{k_{i+1}}\frac{\phi_{i+1}^s-\phi_i^r-\epsilon}{\phi_{i+1}^s-\phi_i^r}.
\end{eqnarray}
where $\prod'$ indicates that the factors with $s=r$ are omitted from the product. 
\begin{figure}
\centering
\fbox{
\includegraphics[width=82mm, bb = 40 20 800 585]{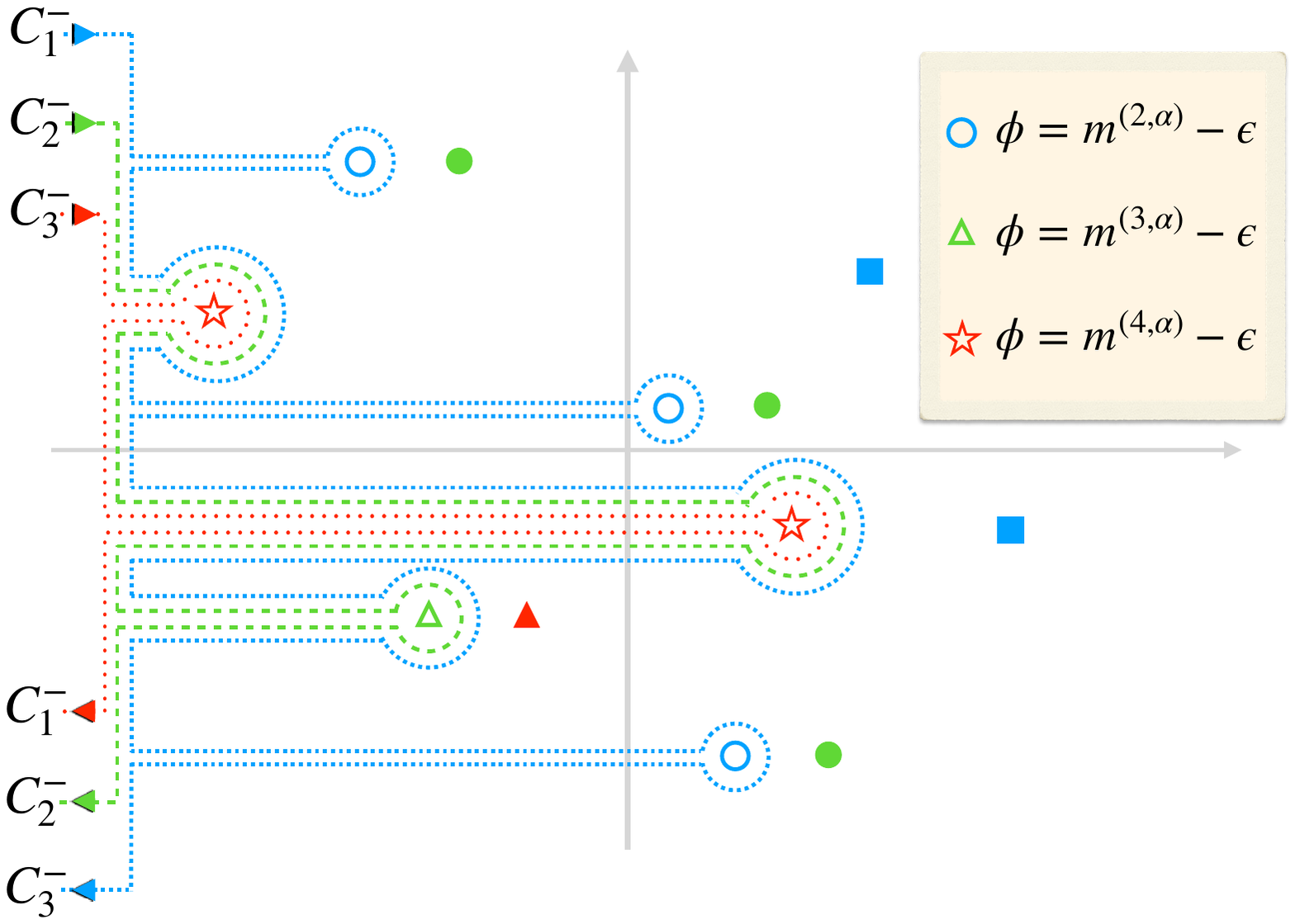}}
\fbox{
\includegraphics[width=82mm, bb = 40 20 800 585]{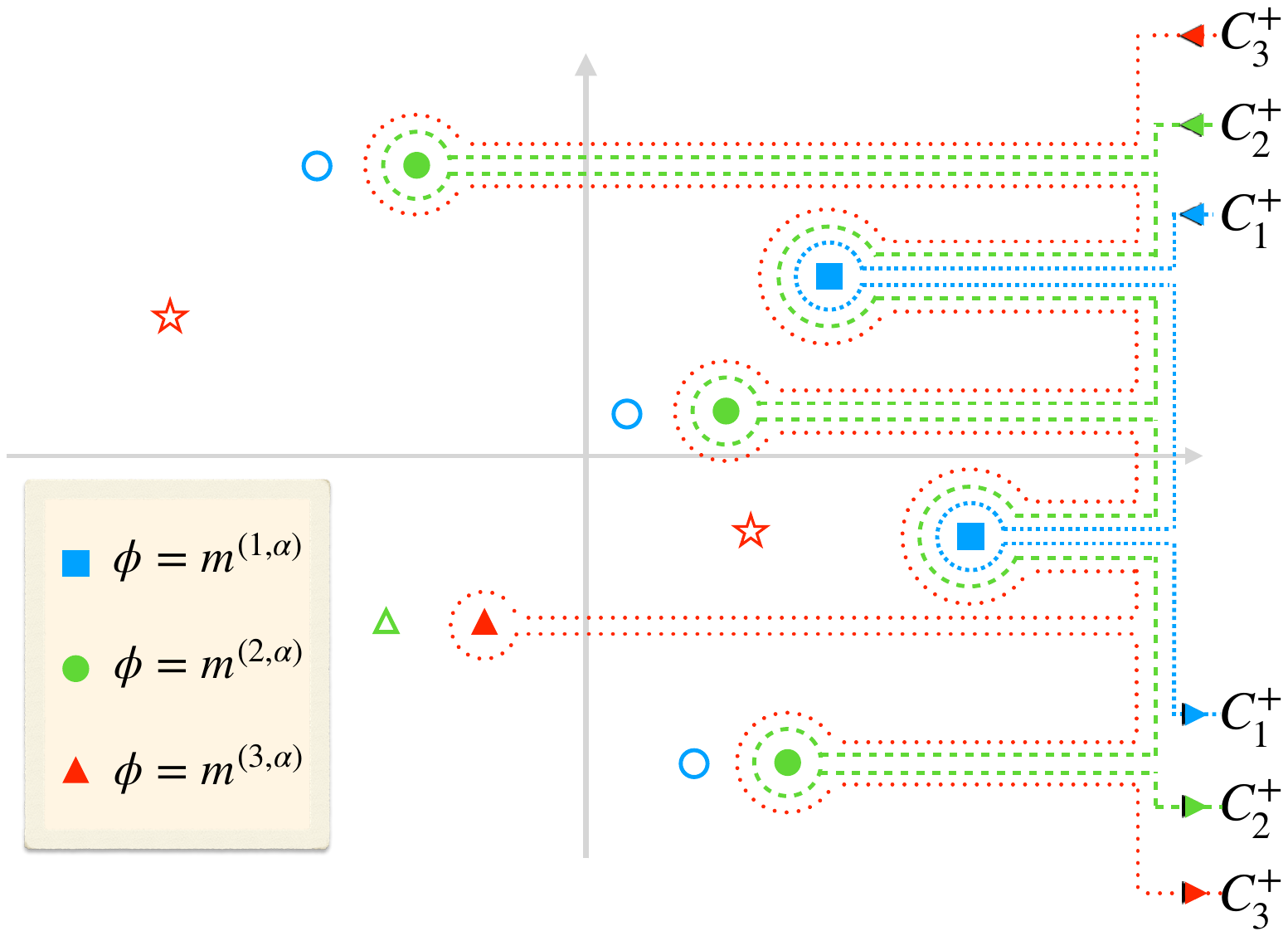}}
\caption{Integration contours. 
For all $r=1,\cdots,N_i$, the integration contours for $\phi_i^r$ are identical path $C_i^+$ (right panel) on the complex $\phi$ plane. Each path $C_i^+$ can be freely deformed as long as it does not cross over the other paths and the poles at $\phi=m^{(i,\alpha)}$ and $\phi=m^{(i+1,\alpha)}-\epsilon$. 
Following this rule, one can deform 
the paths $C_i^+$ to $C_i^-$ (left panel) 
if there is no pole at the infinity.}
\label{fig:contours_duality}
\end{figure}
The integration contours $C_i^+$ are paths 
which are determined from the Jeffrey-Kirwan residue formula to pick up the poles corresponding to the fixed points. 
Explicitly, $C_i^+$ are contours on the complex $\phi$ plane 
starting at a point $\phi = + \infty$, 
surrounding $\phi=m^{(j,\alpha)}~(j=1,\cdots,i)$ 
and $C_j^+~(j=1,\cdots,i-1)$,  
and ending at another point $\phi = + \infty$ 
(see the right panel of Fig.\,\ref{fig:contours_duality}). 
As shown in Appendix \ref{appendix:VP}, 
The contour integral \eqref{eq:topopartition} is given by
the sum of residues at the poles classified by sets of Young tableaux. For set of a Young tableaux 
\beq
{\bf Y} = \left\{Y^{(j,\alpha)} = \left( l_j^{(j,\alpha)},l_{j+1}^{(j,\alpha)},\cdots,l_L^{(j,\alpha)} \right) \right\},
\eeq
the corresponding pole is located at
\beq
\phi_{i}^{(j,\alpha,p)} = m^{(j,\alpha)} + (p-1) \epsilon, \hs{10}
(i=1,\cdots, L,~j=1,\cdots,i,~\alpha=1,\cdots,n_j,~p=1,\cdots,l_i^{(j,\alpha)}),
\eeq
where we have relabeled $\phi_i^r$ as $\phi_i^{(j,\alpha,p)}$. 
Performing the contour integral and evaluating the residues, 
we obtain the partition function of the form
\beq
{\cal Z}^{n_1,n_2,\dots,n_{L+1}}_{~k_1,k_2,\cdots,k_L} &=& \left( \, \prod_{i=1}^L \Lambda_i^{\beta_i k_i} \right) \sum_{{\bf Y}} Z_{\mY }^{\bf Y} Z_{\wt \mY}^{\bf Y} Z_{Z \Phi}^{\bf Y} Z_{W \wt W}^{\bf Y},
\label{eq:Z_residue}
\eeq
where $Z_{\mY}^{\bf Y}$, $Z_{\wt \mY}^{\bf Y}$, $Z_{Z \Phi}^{\bf Y}$ and $Z_{W \wt W}^{\bf Y}$ are contributions from $\mY$, $\wt \mY$, $(Z,\Phi)$ and $(W,\wt W)$, which are respectively given by
\beq
Z_{\mY}^{\bf Y} &=& \prod_{i=1}^L \prod_{\alpha=1}^{n_i} \prod_{k=1}^{i} \prod_{\beta=1}^{n_k} \prod_{q=1}^{l_i^{(k,\beta)}} {\hs{-1} \phantom{\bigg|}}' \frac{1}{m^{(k,\beta)}-m^{(i,\alpha)}+(q-1) \epsilon}, \\
Z_{\wt \mY}^{\bf Y} &=& \prod_{i=1}^L \prod_{j=1}^i \prod_{\alpha=1}^{n_j} \prod_{p=1}^{l_i^{(j,\alpha)}} \prod_{\beta=1}^{n_{i+1}} {\hs{-1}  \phantom{\bigg|}}' \frac{1}{ m^{(i+1,\beta)} - m^{(j,\alpha)} - p \epsilon}, \\
Z_{Z\Phi}^{\bf Y} &=& \prod_{i=1}^L \prod_{j=1}^i \prod_{\alpha=1}^{n_j} \prod_{p=1}^{l_i^{(j,\alpha)}} \prod_{k=1}^i \prod_{\beta=1}^{n_k} \prod_{q=1}^{l_i^{(k,\beta)}} {\hs{-1} \phantom{\bigg|}}' \frac{m^{(k,\beta)} - m^{(j,\alpha)} - (p-q) \epsilon}{m^{(k,\beta)} - m^{(j,\alpha)} - (p-q+1) \epsilon}, \\
Z_{W \wt W}^{\bf Y} &=& \prod_{i=1}^{L-1} \prod_{j=1}^i \prod_{\alpha=1}^{n_j} \prod_{p=1}^{l_i^{(j,\alpha)}} \prod_{k=1}^{i+1} \prod_{\beta=1}^{n_k} \prod_{q=1}^{l_{i+1}^{(k,\beta)}} {\hs{-1} \phantom{\bigg|}}'  \frac{m^{(k,\beta)} - m^{(j,\alpha)} - (p-q+1) \epsilon}{m^{(k,\beta)} - m^{(j,\alpha)} - (p-q) \epsilon},
\eeq
where $\prod'$ indicates that the vanishing factors in the denominator and numerator are omitted. 
We can show that the partition function \eqref{eq:Z_residue} can be rewritten as 
\beq
{\cal Z}^{n_1,n_2,\dots,n_{L+1}}_{~k_1,k_2,\cdots,k_L} = \left( \, \prod_{i=1}^L \Lambda_i^{\beta_i k_i} \right) \sum_{\bf Y} \frac{A^{\bf Y}}{B^{\bf Y}}, 
\label{eq:Z_reduced}
\eeq
with
\beq
A^{\bf Y} &=& \prod_{i=1}^L \prod_{j=1}^i \prod_{\alpha=1}^{n_j} \prod_{k=1}^i \prod_{\beta=1}^{n_k} (-\epsilon)^{l_i^{(j,\alpha)}-l_i^{(k,\beta)}} \left( \frac{m^{(j,\alpha)}-m^{(k,\beta)}}{\epsilon} + 1 \right)_{l_i^{(j,\alpha)}-l_i^{(k,\beta)}}, \\
B^{\bf Y} &=& \prod_{i=1}^L \prod_{j=1}^i \prod_{\alpha=1}^{n_j} \prod_{k=1}^{i+1} \prod_{\beta=1}^{n_k} (-\epsilon)^{l_i^{(j,\alpha)}-l_i^{(k,\beta)}} \left( \frac{m^{(j,\alpha)}-m^{(k,\beta)}}{\epsilon} + 1 \right)_{l_i^{(j,\alpha)}-l_{i+1}^{(k,\beta)}},
\eeq
where $(a)_b$ denotes the Pochhammer symbol 
\beq
(a)_b = \frac{\Gamma(a+b)}{\Gamma(a)}.
\eeq
We can confirm that the result \eqref{eq:Z_reduced}
of the contour integral for 
${\cal Z}^{n_1,n_2,\dots,n_{L+1}}_{~k_1,k_2,\cdots,k_L}$ is proportional to the residue of the integrand for the total vortex partition function in the 2d $\mathcal N=(2,2)$ theory 
\beq
Z = \int \prod_{i=1}^L \prod_{a=1}^{N_i} \left[  \frac{d \sigma_i^a}{2\pi i \epsilon} \, \exp \left( - \frac{2\pi i \sigma_i^a \tau_i}{\epsilon} \right) {\prod_{b=1}^{N_i}} {\hs{-1} \phantom{\bigg|}}' \Gamma \left( \frac{\sigma_i^a -\sigma_i^b}{\epsilon} \right)^{-1} \prod_{c=1}^{N_{i+1}}\Gamma \left(\frac{\sigma_i^a - \sigma_{i+1}^{c}}{\epsilon} \right) \right], 
\label{eq:2d_VP}
\eeq
at the pole 
\beq
\sigma_i^{(j,\alpha)} = - m^{(j,\alpha)} - l_i^{(j,\alpha)}\epsilon.
\eeq
This is consistent with the fact that the total vortex partition function $Z$ can be written as a sum of the contributions from each topological sectors
\beq
Z = {\cal Z}^{n_1,n_2,\dots,n_{L+1}}_{\rm 1-loop} \sum_{k_1=0}^\infty \cdots \sum_{k_L=0}^\infty \left( \, \prod_{i=1}^L \Lambda_i^{\beta_i k_i} \right) {\cal Z}^{n_1,n_2,\dots,n_{L+1}}_{~k_1,k_2,\cdots,k_L}. 
\eeq
Thus, we can confirm that the K\"ahler quotient construction 
gives the correct information on the vortex moduli space and the non-perturbative effects in the parent 2d $\mathcal N=(2,2)$ theory.

\subsection{Duality and partition function}
One of the advantages of using 
the K\"ahler quotient construction is 
that it makes the duality manifest 
in the contour integral expression \eqref{eq:topopartition}. 
We can show that the partition function 
agrees with that of the dual theory as follows.
In the dual theory, the effective vortex action is described 
by the same action as \eqref{eq:vortex_effective_action}
with the duality map of the degrees of freedom \eqref{eq:dual_hADHM} and 
the sign flip of the FI parameters 
\beq
\frac{4\pi}{g_i^2} \rightarrow - \frac{4\pi}{g_i^2}.
\eeq
As the Jeffrey-Kirwan residue formula implies, 
the relevant poles become those enclosed by the contour $C_i^-$ in the left panel of Fig.\,\ref{fig:contours_duality} due to the sign flip of the FI parameters. 
We can see from \eqref{eq:Z_UUt} 
that if $n_i \not =0 ~ (N_i \not = N_{i+1})$ for all $i$, 
the integrand in \eqref{eq:topopartition} 
has no pole at the infinity,  
and hence we can change the contour of integration 
from $C_{i}^+$ to $C_{i}^-$ without changing 
the result of the integration 
(see the left panel of Fig.\,\ref{fig:contours_duality}). 
The residues for the contours $C_i^+$ and $C_i^-$ 
give the vortex partition functions 
in the original and dual theory, respectively, 
and hence they are equivalent.
More precisely, it is easy to check 
by changing the variables as $\phi_i^r \rightarrow - \tilde \phi_{L+1-i}^r$
that 
\begin{eqnarray}
{\cal Z}^{n_1,n_2,\dots,n_{L+1}}_{~k_1,k_2,\dots,k_L}(M_1,M_2,\dots,M_{L+1},\epsilon,\Lambda_i)
 = {\cal Z}^{n_{L+1},n_L\dots,n_2,n_1}_{~k_L,k_{L-1},\dots,k_1}(\epsilon-M_{L+1},\dots,\epsilon-M_2,\epsilon-M_1,\epsilon,\Lambda_{L+1-i}),
\end{eqnarray} 
and thus the duality of the NL$\sigma$M holds also for 
the vortex partition functions. 

\section{Summary and discussion}\label{sec:summary}
In this paper, we analyzed the moduli spaces of the following topological solitons:
\begin{enumerate}
\item 
BPS vortices in the $U(N_1) \times \cdots \times U(N_L)$ GL$\sigma$Ms characterized 
by the linear quiver \eqref{eq:quiver}.
\item 
sigma model instantons (lumps) in the K\"ahler flag manifold sigma models,
which can be obtained in the large gauge coupling limit of the GL$\sigma$Ms.
\end{enumerate}
In these theories, vortices and sigma model instantons 
carrying multiple topological charges $\{ k_i \} \in \mathbb Z^L$ appear.
We analyzed BPS equations for vortices, extracted the data set of the vortex moduli through the moduli matrix method,
and exactly determined the moduli spaces of 1/2 BPS vortices in the GL$\sigma$Ms as shown in Eq.\,\eqref{eq:vtx_moduli}.
We also discussed how to obtain general exact instanton solutions in the K\"ahler flag manifold sigma models.

We showed that the moduli matrix method in the $L=1$ case
can be recast into the ADHM-like quotient construction (the half-ADHM construction) of the moduli space, which has been originally derived from the view point of the D-brane construction in the string theory \cite{Hanany:2003hp}.
Generalizing this technique to the case of general $L$, 
we constructed a quotient construction for the vortex moduli space Eq.\,\eqref{eq:moduli_space_genearl} 
specified by the quiver diagram shown in Fig.\,\ref{fig:ADHM_quiver}.
In this half-ADHM formalism, various features of the moduli space are manifest. 
In particular, we have observed that the duality relations are expressed in the simple form \eqref{eq:duality},
which is realized by reversing all the arrows in Fig.\,\ref{fig:ADHM_quiver}. 
%such that this duality at classical level is manifest in instanton solutions \eqref{eq:inst-sol} and \eqref{eq:GenRmap} in terms of the half-ADHM data. 

As an application, we have used the quotient construction of the vortex moduli space 
to calculate the vortex partition functions in 2d ${\cal N}=(2,2)$ GL$\sigma$Ms. 
Applying the localization formula to the half-ADHM system, 
we have computed the vortex partition from the data of the fixed points of 
the torus action acting on the vortex moduli space. 
From the viewpoint of the vortex partition functions, 
we have confirmed dualities between pairs of quiver gauge theories whose vacuum moduli spaces are identical flag manifolds. 
We have found that the partition functions agree even though 
the structures of the vortex moduli spaces, in particular, 
the fixed point structures, are different between the dual pairs of GL$\sigma$Ms.
The half-ADHM formalism has 
turned out to make the duality manifest even at the quantum level.

One of the future directions is to study the low energy dynamics of vortices and lumps, 
which are described by geodesic motions on the moduli space equipped with a metric. In the case of $L=1$ (Grassmannian sigma models) and local non-Abelian vortices, 
the moduli space metrics have been determined for well-separated vortices \cite{Fujimori:2010fk} 
and low energy dynamics have been studied \cite{Eto:2011pj}. 
Extending these analyses to the quiver gauge theories and the flag manifold NL$\sigma$Ms discussed in this paper would be interesting.

In this paper, we have extended the the half-ADHM formalism for the moduli space of vortices to the gauged linear sigma models characterized by the linear quiver \eqref{eq:quiver}. 
One of the future problems is to investigate to what extent the half-ADHM formalism can be generalized to the vortices 
in arbitrary gauge theories.
We have studied flag manifold NL$\sigma$Ms $G/H$ 
with $G=U(N)$. Extending our work to other groups $G$ 
such as $SO(N)$, $USp(2N)$ and exceptional groups is
an important future work. For $L=1$, such isometry $G$ can be realized by imposing holomorphic constraints (superpotentials in supersymmetric cases)  \cite{Higashijima:1999ki}.

Replacing the base space with a space of different geometry and topology would broaden the application of the half-ADHM formalism. 
However, such a replacement in the base space could drastically change the half-ADHM formalism.
In particular, all the proofs in Appendices \ref{sec:ZPsiPatches} and \ref{sec:non-singular} must 
be reconstructed from scratch. Although research in this direction is challenging, it would deepen our understanding of vortices.

In both theories of the dual pairs discussed in this paper, the quiver diagrams consist of linear chains with all arrows pointing in the same direction. 
According to \cite{Benini:2014mia}, the cluster algebra on quiver diagrams produces more dual pairs of theories. That is, the present theory should be dual to various theories characterized by complicated quiver diagrams,
such as one involving chain loops or arrows in the opposite direction.
It would be interesting to extract the half-ADHM data from the moduli matrices for such models and 
see how those dualities are expressed at both the classical and quantum levels.

Another important future direction is to study sigma model lumps (instantons) in non-\kahler flag NL$\sigma$Ms. 
In this paper, we have studied only BPS lumps (instantons) in \kahler flag NL$\sigma$Ms, 
where there are no forces among lumps, thus admitting the moduli space.
On the other hand, lumps in non-\kahler flag NL$\sigma$Ms are non-BPS; hence there are forces among them. Interaction between non-BPS sigma model instantons would be important when we discuss the non-perturbative aspects of the sigma models from the viewpoint of instantons.

\section*{Acknowledgement}
This work is supported by the Ministry of Education, Culture, Sports, Science, and
Technology(MEXT)-Supported Program for the Strategic Research Foundation at Private Universities
“Topological Science” (Grant No. S1511006) and by the Japan Society for the Promotion of Science
(JSPS) Grant-in-Aid for Scientific Research (KAKENHI) Grant Number (18H01217). This work is also
supported in part by JSPS KAKENHI Grant Numbers JP21K03558 (T. F.) and JP22H01221 (M. N.).

\appendix
\newpage
\section{Riemannian manifolds and \kahler manifolds}
\label{appendix:Riemann}

\def\mA{{\mathcal A}}  
%%%%%%%%%%%%%%%%%%%%%%%%%%%%%%%%%%%%%%%%%%%%%%%%%
\subsection{General Riemannian metric for flag manifolds} \label{sec:appA}
The homogeneous Riemannian metric of the generalized flag manifolds 
$G/H=U(N)/[U(n_1)\times \cdots \times U(n_{L+1})]$ 
has $L(L+1)/2$ parameters (decay constants)
and it becomes a \kahler metric on a $L$ dimensional subspace parameterized by the FI parameters.  
Here we explain the relations between various expressions for the flag manifold sigma models. 
Using $n_i$-by-$N$ matrix valued fields $v_i~(i=1,\cdots,L,L+1)$,  
the flag manifold sigma model can be given in the following form 
\begin{eqnarray}
{\cal L}=\frac12 \sum_{i,j=1}^{L+1} f_{ij}\tr \Big[({\cal D}_\mu v_i v_j^\dagger) ({\cal D}_\mu v_i v_j^\dagger) ^\dagger \Big]
+ \sum_{i,j=1}^{L+1} \tr \Big[ \lambda_{ij}(v_i v_i^\dagger-\delta_{ij}{\bf 1}_{n_i}) \Big],
\end{eqnarray}
where $f_{ij} = f_{ji} > 0$ are coupling constants and 
the covariant derivative on $v_i$ is defined as 
${\cal D}_\mu v_i=\partial_\mu v_i+ia_\mu^i v_i$ with $U(n_i)$ gauge fields $a_\mu^i$.
%We set diagonal elements of $f_{ij}$ to be zero as $f_{ii}=0$ to 
% fix redundancies of $f_{ij}\to f_{ij}+ u_i+ u_j$ with some vector $u_i$.
The $n_i$-by-$n_j$ matrices $\lambda_{ij}$ are
Lagrange multipliers which gives the constraints
\begin{eqnarray}
v_i v_i^\dagger=\delta_{ij}{\bf 1}_{n_i} \quad \Leftrightarrow \quad  
U^\dagger=(v_1^\dagger,v_2^\dagger,\cdots,v_{L+1}^\dagger) \in U(N)
\quad {\rm with} \quad N=\sum_{i=1}^{L+1} n_i,
\end{eqnarray}
and thus the target manifold becomes the generalized flag manifold
when the auxiliary gauge fields are eliminated
\begin{eqnarray}
a_\mu^i=i\partial_\mu v_i v_i^\dagger. 
%\quad  f_{\mu\nu}(a^i)=-i[{\cal D}_\mu, {\cal D}_{\nu}]=-i\partial_\mu v_i({\bf 1}-v_i^\dagger v_i) \partial_\nu v_i^\dagger-(\mu \leftrightarrow \nu)
\end{eqnarray}
Substituting these into the Lagrangian, we obtain
\begin{eqnarray}
{\cal L}= - \frac{1}{2} \sum_{i=1}^L \sum_{j=i+1}^{L+1} f_{ij}\tr[ \partial_\mu P_i \partial^\mu P_j],
\label{eq:L_Riemann}
\end{eqnarray}
where $P_i$ are projection operators 
\beq
P_i\equiv v_i^\dagger v_i, \hs{10} (P_i P_j = \delta_{ij} P_i).
\eeq
Note that the terms proportional to diagonal elements of $f_{ij}$ are introduced only for stability of the auxiliary fields
and disappear in the above Lagrangian since ${\cal D}_\mu v_i v_i^\dagger=0$. 

For some purposes, it might be convenient to express the sigma model with quadratic kinetic terms,
which reduces to certain special cases of the above model
\begin{eqnarray}
\sum_i f_i \tr \Big[ ({\cal D}_\mu v_i) ({\cal D}^\mu v_i)^\dagger \Big]=
\sum_{i,j} f_i \tr \Big[ ({\cal D}_\mu v_i v_j^\dagger) (v_j{\cal D}^\mu v_i^\dagger) \Big]=\frac{1}{2} \sum_{i,j} (f_i+f_j) \tr\Big[ ({\cal D}_\mu v_i v_j^\dagger) ({\cal D}^\mu v_iv_j^\dagger)^\dagger \Big],
\end{eqnarray} where the completeness condition $\sum_j v_j v_j^\dagger={\bf 1}_N$ is used.
For $L>2$, this model does not cover whole space of the homogeneous Riemannian metric,
whereas for $L=1,2$, it reproduces the Riemannian metric with arbitrary decay constants 
since $\#f_i=L+1\ge L(L+1)/2=\# f_{ij}$.

Each field configuration can be viewed as a map form $\mathbb R^2 \cup \{\infty\} =S^2$ 
to the target space $ {\cal M}=G/H$ 
and defines the topological charges 
\begin{eqnarray}
\pi_2 (G/H)
%=\pi_2 \left(\frac{SU(n_1+n_2+\cdots+n_{L+1})} {S[U(n_1)\times U(n_2)\times \cdots\times U(n_{L+1})]}\right) 
=\pi_1(H)=\pi_1 \left(S[U(n_1)\times U(n_2)\times \cdots\times U(n_{L+1})]\right) ={\mathbb Z}^L.
\end{eqnarray}
Explicitly, the topological numbers $m_i~(i=1,\cdots,L+1)$ are given by
\begin{eqnarray}
m_i \equiv -\frac1{2\pi} \int_{\mathbb R^2} dx^2 \, \tr[f_{12}^i]= \frac{i}{2\pi} \int_{\mathbb R^2} \tr [ d v_i \wedge d v_i^\dagger ]
=\frac{1}{2\pi i} \oint_{S^1} \tr [ d v_i  v_i^\dagger ] \in \mathbb Z, \label{eq:Riemancharge}
\end{eqnarray}
where $f_{\mu \nu}^i = \p_\mu a_\nu^i - \p_\nu a_\mu^i + i [a_\mu^i,a_\nu^j]$. 
Note that there are only $L$ independent charges since the total charge vanishes as
\begin{eqnarray}
\sum_{i=1}^{L+1} m_i = \frac{i}{2\pi} \int_{\mathbb R^2}
\sum_{i=1}^{L+1} \tr [ d v_i \wedge d v_i^\dagger ]= -\frac{i}{2\pi} \int_{\mathbb R^2} \tr [ dU U^\dagger \wedge dU U^\dagger]=0.
\end{eqnarray}
%%%%%%%%%%%%%%%%%%%%%%%%%%%%%%%%%%%%%%%%%%%%%%%%%%%%%%%%%%%%%%%%%%%%
\subsection{\kahler condition on decay constants}
Let us discuss the relation between the model with a Rieamannian metric introduced above and 
the nonlinear sigma model discussed in the main text. 
As we have seen in \eqref{eq:kahler_target}, 
the \kahler potential for the \kahler metric is given by
\beq
K =\sum_{i=1}^L r_i \ln \det (\xi_i \xi_i^\dagger). 
\eeq
Note that the set of matrices $\{\xi_i\}$ and $\{v_i\}$ are related as  
\begin{eqnarray}
\left(v_1^\dagger,v_2^\dagger, \cdots, v_i^\dagger \right)^\dagger=\xi_i^o U =\xi_i^o \hat h^{-1}{\cal G}=\hat h_i^{-1} \xi_i^o {\cal G}=\hat h_i^{-1} \xi_i , 
\end{eqnarray} 
where ${\cal G}=\hat h\, U~(\hat h \in \hat H)$, $\hat h_i=\xi_i^o\hat h (\xi_i^o)^\dagger \in GL(N_i,\mathbb C)$.
%and we have used $\hat h_i\hat h_i^\dagger=\xi_i\xi_i^\dagger$.
Using the \kahler metric $g_{A \bar B}(X) = \p_A \bar \p_B K$,
we find that the Lagrangian for the \kahler flag manifold sigma model is given by
\begin{eqnarray}
\mathcal L = - \frac{1}{2} \sum_{i=1}^L r_i \sum_{j,l}^{j\le i <l} \tr[\p_\mu P_j \p^\mu P_l]. 
%= - \frac{1}{2} \sum_{i=1}^L \sum_{j=i+1}^{L+1} f_{ij}\tr[\p_\mu P_i \p^\mu P_j].
\end{eqnarray}
Comparing with \eqref{eq:L_Riemann}, 
we find that the Riemaniann model reduces to the \kahler model 
when the decay constants $f_{ij}$ are given by
\begin{eqnarray}
f_{ij}=\sum_{k=i}^{j-1} r_k \quad   {\rm or~equivalently} \quad  f_{i,i+1}=r_i,\quad f_{ij}=f_{ik}+f_{kj}\quad {\rm for~} i<k<j.
\end{eqnarray}

The vortex numbers $\{k_i\}$ defined in Eq.(\ref{eq:chargeH}) are 
related to topological charges defined in Eq.(\ref{eq:Riemancharge}) as 
\begin{eqnarray}
k_i \equiv\frac1{2\pi i} \int_{\mathbb R^2} \bar \partial \wedge \partial  \ln \det (\xi_i\xi_i^\dagger)
= \sum_{j=1}^i m_j.    
\end{eqnarray}

%%%%%%%%%%%%%%%%%%%%%%%%%%%%%%%%%%%%%%%%%%%%%%%%%%%%%%%%%%%%%%%%%%%%%%%%%%%%
\section{Comments on the master equations}\label{appendix:uniqueness}
\def\Sva{{{\mathfrak S}^{\rm vtx}}}
\def\bomega{{\boldsymbol \omega}}
\def\tu{\tau}
\def\vt{t}
In this appendix, 
we discuss the existence and the uniqueness 
of the solution of the set of the master equations \eqref{eq:mastereq}
\begin{align}
\hat {\cal E}_i \ \equiv \ q_i \Omega_{i+1} q_i^\dagger \Omega_i^{-1} - \Omega_i q_{i-1}^\dagger \Omega_{i-1}^{-1} q_{i-1} +\frac{4}{g_i^2} \partial_{\bar z} \left( \p_z \Omega_i \Omega_i^{-1} \right) -r_i {\bf 1}_{N_i} \ = \ 0.
\label{eq:calE_master}
\end{align}
These are equations for the positive definite Hermitian matrices $\Omega_i \in GL(N_i,\mathbb C)$ 
determined by a given set of the moduli matrices $(q_1,\cdots,q_L)$.
The master equations $\hat {\cal E}_i=0$ are related to 
the original BPS equations ${\cal E}_i = 0$ for the magnetic flux as
\begin{eqnarray}
{\cal E}_i = S_i^{-1}\hat {\cal E}_i S_i
~~~\mbox{with}~~~  
{\cal E}_i \equiv Q_i Q_i^\dagger-Q_{i-1}^\dagger Q_{i-1} -\frac{2}{g_i^2} F_{12}^i-r_i {\bf 1}_{N_i}, 
\end{eqnarray}
where $S_i \in GL(N_i,\mathbb C)$ are the matrices that 
can be obtained from $\Omega_i$ by the Cholesky decomposition $\Omega_i = S_i S_i^\dagger$. 
The matrices $q_i$ and $ S_i$ are related to $Q_i$ and ${\cal D}_\mu=\partial_\mu+iA_\mu^i$ as
\beq
Q_i= S_i^{-1} \mq_i(z) \, S_{i+1}, \hs{10}
A_{\bar z} = \frac{1}{2} (A_1^i+iA_2^i) = -i S_i^{-1} \partial_{\bar z} S_i.
\eeq
The boundary condition for $\Omega_i$ is 
\begin{align}
\lim_{|z|\to \infty }\left\{\mD_i(z)^{-1} \Omega_i(z,\bar z) \, \mD_i(z)^{\dagger -1}\right\}= \Omega_i^{\rm o} \in GL(N_i,\mathbb C), 
\label{eq:bcOmega}
\end{align}
where $\mD_i(z)$ is the $N_i$-by-$N_i$ matrix defined through the relation (see Eq.\,\eqref{eq:def_mD}) 
\beq
\mH_i(z) \equiv q_i(z) \, q_{i+1}(z) \cdots q_L(z) = (\mD_i(z),\wt \mD_i(z))
\eeq 
and $\Omega_i^{\rm o}$ is the constant positive-definite matrix corresponding to the vacuum configuration (see Eq.\,\eqref{eq:general_sol}).
%such that  the flags becomes a full rank constant matrices  at the boundary as 
%\begin{align}
%Q_i Q_{i+1}\cdots Q_L=S^{-1}_i \xi _i (z)  \quad \stackrel{|z|\to \infty}{\to} \quad 
%(S_i^{\rm o})^{-1} \xi_i^{\rm o} ={\rm const.}\quad
%{\rm with~} \Omega_i^{\rm o}=S_i^{\rm o}(S_i^{\rm o})^\dagger.
%\end{align}
By examining the master equation for large $z$, 
we find that the deviation of $\Omega_i$ from the large coupling limit $\Omega_i^{\infty}$ (see Eq.\eqref{eq:sol_omega_large}) is given by,
\begin{align}
\mD_i(z)^{-1} (\Omega_i-\Omega_i^{\infty}) \, \mD_i(z)^{\dagger -1}
=\frac{\p_z \p_{\bar z}}{g_i^2} {\cal O}(|z|^{-2})={\cal O} \left(  |z|^{-4}\right).
\label{eq:OmegaAsymptotic}
\end{align}

Here, let ${\cal F}_\Omega(\mathbb R^2)$ denote 
the space of configurations ${\bf \Omega} = \{\Omega_1,\Omega_2,\dots, \Omega_L \}$ 
where $\Omega_i$ are smooth maps from $\R^2$ to the space of positive definite $N_i$-by-$N_i$ Hermitian matrices 
satisfying the boundary condition \eqref{eq:bcOmega} 
and the asymptotic behavior \eqref{eq:OmegaAsymptotic}
with a given set of moduli matrices $\{q_i(z) \, | \, i=1,2,\dots,L\}$.
Let us take a reference point ${\bf \Omega}^{\rm ref}$ in ${\cal F}_\Omega(\mathbb R^2)$ 
and denote its components as $\Omega_i^{\rm ref}=S_i^{\rm ref} S_i^{\rm ref\dagger}$. 
Since $\Omega_i$ is a positive definite Hermitian matrix, 
we can define an $N_i$-by-$N_i$ Hermitian matrix $\omega_i$ depending on the reference point as
\begin{align}
\omega_i = \log\left[ (S_i^{\rm ref})^{-1}\Omega_i (S_i^{\rm ref \dagger})^{-1}\right]
\quad \Leftrightarrow \quad \Omega_i=S_i^{\rm ref}e^{\omega_i} S_i^{\rm ref \dagger}. \label{eq:omegadef}
\end{align}  
This relation defines a one-to-one map between ${\cal F}_\Omega$ 
and the functional space ${\cal F}_\omega$: 
the space of $\bomega=(\omega_1,\omega_2,\cdots,\omega_L)$
whose components $\omega_i$ are smooth bounded functions from $\R^2$ to 
the space of Hermitian matrices of order $N_i$ 
satisfying the boundary condition and the asymptotic behavior 
\begin{align}
\lim_{|z|\to \infty} \omega_i=0, \hs{5}
\omega_i ={\cal O}(|z|^{-4}). 
\label{eq:bcomega}
\end{align}
It is worth noting that the asymptotic behaviors in Eqs.\,\eqref{eq:OmegaAsymptotic} and 
\eqref{eq:bcomega} implies that the following square-integrable conditions are satisfied
\begin{align}
\langle {\cal E}^2\rangle < \infty ~~~\mbox{for $\forall {\bf \Omega}\in {\cal F}_\Omega$},\qquad 
\langle \bomega^2\rangle <\infty ~~~\mbox{for $\forall \bomega \in{\cal F}_\omega$},
\end{align}
where we have used the following bracket notation
\begin{align}
\langle { \cal O} \rangle \equiv   \int d^2x  \sum_{i=1}^L \tr[ {\cal O}_i].
\end{align}
%%%%%%%%%%%%%%%%%%%%%%%%%%%%%%%%%%%%%%%%%%%%%
%%%%%%%%%%%%%%%%%%%%%%%%%%%%%%%%%%%%%%%%%%%%%%%%%%
\subsection%[Linearization of the master equations]
{$\cal H$: Linearization of master equations} \label{sec:H}
First, let us define a linear operator ${\cal H}={\cal H}(q,S)$ by
\begin{align}
({\cal H} {\bf v})_i\equiv 
-\frac{4}{g_i^2} {\cal D}_{\bar z}{\cal D}_z  v_i +
Q_i \left(Q_i^\dagger v_i-v_{i+1}Q_i^\dagger\right)+
\left( v_i Q_{i-1}^\dagger -
 Q_{i-1}^\dagger v_{i-1} \right)Q_{i-1},
\end{align} 
where ${\bf v} = (v_1,v_2,\cdots,v_L)$ is an element of ${\cal F}_\omega^{\mathbb C}\equiv \{{\bf u}+i {\bf w} \big|  {\bf u, w}\in {\cal F}_\omega \}$.
This operator ${\cal H}$, which depends on $q_i, S_i$, 
appears in the linearized master equations 
and plays a central role in the subsequent subsections.  
Explicitly, one can show that the variation $\delta \hat {\cal E}_i$ of $\hat {\cal E}_i$ 
under the infinitesimal shift
$\delta \Omega_i\equiv S_i \omega_i S_i^\dagger$ 
of $\Omega_i=S_iS_i^\dagger \in {\cal F}_\Omega$
with $\omega_i=\omega_i(z,\bar z) \in {\cal F}_\omega$ 
is given in terms of $\mathcal H$ as
\begin{align}
S_i^{-1} \delta \hat {\cal E}_i S_i = - ({\cal H}\,{\bomega})_i. 
\end{align}
This operator ${\cal H}$ is Hermitian and positive semi-definite
\begin{align}
\langle {\bf v} {\cal H} {\bf v}^\dagger \rangle 
% \int d^2x  \sum_{i=1}^L\tr \left[ \omega_i {\cal H} \omega_i\right]
=\int d^2x  \sum_{i=1}^L\tr \left[ \frac{4}{g_i^2} {\cal D}_{\bar z } v_i  {\cal D}_z v_i^\dagger 
+(v_i Q_i-Q_i v_{i+1})(v_i Q_i-Q_i v_{i+1})^\dagger \right] \, \ge 0 
\hspace{3mm}
\mbox{for arbitrary ${\bf v}\in {\cal F}_\omega^{\mathbb C}$.}
\label{eq:innnerproductH}
\end{align}
Furthermore, we can show that there exists a gap in the spectrum of the linear operator ${\cal H}$ as follows.
Suppose that ${\cal H} {\bf v}^\dagger=0$. 
Then, the above inner product vanishes and hence
\begin{align}
0={\cal D}_z v_i^\dagger ={\cal D}_{\bar z}v_i,
%S^\dagger \partial_z (S^{\dagger-1} \omega_i S^{\dagger} ) S^{\dagger-1}, 
\quad v_i Q_i =Q_i v_{i+1}\quad   (v_{L+1}=0)  \quad {\rm for~} i=1,2,\dots,L.
\end{align}
It follows from the second equation that 
\begin{align}
v_i Q_i Q_{i+1} \cdots Q_L = Q_i v_{i+1} Q_{i+1} \cdots Q_L = \cdots = Q_i Q_{i+1} \cdots Q_L v_{L+1} = 0.
\end{align}
This equation implies that $v_i$ must vanish since $Q_i Q_{i+1} \cdots Q_{L}$ 
has the maximal rank 
except for a finite number of points\footnote{The matrix $\xi_ i =q_i q_{i+1} \cdots q_L$ 
has the maximal rank except for a finite number of points. The matrix $Q_i Q_{i+1} \cdots Q_{L} = S_i^{-1} \xi_i$ has the same property since $S_i(z,\bar z) \in GL(N_i ,\mathbb C)$.}. 
Therefore, we find that ${\cal H}$ has no zero mode
\begin{align}
\ker {\cal H}(q,S) = \{ 0 \} ~~~ 
\mbox{for $\forall (q_i,S_i)$}.
\end{align}
From this property, we can immediately conclude that $\boldsymbol{\Omega}$ has no additional moduli parameter and all the moduli parameters of vortex solutions are contained in $\{ q_i(z) \}$.
 
%One can show that the spectrum of ${\cal H}$  is non other than that of 
%fluctuations of the gauge fields around the vortex solution and 
%the above feature is a reflection of the fact that the vacuum of the theory lies in the Higgs phase.
%The spectrum 
%$\Lambda=\{\lambda_a| a=1,2,\cdots, \lambda_a>0\}$ given by
%\begin{align}
%{\bf v}_p \in {\cal F}_\omega,\quad {\cal H}{\bf v}_p= {\bf v}_p  \lambda_p,
%\end{align}
%is decomposed into 
%a  bulk continuous spectrum $\Lambda_{\rm bulk}$ and a discontinuous spectrum $\Lambda_{\rm bd}$ 
%for bound states around vortices 
%which satisfy 
%\begin{align}
%\Lambda=\Lambda_{\rm bulk}\cup \Lambda_{\rm bd},\quad 0< \Delta_\star \equiv \inf \Lambda.
%\end{align}
%Bounded eigenmodes are known as, so called, {\it breather modes},
%which  are  expected to have one mode corresponding to one vortex.

%begin{align}
%    0< \Delta_\star  \langle {\bf v}^2\rangle \le  \langle {\bf v} {\cal H} {\bf v}\rangle <\infty.
%\end{align}

%%%%%%%%%%%%%%%%%%%%%%%%%%%%%%%%%%%%%
%%%%%%%%%%%%%%%%%%%%%%%%%%%%%
\subsection{Vortex action and proof of the uniqueness} \label{sec:convexity}
For a given set of moduli matrices $\{q_i(z)\}$ (with  an appropriate gauge fixing of the $V$-transformations),  we can show that 
there exists a functional $\Sva$ of $\boldsymbol{\Omega}$ such that 
the variation of $\Sva$ with respect to $\Omega_i$ is given by
\begin{align}
\delta \Sva=-\int d^2x \sum_{i=1}^L \tr \left[\delta \Omega_i \Omega_i^{-1} \hat {\cal E}_i\right]
=-\int d^2x \sum_{i=1}^L \tr \left[\omega_i  {\cal E}_i\right]=-\langle \bomega {\cal E}\rangle. \label{eq:devS}
\end{align}
That is, $\Sva$ is the action which gives the full set of the master equations $\{\hat {\cal E}_i=0\}$.

\paragraph{Existence of $\Sva$ \\}
Although such an $\Sva$ may not be unique due to constant and total derivative terms, 
we can show that the following functional gives the full set of the master equations
\begin{align}
\Sva[{\bf \Omega},{\bf \Omega}^{\rm ref}]\equiv \int d^2 x  \left\{{\cal L}({\bf \Omega})-{\cal L}({\bf \Omega}^{\rm ref})\right\}
%\quad D_R\equiv \left\{z=x^1+i x^2 \in \mathbb C \big| |z|\le R \right\}
,\hs{5} {\cal L}({\bf \Omega})=\sum_{i=1}^L\left\{{\cal L}^D_i+\frac1{g_i^2}{\cal L}^K_i\right\}. 
\label{eq:defSva}
\end{align}
Here ${\cal L}_i^D$ and ${\cal L}_i^K$\footnote{These terms are equivarent to the normal kinetic terms plus the Wess-Zumino-Witten terms up to total derivative terms. } are given by
\begin{align}
{\cal L}^D_i&= \tr \left[ \Omega_i^{-1} q_i \Omega_{i+1}q_i^\dagger + r_i \log \Omega_i \right],\\
{\cal L}^K_i&= \tr \left[ 4 |\p_z \psi_i|^2 +2 e^{-2\psi_i} (L_i^{-1} \p_{\bar z} L_i)  \, e^{2\psi_i} \, (L_i^{-1} \p_{\bar z} L_i)^\dagger \right],
\end{align}
where $\psi_i$ is a $N_i$-by-$N_i$ diagonal matrix and $L_i$ is a lower unitriangular matrix obtained by the Cholesky decomposition
\beq
\Omega_i =L_i \, e^{2\psi_i}L_i^\dagger.
\eeq
%$\rho_i=\rho_i(z,\bar z)$ is a certain field-independent density introduced to cancel an infra-red divergence of the action.
Note that although $\Sva$ is not invariant under the $V$-transformations, 
the shift is independent of $\Omega$. 
Hence, the variation of $\Sva$ with respect to $\Omega_i$ reproduce the master equations, 
which are covariant under the $V$-transformations. 
The term ${\cal L}({\bf \Omega}^{\rm ref})$ in the integrand is added 
to make the integral finite. 

To confirm that surface terms vanish in the r.h.s of \eqref{eq:devS}, 
we need to discuss the boundary conditions for $\{\psi_i, L_i\}$.
Note that,  as discussed in  Appendix \ref{subsec:patch_L=1},   an arbitrary given matrix $\mD_i(z)$
can always be transformed into a lower triangular form by a $V$-transformation as
\begin{align}
\mD_i(z)=\begin{pmatrix}
p_{i,1}(z)& 0 & \cdots  &0\\
\star & p_{i,2}(z) & \ddots & \vdots\\
\vdots & \ddots &\ddots& 0 \\
\star &\cdots&\star & p_{i,N_i}(z)
\end{pmatrix}= \left( {\bf 1}_{N_i} +{\cal O}(z^{-1})\right) {\bf p}_i(z) 
\end{align}
where ``$\star$" stand for polynomials and ${\bf p}_i(z)$ is a diagonal matrix 
\begin{align}
{\bf p}_i(z)={\rm diag} \, ( p_{i,1}(z),p_{i,2}(z),\cdots, p_{i,N_i}(z) )    \label{eq:pi}
\end{align} 
such that the $a$-th diagonal entry $p_{i,a}(z)$ is a monic polynomial of degree $l_{i,a}$ and ${\rm deg}(\det {\bf p}_i(z)) = \sum_{a} l_{i,a} = k_i$.
Under this gauge choice, the boundary condition \eqref{eq:asymptoticS} for $\Omega_i$ can be rewritten in terms of $\{ \psi_i, L_i\}$ as 
 \begin{align}
  \psi_i  \to \frac12 \log \Omega_i^o+ \frac12 \log |{\bf p}_i(z)|^2+{\cal O}(|z|^{-2}), \hs{5} L_i\to \mD_i(z) ({\bf 1}_{N_i}+{\cal O}(|z|^{-2})){\bf p}_i(z)^{-1}\quad {\rm for~} |z|\to \infty. 
\end{align}
From these behaviors, we conclude that  the contributions from the surface terms of 
$\delta \Sva$ vanishes as
\begin{align}
\lim_{R\to \infty} {\rm Im}\oint_{|z|=R} dz\sum_i \tr \left[2 \delta \psi_i \partial_{z} \psi_i + e^{-2\psi_i} (L_i^{-1} \delta L_i) \, e^{2\psi_i} \p_{z} L_i^\dagger L_i^{\dagger-1} \right] = 
\lim_{R\to \infty} {\rm Im}\oint_{|z|=R} dz \sum_i \tr \left[\delta \psi_i \frac{\p_z {\bf p}_i(z)}{{\bf p}_i(z)}\right]=0
\end{align}
where we have used $\lim_{|z|\to \infty}\delta \psi_i=0$.
From this property, we conclude that Eq.\eqref{eq:devS} holds. 
Furthermore, this property implies that terms in ${\cal L}(\bf \Omega)$ that could diverge when integrated are independent of $\Omega_i$ 
and thus the counter term ${\cal L}({\bf \Omega}^{\rm ref})$ in the action cancels such terms.
This cancellation makes $\Sva$ finite and hence its convexity, 
discussed in the next paragraph, is well-defined.
Although $\Sva$ depends on the choice of the coordinate patch of the moduli space (the choice of the label $\lambda_i=(l_{i,1},l_{i,2},\dots,l_{i,N_i})$),
there is no problem with the arguments for existence and uniqueness of the solution with fixed moduli parameters. 

\paragraph{Convexity of $\Sva$}
The action $\Sva$ constructed above is always a convex functional.
To see this, let us take a pair of configurations $\Omega^{(a)}_i=S^{(a)}_i (S_i^{(a)})^\dagger (a=1,2)$
which satisfy the same boundary conditions with a given $\{ q_i=q_i(z) \}$.
For such a pair of configurations, let us define 
an Hermitian matrix $\omega_i\in \mathfrak u (N_i)$ by 
\begin{align}
e^{\omega_i}\equiv (S^{(1)}_i)^{-1} \Omega_i^{(2)} (S_i^{(1)})^{\dagger -1},
\end{align}
which satisfies the boundary condition \eqref{eq:bcOmega}.
Using these quantities $\{\omega_i\}$, 
we obtain a set of functions of a parameter $\tu$
\begin{align}
\Omega_i(\tu\bomega)=  S_i(\tu \bomega)S_i(\tu \bomega)^\dagger\equiv S_i^{(1)}e^{\tu\, \omega_i} (S_i^{(1)})^\dagger
\label{eq:Omegat}
\end{align}
which continuously interpolates two given configurations: 
\beq
\Omega_i(0) =\Omega_i^{(1)}\quad {\rm and} \quad  \Omega_i(\bomega) =\Omega_i^{(2)}.
\eeq
Then, substituting $\Omega_i(\tu\bomega)$ to the action
and using 
\begin{align}
 \frac{d \Omega_i (\tu  \bomega)}{d\tu} =S_i(\tu \bomega)\omega_i S_i(\tu \bomega)^\dagger, \hs{10} 
 \frac{d}{d\tu}\left\{
 \frac{d \Omega_i(\tu\bomega)}{d\tu}\Omega_i(\tu\bomega)^{-1}\right\}=0, 
 %\frac{d \Omega_i(t)}{dt}\Omega_i(t)^{-1}= S_i^{(1)}\omega_i  ( S_i^{(1)})^{-1}: \hbox{$t$-indepedent}
\end{align}
we find that $\Sva$ is always a convex functions of $\tu$: 
\begin{align}\forall \tu \in \mathbb R  :\quad
\frac{ d^2 \Sva[{\bf \Omega}(\tu\bomega), {\bf \Omega}^{\rm ref}]}{d \tu^2}=
\langle \bomega {\cal H} \bomega \rangle\Big|_{S_i\to S_i(\tu\bomega)} > 0.
\end{align}
Therefore, 
if $\{\Omega_i^{(1)}\}$ is the solution of the master equations, 
then the action takes the minimum at $\tu=0$ and 
can never have other extrema since the derivative $\p_\tau \Sva$ is monotonically increasing function
for any choice of $\{\Omega_i^{(2)}\}$. 
Therefore, {\it the solution of the master equations must be unique if it exists}.

%%%%%%%%%%%%%%%%%%%%%%%%%%%%%%%%%%%%%%%%%%%%%%%%%%%%%%%%%%%%%%%%%%%%%%%%%%%%

\subsection{Comments on existence of solutions}
In the main text, it is assumed that a solution to the master equations exists 
for an arbitrarily given set of moduli matrices $\{q_i(z)\}$. 
While the existence of the Abelian vortex solutions has been proven in \cite{Taubes:1979tm}, 
for non-Abelian vortices, however, it is generally difficult to prove the existence of solutions except for some limited cases \cite{Tarantello2011NonabelianVE,Chen2012ExistenceOM,HAN2014117}. 
To the best of our knowledge, the proof of existence of vortex solutions in the general systems is not known.
Let us try to give a circumstantial evidence that the solution exists, 
using a rough argument that is not necessarily mathematically rigorous.

First, let $|\!| \cdot |\!|$ be a norm defined in the functional space ${\cal F}_\omega$. 
Let us write an arbitrary element ${\bomega} \in {\cal F}_\omega$ as $\bomega=\tu \hat \bomega$ by using its norm $\tu \equiv |\!| \bomega |\!|$ and the normalized element $\hat \bomega \in {\cal F}_\omega$ satisfying $|\!| \hat \bomega |\!|=1$.  
Then an arbitrary $\bf \Omega \in {\cal F}_\Omega$ can be expressed using Eq.\eqref{eq:omegadef} as 
\begin{align}
\Omega_i=\Omega_i(\tu \hat \bomega) =S_i(\tu \hat \bomega)  (S_i(\tu \hat \bomega))^\dagger 
\equiv S_i^{\rm ref}e^{ \tu \hat \omega_i} (S_i^{\rm ref})^\dagger,
\qquad \Omega_i(0)=\Omega_i^{\rm ref},
\end{align}
and the vortex action $\Sva$ defined in Eq.\eqref{eq:defSva} can be regarded as a function of $\tau$,
\def\Sf{{ {\mathfrak S}_{\hat \bomega}}} 
\begin{align}
\Sf(\tau)\equiv \Sva[{\bf \Omega}(\tu \hat \bomega),{\bf \Omega}^{\rm ref}].
\end{align}
%As is taken in Eq.\eqref{eq:Omegat}, 
%with an arbitrary chosen initial configuration $\Omega_i^{\rm ini}=S_i^{\rm ini}(S_i^{\rm ini})^\dagger \in {\cal F}_\Omega$,   
%an arbitrary configuration of $\Omega_i \in {\cal F}_\Omega$ can be expressed as
%\begin{align}
%\Omega_i=\Omega_i(\tu \bomega) =S_i(\tu \bomega)  (S_i(\tu \bomega))^\dagger 
%\equiv S_i^{\rm ini}e^{ \tu \omega_i} (S_i^{\rm ini})^\dagger
%\end{align}
%by choosing  $\tu \in \mathbb R_{\ge 0}$ and $\bomega=\{\omega_i \}\in {\cal F}_\omega$ with an appropriate normalization, 
%for instance,  $ \langle \bomega^2\rangle=1$.
%\begin{align}
%1=||\bomega ||_{H^2(\mathbb R^2)}\equiv
%\sqrt{  \langle \bomega^2 +|\p_z\bomega|^2+|\p_z\p_z\bomega|^2+|\p_z\p_{\bar z}\bomega|^2\rangle }
%\end{align}
%where $||\bomega ||_{H^2(\mathbb R^2)}$ is a norm for 
%a direct sum of  $\sum_i N_i^2$ copies of the sobolev space $H^2(\mathbb R^2)\subset {\cal C}^0(\mathbb R^2)$.
%An action for $\bf \Omega$ satisfying Eq.\eqref{eq:devS} can always be written as,
%\begin{align}
%\Sva[{\bf \Omega}(\tu \bomega)]=\Sva[{\bf \Omega}(\tu \bomega),{\bf \Omega}^{\rm ref}]=\Sva[{\bf \Omega}^{\rm ref},
%{\bf \Omega}^{\rm ini}]+\Delta \Sva[\tu \bomega],
%\end{align}
%where  $\Delta \Sva[\tu\bomega]=\Sva[{\bf \Omega}(\tu\bomega), {\bf \Omega}^{\rm ini}]$ is defined such that its %Fr\'echet 
Note that this function is convex everywhere as discussed in the previous subsection. 
In the following, we show that for each choice of $\hat{\bomega}$, the function $\Sf(\tau)$  has a minimum at some point with $\tau = \tau_{\hat \bomega} < \infty$. Then, tracking the decreasing sequence of the function $\Sf(\tau_{\hat \bomega})$ in the space of normalized $\hat \bomega$, we can find the solution of the master equation. 
Thus, roughly speaking, 
showing the existence of a solution to the master equations is equivalent to showing that $\Sf(\tau)$ has a minimum for an arbitrary normalized configuration $\hat \bomega \in {\cal F}_{\omega}$ and the space of normalized $\hat \bomega$ is a complete metric space.

\paragraph{Coerciveness of $\Sva$}
%In the limit $\tu \rightarrow \infty$,  $S_i(\tu\hat \bomega)$ is no longer an element of $GL(N_i,\mathbb C)$ 
%and a Coulomb phase can appear there.  
%To prohibit gradient flows towards this Coulomb phase, an alternative argument is needed.
%
Here, we show that for an arbitrary nonzero element $\hat \bomega \in {\cal F}_\omega$
\begin{align}
\lim_{\tu \to  \infty} \frac{d \Sf(\tau)}{d \tu} = \infty \qquad \mbox{and} \qquad \frac{d^2 \Sf(\tau)}{d \tu^2} >0 ~~~ \mbox{for arbitrary $\tu$}, 
\end{align}
i.e. $\Sf(\tau)$ is a coercive and convex function of $\tu$ for an arbitrary $\hat \bomega \in {\cal F}_\omega\backslash \{0\}$.

The derivative of $\Sf(\tau)$ with respect to $\tu$ is given by\footnote{The above property can be viewed as the definition of the function $\Sf(\tau)$, 
which is independent of the details of the original defining equation \eqref{eq:defSva}.
We can confirm that $\Sf(\tau)$ is finite for an arbitrary
$\tau \in\mathbb R$ %$\bomega \in {\cal F}_{\omega}(\mathbb R^2)$ 
since $\langle \hat \bomega {\cal E} \rangle ^2 \le \langle \hat \bomega^2 \rangle \langle  {\cal E}^2 \rangle < \infty.$}
\begin{align} 
\frac{d \Sf(\tau)}{d \tu}=
- \langle  \hat {\bomega}  {\cal E}\rangle  \Big|_{S_i=S_i(\tu\hat \bomega)}  
=- \langle  {\hat \bomega}  {\cal E}\rangle  \Big|_{S_i=S_i^{\rm ref}}  
+ \int_0^\tu ds \langle {\hat \bomega} {\cal H}{\hat \bomega}  \rangle  \Big|_{S_i=S_i(s\hat \bomega)},\qquad
\Sf(0)=0.\label{eq:deltaS}
\end{align}
Note that $\langle \hat \bomega {\cal H}\hat \bomega\rangle$ implicitly depends on $\tu$, since ${\cal H}$ depends on $S_i=S_i(\tu\hat \bomega)$. Let us show that $\langle \hat \bomega {\cal H}\hat \bomega\rangle$ cannot vanish even in the limit $\tu \rightarrow \infty$, 
whereas it must be positive definite for a finite $\tu$ as discussed in Sec.\ref{sec:H}.
First, we show that
\begin{align}
\lim_{\tu\to \infty} \langle {\hat \bomega} {\cal H}{\hat \bomega}  \rangle  \Big|_{S_i=S_i(\tu\hat \bomega)}=0 
%\quad ( \lim_{t\to \infty} \ker {\cal H} \Big|_{S_i=S_i(t)}\not =\emptyset),
~~~\Longrightarrow~~~\hat \omega_i(z,\bar z)=0. 
\label{eq:statement_B35}
\end{align}
From the assumption given in left-hand side of the above statement, we obtain
\begin{align} 
\lim_{\tu \to \infty }{\cal D}_z \hat \omega_i \Big|_{S_i=S_i(\tu\hat \bomega)}=\lim_{\tu \to \infty }{\cal D}_{\bar z} \hat \omega \Big|_{S_i=S_i(\tu\hat \bomega)}=0.
\end{align}
Since entries of $\hat \bomega \in {\cal F}_\omega$ are continuous functions,
%Under the assumption of absolute continuity of $\hat \bomega$,
it follows that      
\begin{align}
\partial_z \tr \hat \omega_i^p = p  \lim_{\tu\to \infty } \tr [ \hat \omega_i^{p-1} {\cal D}_z \hat \omega_i]=0 ~~~~~
\mbox{and hence}~~~~~
\hbox{eigenvalues of } \hat \omega_i ={\rm constant}.
\end{align}
Combining the above result with the boundary condition $\lim_{|z|\to \infty} \hat \omega_i=0$, we find  all $\hat \omega_i$ must vanish everywhere
\begin{align}
\forall i, \forall z,  ~~~~~ 
\hat \omega_i=\hat \omega_i(z,\bar z)=0.
\end{align}
This implies Eq.\,\eqref{eq:statement_B35} and 
its contraposition 
\begin{align}
\forall \hat \bomega\in {\cal F}_\omega\backslash \{0\}: \quad \lim_{\tu\to  \infty} \langle {\hat \bomega} {\cal H}{\hat \bomega} \rangle 
\Big|_{S_i=S_i(\tu\hat \bomega)} > 0, \label{eq:PositivenessInLimit}
\end{align}
where we have used the fact that ${\cal H}$ is a positive semi-definite operator. 
Applying this statement and the convexity of $\Sf$ to Eq.\eqref{eq:deltaS}, we conclude that 
\begin{align}
\forall \hat \bomega   \in {\cal F}_\omega\backslash \{0\}:\qquad
\lim_{\tu \to  \infty} \frac{d \Sf(\tau)}{d \tu} = \infty,\qquad 
\forall \tu : \frac{d^2 \Sf(\tau)}{d \tu^2} > 0
\end{align}
and thus $\Sf(\tau)$ is a coercive and convex function of $\tu$
for arbitrary $\hat \omega_i \in {\cal F}_\omega\backslash \{0\}$. 
This result implies that
$\Sf(\tau)$ has a minimum, $\Sf(\tau_{\hat \bomega}) (\le 0)$, with a certain $\tu=\tu_{\hat \bomega}$ for each $\hat \bomega$. 
Thus, by collecting these minimum, we can define a map from a hypersurface 
$\hat {\cal F}_\omega\equiv \{ \hat \bomega \in {\cal F}_\omega | \ |\!| \hat \bomega  |\!| =1  \}$ to 
${\cal F}_\omega$ as, $\hat \bomega\in \hat {\cal F}_\omega \,\mapsto\, \tu_{\hat \bomega} \hat \bomega \in {\cal F}_\omega$.
Since
\begin{align}
    \inf_{{\bf \Omega}\in {\cal F}_\Omega} \Sva[{\bf \Omega},{\bf \Omega}^{\rm ref}] =\inf_{\hat \bomega \in \hat {\cal F}_\omega} \Sf(\tu_{\hat \bomega}),
\end{align}
assuming ${\cal F}_{\omega}$ is a complete metric space, the coereciveness and convexity of $\Sf(\tu)$ implies that 
$\Sva$ has a minimum with a certain $\bomega=\bomega_{\rm min}$, 
which gives a solution ${\bf \Omega}^{\rm sol}\equiv {\bf \Omega}(\bomega_{\rm min})$ to the master equations.  
This ``proof" for the existence of the solution is, however, not mathematically rigorous, 
since in this argument the ``solution" obtained using a decreasing Cauchy sequence
$\bomega^{(1)},\bomega^{(2)},\cdots \in {\cal F}_\omega$, the limit $\bomega_{\rm min}=\lim_{a\to \infty} \bomega^{(a)}$ 
is not guaranteed to consists of bounded, smooth functions, ${\bomega}_{\rm min} \in {\cal F}_\omega$. 
For a complete proof, therefore, we need to give more mathematically precise arguments.
Nevertheless, the above arguments, especially Eq\eqref{eq:PositivenessInLimit}, 
are expected to be useful for an intuitive understanding of the existence of a solution, 
and can actually distinguish our system from those where the master equations do not have solutions (see the example below).

\paragraph{Compact cases and Bradlow bound}
Most of the above discussion and results can be applied to the models
defined on a compact base space $\Sigma$ with a finite area $A$ as long as $\ker {\cal H} = \emptyset$.  
It is, however, well known that there is a lower bound $A_{\rm Bb}$, the so-called Bradlow bound, on the area for the existence of solutions.
In our case, the lower bound is given by the following set of inequalities
\begin{align}
0 \ \le \ \int_\Sigma d {\rm vol}\, \tr Q_i Q_i^\dagger \ = \ \sum_{j=1}^i
\left( r_j N_j A - \frac{4\pi k_j}{g_j^2}\right),\qquad  A= \int_\Sigma d {\rm vol},
\qquad{\rm for~} i=1,2,\cdots,L,
\end{align} 
where $ d {\rm vol}$ is the volume form on $\Sigma$.
In the Bradlow limit saturating the above bound, 
at least one of {$Q_i$} must vanish everywhere 
and thus the operator ${\cal H}$ has a non-trivial kernel ($\ker {\cal H}\not = \emptyset$),
which implies the argument above is no longer applicable. 
The most significant difference between the compact and non-compact cases is that 
in the compact case, a (covariantly) non-zero constant $\bomega$, 
for which $\langle \bomega {\cal H} \bomega\rangle$ 
may vanish in the limit of $\tu \to \infty$, 
is allowed since the area is finite and the condition \eqref{eq:bcomega} is absent. 
The following $\hat \bomega^{\rm c}_{(i)}\in {\cal F}_\omega$ is the simplest example of such a $\hat \bomega$
\begin{align}
\hat \bomega^{\rm c}_{(i)} = \ \bomega^{\rm c}_{(i)}/|\!|\bomega^{\rm c}_{(i)} |\!| , \qquad 
\bomega^{\rm c}_{(i)}\equiv ({\bf 1}_{N_1},{\bf 1}_{N_2},\cdots, {\bf 1}_{N_i}, {\bf 0}, {\bf 0},\cdots), 
\qquad i \in \{1,2,\cdots,L\}.     
\end{align}
For this configuration, the $i$-th gauge group is restored in the limit $\tau \rightarrow \infty$
\begin{align}
 \lim_{\tu\to \infty}Q_i(\tu \hat \bomega^{\rm c}_{(i)} ) = \lim_{\tu \to \infty }e^{-\tu/|\!|\bomega^{\rm c}_{(i)} |\!|} Q_i^{\rm ref} =0,
 \qquad  \lim_{\tu \to \infty} \langle \hat \bomega {\cal H} \hat \bomega\rangle\big|_{\hat \bomega=\hat \bomega^{\rm c}_{(i)}}=0,
\end{align}
whereas the other quantities, $Q_j(j\not=i), A^j_{\bar z}$ remain invariant.
Thus, we obtain
\begin{align}  
 \lim_{\tu \to \infty } \frac{d \Sf(\tau)}{d\tu } =\frac1{|\!|\bomega^{\rm c}_{(i)} |\!|}
\sum_{j=1}^i\left( r_j N_j A - \frac{4\pi k_j}{g_j^2}\right),  \quad 
{\rm for~} \hat \bomega=\hat \bomega^{\rm c}_{(i)}, \quad i=1,2,\cdots,L, 
\end{align}
which shows that $\Sf(\tau)$ is no longer coercive for a sufficiently small area $A$.
Note that
for $\Sva[{\bf \Omega},{\bf \Omega}^{\rm ref}]$ to be coercive,
it is necessary that $\Sf(\tau)$ for $\hat \bomega=\hat \bomega^{\rm c}_{(i)}$ must be coercive for all $i=1,\cdots,L$.  
Therefore, $\Sva[{\bf \Omega},{\bf \Omega}^{\rm ref}]$ is coercive only when the area $A$ is larger than the Bradlow bound $A > A_{\rm Bb}$.
%\footnote{
%The bound, $A> A_{\rm c}$, that guarantees the coerciveness of $\Sva$ may be stronger than the Bradlow bound, i.e. $A_{\rm c} \ge A_{\rm Bb}$. 
%Indeed, a stronger condition $A_{\rm Lb} \ge A_{\rm Bb}$ has been found from the non-vanishing requirement for the vortex moduli space volume, 
%which has been computed by the localization method \cite{Ohta_2021}. 
%For $A_{\rm c} \ge A \ge A_{\rm Lb}$, 
%the lack of the coerciveness does not immediately imply the absence of a solution, 
%but may mean the appearance of a new solution corresponding to the limit of $\tu \to \infty$.
%Such a solution must be in the Coulomb phase with a constant magnetic field and never appears when $A>A_{\rm c}$.
%}.
In this way, we can show that the discussion on the existence of solutions 
based on the coerciveness of $\Sva[{\bf \Omega},{\bf \Omega}^{\rm ref}]$ 
is consistent with the Bradlow bound.

The Bradlow bound, $A \ge A_{\rm Bb}$ 
is only a necessary condition on the area $A$ for the existence of solutions and 
a necessary and sufficient condition, $A\ge A_{\rm tb}$ might be stronger than this condition, 
$A_{\rm tb}\ge A_{\rm Bb}$ and $A_{\rm tb}$ might depend on a point of the moduli space.
For a generic moduli point, the condition $A\ge A_{\rm tb}$ would be found indirectly from the non-vanishing requirement for the vortex moduli space volume, 
which has been computed by the localization method \cite{Miyake:2011yr,Miyake:2011fq,Ohta:2019odi,Ohta_2021}.
By refining the above argument on coerciveness of the functional $\Sva$,
we expect that it is possible to prove the condition $A\ge A_{\rm tb}$ directly.

%%%%%%%%%%%%%%%%%%%%%%%%%%%%%%%%%%%%%%%%%%%%%%%%
%%%%%%%%%%%%%%%%%%%%%%%%%%%%%%%%%%%%%%%%%%%%%%%%%%%%%
\paragraph{Relaxation method}
%Proving the existence of a solution is difficult and beyond the scope of this paper, but without this assumption most of our argument in this paper would not be possible.Therefore, 
It is important and useful to provide an explicit procedure for obtaining a numerical solution that minimizing $\Sva$ by discretizing the system. 
To find the minimum of $\Sva$, let us consider the following recurcive relation
\begin{align}
S_i^{(n+1)}=S_i^{(n)}e^{ \delta \omega_i^{(n)}}, \quad S_i^{(0)}=S_i^{\rm ref}.
\end{align}
If we choose $\delta \omega_i^{(n)} ={\cal H}^{-1} {\cal E}_i|_{S\to S^{(n)}}$, 
this procedure can be regarded as Newton's method, 
which is, however, impractical since the calculation of ${\cal H}^{-1}$ is known to be very costly.
One simple and effective method is the relaxation method,
where  $\delta \omega_i$ in each step is given by $\delta \omega_i^{(n)} =\alpha {\cal E}_i|_{S_i\to S_i^{(n)}}$ with an appropriate step size $\alpha \in \mathbb R_{>0}$. 
The parameter $\alpha$ must be a sufficiently small to satisfy the Courant-Friedrichs-Lewy condition, 
$\alpha /a_{\rm lat}^2<{\cal O}(1)$, where $a_{\rm lat}$ is the spatial lattice spacing.
With such small $\alpha$, the convexity and coerciveness of $\Sva$ guarantee that  
this sequence converges without being trapped by meta-stable points.
To see that $\Sva$ decreases at each step with a sufficiently small step size $\alpha$,
let us take the continuous limit of $\alpha \to 0$ and 
rewrite the recursive relation into a differential equation by introducing a fictitious time $\vt$.
%
%The existence of the action $\Sva$ satisfying the property \eqref{eq:devS} implies that we can use the steepest descend method to get such a numerical solution and 
%the convexity of $\Sva$ (or equivalently, the uniqueness of the solution) guarantees that  
%this method ends successfully without being trapped by meta-stable points.
%The steepest descend method to solve the master equations $\hat {\cal E}_i=0$,
% which is sometimes also referred to as the relaxation method, is explicitly given  by the following procedure.
Suppose that $S_i $(or $\Omega_i$) is a function of $\vt$ and define its time evolution as follows
\begin{align}
\frac{\p S_i(\vt)}{\p \vt}= S_i(\vt) {\cal E}_i(\vt)  \quad \left(=\hat {\cal E}_i(\vt) S_i(\vt) \right) \quad {\rm with~} q_i(z) {~\rm  fixed}  \label{eq:relaxation}
\end{align}
where ${\cal E}_i(\vt)$ is the quantity obtained by substituting $S_i=S_i(\vt)$ into ${\cal E}_i$.
Under this time evolution, $\Sva$ monotonically decrease as
\begin{align}
\frac{d}{d\vt} \Sva[{\bf \Omega}(\vt),{\bf \Omega}^{\rm ref}]=   -2\langle {\cal E}(\vt)^2\rangle,
%\int d^2x \sum_{i=1}^{N_i}\tr \left[{\cal E}_i(\vt)^2\right]\le 0,
\end{align}
and this gradient flow stops only when $\Sva$ takes the minimum value with ${\cal E}_i=0$. 
The relaxation time needed to obtain a numerical solution with a given accuracy can be estimated as follows. 
From Eq.\,\eqref{eq:relaxation}, 
we can derive the time evolution of ${\cal E}_i(\vt)$ as
\begin{align}
\frac{\p {\cal E}_i(\vt)}{\p \vt}=-2{\cal H} {\cal E}_i(\vt) \qquad {\rm with~~~} {\cal D}_{\bar z} Q_i=0. 
\label{eq:ESchrodinger}
\end{align}
The operator $\cal H$ hss positive definite eigenvalues, 
and hence if ${\cal E}_i$ is expanded in terms of the eigenmodes of $\cal H$, 
each eigenmode decays exponentially.
After a sufficient time of relaxation, 
the deviation from the true solution is dominated by 
the lowest eigenmode ${\cal E}_i^\star$ and decreases exponentially as
\begin{align}
{\cal E}_i(\vt) \approx {\cal E}^\star_i \exp(-{2 \Delta_\star }\vt ) \qquad {\rm with~~~}  {\cal H}{\cal E}^\star_i =\Delta_\star {\cal E}^\star_i,
\end{align}	
where $\Delta_\star \in \mathbb R_{>0}$ is the lowest eigenvalue. 
Thus, we can estimate the accuracy of the numerical solution using the relaxation time $\vt$, 
as long as calculation errors can be ignored. 
In the limit $t \rightarrow \infty$, $S_i(\vt)$ converges to the solution $S_i^{\rm sol}$ 
\begin{align}
\lim_{\vt\to \infty} {\cal E}_i(\vt)=0 \qquad {\rm with} \qquad  \lim_{\vt\to \infty}S_i(\vt)=S_i^{\rm sol}. 
\end{align}

\subsection{K\"ahler metric and potential for the vortex moduli space} \label{sec:Kahlermetric}
By using the vortex action $\Sva$, 
the K\"ahler potential ${\cal K}^{\rm vtx}$ giving  the metric for the vortex moduli space can be naturally introduced. 
Note that since the counter term ${\cal L}({\bf \Omega}^{\rm ref})$ introduced in Eq.\eqref{eq:defSva} has moduli-dependence, 
for the definition of ${\cal K}^{\rm vtx}$,
it is more natural to regularize the integral by introducing a spatial cut-off $R \in \mathbb R_{>0}$. 
The K\"ahler potential on the moduli space ${\cal K}^{\rm vtx}$ can be obtained by substituting 
the solution ${\bf \Omega}={\bf \Omega}^{\rm sol}$ of the master equations to the vortex action with such a regularization
\begin{align}
{\cal K}^{\rm vtx}={\cal K}^{\rm vtx}(\phi^A,\bar \phi^{\bar A})\equiv \int_{D_R}d^2x \, {\cal L}({\bf \Omega}^{\rm sol}),
\quad D_R\equiv \left\{z=x^1+i x^2 \in \mathbb C \, \big| \, |z|\le R \right\}.
\end{align}
Here, $\phi^A \in \mathbb C$ are moduli parameters, which linearly appear in the moduli matrices $q_i(z)=q_i(z,\phi^A)$ when the $V$-transformations are properly fixed. 
There is a convenient formula for the K\"ahler metric 
which is calculable without going back to the definition of ${\cal K}^{\rm vtx}$. 
If the solution $\Omega_i=\Omega_i^{\rm sol}$ is given, 
the K\"ahler metric can be calculated by using the formula 
\begin{align}
g_{A\bar B}\equiv 
\frac{\partial^2 {\cal K}^{\rm vtx}}{\partial \bar \phi^{\bar B}\partial \phi^A}=\int d^2x 
\frac{\p}{\p \bar \phi^{\bar B}}\sum_{i=1}^L \tr \left[ \Omega_{i}^{-1} 
\frac{\partial q_i}{\partial \phi^A} \Omega_{i+1}q_i^\dagger \right]_{\Omega=\Omega^{\rm sol}}. 
\label{eq:Kmetricformula}
\end{align}
In this formula, we can chenck the invariance under the $V$-transformations, 
which were fixed to define ${\cal K}^{\rm vtx}$.
Since the $V$-transformations naturally induce coordinate transitions on the vortex moduli space as explaind in Appendix.\ref{sec:ZPsiPatches}, the invariance under the $V$-transformations allows us to 
choose an arbitrary coordinate patch to describe the moduli space metric.
This metric $g_{A\bar B}$ turns out to be positive definite and thus invertible, as will be shown in the next paragraph.
Furthermore, we can show the regularity of the Riemann curvature tensor, $R^A{}_{BC\bar D}$, 
using the fact that $R^A{}_{BC\bar D}$ can be expressed in terms of the higher derivatives of $\Omega_i$ 
such as $\p_{\phi_A} \p_{\phi_B} \p_{\bar \phi_C} \Omega_i$, 
which can be determined through the differentiated forms of the master equations. 
Thanks to the existence of ${\cal H}^{-1}$, those equations are algebraically solvable. 
Therefore, the above formula implies that {\it the K\"ahler manifold defined by this metric is regular everywhere.}

\paragraph{Moduli space approximation}
The K\"ahler metric $g_{A\bar B}$ defined above is equivalent to that describing the dynamics of vortices.
In the moduli space approximation \cite{Manton:1981mp},  
moduli parameters are promoted to slowly varying functions of time $t$ 
\begin{align}
\phi^A \quad \to \quad \phi^A(t)
\end{align}
The physical fields $Q_i$ and  $A^i_{\bar z}$ 
depend on $t$ only through the moduli parameters 
\begin{align}
Q_i =Q_i^{\rm sol}(z,\bar  z, \phi^A(t), \bar \phi^A(t)), \quad A^i_{\bar z}=A_{\bar z}^{i,\rm sol}(z,\bar z, \phi^A(t), \bar \phi^A(t)).
\end{align}
The gauge potentials $A_t^i$ are given by Eq.\eqref{eq:Amu}, 
for which the linearized equations of motion are satisfied.
By substituting these approximations to the original action,
we obtain the fowlloing terms form the kinetic term:
% \begin{align}
%\frac{\p \Omega_i}{\p \phi^A}=S_i \left( {\cal H}^{-1} T_{i,A} \right)S_i^{\dagger}
%\end{align}
%\begin{align}
%q_i=q_i(z, \phi^A), \quad\p_{\bar z} q_i= \frac{\p q_i}{\p \bar \phi^{\bar B}}=0
%\end{align}
%\begin{align}
%Q_{i, A}\equiv S_i^{-1}\frac{\p q_i}{\p \phi^A} S_{i+1},\quad T_{i,A}\equiv Q_{i,A}Q_i^\dagger -Q_{i-1}^\dagger Q_{i-1,A}
%\end{align}
\begin{align}
\int d^2x \sum_{i=1}^L \tr\left[ {\cal D}_t Q_i {\cal D}_t Q_i^\dagger+\frac{4}{g_i^2} F_{tz}^i F_{t\bar z}^i\right]
=g_{A \bar B} \frac{d\phi ^A}{dt} \frac{d \bar \phi^{\bar B}}{dt},
\end{align}
where $g_{A \bar B}$ is the K\"ahler metric defined in Eq.\eqref{eq:Kmetricformula}. 
Note that this equation show that the metric defined in Eq.\eqref{eq:Kmetricformula} is positive definite.
The coincidence of the two different definitions of the metric $g_{A\bar B}$ is not accidental, 
but is due to the supersymmetry behind the system as shown in \cite{Eto:2006uw}.

\paragraph{Large coupling limit}
In the large coupling limit $g_i\to \infty $ for all $i$,  the K\"ahler potential becomes
\begin{align}
\lim_{g_i\to \infty} {\cal K}^{\rm vtx}= \int_{D_R} \! d^2x \, \sum_{i=1}^L r_i \log \det ( \xi_i(z) \xi_i(z)^\dagger) +{\rm const.} ,
\end{align} 
since   $\det \Omega_i  \sim  \log  \det ( \xi_i(z) \xi_i(z)^\dagger) $ in this limit.
This result is consistent with the moduli space approximation for instantons in the sigma models.  
Note that the above quantity diverges and thus a divergent constants must be subtracted 
to obtain a finite quantity by introducing IR cut-off $R$ and using K\"ahler transformation.
It is convenient to decompose the integrand into the two parts as
\begin{align}
\log \det ( \xi_i(z) \xi_i(z)^\dagger) =\log |\det \mD_i(z)|^2 + \log\det\left( {\bf 1}_{N_i}+\varphi_i(z)\varphi_i(z)^\dagger \right) , 
\end{align}
where we have used $\mD_i(z)^{-1}\xi_i(z)=\left( {\bf 1}_{N_i}, \varphi_i(z) \right) $, and 
$\varphi_i(z)$ is an instanton solution.
These two terms are calculated separately in the subsequent paragraphs.
%Note that since $\varphi_i(z)$ has singular configuration, even for considering instantons in the sigma models,
% the first term in the r.h.s. of the above is needed to calculate the K\"ahler potential in the well-defined way.
\paragraph{Position moduli  for  vortices}
At first glance, the contribution from the first term may seem to disappear since it can be cancelled by a K\"ahler transformation. 
However, after the regularization, careful calculations lead to the following important term 
\begin{align}
\int_{D_R}d^2x \,r_i \log |\det \mD_i(z)|^2 = k_i  \pi r_i  R^2\log \frac{R^2}{e} + \pi r_i \sum_{\alpha=1}^{k_i} |z_{(i,\alpha)}|^2
\end{align}
where $\{z_{(i,\alpha)}\}$ are zeros of $\det \mD_i(z)$. 
Here, the translational invariance is broken due to the regularization. 
Note that since $2\pi r_i$ is the tension (mass) of a vortex,
the second term in the r.h.s. of the above equation gives 
the (dominant parts of) kinetic terms of the position moduli $\{z_{(i,\alpha)}\}$,
whereas the first term is divergent in the limit of $R\to \infty$ and eliminated by the K\"ahler transformation.
It should be noted that this contribution cannot be ignored even in the sigma model.

\paragraph{Non-normalizable moduli}
The contribution from the second term gives the following divergent term in the limit of $R=\infty$ 
\begin{align}
\int_{D_R} d^2x\, r_i \log\det\left( {\bf 1}_{N_i}+\varphi_i(z)\varphi_i(z)^\dagger \right) =
\pi r_i \tr[\Phi_i^{\rm non}  (\Phi_i^{\rm non} )^\dagger] \log R +{\cal O}\left( (R)^0 \right) ,
\end{align}
where an $N_i$-by-$(N-N_i)$ matrix $\Phi_i^{\rm non}\equiv \Psi_i \widetilde \Psi_i$ appears 
in the dominant term of $\varphi_i(z)$ for large $|z|$ as
\begin{align}
\varphi_i(z)=\mD_i(z)^{-1}\widetilde \mD_i(z)= \Psi_i (z{\bf 1}_{k_i}-Z_i)^{-1}\widetilde \Psi_i = \frac1z \Psi_i \widetilde \Psi_i 
+{\cal O}\left(\frac1{z^2}\right),
\end{align}
where we used the half ADHM data $\{Z_i,\Psi_i,\widetilde \Psi_i\}$ discussed in Sec.\ref{sec:kahlerquotient}.
The above divergent term can not be eliminated by a K\"ahler transformation 
and gives a divergent kinetic terms for $\Phi_i^{\rm non}$. 
Thus, the entries of $\Phi_i^{\rm non}$ are non-normalizable moduli. 
Intuitively, this divergence is due to the fact that there is no mass gap in the bulk and 
any fluctuations of $\Phi_i^{\rm non}$ excite zero modes in the bulk 
where the moduli approximation is invalid. 
To describe the dynamics of these moduli, we need to go back to the full field equation.

%%%%%%%%%%%%%%%%%%%%%%%%%%%%%%%%%%%%%%%%
\section{Coordinate patches of moduli space for local votices}\label{sec:ZPsiPatches}
%%%%%%%%%%%%%%%%%%%%%%%%%%%%%%%%%%%%%%%%%%%%%%%
In this appendix, we present more details of 
the vortex moduli space. 
We show the equivalence of the two expressions of the vortex moduli space: one is written in terms of the moduli matrix \eqref{eq:vtx_moduli} and 
the other is written in terms of the half ADHM data \eqref{eq:moduli_space_genearl}. 
In particular, we focus on the $L=1,\,N=n$ case where the two expressions of the vortex moduli space take the forms 
(see below for more precise definitions)
\begin{align}
{\cal M}\equiv& ~ \{ \mD(z) | \det \mD(z) ={\cal O}(z^k) \}/\hbox{$V$-trf.}  \label{eq:AppcalM} \\
\wt {\cal M}\equiv& ~ \{(Z,\Psi) |  \{Z,\Psi\} \hbox {\ on which $GL(k,\mathbb C)$ action is free }\} /GL(k,\mathbb C).
\label{eq:AppcalMt}
\end{align}
We will show that both of these two spaces correspond 
to the vortex moduli space $\mathcal M_{\rm vtx}{\,}_k^{n,0}$ in the $U(n)$ gauge theory with $n$ flavors. 
The space $\mathcal M_{\rm vtx}{\,}_k^{n,0}$ 
can also be viewed as 
the local vortex moduli subspace 
in the total moduli space ${\cal M}_{\rm vtx \,}{}^{n,N-n}_k$ 
for a general flavor number $N \geq n$. 
Once the equivalence of the local vortex moduli subspace
is shown, 
the equivalence of the total moduli space immediately follows\footnote{
The total moduli space ${\cal M}_{\rm vtx \,}{}^{n,N-n}_k$
has additional directions described by $\wt \mD(z)$ and $\wt \Psi$,
which are fibered over the local vortex moduli subspace.  
The equivalence such fiber directions in \eqref{eq:vtx_moduli} and \eqref{eq:moduli_space_genearl}
follows from the one-to-one relation 
$\wt \mD(z) = \mJ(z) \wt \Psi$.}. 
Hence, we focus on the local vortex moduli subspace. 
The equivalence the two spaces above play a fundamental and important role also in the general case with $L \ge 1$,  
where the moduli space can be viewed as 
a set of $L$ copies of the $L=1$ moduli space
subject to the additional conditions.

%%%%%%%%%%%%%%%%%%%%%%%%%%%%%%%%%%%%%
\subsection%[Moduli space of moduli matrix]
{Moduli space ${\cal M}$ of moduli matrix $\mD(z)$ }\label{subsec:patch_L=1}
Let $\C_{n,k}[z]$ denote the set of $n$-by-$n$ matrices with polynomial entries whose determinants are degree $k$ polynomials.
The definition of the moduli space \eqref{eq:moduli_v_L=1} 
can be rephrased as 
\begin{align}
{\cal M} \, \equiv \, \C_{n,k}[z] / \C_{n,0}[z] \, = \, \C_{n,k}[z]/ \sim. 
\label{eq:calM_def}
\end{align}
where the equivalence relation ``$\sim$'' 
for $\mD(z),\,\mD'(z) \in \C_{n,k}[z]$ is defined as 
\begin{align}
\mD(z) \sim \mD'(z) 
~~~\Longleftrightarrow~~
{}^\exists V(z) \in \C_{n,0}[z] ~~~~
\mbox{such that} ~~~~
\mD'(z) = V(z) \mD(z).
\label{eq:V-equiv_L=1}
\end{align}
In this subsection, we provide some details on the coordinates of the moduli space. 

\subsubsection%[Atlas]
{Atlas $\{(\phi_\lambda, {\cal M}_\lambda)\}$ of $\mathcal M$}
For a given set of non-negative integers $\lambda=(l_1,l_2,\cdots,l_n)$ such that $l_1+\cdots+l_n=k$, 
let ${\cal M}_\lambda$ be the space of matrices of the form
\beq
{\renewcommand{\arraystretch}{1.0}
{\setlength{\arraycolsep}{1.0mm} 
\mD_{\lambda}(z) \equiv 
\ba{ccc} z^{l_1} & & \\ & \ddots & \\ & & z^{l_n} \ea
 - 
\ba{ccc} P^{11} & \cdots & P^{1n} \\ \vdots & \ddots & \vdots \\ P^{n1} & \cdots & P^{nn} \ea}}, 
\hs{10}
P^{ab} = \sum_{m=1}^{l_b} T^{ab}_{m} z^{m-1}. 
\label{eq:M_patch}
\eeq
Note that if two matrices of the form \eqref{eq:M_patch} are $V$-equivalent 
$\mD_{\lambda}(z) \sim \mD'_{\lambda}(z)$, 
it follows that they are actually identical matrices 
$\mD_{\lambda}(z)= \mD'_{\lambda}(z)$. 
This is because if a $V$-equivalent pair $\mD(z)$ and $\mD'(z)$ have the same leading order behavior as \eqref{eq:M_patch}, 
the $V$-transformation relating the pair 
$V(z) = \mD'(z) \, \mD(z)^{-1}$ 
behaves as $V(z)={\bf 1}_n+{\cal O}(z^{-1})$ for large $|z|$ 
and hence its regularity implies that $V(z) = \mD'(z) \, \mD(z)^{-1} = {\bf 1}_n$.
Therefore, each element of $\mathcal M_{\lambda}$ specifies a distinct $V$-equivalence class and hence $\mathcal M_{\lambda}$ can be viewed as a subspace of $\mathcal M$. 
The local coordinate system on the coordinate patch (chart) ${\cal M}_{\lambda}$ with $\lambda=(l_1,\cdots,l_n)$ is given by
\beq
\phi_\lambda ({\cal M}_\lambda) \simeq \left\{ T^{ab}_{m} \, \big| \, 1\le a,b \le n, 1\le m\le l_b\right\} \simeq \mathbb C^{kn},
\eeq
where $T^{ab}_{m}$ are coefficients of the polynomials $P^{ab}(z)$. 
Two coordinate patches $\mathcal M_{\lambda}$ and $\mathcal M_{\lambda'}$ with $\lambda \not = \lambda'$ are glued by the coordinate transformation $\tau_{\lambda', \lambda} = \phi_{\lambda'} \circ \phi_{\lambda}^{-1}$ which can be read off from the 
$V$-transformation relating $\mD_{\lambda}(z) \in \mathcal M_{\lambda}$ and $\mD_{\lambda'}(z) \in \mathcal M_{\lambda'}$:
\beq
\mD_{\lambda'}(z) = V_{\lambda',\lambda}(z) \, \mD_{\lambda}(z), \hs{10}
V_{\lambda',\lambda}(z) \in \mathbb C_{n,0}[z],
\eeq
the explicit form of the coordinate transformation 
can be determined by requiring that all the entries of 
$V_{\lambda',\lambda}(z) = \mD_{\lambda'}(z) \mD_{\lambda}(z)^{-1}$ are regular.
Gluing all the coordinate patches $\mathcal M_{\lambda}$,
we obtain a complex manifold
\begin{align}
{\cal M}'= \bigcup_{\lambda\in \Lambda} {\cal M}_{\lambda} \quad {\rm with} \quad  
\Lambda =\left\{ \, (l_1,l_2,\cdots,l_n) \ \bigg| \ l_a \in \mathbb Z_{\ge 0}, \ \sum_{a=1}^nl_a=k\right\}.
\label{eq:calMatlas}
\end{align}
This is a submanifold of $\mathcal M$ ($\mathcal M' \subseteq \mathcal M$), since $\mathcal M_{\lambda} \subset \mathcal M$ for all $\lambda \in \Lambda$. 
In subsection \eqref{subsec:disjoint_union}, 
we show that
$\mathcal M$ can be decomposed into a disjoint union of subspaces
$\mathcal M_{\lambda}^{\rm tri}$ such that $\mathcal M_{\lambda}^{\rm tri} \subset \mathcal M_{\lambda}$. 
This fact implies that $\mathcal M$ is a subspace of $\mathcal M'$: 
\beq
\mathcal M = \bigsqcup_{\lambda \in \Lambda} \mathcal M_{\lambda}^{\rm tri} \subseteq \bigcup_{\lambda\in \Lambda} {\cal M}_{\lambda} = \mathcal M'.
\eeq
Since ${\cal M}' \subseteq {\cal M}$ and ${\cal M}' \supseteq {\cal M}$, we conclude that ${\cal M}={\cal M}'$ and 
$\{(\phi_\lambda, {\cal M}_\lambda)\}$ gives an atlas of $\mathcal M$. 

%%%%%%%%%%%%%%%%%%%%%%%%%%
\subsubsection{Decomposition into  disjoint union}
\label{subsec:disjoint_union}
Here we show the decomposition of the moduli space $\mathcal M$ 
into a disjoint union of subspaces 
$\mathcal M_{\lambda}^{\rm tri}$ such that 
$\mathcal M_{\lambda}^{\rm tri} \subset \mathcal M_{\lambda}$, 
which we have used to show that ${\cal M}' \supseteq {\cal M}$.  
The moduli space $\mathcal M$ is the space of the $V$-equivalence classes of the matrix $\mD(z)$. 
We can pick up a representative in each $V$-equivalence class by fixing the ``gauge redundancy" of the $V$-transformation \eqref{eq:V-equiv_L=1} in the following way. 
Here we focus on the case of $n=2$ for simplicity. 
Let $\mD(z)$ be a generic element in $\C_{2,k}[z]$
\beq
\mD(z) = \ba{cc} f(z) & h(z) \\ g(z) & i(z) \ea \in \C_{2,k}[z],
\eeq
where $f(z),g(z),h(z)$ and $i(z)$ are polynomials. 
Using the Euclidean algorithm (B\'ezout's identity), 
we can show that there exist polynomials $\tilde f(z)$ and $\tilde g(z)$ such that 
\beq
f(z) \tilde f(z) + g(z) \tilde g(z) = p(z),
\eeq
where $p(z)$ is the polynomial greatest common divisor of $f(z)$ and $g(z)$. 
Using $\tilde f(z), \tilde g(z)$ and $p(z)$, 
we can construct $\mathcal V(z) \in \C_{2,0}[z]$ with which 
$\mD(z)$ is transformed into an upper triangular form
\beq
\mD(z) \rightarrow \mathcal V(z) \, \mD(z) = \ba{cc} \tilde f(z) & \tilde g(z) \\ - q_g(z) & q_f(z) \ea \ba{cc} f(z) & h(z) \\ g(z) & i(z) \ea = \ba{cc} p(z) & h'(z) \\ 0 & i'(z) \ea,
\eeq
where $q_f(z)$ and $q_g(z)$ are the polynomials defined by $f(z)=q_f(z)p(z)$ and $g(z) = q_g(z) p(z)$. 
Note that $\tilde f(z)$ and $\tilde g(z)$ are not unique, 
that is, we can further multiply another 
$V$-transformation without changing the upper-triangular form
\beq
\mathcal V(z) \, \mD(z) \rightarrow \mathcal V'(z) \mathcal V(z) \, \mD(z) = 
\ba{cc} a_1 & j(z) \\ 0 & a_2 \ea \ba{cc} p(z) & h'(z) \\ 0 & i'(z) \ea,
\eeq
where $a_1$ and $a_2$ are constants and $j(z)$ is a polynomial.
This redundancy can be fixed by requiring that the diagonal entries are monic polynomials and minimizing the degree of the upper-right element 
\beq
\mathcal V'(z) \mathcal V(z) \, \mD(z) = 
\ba{cc} z^{l_1} & 0 \\ 0 & z^{l_2} \ea + 
\ba{cc} P_{11}(z) & P_{12}(z) \\ 0 & P_{22}(z) \ea,
\label{eq:standard_form}
\eeq
where $l_i~(i=1,2)$ are the degrees of the diagonal engries $(l_1 = {\rm deg} \, p(z),\,l_2={\rm deg}\,i'(z))$ and $P_{ab}~(1 \leq a \leq b \leq 2)$ are polynomials of degree less than $l_b$.
In each $\mathcal V$-equivalence class, 
the form \eqref{eq:standard_form} is unique 
and hence the gauge redundancy is completely fixed. 
This procedure can be generalized to the case of general $n$; 
for any $\mD(z) \in \mathbb C_{n,k}[z]$, we can find 
$\mathcal V(z) \in \mathbb C_{n,0}[z]$ such that
\begin{align}
\mD(z) \rightarrow \mathcal V(z) \mD(z) = \mD_{\lambda}^{\rm tri}(z) = 
{\renewcommand{\arraystretch}{0.8}
{\setlength{\arraycolsep}{0.7mm} 
\ba{cccc} z^{l_1} & & & \\ & z^{l_2} & & \\ & & \ddots & \\ & & & z^{l_n} \ea
 -
\ba{cccc} P_{11} & P_{12} & \cdots & P_{1n} \\ 0 & P_{22} & \ddots & \vdots\\
\vdots &\ddots &\ddots & P_{n-1,n} \\ 0& \cdots & 0 & P_{nn} \ea
}},
\label{eq:Dutrigform}
\end{align}
where $\lambda=(l_1,l_2,\cdots,l_n)$ is a set of non-negative integers 
such that $l_1+\cdots+l_n=k$ and 
$P_{ab}(z)~(1 \leq a \leq b \leq n)$ are polynomials of degree less than $l_b$. 
Since any $\mathcal V$-equivalence class has a unique representative 
of the form \eqref{eq:Dutrigform}, 
the moduli space \eqref{eq:AppcalM2} can be decomposed 
into the disjoint union of the subspaces 
\beq
\mathcal M = \bigsqcup_{\lambda \in \Lambda} \mathcal M_{\lambda}^{\rm tri},
\label{eq:M_disjoint}
\eeq
where $\mathcal M_{\lambda}^{\rm tri}$ are
the sets of matrices $\mD_{\lambda}^{\rm tri}(z)$ 
of the form \eqref{eq:Dutrigform} specified by the set of non-negative integers $\lambda=(l_1,\cdots,l_n)$. 
Note that there is no overlap between them, i.e. 
$\mathcal M_{\lambda}^{\rm tri} \cap \mathcal M_{\lambda'}^{\rm tri} = \emptyset$ for $\lambda \not = \lambda'$. 
Since the form of the matrix \eqref{eq:Dutrigform} 
is a special case of \eqref{eq:M_patch}, 
it follows that $\mathcal M_{\lambda}^{\rm tri}$ is a subspace of $\mathcal M_{\lambda}$. 
Therefore, we conclude that $\mathcal M$ is a subspace of $\mathcal M'$
\beq
\mathcal M = \bigsqcup_{\lambda \in \Lambda} \mathcal M_{\lambda}^{\rm tri} \subseteq \bigcup_{\lambda\in \Lambda} {\cal M}_{\lambda} = \mathcal M'.
\eeq

\subsubsection{Matrices with non-polynomial entries}
So far, entries of the matrices $\mD(z)$ and $\mathcal V(z)$ are 
assumed to be polynomials for simplicity.
Strictly speaking, the most general solution
can have arbitrary entire functions 
of $z \in \mathbb C$ as their entries. 
That is, $\mD(z)$ and $\mathcal V(z)$ are not necessarily elements of 
$\mathbb C_{n,k}[z]$ and $\mathbb C_{n,0}[z]$ 
but they can be elements of larger spaces ${\cal G}_{n,k}[z]$ and ${\cal G}_{n,0}[z]$: 
\begin{align}
\mD(z) \in {\cal G}_{n,k}[z], \quad \quad 
\mathcal V(z) \in {\cal G}_{n,0}[z],
\end{align}
where ${\cal G}_{n,k}[z]$ are
the space of maps from the complex $z$-plane $\C$ 
to the space of $n$-by-$n$ square matrices 
that have the following properties 
\begin{itemize}
\item 
If $X(z) \in \mathcal G_{n,k}[z]$, all the entries of $X(z)$ 
are entire functions of $z \in \mathbb C$ ($\partial_{\bar z}X(z)=0$),
\item  
If $X(z) \in \mathcal G_{n,k}[z]$, $\det X(z)$ has $k$ zeros on $\mathbb C$ and 
${\rm rank} \, X(z) = n$ except at the zeros of $\det X(z)$. 
\end{itemize}  
It worth noting that ${\cal G}_{n,0}[z]$ forms a group under  matrix multiplication\footnote{
For general integers $p \geq 0$ and $q \geq 0$, 
multiplication of elements of ${\cal G}_{n,p}[z]$ and ${\cal G}_{n,q}[z]$ defines 
a map ${\cal G}_{n,p}[z] \times {\cal G}_{n,q}[z] \rightarrow {\cal G}_{n,p+q}[z]$  
\begin{align}
(X(z), Y(z)) \in {\cal G}_{n,p}[z] \times {\cal G}_{n,q}[z] \quad \mapsto \quad X(z) \, Y(z) \in {\cal G}_{n,p+q}[z]. \notag
\end{align}
}.

\paragraph{Definition of $\cal M_{\mathcal G}$}
We define the space $\mathcal M_{\mathcal G}$ as 
\begin{align}
{\cal M}_{\mathcal G} \equiv {\cal G}_{n,k}[z] / {\cal G}_{n,0}[z]= {\cal G}_{n,k}[z]/ \sim. 
\label{eq:AppcalM2}
\end{align}
where the equivalence relation ``$\sim$'' 
for elements $\mD(z)$ and $\mD'(z)$ in ${\cal G}_{n,k}[z]$ is defined by 
\begin{align}
\mD(z) \sim \mD'(z) 
~~~\Longleftrightarrow~~
{}^\exists V(z) \in {\cal G}_{n,0}[z] ~~~~
\mbox{such that} ~~~~
\mD'(z) = V(z) \mD(z).
\end{align}
Replacing ${\cal G}_{n,k}[z]$ with the subspace ${\C}_{n,k}[z] \subset {\cal G}_{n,k}[z]$ 
consisting of matrices with polynomial entries, 
we can go back to the definition of the moduli space $\mathcal M$ given in \eqref{eq:calM_def}.
Although \eqref{eq:AppcalM2} is the most general definition of the moduli space, it actually gives the same space as \eqref{eq:calM_def}, that is, $\mathcal M_{\mathcal G} = \mathcal M$. 
Therefore, we can use the simpler definition of the moduli space $\mathcal M$ based on the space of matrices with polynomial entries ${\C}_{n,k}[z]$. 

To see that $\mathcal M_{\mathcal G} = \mathcal M$, 
let us show that any matrix $\mD(z) \in \mathcal G_{n,k}[z]$ 
can be fixed into the upper-triangular form \eqref{eq:Dutrigform} with polynomial entries
by an element of the $V$-transformation $\mathcal V(z) \in \mathcal G_{n,0}[z]$. 
We will use the following two theorems, 
which will be proven in 
subsections \ref{subsec:division} and \ref{subsec:Euclid}: 

\begin{thm}\label{thm:division}
Let $(f(z), g(z))$ an arbitrary pair of entire functions.
If $g(z)$ has $m$ zeros, 
then, there is a pair of a polynomial $p(z)$ of degree less than $m$ and an entire function $h(z)$
such that,
\begin{align}
f(z)=p(z)+h(z) g(z).
\end{align}
\end{thm}

\begin{thm}\label{thm:Euclid}
For a pair of entire functions $(f(z), g(z))$,
there is an element $V(z)$ of ${\cal G}_{2,0}[z]$ such that 
\begin{align}
\exists V(z) \in {\cal G}_{2,0}[z], \exists p(z), \quad 
V(z) \begin{pmatrix}
f(z)\\ g(z)
\end{pmatrix}=\begin{pmatrix}
p(z)\\0
\end{pmatrix},
\end{align}
where $p(z)$ is a certain entire function whose a set of zeros is the intersection of those of $f(z)$ and $g(z)$ including their multiplicities. 
In particular, the first row of the above equation indicates that
their is a pair of entire functions, $(\tilde f(z), \tilde g(z))$ 
such that 
\begin{align}
f(z) \tilde f(z)+g(z) \tilde g(z) =p(z).
\label{eq:Bezout_entire}
\end{align}
\end{thm}

\paragraph{Decomposition into  disjoint union}
Replacing the polynomials in subsubsection \ref{subsec:disjoint_union} with entire functions and using 
Theorems \ref{thm:division} and \ref{thm:Euclid}, 
we can show that $\mathcal M_{\mathcal G}$ can be decomposed into 
the disjoint union of $\mathcal M_\lambda^{\rm tri}$. 
We again focus on the case of $n=2$ for simplicity. 
Let $\mD(z)$ be a generic element in $\mathcal G_{2,k}[z]$
\beq
\mD(z) = \ba{cc} f(z) & h(z) \\ g(z) & i(z) \ea \in \mathcal G_{2,k}[z],
\eeq
where $f(z),g(z),h(z)$ and $i(z)$ are 
entire functions. 
Using Theorem \ref{thm:Euclid}, 
we can find $\tilde f(z), \tilde g(z)$ and $p(z)$ satisfying Eq.\,\eqref{eq:Bezout_entire}. 
Using such $\tilde f(z)$, $\tilde g(z)$ and $p(z)$ 
we can construct $\mathcal V(z) \in \mathcal G_{2,0}[z]$ with which 
$\mD(z)$ is transformed into an upper triangular form
\beq
\mD(z) \rightarrow \mathcal V(z) \, \mD(z) = \ba{cc} \tilde f(z) & \tilde g(z) \\ - q_g(z) & q_f(z) \ea \ba{cc} f(z) & h(z) \\ g(z) & i(z) \ea = \ba{cc} p(z) & h'(z) \\ 0 & i'(z) \ea,
\eeq
where $q_f(z)$ and $q_g(z)$ are the entire functions defined by $f(z)=q_f(z)p(z)$ and $g(z) = q_g(z) p(z)$. 
Note that $p(z)$ and $i'(z)$ have finite number of zeros since $\mD(z) \in \mathcal G_{n,k}[z]$ with finite $k$. 
The functions $\tilde f(z)$ and $\tilde g(z)$ are not unique, 
that is, we can further multiply another 
$V$-transformation without changing the upper-triangular form
\beq
\mathcal V(z) \, \mD(z) \rightarrow \mathcal V'(z) \mathcal V(z) \, \mD(z) = 
\ba{cc} a_1(z) & j(z) \\ 0 & a_2(z) \ea \ba{cc} p(z) & h'(z) \\ 0 & i'(z) \ea,
\eeq
where $j(z)$ is an arbitrary entire functions and $a_1(z)$, $a_2(z)$ are entire functions without zero. 
This redundancy can be fixed by requiring that the diagonal entries are monic polynomials and fixing the upper-right element to be the minimum degree polynomial using Theorem \ref{thm:division}
\beq
\mathcal V'(z) \mathcal V(z) \, \mD(z) = 
\ba{cc} z^{l_1} & 0 \\ 0 & z^{l_2} \ea + 
\ba{cc} P_{11}(z) & P_{12}(z) \\ 0 & P_{22}(z) \ea,
\label{eq:standard_formG}
\eeq
where $l_i~(i=1,2)$ are the numbers of zeros of the diagonal engries and $P_{ab}~(1 \leq a \leq b \leq 2)$ are polynomials of degree less than $l_b$.
In each $\mathcal V$-equivalence class, 
the form \eqref{eq:standard_formG} is unique 
and hence the gauge redundancy is completely fixed. 
This procedure can be generalized to the case of general $n$
\begin{align}
\mD(z) \rightarrow \mathcal V(z) \mD(z) = \mD_{\lambda}^{\rm tri}(z) = 
{\renewcommand{\arraystretch}{0.8}
{\setlength{\arraycolsep}{0.7mm} 
\ba{cccc} z^{l_1} & & & \\ & z^{l_2} & & \\ & & \ddots & \\ & & & z^{l_n} \ea
 -
\ba{cccc} P_{11} & P_{12} & \cdots & P_{1n} \\ 0 & P_{22} & \ddots & \vdots\\
\vdots &\ddots &\ddots & P_{n-1,n} \\ 0& \cdots & 0 & P_{nn} \ea
}},
\label{eq:DutrigformG}
\end{align}
where $\lambda=(l_1,l_2,\cdots,l_n)$ is a set of non-negative integers 
such that $l_1+\cdots+l_n=k$ and 
$P_{ab}(z)~(1 \leq a \leq b \leq n)$ are polynomials of degree less than $l_b$. 
Since any $\mathcal V$-equivalence class has a unique representative 
of the form \eqref{eq:DutrigformG}, 
the moduli space \eqref{eq:AppcalM2} can be decomposed 
into the disjoint union in a similar way as $\mathcal M$
\beq
\mathcal M_{\mathcal G} = \bigsqcup_{\lambda \in \Lambda} \mathcal M_{\lambda}^{\rm tri},
\label{eq:M_G_disjoint}
\eeq
where $\mathcal M_{\lambda}^{\rm tri}$ are
the sets of matrices $\mD_{\lambda}^{\rm tri}(z)$ 
of the form \eqref{eq:Dutrigform} specified by the set of non-negative integers $\lambda=(l_1,\cdots,l_n)$. 

The decomposition implies that $\mathcal M_{G}$ is 
a subspace of $\mathcal M'$ defined in Eq.\,\eqref{eq:calMatlas} 
\beq
\mathcal M_{\mathcal G} = \bigsqcup_{\lambda \in \Lambda} \mathcal M_{\lambda}^{\rm tri} \subseteq \bigcup_{\lambda\in \Lambda} {\cal M}_{\lambda} = \mathcal M'.
\eeq
On the other hand, $\mathcal M'$ is obviously 
a submanifold of $\mathcal M_{\mathcal G}$. 
Therefore, $\mathcal M_{\mathcal G} = \mathcal M'$ and 
hence definition of the moduli space $\mathcal M = \mathcal M'$ 
defined in terms of polynomials
and $\mathcal M_{\mathcal G}$ 
defined in terms of entire functions are equivalent.

%We can show that \eqref{eq:AppcalM2} and \eqref{eq:moduli_v_L=1} are equivalent by using the fact that in the equivalence class of $\mD(z)$, there exists a unique matrix $\mD^{\rm tri}(z)$ of the form
%\beq
%{\renewcommand{\arraystretch}{1.0}
%{\setlength{\arraycolsep}{1.0mm} 
%\mD^{\rm tri}(z) = 
%\ba{cccc} z^{l_1} & & & \\ & z^{l_2} & & \\ & & \ddots & \\ & & & z^{l_n} \ea
% - 
%\ba{cccc} P_{11} & P_{12} & \cdots & P_{1n} \\ 0 & P_{22} & \ddots & \vdots\\
%\vdots &\ddots &\ddots & P_{n-1,n} \\ 0& \cdots & 0 & P_{nn} \ea,
%}}
%\eeq
%where $(l_1,l_2,\cdots,l_N)$ is a set of non-negative integers and $P_{ab}~(a \leq b)$ are polynomials of degree less that $l_b$. In the following, we show that \eqref{eq:AppcalM2} is equivalent to \eqref{eq:moduli_v_L=1} by looking at the atlas covering the manifold $\cal M$.

%%%%%%%%%%%%%%%%%%%%%%%%%%%%%%%
 \subsubsection{Proof of Theorem \ref{thm:division}}
 \label{subsec:division}
In the previous subsection we have used Theorem \ref{thm:division} 
to show that $\mD(z)$ can always be fixed as \eqref{eq:Dutrigform}. 
Here we give the proof of the theorem.
\paragraph{Theorem \ref{thm:division}.}
{\it Let $(f(z), g(z))$ an arbitrary pair of entire functions on $\C$.
If $g(z)$ has $m$ zeros, 
then, there is a pair of a polynomial $p(z)$ of degree less than $m$ and an entire function $h(z)$}
such that,
\begin{align}
f(z)=p(z)+h(z) g(z).
\end{align}
\begin{proof}
We first assume that $g(z)$ is a polynomial of degree $m$. 
Suppose $z=a$ is a zero of $g(z)$ with multiplicity $l_a$. 
Let $p_a(z)$ be the polynomial related to the principal part of the Laurent series of $f(z)/g(z)$ at $z=a$ as
\begin{align}
\frac{f(z)}{g(z)}&=\frac1{(z-a)^{l_a}}p_a(z) + {\rm regular~term}.
\end{align}
Using $p_a(z)$ defined for the all zeros of $g(z)$, 
we can define an entire function $h(z)$ as 
\begin{align}
h(z) \equiv \frac{f(z)}{g(z)}- \sum_{a \in Z(g)}\frac{p_a(z)}{(z-a)^{l_a}},
\end{align}
where $Z(g)$ denotes the set of zeros of $g(z)$. 
Multiplying $g(z)$ to the both sides of the above,
we find that
\begin{align}
f(z)= h(z) g(z) + p(z),
\label{eq:devision_proof}
\end{align}
where $p(z)$ is given by
\begin{align}
p(z) = \sum_{a \in Z(g)} p_a(z) \frac{g(z)}{(z-a)^{l_a}}. 
\end{align}
Note that $p(z)$ is the polynomial of degree less than $m$  
\beq
p(z) = \sum_a\sum_{n=1}^{l_a} \frac{f^{(n-1)}(a)}{(n-1)!} e_{a,n}(z)
\eeq
where $\{e_{a,n}\}$ is the basis of polynomials of degree less than $m$ 
defined by
\begin{align}
e_{a,n}(z)=\sum_{q=0}^{l_a-n} \frac{1}{q!} \partial_z^{q} \left[ \frac{(z-a)^{l_a}}{g(z)} \right]_{z=a} \times (z-a)^{q+n-1-l_a} g(z),\quad n=1,\cdots,l_a.
\end{align}
%In particular, if $l_a=1$ for all $a$, then $p(z)$ is regarded as the interpolation polynomial in the Legendre form:
%\begin{align}
%p(z)= \sum_{a \in Z(g)} f(a) \, e_{a,1}(z) ~~
%{\rm with~~} e_{a,1}(z) = \prod_{b \in Z(g) \backslash \{a\}}
%\frac{z-b}{a-b},\quad {\rm for~} a \in Z(g). 
%\end{align}

If $g(z)$ is not a polynomial but a generic entire function with $m$ zeros, there exists a polynomial $\tilde g(z)$ of degree m and a entire function $X(z)$ such that $g(z) = e^{X(z)} \tilde g(z)$. As shown above, there exist a polynomial $p(z)$ of degree less than m and a entire function $\tilde h(z)$ such that $f(z) = \tilde h(z) \tilde g(z) + p(z)$. Writing $h(z)=\tilde h(z) e^{-X(z)}$, we find that 
\beq
f(z) = h(z) g(z) + p(z).
\eeq
\end{proof}

 %%%%%%%%%%%%%%%%%%%%%%%%%%%%
\subsubsection{Proof of Theorem \ref{thm:Euclid}}\label{subsec:Euclid}
To prove Theorem \ref{thm:Euclid}, let us first show the following lemma
\begin{lmm}
If  $(f(z), g(z))$ is a pair of entire functions of $z$ 
which have no common zero,
then there is a pair of entire functions, $(\tilde f(z), \tilde g(z))$ 
such that 
\begin{align}
f(z) \tilde f(z)+g(z) \tilde g(z) =1.
\end{align}
\end{lmm}
\begin{proof}	
Let $A_f=\{a_n|n\in \mathbb N\}$ and $A_g=\{b_n| n\in \mathbb N\}$ 
be the ordered sets of zeros of $f(z)$ and $g(z)$, respectively.
According to Mittag-Leffler's theorem, 
we can construct a function $h_f(z)$ such that 
the set of poles of $h_f(z)$ is 
in one-to-one correspondence with the set of zeros of $f(z)$, 
and the principal part at $z = a_n \in A_f$ is $P_{a_n}(z)$
\beq
h_f(z) = P_{a_n}(z) + \{ \, \mbox{terms regular at $z=a_n$} \, \}, \hspace{5mm} {}^\forall a_n \in A_f, 
\eeq
where $P_a(z)$ is the principal part of the function $h(z)=(f(z) g(z))^{-1}$ at $z=a \in A_f \cup A_g$
\beq
h(z) = \frac{1}{f(z) g(z)} = P_a(z) + \{ \, \mbox{terms regular at $z=a$} \, \}. 
\eeq
Similarly, we can construct a function $h_g(z)$. 
The, the function $r(z)$ defined by
\begin{align}
r(z) \equiv  h(z)-h_f(z)-h_g(z) \label{eq:r(z)}
\end{align}
is an entire function since the set of poles of $h(z)$ is $A_f \cup A_g$ and all the poles in the r.h.s. are exactly cancelled. 
In addition, the following two functions are also entire functions:
\begin{align}
\tilde f(z)\equiv g(z) (h_g(z)+r(z)), \quad \tilde g(z) \equiv f(z) h_f(z)  \label{eq:tildefg}
\end{align}
where all poles of $h_f(z)$ and $h_g(z)$ are cancelled with 
the corresponding zeros of $f(z)$ and $g(z)$, respectively. 
By multiplying $h(z)^{-1}=f(z) g(z)$ 
to the both side of Eq.\eqref{eq:r(z)}, we find that
\begin{align}
1=f(z)\tilde f(z)+g(z) \tilde g(z).
\end{align}
Here the pair $(\tilde f(z), \tilde g(z))$ constructed above is not the general solution but a special solution.
Different point sequences $A_f'$ and $A_g'$ obtained by switching the order infinitely many times can give a different solution.
\end{proof}
As an example, let us consider $(f(z),g(z))=(\sin^3(z),\cos(z))$, $h_f(z)$. The functions $h_f(z)$, $h_g(z)$ and $r(z)$ are given by
\begin{align}
h_f(z)&=\frac{1}{z^3}+\frac{1}{z}+\sum_{n=1}^\infty\left(\frac1{(z-n\pi)^3}+\frac1{z-n\pi}+\frac1{(z+n\pi)^3}+\frac1{z+n\pi}\right)=\frac{1}{\tan z}\left(1+\frac1{\sin^2z}\right),\nn
h_g(z)&=-\sum_{n=1}^\infty\left(\frac1{z-(n-\frac12)\pi}+\frac1{z+(n-\frac12)\pi}\right)=\tan z,\quad 
r(z)=\frac1{\sin^3 (z) \cos (z)}-h_f(z)-h_g(z)=0,
\end{align}
and thus a special solution of  $(\tilde f(z), \tilde g(z))$ is given as
\begin{align}
\tilde f(z)=\cos(z) h_g(z)=\sin(z), \quad \quad \quad \tilde g(z)=\sin^3 (z) h_f(z)=\cos(z) (1+\sin^2(z)).
\end{align}
\paragraph{Theorem \ref{thm:Euclid}.}
{\it For a pair of entire functions $(f(z), g(z))$,
there is an element $V(z)$ of ${\cal G}_{2,0}[z]$ such that 
\begin{align}
{}^\exists V(z) \in {\cal G}_{2,0}[z], ~ {}^\exists p(z), \quad \quad 
V(z) \begin{pmatrix}
f(z) \\ g(z)
\end{pmatrix}=\begin{pmatrix}
p(z)\\0
\end{pmatrix},
\label{eq:c.2}
\end{align}
where $p(z)$ is a certain entire function whose set of zeros 
is the intersection of those of $f(z)$ and $g(z)$ including their multiplicities. 
In particular, the first row of the above equation indicates that
their is a pair of entire functions, $(\tilde f(z), \tilde g(z))$ 
such that 
\begin{align}
f(z) \tilde f(z)+g(z) \tilde g(z) =p(z).
\end{align}}
\begin{proof}
Let $A(f,g)$ be an intersection of sets of zeros of $f(z)$ and $g(z)$ including their multiplicities.
According to Weierstrass factorization theorem,
there exists an entire function $p(z)$ whose set of zeros is $A(f,g)$.
Then, functions $f_0(z)$ and $g_0(z)$ defined by
\begin{align}
f_0(z) \equiv \frac{f(z)}{p(z)}, \quad 
g_0(z) \equiv \frac{g(z)}{p(z)}, 
\end{align}
are entire functions without common zero, 
and hence we can apply the above lemma to find 
a pair of entire functions $(\tilde f(z), \tilde g(z))$ satisfying   
\begin{align}
f_0(z) \tilde f(z)+g_0(z) \tilde g(z)=1 \quad \Rightarrow \quad 
f(z) \tilde f(z)+g(z) \tilde g(z)=p(z).
\end{align}
Then, we can construct the matrix $V(z) \in {\cal G}_{2,0}[z]$ satisfying \eqref{eq:c.2} as 
\begin{align}
V(z) =\begin{pmatrix}
\tilde f(z) & \tilde g(z) \\ -g_0(z) & f_0(z) 
\end{pmatrix}.
\end{align}
\end{proof}
%%%%%%%%%%%%%%%%%%%%%%%%%%%%%%%%%%%%%%%%%%

%%%%%%%%%%%%%%%%%%%%%%%%%
\subsection%[MOduli space of the half-ADHM data]
{Moduli space $\wt {\cal M}$ of the half-ADHM data}\label{sec:ADHMdata}
In this subsection, we discuss the moduli space of the half-ADHM data, which is given by the $GL(N,\C)$ quotient of the vector space of matrices $\{ Z, \Psi\}$ 
\begin{eqnarray}
\wt{\cal M} \, \cong \, \left\{ (Z,\Psi) \, \Big| \, \{Z,\Psi\} \ \mbox{on which $GL(k,\mathbb C)$ action is free} \right\}/GL(k,\mathbb C),
\end{eqnarray}
where $Z$ $k$-by-$k$ matrix and $\Psi$ is a $n$-by-$k$ matrix on which $GL(k,\C)$ acts
\beq
Z \rightarrow g^{-1} Z g, \hs{5} 
\Psi \rightarrow \Psi g. 
\eeq
The condition that the $GL(k,\mathbb C)$ action is free 
means that there is no non-trivial $g \in GL(k,\mathbb C)$ 
that fixes $(Z,\Psi)$. 
To examine the moduli space $\wt{\cal M}$, 
it is convenient to rewrite the $GL(k,\C)$ free condition 
as we show below. 

\subsubsection%[Two expressions of GL(k,C) free condition]
{Two expressions of $GL(k,\mathbb C)$ free condition}
\label{sec:TwoExpGLkCfree}
The moduli space of the half-ADHM data is the $GL(k,\C)$ quotient 
of the space of matrices $\{\Psi,Z\}$ on which the $GL(k,\C)$ action is free. 
There are two equivalent conditions for $(\Psi, Z)$ to be 
a pair of matrices on which $GL(k,\mathbb C)$ action is free:\footnote{It is convenient to rewrite the second condition, 
${\cal C}_2(\Psi,Z)$, as 
\begin{align}
    \exists \vec v: {\rm ~a~column~vector~}, \forall z\in \mathbb C,
    \quad  \Psi (z{\bf 1}_n-Z)^{-1} \vec v=0 \quad 
    \Rightarrow\quad \vec v=0.
\end{align}} 
\begin{alignat}{2}
&{\cal C}_1(\Psi,Z):~ 
\mbox{for}~X \in \mathfrak{gl}(k,\C), \quad 
&\quad& \Psi X=0,~[Z,X]=0 ~ \Rightarrow ~ X=0, \\
&{\cal C}_2(\Psi,Z):~ 
\mbox{for}~\vec v \in \C^k, \quad &\quad&
\Psi Z^a \vec v=0 ~~ \mbox{for} ~~ a=0,1,\dots k-1 ~ \Rightarrow ~ \vec v=0. \label{eq:cond_2}
\end{alignat}
To show the equivalence of these conditions, 
let us consider the inclusion relation between the following two sets 
\begin{align}
F_1\equiv \left\{(\Psi,Z) \, \big| \, {\cal C}_1(\Psi,Z) \right\}, 
\quad F_2\equiv \left\{(\Psi,Z) \, \big| \, {\cal C}_2(\Psi,Z) \right\}. \label{eq:TwoF}
\end{align}

\begin{proof}[Proof of $F_2 \subseteq F_1$]
If $\Psi X=0$ and $[Z,X] = 0$ for $X \in \mathfrak{gl}(k,\mathbb C)$, 
it follows that $\Psi Z^a X = 0$ for $a=1,2,\cdots$. 
For an element $(\Psi,Z) \in F_2$, $\Psi Z^a X=0$ implies that $X=0$
and hence 
\begin{align}
\Psi X=0,~[Z,X]=0 \quad \Rightarrow \quad \Psi Z^a X = 0 ~~{\rm for}~~ a=1,2,\dots \quad \Rightarrow \quad X=0.
\end{align}
This shows that $(\Psi,Z) \in F_2 \Rightarrow (\Psi,Z) \in F_1$, or equivalently $F_2 \subseteq F_1$.
\end{proof}
\begin{proof}[Proof of $F_2 \supseteq F_1$]
Here we prove that $F_2 \supseteq F_1$ by 
showing that $\overline{F_2} \subseteq \overline{F_1}$, 
where $\overline{F_1}$ and $\overline{F_2}$ 
are the complements of $F_1$ and $F_2$, respectively.
To show $\overline{F_2} \subseteq \overline{F_1}$, 
we show that there exists a nontrivial $X \in \mathfrak{gl}(k,\C)$ satisfying $\Psi X = 0$ and $[Z,X]=0$
for any element $(\Psi, Z) \in \overline{F_2} = \{(\Psi,Z)\} \backslash F_2$.
If $(\Psi, Z) \in \overline{F_2}$, 
there exist a set of $l$ linearly independent column vectors 
$\{\vec v_p| p=1,\dots, l \le k \}$ satisfying 
$\Psi Z^a \vec v_p=0$ for all $a \in \mathbb Z_{\ge 0}$.
After an appropriate $GL(k,\mathbb C)$ transformation,
therefore,
$\Psi Z^a$ take the following form
\begin{align}
\Psi Z^a =
\begin{pmatrix}
\star \, \Big| \, {\bf 0}_{\tiny k\hbox{-by-} l}
\end{pmatrix} 
\quad {\rm for~} a=0,1,\dots,k-1.
\end{align}
Under this gauge choice,  $Z$ takes the following form
\begin{align}
Z=\begin{pmatrix}
Z_+& {\bf 0}\\  W& Z_-
\end{pmatrix}
\end{align} 
where $Z_+ \in M_{k-l,k-l}$, $Z_- \in M_{l,l}$ and $W \in M_{l,k-l}$ ($M_{n,m}$: space of $n$-by-$m$ matrices). 
Let $h$ be the endomorphism on $M_{l,k-l}$ given by
\begin{align}
h: B \quad \mapsto\quad h(B) = Z_- B - B Z_+.  
\end{align}
If ${\rm dim} ({\rm Ker}(h)) \not = 0$, 
then a non-trivial $X \in \mathfrak{gl}(k,\C)$ can be constructed as 
\begin{align}
X=\begin{pmatrix}
{\bf 0} & {\bf 0}\\ B_0 & {\bf 0}
\end{pmatrix} \quad {\rm with} \quad 
B_0 \in {\rm Ker}(h)\backslash \{ {\bf 0} \}.
\end{align}
This non-trivial element $X \in \mathfrak{gl}(k,\C)$ satisfies 
\begin{align}
\Psi X=0,\quad [Z,X]=\begin{pmatrix}
{\bf 0} & {\bf 0}\\ h(B_0) & {\bf 0}
\end{pmatrix}={\bf 0}.
\end{align}
If ${\rm dim} ({\rm Ker}(h))=0$,
a squared matrix $X$ can be constructed 
by using the inverse map $h^{-1}$. 
For example, 
\begin{align}
X=\begin{pmatrix}
{\bf 0} & {\bf 0}\\ h^{-1}(W) & {\bf 1}
\end{pmatrix} \not={\bf 0} 
\end{align}
satisfies
\begin{align}
\Psi X=0,\quad [Z,X]=\begin{pmatrix}
{\bf 0} & {\bf 0}\\ Z_-h^{-1}(W)-h^{-1}(W) Z_+-W & {\bf 0}
\end{pmatrix}={\bf 0}.
\end{align}
Therefore, we can always construct 
a nontrivial $X \in \mathfrak{gl}(k,\C)$ satisfying 
$\Psi X = 0$ and $[Z,X]=0$
for any element $(\Psi, Z) \in \overline{F_2}$. 
Thus we conclude that $\overline{F_2} \subseteq \overline{F_1}$, 
that is, $F_2 \supseteq F_1$.
\end{proof}
Combining these two facts, $F_1 \supseteq F_2$ and $F_2 \supseteq F_1$, 
we conclude that $F_1=F_2$. 

%%%%%%%%%%%%%%%%%%%%%%%%%%%%
\subsubsection%[Atlas]
{Atlas $\{(\wt \phi_\lambda, \wt {\cal M}_\lambda)\}$ of $\wt {\mathcal M}$ from $\{Z, \Psi \}$}
Using the spaces $F_1$ or $F_2$ given in Eq.\,\eqref{eq:TwoF}, 
we can rewrite the definition of the manifold $\wt {\cal M}$ as 
\begin{align}
\wt {\cal M} \equiv F_1/GL(k,\mathbb C) = F_2 /GL(k,\mathbb C). 
\end{align}
The atlas of this manifold is given as follows. 
Let $\Lambda$ be the same index set $\Lambda$ as that given in Eq.\eqref{eq:calMatlas}
\beq
\Lambda =\left\{(l_1,l_2,\cdots,l_n) \, \Big| \, l_a \in \mathbb Z_{\ge 0},~\sum_{a=1}^nl_a=k\right\}. 
\eeq
For $\lambda \in \Lambda$, 
let $\wt{\mathcal M}_\lambda$ be the subspace of $\wt{\mathcal M}$ 
given by the equivalence classes of the data $(\Psi,Z)$ of the form
\beq
{\renewcommand{\arraystretch}{0.8}
{\setlength{\arraycolsep}{1.0mm} 
\Psi_\lambda = \ba{ccc} \psi_1 & & \\ & \ddots & \\  & & \psi_n \ea + \ba{ccc} \Psi_{11} & \cdots & \Psi_{1n} \\ \vdots ~ & \ddots & \vdots ~ \\ \Psi_{n1} & \cdots & \Psi_{nn} \ea, \hs{10}
Z_\lambda = \ba{ccc} Z_1 & & \\ & \ddots & \\  & & Z_n \ea + \ba{ccc} Z_{11} & \cdots & Z_{1n} \\ \vdots ~ & \ddots & \vdots ~ \\ Z_{n1} & \cdots & Z_{nn} \ea, 
}}
\label{eq:matrixform_tildeM}
\eeq
where $\psi_a$ are $l_a$-component row vectors, 
$\Psi_{ab}$ are $l_b$-component row vectors, 
$Z_a$ are $l_a$-by-$l_a$ matrices and 
$Z_{ab}$ are $l_a$-by-$l_b$ matrices such that
\beq 
\ba{c} \psi_a \\ Z_a \ea = 
{\renewcommand{\arraystretch}{0.9}
{\setlength{\arraycolsep}{1.0mm} 
\overset{l_a}{\ba{ccc} 1 & & \\ & \ddots & \\ & & 1 \\ \hline 0 & \cdots & 0 \ea} \scriptstyle{l_a+1}, \hs{10}
\ba{c} \Psi_{ab} \\ Z_{ab} \ea}} = 
{\renewcommand{\arraystretch}{0.9}
{\setlength{\arraycolsep}{0.5mm} 
\overset{l_b}{\ba{ccc} 0 & \cdots & 0 \\ \vdots & \ddots & \vdots \\ 0 & \cdots & 0 \\ \hline T_{ab,1} & \cdots & T_{ab,l_b} \ea} \scriptstyle{l_a+1}}}. \label{eq:datainPsiZ}
\eeq
To check that $GL(k,\mathbb C)$ action is free on $(\Psi,Z)$ given above,
it is convenient to map $(\Psi, Z)$ into the infinite dimensional
Grassmannian $G(k,\infty)$ given by 
the set of an infinite number of $k$-component row vectors constructed 
from the rows of $\Psi Z^{p-1}$ with $p \in \mathbb N$. 
As shown in Sec.\ref{sec:TwoExpGLkCfree},
the $GL(k,C)$ action is free on $(\Psi,Z)$ if and only if
the image of the mapping to $G(k,\infty)$ contains 
a bases of the $k$-dimensional vector space.
Let us define $k$-component row vectors $\{{\bf e}_{a,p}\}$ as
\begin{align}
{\bf e}_{a,p} \equiv \hbox{the $a$-th row of~} \Psi Z^{p-1}, \quad  a=1,\cdots,n,\quad p=1,\cdots,k.
\end{align}
From $(\Psi,Z)$ given in Eq.\eqref{eq:datainPsiZ}, 
we can construct the following $(l_a+1)$-by-$k$ matrix for each $a$,
\begin{align}
\begin{pmatrix}
{\bf e}_{a,1} \\ \vdots \\ {\bf e}_{a,l_a} \\ \hline {\bf e}_{a,l_a+1}   
\end{pmatrix} =
\left(\begin{array}{c|c|c|c|c|c|c}
& & & & & & \\
{\bf 0} & ~\cdots~ &{\bf 0}& ~~ {\bf 1}_{l_a} ~~ & {\bf 0} & ~\cdots~ & {\bf 0} \\ 
& & & & & & \\ \hline
\vec T_{a,1}&\cdots&\vec T_{a,a-1}&\vec T_{a,a}&\vec T_{a,a+1}&\cdots &
\vec T_{a,n}
\end{array}\right),
\end{align}
with a $l_b$-component row vector 
$\vec T_{ab}=(T_{ab,1},\cdots,T_{ab,l_b})$.
By correcting the first $l_a$ rows for all $a$, 
we find the identity matrix
\begin{align}
E_\lambda = {\bf 1}_k \quad~ {\rm with} \quad~
E_\lambda \equiv 
\begin{pmatrix}
\hat {\bf e}_1(l_1)\\ \hat {\bf e}_2(l_2) \\ \vdots \\ \hat {\bf e}_n(l_n)  
\end{pmatrix} \quad 
{\rm and} \quad 
\hat {\bf e}_a(m) \equiv  
\begin{pmatrix}
{\bf e}_{a,1} \\ {\bf e}_{a,2} \\ \vdots \\ {\bf e}_{a,m}  
\end{pmatrix} 
\label{eq:1inPsiZ}
\end{align}
which immediately indicate that the $GL(k,\mathbb C)$ action is free.
All the moduli parameters in $(\Psi, Z)$ are contained in 
$\{ {\bf e}_{a,l_a+1} \,|\, a=1,\dots,n \} \simeq \mathbb C^{kn}$ as
${\bf e}_{a,l_a+1}=(\vec T_{a,1}\cdots \vec T_{a,n})$
and all the entries are independent.
Therefore, the local coordinate system on the coordinate patch (chart) $\wt{\cal M}_{\lambda}$ with $\lambda=(l_1,\cdots,l_n)$ is given by
\beq
\tilde \phi_\lambda (\wt{\cal M}_\lambda) \simeq \left\{ T^{ab}_{m} \, \big| \, 1\le a,b \le n, 1\le m\le l_b\right\} \simeq \mathbb C^{kn}.
\eeq
The coordinate transformation to another patch $\wt{\mathcal M}_{\lambda'}$
can be constructed by using the matrix $g \in GL(k,\C)$ defined by 
\begin{align}
g=E_{\lambda'}
%\equiv \begin{pmatrix}
%\hat {\bf e}_1(l_1')\\ \vdots \\  \hat {\bf e}_n(l_n') 
%\end{pmatrix}  
\quad {\rm with} \quad \lambda'=(l_1',l_2',\cdots,l_n').
\end{align} 
Using the matrix $g$, 
we can read off the coordinate transformation 
$\wt f_{\lambda'\lambda} =\wt \phi_{\lambda'} \circ \wt \phi_{\lambda}^{-1}$ from the relation
\begin{align}
(\Psi',Z')=(\Psi g^{-1},g Z g^{-1})   
\in \wt {\cal M}_{\lambda}\cap \wt {\cal M}_{\lambda'}.
\end{align}
Gluing all the coordinate patches $\wt{\mathcal M}_{\lambda}$,
we obtain a complex manifold
\begin{align}
\wt{\cal M}'= \bigcup_{\lambda\in \Lambda} \wt{\cal M}_{\lambda}, 
\quad {\rm with} \quad  
\Lambda =\left\{(l_1,l_2,\cdots,l_n) \, \Big| \, l_a \in \mathbb Z_{\ge 0},~\sum_{a=1}^nl_a=k\right\}.
\label{eq:wtcalMatlas}
\end{align}
This is a submanifold of $\wt{\mathcal M}$ since $\wt{\mathcal M}_{\lambda} \subset \wt{\mathcal M}$ for all $\lambda \in \Lambda$. 
In subsection \eqref{subsec:wtdisjoint_union}, 
we show that
$\wt{\mathcal M}$ can be decomposed into a disjoint union of subspaces
$\wt{\mathcal M}_{\lambda}^{\rm tri}$ such that $\wt{\mathcal M}_{\lambda}^{\rm tri} \subset \wt{\mathcal M}_{\lambda}$. 
This fact implies that $\wt{\mathcal M}$ is a subspace of $\wt{\mathcal M}'$: 
\beq
\wt{\mathcal M} = \bigsqcup_{\lambda \in \Lambda} \wt{\mathcal M}_{\lambda}^{\rm tri} \subseteq \bigcup_{\lambda\in \Lambda} \wt{\cal M}_{\lambda} = \wt{\mathcal M}'.
\eeq
Since $\wt{\cal M}' \subseteq \wt{\cal M}$ and $\wt{\cal M}' \supseteq \wt{\cal M}$, we conclude that $\wt{\cal M}=\wt{\cal M}'$ and 
$\{(\wt{\phi}_\lambda, \wt{\cal M}_\lambda)\}$ gives an atlas of $\wt{\mathcal M}$. 

\subsubsection{Decomposition into disjoint union}\label{subsec:wtdisjoint_union}
Here, we show that
$\wt{\mathcal M}$ can be decomposed 
into a disjoint union of subspaces
$\wt{\mathcal M}_{\lambda}^{\rm tri}$ such that $\wt{\mathcal M}_{\lambda}^{\rm tri} \subset \wt{\mathcal M}_{\lambda}$.

The manifold $\wt{\mathcal M}$ is the space of equivalent classes of the matrices $[ (\Psi,Z) ]$ 
satisfying the $GL(k,\C)$-free condition \eqref{eq:cond_2}. 
For each equivalence class, 
we can associate a set of integers $\lambda=(l_1,l_2,\cdots,l_n)$ as follows. 
Let $(\Psi,Z)$ is an representative of a equivalence class 
and ${\bf e}_{a,p}$ be the $k$ component row vectors defined by
\begin{align}
{\bf e}_{a,p} \equiv \hbox{the $a$-th row of~} \Psi Z^{p-1}.
\end{align}
Let $\wt {\mathcal V}_a~(a=1,\cdots,N)$ be the vector spaces spanned by ${\bf e}_{b,p}$ with $b=1,\cdots,a$
\beq
\wt {\mathcal V}_a = {\rm span} \left\{ {\bf e}_{b,p} \, | \, p \in {\mathbb N},~b=1,\cdots,a \right\}.
\eeq
These vector spaces form a flag 
\begin{align}
\{0\}=\wt {\mathcal V}_0\subseteq \wt {\mathcal V}_1 \subseteq  \wt {\mathcal V}_2 \subseteq \cdots \subseteq \wt 
{\mathcal V}_n.
\end{align}
Then, we define $l_a$ as 
\beq
l_a ~=~ {\rm dim}_\C \, \wt{\mathcal V}_a - {\rm dim}_\C \, \wt{\mathcal V}_{a-1}. 
\eeq
Since $(\Psi,Z)$ satisfies the $GL(k,\C)$-free condition \eqref{eq:cond_2}, it follows that
\beq
{\rm dim}_\C \, \wt{\mathcal V}_n = l_1 + \cdots l_n = k.
\eeq
Since the set of integers $\lambda=(l_1,\cdots,l_n)$ is invariant under the $GL(k,\C)$ transformation, each equivalent class $[ (\Psi,Z) ]$ has unique $\lambda$. Therefore, $\wt {\mathcal M}$ can be decomposed into the disjoint union of the spaces of equivalence classes $\wt{\mathcal M}_{\lambda}^{\rm tri}$ classified by $\lambda=(l_1,\cdots,l_n)$
\beq
\wt{\mathcal M} = \bigsqcup_{\lambda \in \Lambda} \wt{\mathcal M}_{\lambda}^{\rm tri}.
\eeq
We can determine ${\rm dim}_\C \, \wt{\mathcal V}_a$ by constructing 
the basis of $\wt {\mathcal V}_a$. 
It can be obtained inductively 
from the basis of $\wt {\mathcal V}_{a-1}$ 
by adding the vectors ${\bf e}_{a,1},\cdots,{\bf e}_{a,l_a}$.
Here, $l_a$ is the maximum number such that ${\bf e}_{a,l_a}$ is linearly independent of 
$\{ {\bf e}_{a,1}, \cdots , {\bf e}_{a,l_a-1} \}$ and any element of $\wt {\mathcal V}_{a-1}$, 
or equivalently, $l_a$ is the minimum number such that
\beq
{\bf e}_{a,l_a+q} ~\in~ \wt {\mathcal V}_{a-1} \cup {\rm span}(\{ {\bf e}_{a,p} \, | \, p=1,\cdots,l_a \})~~~\mbox{for $q=1,2,\cdots$}.
\label{eq:vec_l_aplusq}
\eeq
 
Next, let us show that $\wt{\mathcal M}_\lambda^{\rm tri}$ is a subspace of 
$\wt{\mathcal M}_\lambda$. 
Let $E_\lambda$ be the $k$-by-$k$ matrix whose 
row vectors are ${\bf e}_{a,p}$
\beq
E_\lambda \equiv 
\begin{pmatrix}
\hat {\bf e}_1(l_1)\\ \hat {\bf e}_2(l_2) \\ \vdots \\ \hat {\bf e}_n(l_n)  
\end{pmatrix} \quad~~ 
{\rm with} \quad~~
\hat {\bf e}_a(m) \equiv  
\begin{pmatrix}
{\bf e}_{a,1} \\ {\bf e}_{a,2} \\ \vdots \\ {\bf e}_{a,m}  
\end{pmatrix} 
\eeq
Since $E_\lambda$ is the basis of $\wt{\mathcal V}_n \cong \C^k$, 
there is an element of $GL(k,\C)$ such that
\beq
E_\lambda \rightarrow E_\lambda \, g = \mathbf 1_k, \hs{5} g \in GL(k,\C).
\eeq
After fixing the $GL(k,\C)$ redundancy as $E_\lambda = \mathbf 1_k$, 
${\bf e}_{a,p}$ take the form
\begin{align}
\begin{pmatrix}
{\bf e}_{a,1}\\  \vdots \\ {\bf e}_{a,l_a} \\ \hline  {\bf e}_{a,l_a+1}   
\end{pmatrix} =
\left(\begin{array}{c|c|c|c|c|c|c}
& & & & & & \\
{\bf 0} & ~\cdots~ &{\bf 0}& ~~ {\bf 1}_{l_a} ~~ & {\bf 0} & ~\cdots~ & {\bf 0} \\ 
& & & & & & \\ \hline
\vec T_{a,1}&\cdots&\vec T_{a,a-1}&\vec T_{a,a}&{\bf 0}&\cdots &{\bf 0}
\end{array}\right),
\label{eq:disjoint_tildeM_lambda}
\end{align}
where  ``$\bf 0$''s in ${\bf e}_{a,l_a+1}$ are due to the property in Eq.\eqref{eq:vec_l_aplusq}.
From the above set of row vectors $\{{\bf e}_{a,p} \,|\, a=1,\dots,n;\, p=1,\dots,l_a+1\}$,
the two matrices $\Psi,Z$ can be reconstructed as
\begin{align}
\Psi=
\begin{pmatrix}
{\bf e}_{1,1} \\ {\bf e}_{2,1} \\ \vdots \\{\bf e}_{n,1}
\end{pmatrix},
\hs{10}
Z = E_\lambda Z=
\begin{pmatrix}
\hat {\bf e}_{1}(l_1) \, Z \\ \hat {\bf e}_{2}(l_2) \, Z \\ \vdots\\ \hat {\bf e}_{n}(l_n) \, Z
\end{pmatrix} 
~~~ {\rm with} ~~~ \hat {\bf e}_{a}(l_a) \, Z = 
\begin{pmatrix}
{\bf e}_{a,2} \\ \vdots\\ {\bf e}_{a,l_a} \\{\bf e}_{a,l_a+1}
\end{pmatrix}. \label{eq:e2PsiZ}
\end{align}
From these expression, we find that $\Psi$ and $Z$ take the block lower triangular form 
$\Psi = \Psi_\lambda^{\rm tri}$ and $Z = Z_\lambda^{\rm tri}$ with
\beq
{\renewcommand{\arraystretch}{0.8}
{\setlength{\arraycolsep}{1.0mm} 
\Psi_\lambda^{\rm tri} = \ba{ccc} \psi_1 & & \\ & \ddots & \\  & & \psi_n \ea +
 \ba{cccc} \Psi_{11} & {\bf 0}& \cdots & {\bf 0} \\ \Psi_{21}& \Psi_{22} & \ddots & \vdots ~ \\ 
 \vdots ~  &&\ddots & {\bf 0}\\ \Psi_{n1} & \cdots &\cdots & \Psi_{nn} \ea, \hs{2}
Z_\lambda^{\rm tri} = \ba{ccc} Z_1 & & \\ & \ddots & \\  & & Z_n \ea + 
\ba{cccc} Z_{11} & {\bf 0}& \cdots & {\bf 0} \\  Z_{21}& Z_{22} & \ddots & \vdots ~ \\ 
 \vdots ~  &&\ddots & {\bf 0}\\ Z_{n1} & \cdots &\cdots & Z_{nn} \ea, 
}}
\eeq
with $\psi_a,\Psi_{ab},Z_a,Z_{ab}$ given in Eq.\eqref{eq:datainPsiZ}. 
Since $\wt {\cal M}_\lambda^{\rm tri}$ is the space of $(\Psi,Z)$ of these forms, it is a subspace of $\wt {\cal M}_\lambda$, the space of matrices of the form \eqref{eq:matrixform_tildeM}
\begin{align}
\wt {\cal M}_\lambda^{\rm tri} \subseteq 
\wt {\cal M}_\lambda.
\end{align} 
This fact implies that $\wt{\mathcal M}$ is a subspace of $\wt{\mathcal M}'$: 
\beq
\wt{\mathcal M} = \bigsqcup_{\lambda \in \Lambda} \wt{\mathcal M}_{\lambda}^{\rm tri} \subseteq \bigcup_{\lambda\in \Lambda} \wt{\cal M}_{\lambda} = \wt{\mathcal M}'.
\eeq
Since $\wt{\cal M}' \subseteq \wt{\cal M}$ and $\wt{\cal M}' \supseteq \wt{\cal M}$, we conclude that $\wt{\cal M}=\wt{\cal M}'$ and 
$\{(\wt{\phi}_\lambda, \wt{\cal M}_\lambda)\}$ gives an atlas of $\wt{\mathcal M}$. 

%%%%%%%%%%%%%%%%%%%%%%%%%%%%%%%%%%%%%%%
 \subsection%[Equivalence]
 {Equivalence of $\mathcal M$ and $\wt {\mathcal M}$}
\subsubsection%[Mapping through the half-ADHM mapping relation ]
{ Mapping from ${\cal M}_\lambda $ to $\wt {\cal M}_\lambda $
 through the half-ADHM mapping relation}
Let us recapitulate how to extract the data $(\Psi,Z)$ from $\mD(z), \mJ(z)$ 
through the half-ADHM mapping relation.
The half-ADHM mapping relation is given by
\beq
\mD(z) \Psi = \mJ(z) (z \mathbf 1_k - Z), 
\label{eq:hadhm}
\eeq
which is covariant under both of the $V$ transformation and $GL(k,\mathbb C)$ transformation.
Therefore, this relation defines a mapping between the equivalence classes $[\mD(z)] \in {\cal M} \mapsto [(\Psi,Z)] \in \wt {\cal M}$. 
For a point in the $\lambda=(l_1,l_2,\cdots,l_n)$-patch of $\cal M$,
where $\mD(z)$ takes the form \eqref{eq:M_patch}, 
the corresponding matrix $\mJ(z)$ is given by
\beq
{\renewcommand{\arraystretch}{1.0}
{\setlength{\arraycolsep}{1.0mm} 
\mJ_\lambda(z) = \ba{ccc} j_1 & & \\ & \ddots & \\ & & j_n \ea 
+
\ba{ccc} \mJ_{11} & \cdots & \mJ_{1n} \\ \vdots & \ddots & \vdots \\ \mJ_{n1} & \cdots & \mJ_{nn} \ea,}}
\label{eq:J_patch}
\eeq
where 
$j_a$ and $\mJ_{ab}$ are the following block matrices (row vectors with $l_a$ and $l_b$ components, respectively)
\beq
j_a = \bigg( z^{l_a-1},z^{l_a-2},\cdots,1 \bigg), \hs{10}
\mJ_{ab} = - \left( \sum_{m=1}^{l_b-1} T_{ab,m+1} z^{m-1} \,,\, \sum_{m=1}^{l_b-2}T_{ab,m+2} z^{m-1} \,,\, \cdots \,,\ T_{ab,l_b} \,,\, 0 \right). 
\eeq
Since $\mJ(z)$ also depends on the gauge choice of $GL(k,\mathbb C)$,  
taking $\mJ(z)$ in this form implicitly means that 
we have chosen a certain coordinate patch for $\wt {\cal M}$. 
Noting that $\mJ_{ab}$ satisfies the relation 
\beq
{\renewcommand{\arraystretch}{0.8}
{\setlength{\arraycolsep}{1.1mm} 
z \mJ_{ab} = - P_{ab} ( 1, 0 , \cdots, 0 ) + ( T_{ab,1} , T_{ab,2} , \cdots, T_{ab,l_b} ) + \mJ_{ab} \ba{c|ccc} 0 & 1 & & \\ \vdots & & \ddots & \\ 0 & & & 1 \\ \hline 0 & \cdots & \cdots & 0 \ea,}}
\eeq
we can read off $(\Psi,Z)$ from the half-ADHM mapping relation \eqref{eq:hadhm} and the resulting 
$(\Psi,Z)$ turns out to be exactly equal to those  of $\wt \phi_{\lambda}(\wt {\cal M}_\lambda)$
given in Eq.\eqref{eq:matrixform_tildeM}. 
Therefore the half-ADHM mapping relation defines a one-to-one map between the coordinate patches
\begin{align}
&{i}_\lambda: \phi_\lambda( {\cal M}_\lambda)  \quad \mapsto \quad 
\wt \phi_\lambda( \wt {\cal M}_\lambda) ,\\
&\wt \phi_\lambda^{-1} \circ i_\lambda \circ \phi_\lambda :\quad  {\cal M}_\lambda  \quad \to 
\quad  \wt  {\cal M}_\lambda 
\end{align}
for all $\lambda \in \Lambda$.
Furthermore, since the half-ADHM mapping relation is covariant under the $V$-transformation $\mD'(z)=V(z) \mD(z)$  and the $GL(k, \mathbb C)$ transformation, $(\Psi',Z')=(\Psi g, g^{-1}Z g)$ as
\begin{align}
\mD(z)\Psi =\mJ(z)(z{\bf 1}-Z)\quad \Rightarrow \quad \mD'(z)\Psi' =\mJ'(z)(z{\bf 1}-Z'),\quad 
{\rm with~} \mJ'(z)=V(z)\mJ(z) g,
\end{align}
we find that the following diagram commutes:  for  
$\lambda, \lambda' \in \Lambda,  \lambda\not=\lambda'$,
\begin{align}
    \mD(z) \in \phi_\lambda ({\cal M}_\lambda \cap{\cal M}_{\lambda'}) & \hspace{2cm}\stackrel{i_\lambda}{\to} & (\Psi,Z)\in \wt
     \phi_\lambda (\wt {\cal M}_\lambda \cap\wt {\cal M}_{\lambda'}) \nn
    f_{\lambda' \lambda}=\phi_{\lambda'}\circ  \phi_{\lambda}^{-1}  ~\downarrow \qquad &
    &\qquad \downarrow ~\wt f_{\lambda' \lambda}=\wt \phi_{\lambda'}\circ  \wt \phi_{\lambda}^{-1} \nn
     \mD'(z) \in \phi_\lambda ({\cal M}_\lambda \cap{\cal M}_{\lambda'}) & \hspace{2cm}\stackrel{i_{\lambda'}}{\to} & (\Psi',Z')\in 
     \wt \phi_{\lambda'} (\wt {\cal M}_\lambda \cap\wt {\cal M}_{\lambda'}) 
\end{align}
This fact indicates that  two different maps $i_\lambda, i_{\lambda'}$ with the domain ${\cal M}_\lambda \cap {\cal M}_{\lambda'}$ are consistent 
\begin{align}
\wt \phi_\lambda^{-1} \circ  i_\lambda \circ \phi_\lambda 
=\wt \phi_{\lambda'}^{-1} \circ  i_{\lambda'} \circ \phi_{\lambda'}\quad : \quad 
{\cal M}_\lambda \cap {\cal M}_{\lambda'}  \quad \to \quad \wt {\cal M}_\lambda \cap \wt {\cal M}_{\lambda'}
\end{align}
and the transition function $\wt f_{\lambda'\lambda}$ on the manifold $\wt {\cal M}$ induced by the 
$GL(k,\mathbb C)$ transformation is consistent with 
 $f_{\lambda'\lambda}$ on the manifold ${\cal M}$ induced by the $V$-transformation,
\begin{align}
\wt f_{\lambda'\lambda} =i_{\lambda'}\circ f_{\lambda'\lambda}\circ i_\lambda^{-1}
 : \quad 
\wt \phi_\lambda(\wt {\cal M}_\lambda \cap \wt {\cal M}_{\lambda'})  \quad \to \quad 
\wt \phi_{\lambda'}(\wt {\cal M}_\lambda \cap \wt {\cal M}_{\lambda'}).
\end{align}
Therefore, we conclude that   the two complex manifolds ${\cal M}$ and $\wt {\cal M}$ are 
 biholomorphically equivalent
\begin{align}
{\cal M} \simeq \wt {\cal M},
\end{align}
and the ADHM  relation defines the unique  one-to-one map  between them.

\subsubsection{Examples}
Let us see an example in the case of $k=2$ and $n=2$.
In the $(2,0)$ patch ($\lambda=(2,0)$), the matrix $\mD(z)$ and the corresponding matrix $\mJ(z)$ are given by (see Eqs.\,\eqref{eq:M_patch} and \eqref{eq:J_patch})
\begin{eqnarray}
\mD_{(2,0)}(z)= \left( \begin{array}{cc} z^2-a_2 z - a_1  & 0\\ - b_2 z - b_1& 1 \end{array} \right), \hs{10}
\mJ_{(2,0)}(z) =\left( \begin{array}{cc} z-a_2 & 1 \\ - b_2& 0
\end{array} \right). 
\end{eqnarray}
From the half-ADHM mapping relation, the data $(Z,\Psi)$ can be read off as
\begin{eqnarray}
\Psi_{(2,0)} = \left( \begin{array}{cc} 1 & 0 \\ b_1 & b_2 \end{array} \right),\quad
Z_{(2,0)} = \left( \begin{array}{cc} 0 & 1\\ a_1 & a_2\end{array} \right).
\end{eqnarray}
These matrices correspond to those in Eq.\,\eqref{eq:matrixform_tildeM} with 
\begin{eqnarray}
\psi_1= \left( \begin{array}{cc} 1 & 0 \end{array} \right),\quad
\Psi_{11} = \left( \begin{array}{cc} 0 & 0 \end{array} \right), \quad
\Psi_{21} = \left( \begin{array}{cc} b_1 & b_2 \end{array} \right),\quad
Z_1 = \left( \begin{array}{cc} 0 & 1 \\ 0 & 0 \end{array} \right), \quad
Z_{11} = \left( \begin{array}{cc} 0 & 0 \\ a_1 & a_2 \end{array} \right).
\end{eqnarray}
One can move to the $(1,1)$ patch ($\lambda=(1,1)$) by performing the transformation
\beq
\mD_{(1,1)}(z) = V(z) \mD_{(2,0)}(z), \hs{5}
\mJ_{(1,1)}(z) = V(z) \mJ_{(2,0)}(z) g, \hs{10}
V(z) = \ba{cc} 0 & - v \\ b_2 & z- \tilde u \ea, 
\eeq
as 
\begin{eqnarray}
\mD_{(1,1)}(z)= \left( \begin{array}{cc} z-u & - v  \\ - \tilde v & z - \tilde u \end{array} \right), \hs{10}
\mJ_{(1,1)}(z) =\left( \begin{array}{cc} 1 & 0 \\ 0 & 1 
\end{array} \right),
\end{eqnarray}
where
\beq
u = - \frac{b_1}{b_2}, \hs{5} \tilde u = a_2 + \frac{b_1}{b_2}, \hs{5}
v = \frac{1}{b_2}, \hs{5} \tilde v = a_1 b_2 - b_1 \left( a_2 + \frac{b_1}{b_2} \right).
\eeq
The corresponding half-ADHM data are given by
\begin{eqnarray}
\Psi_{(1,1)} = \left( \begin{array}{cc} 1 & 0\\ 0 & 1 \end{array} \right),\quad
Z_{(1,1)} = \left( \begin{array}{cc} u & v \\ \tilde v & \tilde u \end{array} \right).
\end{eqnarray}
These matrices correspond to those in Eq.\,\eqref{eq:matrixform_tildeM} with
\begin{eqnarray}
\psi_1=\psi_2= 1, \quad \Psi_{ij} = 0 , \quad Z_i = 0 ~~ (i=1,2), \quad Z_{11} = u, \quad Z_{12}=v, \quad Z_{21} = \tilde v, \quad Z_{22}= \tilde u.
\end{eqnarray}
These data in the $(1,1)$ patch is related to that $(2,0)$ patch by
a $GL(2,\C)$ transformation
\beq
\Psi_{(1,1)} = \Psi_{(2,0)} g, \hs{5} 
Z_{(1,1)} = g^{-1} Z_{(2,0)} g, \hs{10}
g = \ba{cc} 1 & 0 \\ u & v \ea \in GL(2,\C). 
\eeq
Similarly, we can show that the data $(0,2)$ patch is also related to $(Z_{(2,0)},\Psi_{(2,0)})$ and $(Z_{(1,1)},\Psi_{(1,1)})$ by $GL(2,\C)$ transformation. 

%%%%%%%%%%%%%%%%%%%%%%%%%%%%%%%%%%%%%%%%%%%%%

\section{Condition of non-singular instanton solution}\label{sec:non-singular}
In this appendix, we discuss the condition for $\xi(z)$ and $(Z,\Psi,\wt \Psi)$ to be the data for non-singular sigma model instanton solutions.
Throughout this section, we focus on the $L=1$ case.

\subsection{Local-semilocal decomposition} \label{sec:LSdecomposition}
Let us define the moduli space of semi-local vortices ${\cal M}_{\rm semi \,}{}^{n,m}_{k}$
as the following subspace of the moduli space of vortices ${\cal M}_{\rm vtx \,}{}^{n,m}_k$
\begin{align}
{\cal M}_{\rm semi \,}{}^{n,m}_{k}\equiv \big\{ \,[\xi(z)]\in  {\cal M}_{\rm vtx \,}{}^{n,m}_k \, \big| \,
 {\rm rank}(\xi(z))=n ~\mbox{for}~\forall z\in \mathbb C \, \big\} \subset  {\cal M}_{\rm vtx \,}{}^{n,m}_k.
\end{align}
In this subsection, we prove 
four lemmas \ref{lmm:1}-\ref{lmm:4} 
which will be used in the later subsections.

\begin{lmm}\label{lmm:1}
Any matrix $\xi(z)$ such that $[\xi(z)] \in {\cal M}_{\rm vtx \,}{}^{n,m}_k$ can be decomposed as
\begin{align}
 \xi(z) = \mD^{\rm lc}(z) \, \xi^{\rm sm}(z) 
 ~~~~ \mbox{with} ~~~ \mD^{\rm lc}(z)\in {\cal G}_{n,k-l}[z] 
 ~~~ \mbox{and} ~~~ [\xi^{\rm sm}(z)]\in {\cal M}_{\rm semi \,}{}^{n,m}_{l}
\end{align}
where $l$ is an integer such that $0 \leq l \leq k$.
\end{lmm}
\begin{proof}[Proof of Lemma \ref{lmm:1}]
Applying the method used in Sec.\,\ref{subsec:patch_L=1} to the rows of $\xi(z) = (\mD(z), \wt \mD(z))$ instead of columns,
we can find matrices $V(z) \in {\cal G}_{n+m,0}[z]$ and $\mD^{\rm lc}(z) \in {\cal G}_{n,l}[z]~(0 \leq l \leq k)$ such that 
\begin{align}
\xi(z) = (\mD^{\rm lc}(z), {\bf 0}) \, V(z).
\end{align}
Then, this matrix $\xi(z)$ can be rewritten as 
\begin{align}
\xi(z) = \mD^{\rm lc}(z) \, \xi^{\rm sm}(z) 
~~~ \mbox{with} ~~~ \xi^{\rm sm}(z) = ({\bf 1}_n, {\bf 0}) \, V(z).
\end{align}
Since $\det V(z) \not = 0$ for all $z \in \C$, it follows that ${\rm rank}(\xi^{\rm sm}(z)) = n$ for all $z \in \C$ and hence $[\xi^{\rm sm}(z)] \in {\cal M}_{\rm semi \,}{}^{n,m}_{k-l}$. 
\end{proof}
Note that the decomposition is not unique 
since the following transformation does not change the matrix $\xi(z)$
\beq
\mD^{\rm lc}(z) \rightarrow \mD^{\rm lc}(z) V'(z)^{-1}, \hs{5}
\xi_{\rm sm}(z) \rightarrow V'(z) \xi_{\rm sm}(z) ~~\mbox{with}~~ V'(z) \in {\cal G}_{n,0}[z]. 
\label{eq:lc_sm_Vtransf}
\eeq
We can show the the equivalence class $[\xi_{\rm sm}(z)]$ is unique. 
\begin{lmm}\label{lmm:2}
The decomposition in Lemma \ref{lmm:1} is unique up to $V$-equivalence relations \eqref{eq:lc_sm_Vtransf}.
\end{lmm}
\begin{proof}[Proof of Lemma \ref{lmm:2}]
Let us assume that $[\xi(z)]\in {\cal M}_{\rm semi \,}{}^{n,m}_{k}$ is rewritten in the two different ways as
\begin{align}
\xi(z) = \mD^{\rm lc}_1(z) \xi_1^{\rm sm}(z) = \mD^{\rm lc}_2(z) \xi_2^{\rm sm}(z)
~~~ \mbox{with}~~~
\xi_i^{\rm sm}(z) = ({\bf 1}_n, {\bf 0}) V_i(z)~~~(i=1,2).
\end{align}
The above equation can be rewritten as
\begin{align}
(\mD^{\rm lc}_1(z), {\bf 0}) = (\mD^{\rm lc}_2(z), {\bf 0}) \, V_2(z) V_1(z)^{-1}, \quad 
\Big( V_2(z) V_1(z)^{-1} \in {\cal G}_{n+m,0}[z] \Big).
\end{align}
This implies that $V_2(z) V_1(z)^{-1}$ takes the form
\begin{align}
V_2(z) V_1(z)^{-1} = \ba{cc} V_{12}(z) & {\bf 0} \\ \ast & \ast \ea,
\end{align}
where $V_{12}(z) \in {\cal G}_{n,0}[z]$ and $\ast$'s are undetermined entries.  
From this equation, we find that  $\mD_1^{\rm lc}(z)=\mD_2(z)^{\rm lc} \, V_{12}(z)$ and hence 
$\xi_1^{\rm sm}(z)$ and $\xi_2^{\rm sm}(z)$ are related as
\begin{align}
\xi_1^{\rm sm}(z) =  V_{12}(z)^{-1}\xi_2^{\rm sm}(z).
\end{align}
Therefore, the decomposition $\xi(z) = \mD^{\rm lc}(z) \xi^{\rm sm}(z)$ 
defines an unique equivalent class $[\xi^{\rm sm}(z)] \in {\cal M}_{\rm semi \,}{}^{n,m}_{k}$.
%\begin{align}
%(\mD^{\rm lc}(z), \xi''(z)) \sim  (\mD^{\rm lc}(z) V'(z)^{-1}, V'(z) \xi''(z)) \quad {\rm with ~} V'(z)\in {\cal G}_{n,0}[z], \label{eq:Vdecompose}
%\end{align}
\end{proof}

\begin{lmm}\label{lmm:3}
Any element $[\xi^{\rm sm}(z)] \in {\cal M}_{\rm semi \,}{}^{n,m}_{k}$ can be written as
\begin{align}
\xi^{\rm sm}(z)=({\bf 1}_n,{\bf 0}) \, V(z)\quad {\rm with~} V(z) \in {\cal G}_{n+m,0}[z].
\end{align}
%which implies 
%\begin{align}
%{\cal M}_{\rm semi \,}{}^{n,m}_{k}  \subset  \frac{{\cal G}_{n+m,0}[z]}{{\cal G}_{n,0}[z]\times {\cal G}_{m,0}[z]}.
%\end{align}
\end{lmm}
\begin{proof}[Proof of Lemma \ref{lmm:3}]
If $\det \mD^{\rm lc}(z) $ has a zero, 
then rank of $\xi(z)$ decreases at that point.       
Equivalently, if $\xi(z)$ has the maximal rank everywhere ($[\xi(z) ]\in {\cal M}_{\rm semi \,}{}^{n,m}_{k}$), 
then $\det \mD^{\rm lc}(z)$ must have no zero 
($\mD^{\rm lc}(z) \in {\cal G}_{n,0}[z]$) and 
hence $\mD^{\rm lc}(z)$ can be absorbed by the $V$-transformation $V(z) \in {\cal G}_{n,0}[z]$.   
Therefore, for any $[\xi(z) ]\in {\cal M}_{\rm semi \,}{}^{n,m}_{k}$, there is a representative $\xi(z)$ of the form 
\begin{align}
\xi(z) = ({\bf 1}_n,{\bf 0}) \, V(z).
\end{align}
\end{proof}  

\begin{lmm}\label{lmm:4}
If two equivalent classes $[\xi_i(z)] = [(\mD_i(z),\wt \mD_i(z)] \in {\cal M}_{\rm semi \,}{}^{n,m}_{k_i}~(i=1,2)$ satisfy
\begin{align}
\mD_1^{-1}(z) \, \wt \mD_1(z) = \mD_2^{-1}(z) \, \wt \mD_2(z) \quad \hbox{or equivalently} \quad  
\mD_1^{-1}(z) \, \xi_1(z)=\mD_2^{-1}(z) \, \xi_2(z),
\label{eq:lmm4_assumption}
\end{align}
then, they are equivalent 
\begin{align}
k_1=k_2, \quad\quad [\xi_1(z)]=[\xi_2(z)].
\end{align}
\end{lmm}

\begin{proof}[Proof of Lemma \ref{lmm:4}]
By applying Lemma \ref{lmm:3} to $[\xi_i(z)] \in {\cal M}_{\rm semi \,}{}^{n,m}_{k_i}$, 
the assumption \eqref{eq:lmm4_assumption}
can be further rewritten as
\begin{align}
\mD_1(z)^{-1}({\bf 1}_n,{\bf 0}) V_1(z)=\mD_2(z)^{-1}({\bf 1}_n,{\bf 0})V_2(z).
\end{align}
This equation can be rewritten as 
\begin{align}
\mD_2(z)\mD_1(z)^{-1}({\bf 1}_n,{\bf 0})=({\bf 1}_n,{\bf 0})V_2(z)V_1(z)^{-1},\quad 
\Big( V_2(z)V_1(z)^{-1} \in {\cal G}_{n+m,0}[z] \Big).
\end{align}
Applying the same argument as Lemma \ref{lmm:2}, 
we conclude that $\mD_2(z)=V_{12} (z) \mD_1(z)$ with $V_{12}(z)\in {\cal G}_{n,0}[z]$,
and hence $\xi_2(z)=V(z) \xi_1(z)$, which means 
that they belong to the same equivalent class.
\end{proof}

%%%%%%%%%%%%%%%%%%%%%%%%%%%%%%%%%%%%%%%%%
\subsection{Condition of semilocal vortices}\label{sec:cond_semi}
For $[\xi(z)] \in {\cal M}_{\rm vtx \,}{}_k^{n,m}$, 
the condition that $[\xi(z)]$ belongs to the semilocal vortex moduli space ${\cal M}_{\rm semi \,}{}_k^{n,m}$ is 
given by 
\begin{align}
{\cal C}_{\rm semi}(\xi(z)) ~ : ~ \det \xi(z) \xi(z)^\dagger \not=0 ~~~ \mbox{for} ~~~ \forall z \in \mathbb C.
\end{align}
This condition on $\xi(z)$ can be translated into the following condition on the corresponding half-ADHM data $(Z,\Psi,\wt \Psi)$
\begin{align}
{\cal C}_{\rm free}(Z,\wt \Psi) ~ : ~ \mbox{If $\exists \, \vec v \in \C^k~\mbox{(row vector)}$ s.t. $\vec v Z^{p-1} \wt \Psi = 0$ for $\forall p \in \mathbb N ~ \Rightarrow ~ \vec v=0$}.
\end{align}
In this subsection, 
we prove the equivalence of the two conditions ${\cal C}_{\rm semi}(\xi(z))$ and ${\cal C}_{\rm free}(Z,\wt \Psi)$\footnote{
The proof here is a more concise version of the one given in Appendix C of \cite{Eto:2007yv}.}.

\begin{thm}
The conditions ${\cal C}_{\rm semi}(\xi(z))$ and ${\cal C}_{\rm free}(Z,\wt \Psi)$ are equivalent. 
\end{thm}
\begin{proof}[Proof of ${\cal C}_{\rm free}(Z,\wt \Psi) \to {\cal C}_{\rm semi}(\xi(z))$] 
Let us prove the contrapositive $\neg \, {\cal C}_{\rm semi}(\xi(z)) \to \neg \, {\cal C}_{\rm free}(Z,\wt \Psi)$. If $[\xi(z)]\in {\cal M}_{\rm vtx \,}{}_k^{n,m}$ does not satisfy the condition ${\cal C}_{\rm semi}(\xi(z))$, then Lemma \ref{lmm:1} in Appendix \ref{sec:LSdecomposition} implies that the matrix $\xi(z)$ can be decomposed as
\begin{align}
\xi(z)=(\mD(z),\wt \mD(z))=\mD_{\rm lc}(z)(\mD_{\rm sm}(z),\wt \mD_{\rm sm}(z)) = \mD_{\rm lc}(z)\xi_{\rm sm}(z) ,
\end{align}
where the $n$-by-$n$ matrices $\mD(z), \mD_{\rm lc}(z),\mD_{\rm sm}(z)$ 
and the $n$-by-$m$ matrices $\wt \mD(z), \wt \mD_{\rm sm}(z)$ satisfy
\begin{align}
\det \mD(z)={\cal O}(z^k),\quad 
\det \mD_{\rm lc}(z)={\cal O}(z^l),\quad \det \mD_{\rm sm}(z)={\cal O}(z^{k'}),\nn
\det \mD(z)^{-1}\wt \mD(z)={\cal O}(z^{-1}),\quad \det \mD_{\rm sm}(z)^{-1}\wt \mD_{\rm sm}(z)={\cal O}(z^{-1}),
\end{align}
with a certain nonzero positive integer $l>0$ and $k'=k-l\ge 0 $. 
For these matrices, we can obtain $(Z,\Psi,\wt \Psi)$, $(Z_{\rm lc},\Psi_{\rm lc})$ and $(Z_{\rm sm},\Psi_{\rm sm},\wt \Psi_{\rm sm})$ through the half-ADHM mapping relations 
\begin{alignat}{2}
z \mJ(z) &=\mD(z) \Psi+\mJ(z)Z, & \wt \mD(z)  &=\mJ(z)\wt \Psi, \label{eq:hADHM1_AppendixD} \\
z \mJ_{\rm lc}(z) &=\mD_{\rm lc}(z) \Psi_{\rm lc}+\mJ_{\rm lc}(z)Z_{\rm lc}, \label{eq:hADHM2_AppendixD} \\
z \mJ_{\rm sm}(z) &=\mD_{\rm sm}(z) \Psi_{\rm sm}+\mJ_{\rm sm}(z)Z_{\rm sm}, \quad \quad & \wt \mD_{\rm sm}(z) &=\mJ_{\rm sm}(z)\wt \Psi_{\rm sm}, \label{eq:hADHM3_AppendixD}
\end{alignat} 
where $\mJ(z)$, $\mJ_{\rm lc}(z)$ and $\mJ_{\rm sm}(z)$ are matrices satisfying
\begin{align}
\mD(z)^{-1} \mJ(z)={\cal O}(z^{-1}),\quad\mD_{\rm lc}(z)^{-1} \mJ_{\rm lc}(z)={\cal O}(z^{-1}),
\quad \mD_{\rm sm}(z)^{-1} \mJ_{\rm sm}(z)={\cal O}(z^{-1}).
\end{align}
Note that the matrices are related as
\begin{align}
\mD(z)=\mD_{\rm lc}(z)\mD_{\rm sm}(z),\quad \wt \mD(z)=\mD_{\rm lc}(z)\wt \mD_{\rm sm}(z).
\end{align}
The matrix $\mJ(z)$ can always be chosen as
\begin{align}
\mJ(z)=(\mJ_{\rm lc}(z),\mD_{\rm lc}(z)\mJ_{\rm sm}(z)).
\end{align}
We can check the condition $\mD(z)^{-1} \mJ(z) = {\cal O}(z^{-1})$ as \begin{align}
\mD(z)^{-1}\mJ_{\rm lc}(z)=\mD_{\rm sm}(z)^{-1}{\cal O}(z^{-1})={\cal O}(z^{-1}), \hs{5}
\mD(z)^{-1} \mD_{\rm lc}(z) \mJ_{\rm sm}(z) = \mD_{\rm sm}(z)^{-1}  \mJ_{\rm sm}(z) = {\cal O}(z^{-1}),
\end{align}
where we have used the transformation \eqref{eq:lc_sm_Vtransf} to fix $\mD_{\rm sm}(z)$ to the form \eqref{eq:M_patch} 
so that it satisfies $\mD_{\rm sm}(z)^{-1}={\cal O}(1)$. 
From the half-ADHM mapping relations Eqs.\,\eqref{eq:hADHM1_AppendixD}-\eqref{eq:hADHM3_AppendixD}
we find that
\begin{align}
z \mJ(z) \ = \ (z \mJ_{\rm lc}(z),\mD_{\rm lc}(z)(z \mJ_{\rm sm}(z))) \ = \ (\mD_{\rm lc}(z) \Psi_{\rm lc}+\mJ_{\rm lc}(z)Z_{\rm lc},\,\mD_{\rm lc}(z)\mD_{\rm sm}(z) \Psi_{\rm sm}+\mD_{\rm lc}(z)\mJ_{\rm sm}(z)Z_{\rm sm} ). \label{eq:lsdata0}
\end{align}
Using the constant matrices $P_{\rm sm}^\mD,P_{\rm sm}^\mJ$ defined by
\begin{align}
{\bf 1}_n=\mD_{\rm sm}(z) P^\mD_{\rm sm}+\mJ_{\rm sm}(z)P^\mJ_{\rm sm},
\end{align}
Eq.\,\eqref{eq:lsdata0} can be further rewritten as
\begin{align}
z \mJ(z) =\mD_{\rm lc}(z)\mD_{\rm sm}(z) (P^\mD_{\rm sm}\Psi_{\rm lc}, \,\Psi_{\rm sm})+
(\mJ_{\rm lc}(z),\mD_{\rm lc}(z)\mJ_{\rm sm}(z))
\begin{pmatrix}
Z_{\rm lc} & {\bf 0}\\ P_{\rm sm}^\mJ \Psi_{\rm lc}&Z_{\rm sm}
\end{pmatrix}.
\label{eq:lsdata1}
\end{align}
In addition, we find the following relation for $\wt \mD(z)$
\begin{align}
\mJ(z)\wt \Psi=\wt \mD(z)=\mD_{\rm lc}(z)\wt \mD_{\rm sm}(z)
=\mD_{\rm lc}(z)\mJ_{\rm sm}(z)\wt \Psi_{\rm sm}
=(\mJ_{\rm lc}(z),\mD_{\rm lc}(z)\mJ_{\rm sm}(z))
\begin{pmatrix}
{\bf 0}\\  \wt \Psi_{\rm sm}
\end{pmatrix}.\label{eq:lsdata2}
\end{align}
From Eqs.\,\eqref{eq:lsdata1} and \eqref{eq:lsdata2},
the half-ADHM data $(Z,\Psi, \wt \Psi)$ can be read off as
\begin{align}
\Psi= (P^\mD_{\rm sm}\Psi_{\rm lc}, \,\Psi_{\rm sm}),\quad 
Z=\begin{pmatrix}
Z_{\rm lc} & {\bf 0}\\ P_{\rm sm}^\mJ \Psi_{\rm lc}&Z_{\rm sm} 
\end{pmatrix},\quad 
\wt \Psi=\begin{pmatrix}
{\bf 0}\\  \wt \Psi_{\rm sm}
\end{pmatrix}.
\end{align}
Using these data, we find that
\begin{align}
Z^{p-1} \wt \Psi = \begin{pmatrix}
{\bf 0}\\ Z_{\rm sm}^{p-1} \wt \Psi_{\rm sm}
\end{pmatrix},\quad {\rm with~} p=1,2,\cdots.
\end{align}
From this expression, we find that there are nonzero row vectors such that $\vec v Z^{p-1} \wt \Psi = 0$.
Therefore, we find that the contrapositive ``$\neg \, {\cal C}_{\rm semi}(\xi(z)) \to \neg \, {\cal C}_{\rm free}(Z,\wt \Psi)$'' is true 
and hence the lemma ``$ {\cal C}_{\rm free}(Z, \wt \Psi) \to  {\cal C}_{\rm semi}(\xi(z))$'' is also true. 
\end{proof}
%%%%%%%%%%%%%%%%%%%%%%%%%%%%%%%%%%%%%%%%%%
\begin{proof}[Proof of ${\cal C}_{\rm semi}(\xi(z)) \to {\cal C}_{\rm free}(Z,\wt \Psi)$]
Let us prove the contrapositive $\neg \, {\cal C}_{\rm free}(Z,\wt \Psi) \to \neg \, {\cal C}_{\rm semi}(\xi(z))$. 
If $GL(k,\mathbb C)$ does not freely act on  $(Z,\wt \Psi)$,  
that is, 
there exists  a certain non-zero row vector $\vec v$ satisfying 
\begin{align}
\vec v Z^{p-1} \wt \Psi = 0~~\mbox{for}~~\forall p \in \mathbb N \quad (\neg \, {\cal C}_{\rm free}(Z,\wt \Psi)~\mbox{is true}),
\end{align}
$(Z,\Psi,\wt \Psi)$ can be transformed by a $GL(k,\mathbb C)$ transf  into the form 
\begin{align}
\Psi=(\Psi_{\rm lc}, \Psi_{\rm sm}),\quad 
Z=\begin{pmatrix}
Z_{\rm lc} & {\bf 0}\\ W_{\rm ls} & Z_{\rm sm}
\end{pmatrix}, 
\quad \wt \Psi =\begin{pmatrix}
{\bf 0}\\  \wt \Psi_{\rm sm}
\end{pmatrix},
\label{eq:D33}
\end{align}
where $Z_{\rm sm}, \Psi_{\rm sm}$ and $\wt \Psi_{\rm sm}$ 
are $k'$-by-$k'$ $(k'<k)$, $n$-by-$k'$ and $k'$-by-$m$ matrices, respectively. 
The set of matrices $(Z_{\rm sm}, \Psi_{\rm sm},\wt \Psi_{\rm sm})$ can be regarded as the half-ADHM data satisfying ${\cal C}_{\rm free}(Z_{\rm sm},\wt \Psi_{\rm sm})$ with a smaller vortex number $k'(<k)$ since there is no $k'$-component column vector $\vec v$ such that $\Psi_{\rm sm} Z_{\rm sm}^{p-1} \vec v=0$
\begin{align}
\Psi_{\rm sm} Z_{\rm sm}^{p-1} \vec v=0 \quad \to \quad
\Psi Z^{p-1} 
\begin{pmatrix}
{\bf 0}\\ \vec v
\end{pmatrix}=\begin{pmatrix}
{\bf 0}\\ \Psi_{\rm sm} Z_{\rm sm}^{p-1} \vec v
\end{pmatrix}={\bf 0}\quad \to \quad \vec v=0.
\end{align}
According to the lemma ``$ {\cal C}_{\rm free}(Z, \wt \Psi) \to  {\cal C}_{\rm semi}(\xi(z))$'' shown above, the condition ${\cal C}_{\rm free}(Z_{\rm sm},\wt \Psi_{\rm sm})$
immediately indicates that the equivalence class of the corresponding matrix $[\xi_{\rm sm}(z)]$ is an element of ${\cal M}_{\rm semi \,}{}_{k'}^{n,m}$.
Furthermore, we can show that for $\xi_{\rm sm}(z)=(\mD_{\rm sm}(z),\wt \mD_{\rm sm}(z))$ and $\xi=(\mD(z),\wt \mD(z))$ corresponding to $(Z,\Psi,\wt \Psi)$,
\begin{align}
\mD_{\rm sm}(z)^{-1}\wt \mD_{\rm sm}(z) = \Psi_{\rm sm} (z{\bf 1}_{k'}-Z_{\rm sm})^{-1}\wt \Psi_{\rm sm} =
\Psi (z{\bf 1}_k-Z)^{-1}\wt \Psi = \mD(z)^{-1}\wt \mD(z),
\label{eq:D35}
\end{align}
where we have used Eqs.\,\eqref{eq:hADHM1_AppendixD},  \eqref{eq:hADHM2_AppendixD} and the relation $\Psi Z^{p-1}\wt \Psi =\Psi_{\rm sm} Z^{p-1}_{\rm sm}\wt \Psi_{\rm sm}$ which follows from Eq.\,\eqref{eq:D33}. The relation $\mD_{\rm sm}(z)^{-1}\wt \mD_{\rm sm}(z)=\mD(z)^{-1}\wt \mD(z)$ implies that $[\xi(z)] \not \in {\cal M}_{\rm semi \,}{}_{k}^{n,m}$. 
This is because if the opposite is true ($[\xi(z)] \in {\cal M}_{\rm semi \,}{}_{k}^{n,m}$), 
Lemma \ref{lmm:4} implies that 
the relation \eqref{eq:D35} leads to 
$[\xi(z)]=[\xi_{\rm sm}(z)]$ and $k=k'$, 
which is inconsistent with $k'<k$.
Thus, we find that ``$ \neg \, {\cal C}_{\rm free}(Z,\wt \Psi) \to \neg \, {\cal C}_{\rm semi}(\xi(z))$'', and hence the lemma ``$ {\cal C}_{\rm semi}(\xi(z)) \to  {\cal C}_{\rm free}(Z, \wt \Psi) $'' is shown. 
\end{proof}
%%%%%%%%%%%%%%%%%%%%%%%%%%%%%%%%%%%%%%%%%%%%%%%%%%%
%%%%%%%%%%%%%%%%%%%%%%%%%%%%%%%%%%%%%%%%%%%%%%%%%
\subsection{Instanton solutions in the Grassmannian sigma model}
\label{sec:inst-sol}
Any semilocal vortex solution becomes a instanton solution in the sigma model limit $g \rightarrow \infty$. 
In this subsection, we show that 
there is actually a one-to-one correspondence 
between the semilocal vortex and instanton solutions.
\begin{thm} 
Let $\varphi=\varphi(z,\bar z)$ be 
an $n$-by-$m$ matrix valued field 
(inhomogeneous coordinates of $G(n,n+m)$) 
on the base space $\mathbb C$.
If $\varphi$ satisfies the BPS instanton equation
\begin{align}
\partial_{\bar z} \varphi(z,\bar z) = 0, 
\quad\lim_{|z|\to \infty}\varphi(z,\bar z)=0,
\label{eq:sm_eq}
\end{align}
and each matrix entry of $\varphi(z)$ has 
a finite number of poles, 
one can uniquely determine 
the corresponding equivalent class $[\xi(z)] = [(\mD(z),\wt \mD(z))] \in {\cal M}_{\rm semi \,}{}^{n,m}_k$ 
and the half-ADHM data $[(Z,\Psi,\wt \Psi)]$. 
Explicitly, the instanton solution 
$\varphi(z)$ can be always written as
\begin{align}
\varphi(z)=\mD(z)^{-1}\wt \mD(z)=\Psi (z{\bf 1}_k-Z)^{-1}\wt \Psi.
\end{align}
\end{thm}
\begin{proof}
The solution of Eq.\,\eqref{eq:sm_eq} can always be written 
in the following form
\begin{align}
\varphi(z)=\sum_{\alpha}\sum_{p=1}^{k_\alpha}\frac{C_{\alpha,p}}{(z-z_\alpha)^p}
\end{align}
with $n$-by-$m$ constant matrices $C_{\alpha,p}$.
For this solution, let us consider an $n$-by-$(n+m)$ matrix $\xi'(z)$ given by
\begin{align}
\xi'(z)=(p(z){\bf 1}_n, p(z) \varphi(z))\quad {\rm with} \quad   
p(z)=\prod_{\alpha}(z-z_\alpha)^{k_\alpha},\quad k'=\sum_\alpha k_\alpha.
\end{align}
Note that all the entries of a matrix $p(z) \varphi(z)$ are polynomials since all the poles in $\varphi(z)$ are cancelled with zeros of $p(z)$.
The equivalence class $[\xi'(z)]$
is an element of ${\cal M}_{\rm vtx \,}{}^{n,m}_{nk'}$
and hence, according to Lemma \ref{lmm:1}, 
there exists a unique equivalence class 
$[\xi(z)]=[(\mD(z),\wt \mD(z))] \in {\cal M}_{\rm semi \,}{}^{n,m}_k$ with an integer $k \in \mathbb Z_{\ge 0}$ 
such that 
\begin{align}
\xi'(z)=\mD_{\rm lc}(z) \xi(z)=\mD_{\rm lc}(z)(\mD(z),\wt \mD(z)), \quad 0 \le k \le nk',
\end{align}
with $\mD_{\rm lc}(z) \in {\cal G}_{n,nk'-k}[z]$  and  $\mD(z)\in {\cal G}_{n,k}[z]$. 
These matrices are related as 
\begin{align}
p(z) {\bf 1}_n=\mD_{\rm lc}(z)\mD(z),  \quad p(z) \varphi(z)=\mD_{\rm lc}(z)\wt \mD(z).
\end{align}
In terms of the matrices $(\mD(z), \wt \mD(z))$,  
the instanton solution $\varphi(z)$ can always be rewritten as  
\begin{align}
\varphi(z) = \frac1{p(z)} (p(z) \varphi(z))=(\mD_{\rm lc}(z) \mD(z))^{-1} \mD_{\rm lc}(z) \wt \mD(z) =\mD(z)^{-1}\wt\mD(z).
\end{align}
Furthermore, using the half-ADHM mapping relation, we can rewrite $\varphi(z)$ in terms of the corresponding half-ADHM data $(Z,\Psi, \wt \Psi)$ as,
\begin{align}
\varphi(z)=\mD(z)^{-1}\wt\mD(z)=\mD(z)^{-1}\mJ(z)\wt \Psi=\Psi (z{\bf 1}_k-Z)^{-1}\wt \Psi.
\end{align}
According to Lemma \ref{lmm:4}, for a given $\varphi(z)$, 
the equivalent class $[(\mD(z), \wt \mD(z))] \in {\cal M}_{\rm semi \,}{}^{n,m}_k$ satisfying the above is unique,
and the equivalent class of the half-ADHM data $[(Z,\Psi,\wt \Psi)]$ satisfying ${\cal C}_{\rm free}(Z,\wt \Psi)$ is also unique.
\end{proof}
%\begin{align}
%\varphi(z)\quad  \mapsto \quad \left([(\mD(z), \wt \mD(z))], k\right) \in \bigsqcup_{k\in \mathbb Z_{\ge 0}} {\cal M}_{\rm semi \,}{}^{n,m}_k: \quad \varphi(z)=\mD(z)^{-1}\wt \mD(z).
%\end{align}

%%%%%%%%%%%%%%%%%%%%%%%%%%%%%%%%%%%%%%%%%%%%%%%%
\section{Embedding of Grassmannian case}
\label{appendix:embedding}
In this appendix, we discuss vortices 
obtained by embedding from the $L=1$ case.
\subsection%[Embedding of vortices from L=1 to L=2]
{Embedding of vortices from $L=1$ to $L=2$ }
Let us consider first consider the embedding of the matrix $\mH$ from $L=1$ to $L=2$. 
For example, in the case with $(k_1,k_2)=(k,0)$, 
$\mH_i$ are given by,    
\begin{eqnarray}
\mH_1(z)=( \mD(z), \widetilde \mD(z), {\bf 0} ), \hs{10}
\mH_2(z)=\left( \begin{array}{ccc} {\bf 1}_{n_1} & {\bf 0} & {\bf 0} \\ {\bf 0} & {\bf 1}_{n_2}& {\bf 0} \end{array} \right), 
\end{eqnarray}
where $\xi=(\mD(z), \widetilde \mD(z))$ is the matrix for the $L=1$ case 
with $n=N_1=n_1$ and $N=N_2=n_1+n_2$.
For the case with $(k_1,k_2)=(0,k)$, 
one can find that the general solution turns out to be     
\begin{eqnarray}
\mH_1 (z)=( {\bf 1}_{n_1},  {\bf 0},  {\bf 0} ), \hs{10}
\mH_2(z)=\left( \begin{array}{ccc} {\bf 1}_{n_1}& {\bf 0} &{\bf 0} \\ {\bf 0} & \mD(z) &\widetilde \mD(z)\end{array} \right), 
\end{eqnarray}
where $\mD(z)$ and $\widetilde \mD(z)$ are 
those for the  $L=1$ case with 
$n=n_2$ and $N=n_2+n_3$. 
Thus, we find that the moduli spaces of vortices 
for $(k_1,k_2)=(k_1,0)$ and $(k_1,k_2)=(0,k_2)$ are identical with 
those of the $L=1$ case
\begin{align}
{\cal M}_{\rm vtx \,}{\,}^{n_1,n_2,n_3}_{k_1,k_2=0} \simeq {\cal M}_{\rm vtx \,}{\,}^{n_1,n_2}_{k_1} , \quad  {\cal M}_{\rm vtx \,}{\,}^{n_1,n_2,n_3}_{k_1=0,k_2} \simeq {\cal M}_{\rm vtx \,}{\,}^{n_2,n_3}_{k_2}.
\end{align}
For $(k_1,k_2)=(k,k)$,  
there exist a subspace in the moduli space 
where $\mH$ of the $L=1$ case can be embedded as 
\begin{eqnarray}
\mH_1(z)=( \mD(z),  {\bf 0}, \widetilde \mD(z) ), \quad \mH_2(z)=\left( \begin{array}{ccc} \mD(z) & {\bf 0} &\widetilde \mD(z) \\ 0 & {\bf 1}_{n_2}& 0 \end{array} \right),
\end{eqnarray}
where $\mD(z)$ and $\widetilde \mD(z)$ 
are those for the  $L=1$ case with taking $n=n_1$ 
and $N=n_1+n_3$. 
This means that the moduli space of vortices with $(k_1,k_2)=(k,k)$ contains the $L=1$, $k$-vortex moduli space 
\begin{align}
{\cal M}_{\rm vtx \,}{\,}^{n_1,n_2,n_3}_{k_1=k,k_2=k} \supset {\cal M}_{\rm vtx \,}{\,}^{n_1,n_3}_{k}.
\end{align}

\paragraph{Embedding of half-ADHM data \\}
The half-ADHM data can also be obatained by embedding that of the $L=1$ case.

\noindent $\bullet ~ (k_1,k_2)=(k,0)$ \\
If $(k_1,k_2)=(k,0)$, we can determine $\mD_1$ and $\mD_2$ 
from the condition ${\rm deg} (\det \mD_i) = k_i$
\beq
\mD_1 = z^k + \sum_{n=0}^{k-1} a_n z^n, \hs{10}
\mD_2 = \ba{cc} 1 & 0 \\ 0 & 1 \ea, \hs{10}( \mbox{up to $V$-transformations}). 
\eeq
From the relations
\beq
\xi_1 = (\mD_1 , \wt \mD_1) = q_1 q_2 , \hs{10}
\xi_2 = (\mD_2 , \wt \mD_2) = q_2, \hs{10} 
\mD_i^{-1} \wt \mD_i = \mathcal O(z^{-1}),
\eeq
the matrices $q_1$, $q_2$, $\wt \mD_1$ and $\wt \mD_2$ can be determined as 
\beq
q_1 = ( P(z), Q(z) ) , \hs{5}
q_2 = \ba{ccc} 1 & 0 & 0 \\ 0 & 1 & 0 \ea, \hs{10}
\wt \mD_1 = (Q(z),0), \hs{5}
\wt \mD_2 = \ba{c} 0 \\ 0 \ea. 
\label{eq:q_(k,0)}
\eeq
where $P(z)$ and $Q(z)$ are polynomials of the forms
\beq
P(z) = z^k + \sum_{n=0}^{k} a_n z^n, \hs{7}
Q(z) = \sum_{n=0}^{k-1} b_n z^n. 
\label{eq:PandQ}
\eeq
The $k$-component row vector $\mJ_1(z)$ can be determined 
from the condition $\mD_1^{-1} \mJ_1 = \mathcal O(z^{-1})$ as
\beq
\mJ_1 = (\tilde P_{k-1}, \tilde P_{k-2} , \cdots , \tilde P_0 ), ~~~\mbox{with}~~~ \tilde P_l = \sum_{n=0}^{l} a_{n+k-l} z^{n}.
\label{eq:Jexample}
\eeq
Note that the $N_2$-by-$k_2$ matrix $\mJ_2$ does not exist since $k_2=0$ in this case. 
From $\mD_1$, $\wt \mD_1$ and $\mJ_1$, 
the matrices $Z_1$, $\mY_1$ and $\mY_2$ 
can be determined as 
\beq
Z_1 =
{\renewcommand{\arraystretch}{1.0}
{\setlength{\arraycolsep}{1.0mm} 
\ba{c|ccc}
0 & 1 & & \\
\vdots & & \ddots & \\
0 & & & 1 \\ \hline
a_0 & a_1 & \cdots & a_{k-1}
\ea, 
\hs{5}
\mY_1 = \ba{cccc} 1 & 0 & \cdots & 0 \ea, 
\hs{5}
\wt \mY_1 = \ba{c} \tilde b_1 \\ \tilde b_2 \\ \vdots \\ \tilde b_k \ea,}}
\eeq
where $\tilde b_l~(l=1,\cdots,k)$ are constants such that
\beq
Q(z) = \sum_{l=1}^k \tilde b_l \tilde P_l . 
\label{eq:tildeb}
\eeq
Note that $(Z_2,\mY_2,\wt\mY_2)$ and $(W_1 \tilde W_1)$ do not exist since $k_2=0$ in this case. 
All these moduli data are identical to those of $U(1)$ semi-local vortex $(L=1,~N_1=1,~\NF=2)$ with $q_1=(P(z),Q(z))$. 

\noindent $\bullet ~ (k_1,k_2)=(0,k)$ \\
For $(k_1,k_2)=(0,k)$, the matrices $q_1$ and $q_2$ are given by
\beq
q_1 = (1, 0) \hs{10}
q_2 = \ba{ccc} 1 & 0 & 0 \\ 0 & P(z) & Q(z) \ea,
\eeq
where $P(z)$ and $Q(z)$ are polynomials of the same form 
as the previous case \eqref{eq:PandQ}. 
The matrices $\xi_i =(\mD_i,\wt \mD_i)$ take the forms
\beq
{\renewcommand{\arraystretch}{0.8}
{\setlength{\arraycolsep}{0.8mm} 
\xi_1 = (1, 0, 0), \hs{5}
\xi_2 = \ba{ccc} 1 & 0 & 0 \\ 0 & P(z) & Q(z) \ea, \hs{3}
\bigg( \mD_1 = 1, ~ \wt \mD_1 = (0,0), ~
\mD_2 = \ba{cc} 1 & 0 \\ 0 & P(z) \ea, ~
\wt \mD_2 = \ba{c} 0 \\ Q(z) \ea \bigg).}}
\eeq
The matrix $\mJ_1(z)$ does not exist and 
$\mJ_2$ can be determined as
\beq
J_2 = \ba{cccc} 0 & 0 & \cdots & 0 \\ \tilde P_{k-1} & \tilde P_{k-2} & \cdots & \tilde P_{0} \ea,
\eeq
where $\tilde P_l~(l=0,\cdots,k-1)$ are the same polynomials as \eqref{eq:Jexample}.
Since $k_1=0$, the matrices $(Z_1,\mY_1,\wt \mY_1)$ and $(W_1,\wt W_1)$ do not exist and $(Z_2,\mY_2,\wt \mY_2)$ are given by
\beq
Z_2 = 
{\renewcommand{\arraystretch}{1.0}
{\setlength{\arraycolsep}{1.0mm} 
\ba{c|ccc}
0 & 1 & & \\
\vdots & & \ddots & \\
0 & & & 1 \\ \hline
a_0 & a_1 & \cdots & a_{k-1}
\ea}},
\hs{5}
\mY_2 = \ba{cccc} 0 & 0 & \cdots & 0 \\ 1 & 0 & \cdots & 0 \ea, 
\hs{5}
\wt \mY_2 = \ba{c} \tilde b_1 \\ \tilde b_2 \\ \vdots \\ \tilde b_k \ea,
\eeq
where $\tilde b_l~(l=1,\cdots,k)$ are constants defined in \eqref{eq:tildeb}. Again, the moduli data are identical to those of $U(1)$ semi-local vortex $(L=1,~N_1=1,~\NF=2)$ with $q_1=(P(z),Q(z))$.

\subsection%[Embedding of vortices from L=1 to general L]
{Embedding of vortices from $L=1$ to general $L$ }
Let $I,J$ be integers such that $1 \leq I < J \leq L+1$
and prepare a $n_I$-by-$n_I$ matrix $\mD(z)$ and 
a $n_I$-by-$n_J$ matrix $\widetilde \mD(z)$
satisfying $\det \mD(z)={\cal O}(z^k)$ and $\mD(z)^{-1}\widetilde \mD(z)={\cal O}(z^{-1})$.
Next, let us embed these matrices into $q_{J-1}$ as 
\begin{align}
q_{J-1}(z)=\left( 
\begin{array}{ccc|c}
{\bf 1}_{N_{I-1}}&{\bf 0}&{\bf 0}&{\bf 0}\\
{\bf 0}&\mD(z)&{\bf 0}&\widetilde \mD(z)\\
{\bf 0}&{\bf 0}&{\bf 1}_{N_{J-1}-N_{I}}&{\bf 0} 
\end{array}
\right), \hs{20}
\end{align}
and set the other $q_i$ to be trivial $q_i=({\bf 1}_{N_i},\bf 0)$ for $i \not= J-1$. This setting gives 
\begin{eqnarray}
&& \hs{38} \overbrace{}^{N_{J-1}-N_i} \hs{6}  \overbrace{}^{N-N_{J}} \\
\mH_i(z) &=& \left(\begin{array}{ccc|ccc}
{\bf 1}_{N_{I-1}}&{\bf 0}&{\bf 0}&\bf 0&{\bf 0}&
\bf 0\\
\, {\bf 0} \, &\mD(z)&{\bf 0}&{\bf 0}&\widetilde \mD(z)& \, {\bf 0} \,\\
{\bf 0}&{\bf 0}&{\bf 1}_{N_{i}-N_{I}}&{\bf 0} &{\bf 0}&{\bf 0}
\end{array}
\right) ~~~ {\rm for} ~~~ i\in [I, J-1], 
\end{eqnarray}
and $\mH_i(z)=({\bf 1}_{N_i},{\bf 0})$ for $i\not \in [I, J-1]$. The vortex numbers are given by
\begin{eqnarray}
(k_1,k_2,\cdots,k_L)=(\underbrace{0,\cdots,0}_{I-1},\underbrace{k,k,\cdots,k}_{J-I},0,\cdots,0).
\end{eqnarray}
In the case of $d=1$, this construction gives general configurations with $k$ elementary vortices.
In cases with $d \ge 2$, however, 
the above construction gives special configurations 
where each of $k$ objects can be regarded as 
a composite state of $J-I$ types of elementary vortices.
This construction also gives the following sigma model instanton solution 
\begin{eqnarray}
\varphi_{ij}=\delta_i^{I} \delta_j^{J}  \mD(z)^{-1}\widetilde \mD(z),
\end{eqnarray}
where $\varphi_{ij}$ are inhomogeneous coordinates of the flag manifold defined in Eq.\,\eqref{eq:inhomogeneous}. 
%Note that $d$ implies a level of compression of vortices and even in the nonlinear sigma model it still keeps physical meaning due to the existence of the complex structure which giving the ordering of the unitary groups $\{U(n_1),\cdots,U(n_{L+1})\}$, otherwise we can freely replace a role for $U(n_{i_c+d})$ and that for $U(n_{i_c+1})$. If the complex structure and tensions of vortices are not taken  into account, there is no  difference between an  elemental vortex and  a composite vortex.

The corresponding half-ADHM data takes the form
\begin{eqnarray}
\Upsilon_{i}={\bf 0} \quad{\rm for}~ i\not=I,\quad \widetilde \Upsilon_i={\bf 0} \quad
{\rm for} ~i\not=J-1,\\
Z_{I}=Z_{I+1}=\cdots=Z_{J-1},\quad\\
W_{I}=W_{I+1}=\cdots =W_{J-2}={\bf 1}_k
\end{eqnarray}
Here $d = J-I$ implies a level of compression of vortices and 
turns out to corresponds to height of the Young tableaux.

%%%%%%%%%%%%%%%%%%%%%%%%%%%%%%%%%%%%
%%%%%%%%%%%%%%%%%%%%%%%%%%%%%%%%%%%%%%%%%%%%%%%%

\section{Brane construction of vortices}\label{sec:brane}
In this appendix, we discuss the D-brane construction of BPS vortices. 
By embedding our system into a 4d $\mathcal N =2$ supersymmetric gauge theory, 
we can identify the D-brane configuration corresponding to the BPS vortex configurations. 
For $L=1$, the D-brane construction of the vortex moduli space has been mentioned in \cite{Hanany:2003hp} and for $L>1$ with $N_1=\cdots=N_L=N$, 
the D-brane configuration for the local vortices has been discussed in \cite{Chen:2011sj}. 
The left figure in Fig.\,\ref{fig:brane-vacua} shows the brane configuration for the Coulomb branch of the model. 
The 4d $\mathcal N =2$ quiver gauge theory corresponding to our system can be realized as the worldvolume effective theory on D4-branes attached to NS5-branes. 
There are $N_i$ D4-branes between neighboring NS5 branes 
and they correspond to the $U(N_i)$ subgroup of the gauge group. 
The gauge coupling constants $1/g_i^2~(i=1,\cdots,L)$ are proportional to 
the separations of the NS5 branes $\Delta x^6_{{\rm NS5},i} = x^6_{{\rm NS5},i+1}-x^6_{{\rm NS5},i}$. 
The bi-fundamental fields $Q_i~(i=1,\cdots,L)$ (hypermultiplets) 
corresponds to the fundamental strings 
stretched between D4-branes in the $i$-th and $i+1$ intervals.
In the presence of the FI parameters, 
which correspond to 
$\Delta x^{7,8,9}_{{\rm NS5},i} = x^{7,8,9}_{{\rm NS5},i+1}-x^{7,8,9}_{{\rm NS5},i}~(i=1,\cdots,L)$, 
the vacuum is in the Higgs phase as shown in the right figure of Fig.\,\ref{fig:brane-vacua}. 
There are $n_i=N_i-N_{i-1}$ D4-branes attached to the $i$-th NS5-brane. 
Fig.\,\ref{fig:brane-vortex} shows an example of the D-brane configurations for BPS vortices. 
The vortices with $i$-th magnetic flux correspond to D2-branes stretched 
between the $i$-th and $(i+1)$-th D4-branes. 
The vortex worldsheet theory 
is a 2d $\mathcal N =(2,2)$ quiver gauge theory 
on the $(x^0,x^1)$ plane in Table \ref{tab:brane}.
The matrices $(Z_i,\Upsilon_i,\wt \Upsilon_i,W,\wt W)$ 
are identified with the component fields of the chiral multiplets
which are identified with the degrees of freedom in the brane configuration as follows 
\begin{itemize}
\item
$Z_i$ : positions of $i$-th D2-branes on $(x^2,x^3)$ plane.
\item
$\Upsilon_i$ : F1 strings between $i$-th D4-branes and $i$-th D2-branes. 
\item
$\wt \Upsilon_i$ : F1 strings between $(i+1)$-th D4-branes and $i$-th D2-branes.
\item
$W, \wt W$ : F1 strings between $i$-th and $(i+1)$-th D2-branes.
\end{itemize}
The moduli space of BPS vortices are identified with 
that of vacua of this quiver gauge theory 
determined by solving the $D$-term condition and 
the $F$-term constraint coming from the cubic superpotential
\beq
\supW=\sum_{i=1}^{L-1}\tr\left[ \wt \mW_i  
\left(\wt \mY_i \mY_{i+1}- Z_i\mW_i+\mW_i Z_{i+1}\right)\right]. 
\eeq
If we turn on hypermultiplets masses, 
which correspond to the positions of D6-branes on the $(x^4,x^5)$ plane, 
only the fixed points of the $SU(N)$ flavor symmetry are left 
as stable BPS configurations. 
Fig.\,\ref{fig:brane-massive} shows an example of the D-brane configurations for such fixed point configurations.
In the presence of the $\Omega$-deformation on the $(x^0,x^1)$ plane, 
all D2-branes are localized at the origin 
and they form clusters which are characterized by Young tableaux. 
Such configurations are the fixed points of the torus action, 
which are relevant to the supersymemtric localization. 

\begin{table}[h!]
\centering
\begin{tabular}{c|ccccccccccccc}
& $x^0$ & $x^1$ & $x^2$ & $x^3$ & $x^4$ & $x^5$ & $x^6$ & $x^7$ & $x^8$ & $x^9$ \\ \hline
D4 & $\times \ $ & $\times \ $ & $\times \ $ & $\times \ $ & & & $\times \ $ & & &  \\
D6 & $\times \ $ & $\times \ $ & $\times \ $ & $\times \ $ & & & & $\times \ $ & $\times \ $ & $\times \ $ \\
NS5 & $\times \ $ & $\times \ $ & $\times \ $ & $\times \ $ & $\times \ $ & $\times \ $ & & & & \\
D2 & $\times \ $ & $\times \ $ & & & & & & & & $\times \ $ 
\end{tabular}
\caption{Brane configuration ($\times$'s indicate the directions in which the branes extend).}
\label{tab:brane}
\end{table}
\begin{figure}[!ht]
\begin{center}
\fbox{
\includegraphics[bb = 20 20 800 600, width=80mm]{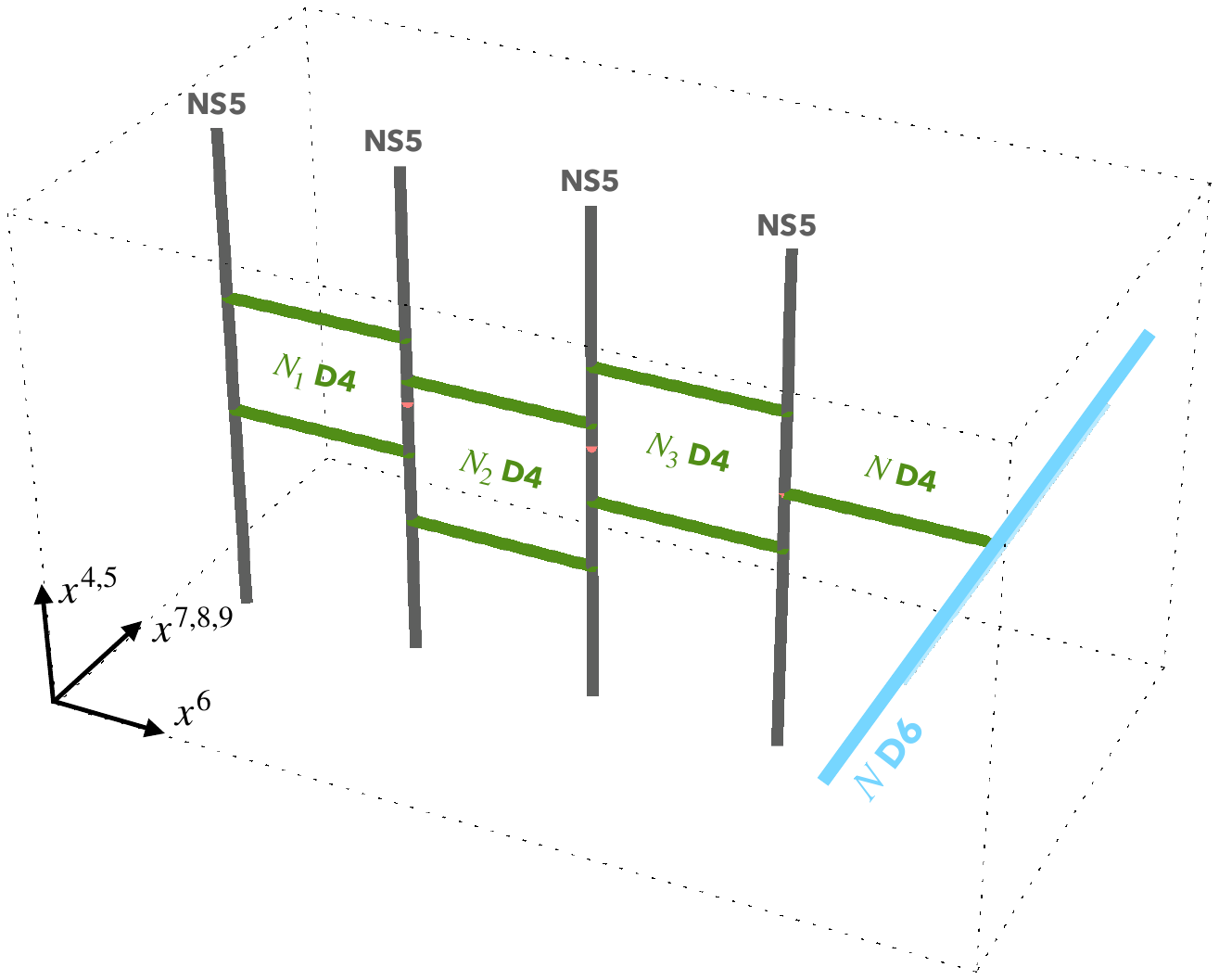}}
\fbox{
\includegraphics[bb = 20 20 800 600, width=80mm]{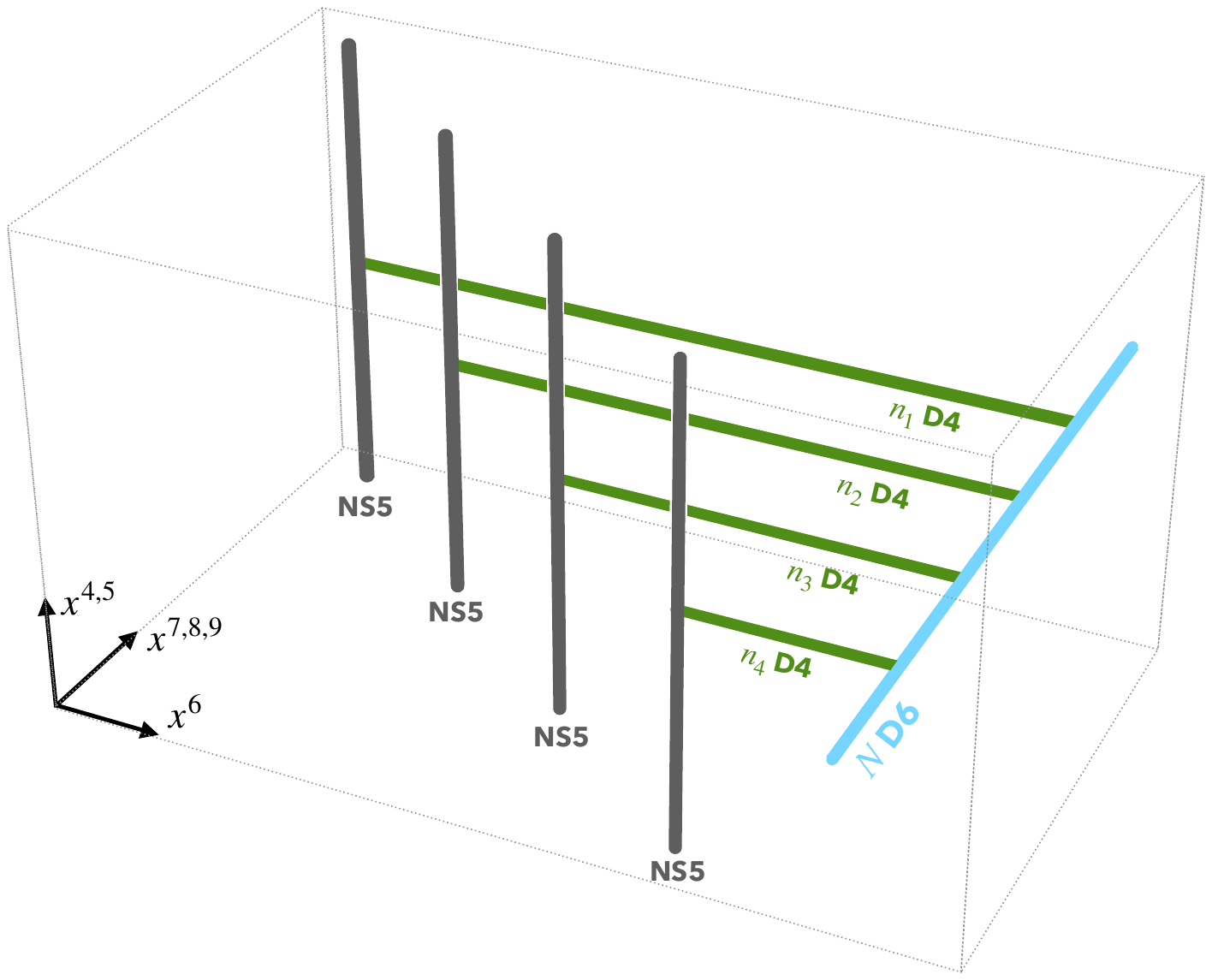}}
\caption{D-brane configurations for the Coulomb branch (left) and the Higgs phase in the presence of FI parameters (right).}
\label{fig:brane-vacua}
\end{center}
\end{figure}

\begin{figure}[!ht]
\begin{center}
\fbox{
\includegraphics[bb = 20 20 800 600, width=120mm]{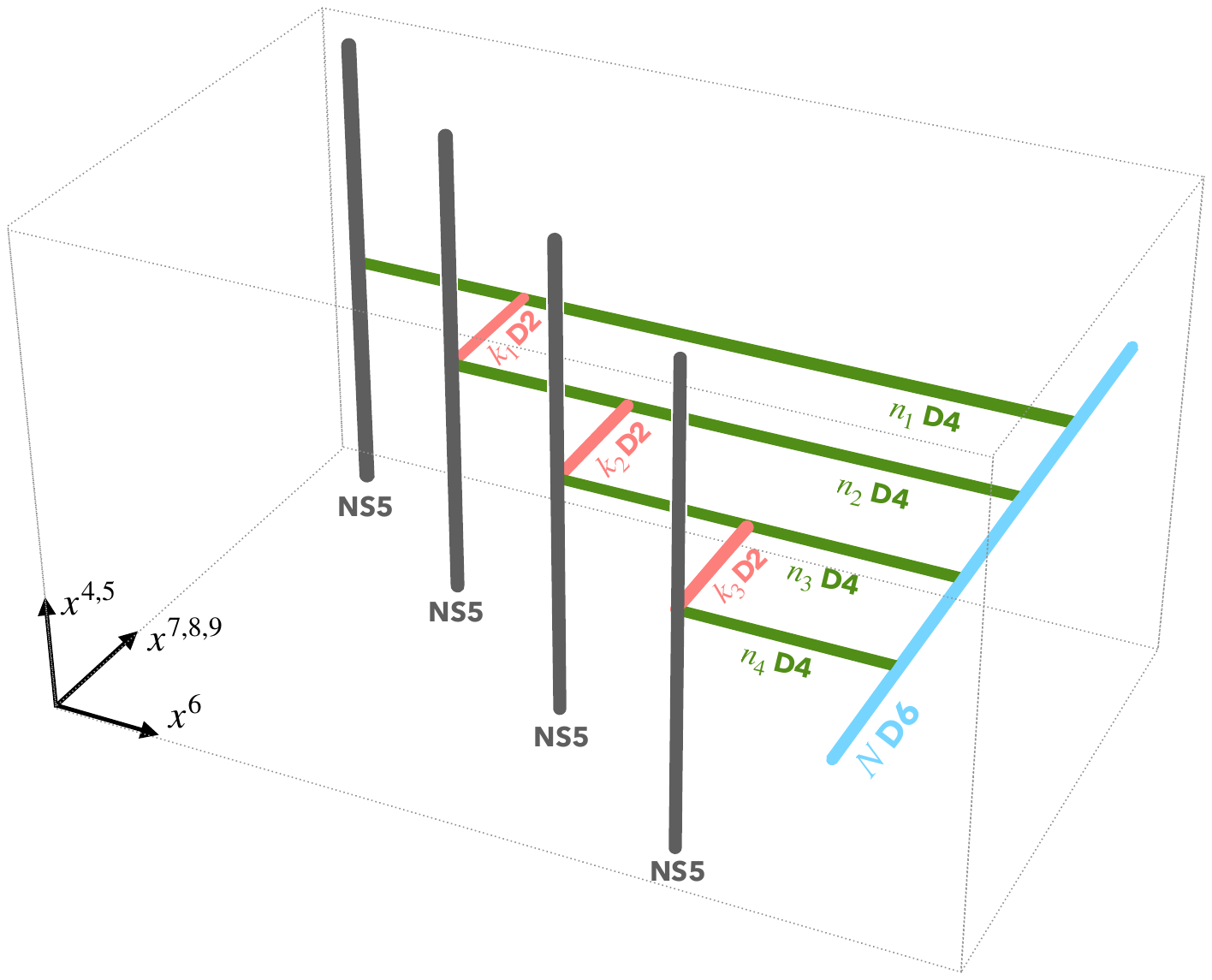}}
\caption{D-brane configuration for BPS vortices. This example shows $(k_1,k_2,k_3)$ vortices in 4d $\mathcal N = 2$ $U(n_1) \times U(n_1+n_2) \times U(n_1+n_2+n_3)$ gauge theory with $\NF=n_1+n_2+n_3+n_4$ hypermultiplets.}
\label{fig:brane-vortex}
\end{center}
\end{figure}

\begin{figure}[!ht]
\begin{center}
\fbox{
\includegraphics[bb = 20 20 800 600, width=120mm]{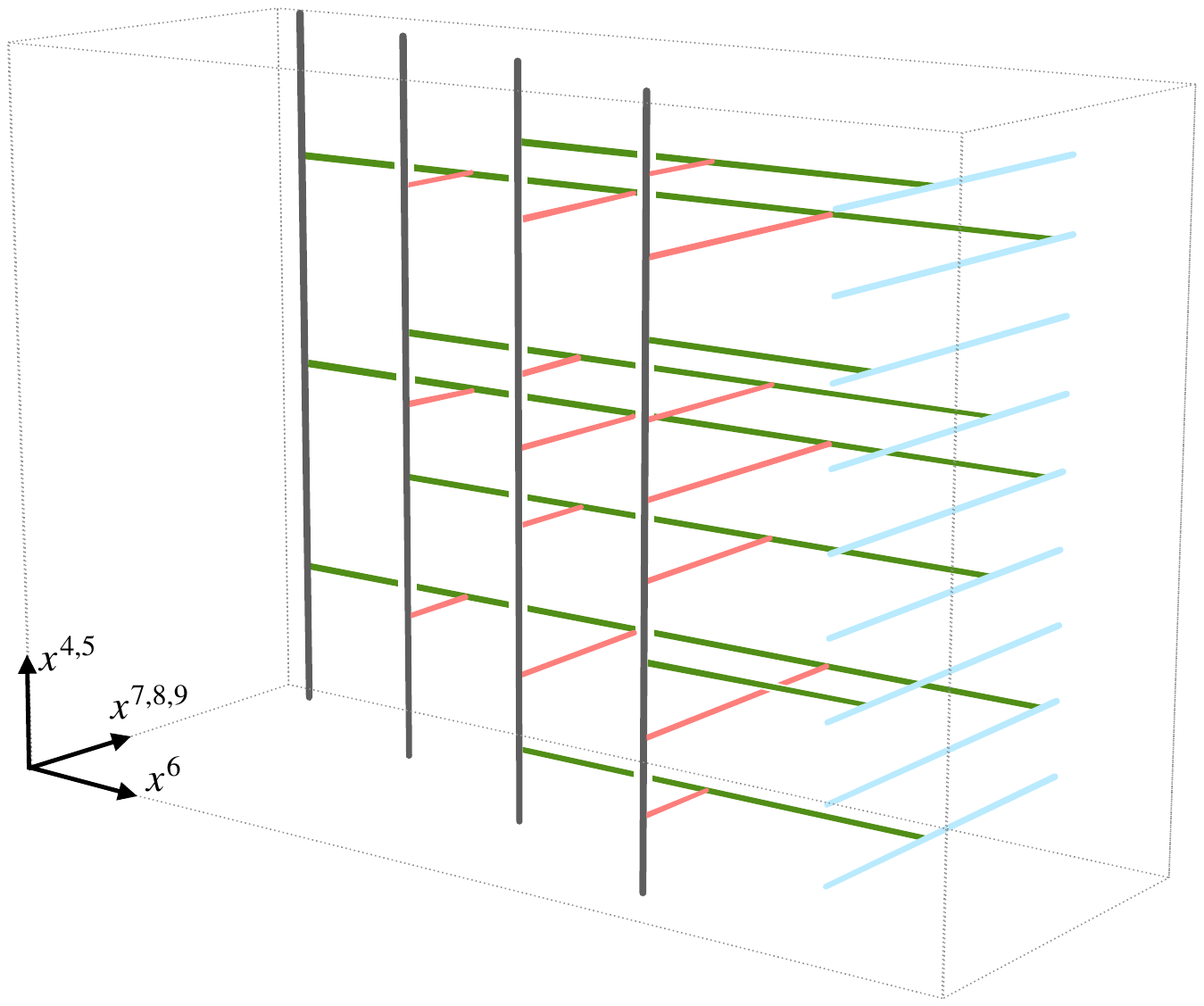}}
\caption{D-brane configuration for BPS vortices in the massive theory. The positions of D6-branes in the $x_4$ and $x_5$ directions correspond to hypermultiplet masses. In the presence of the $\Omega$-background, each cluster of D2-branes (pink line) corresponds to a composite of vortices characterized by a Young tableau.}
\label{fig:brane-massive}
\end{center}
\end{figure}

%%%%%%%%%%%%%%
\section{Smoothness of the moduli space}\label{sec:tW}
The moduli space of vortices for $L>1$ is constructed from the space of matrices satisfying the constaraitns \eqref{eq:constraint}. 
In this appendix, we show that those constraints do not cause any singularities on the moduli space.

\subsection{Singular points on algebraic varieties}
Let us first recall that a singularity on an algebraic variety is 
a point where a tangent space is ill-defined. 
For example, for a subspace $\mathcal M_0$ in 
$\mathbb C^K = \{ \phi_1,\cdots,\phi_K \}$ 
defined as the intersection of the zero loci of polynomials $F_I(\phi_1,\cdots,\phi_K)~(I=1,2,\cdots,n')$,  
a singularity on ${\cal M}_0$ is defined 
as a point where the rank of the matrix
$(J_F)_I{}^i \equiv \partial F_I/\partial \phi_i$ decreases. 
Let $H_{ij}$ be the hessian of the function $V$ defined as 
\begin{align}
V = \sum_{I} |F_I|^2 , \hs{10}
H_{ij}= \frac{\partial^2 V}{\partial \phi_i \partial \bar \phi_j}= \sum_I (J_F)_I^i (J_F^*)_I^j.
\end{align}
Then, at the singular point, 
an extra flat directions (zero eigenvectors) appears 
since $H_{ij}$ has a lower rank. 
If there is a symmetry group $G$ that preserves $V$, the equation $F=0$ reduces to a constraint equation $\tilde F=0$ that defines a subspace $\cal M$ in the quotient space $(\mathbb C^K-\{0\})/G$. 
Since ${\rm rank} \, J_{\tilde F} = {\rm rank} \, J_{F}$ for any smooth quotient space,
singularities of $\mathcal M$ can be determined by looking at the rank of $J_F$ on $\mathcal M_0$.
In particular, $\mathcal M$ is smooth if the rank of $J_F$ is constant everywhere. 
%For example, embedding of a Grassmannian $G(M,N)$ into $\mathbb CP^{K}$ with $K=N!/M!/(N-M)!-1$ causes several Pl\"ucker relations and $K-N(N-M)$ within them are linearly independent.   
If $J_F$ has the maximal rank $n'$ everywhere, that is,  
\begin{align}
0=\Lambda^I  (J_F)_I{}^i=\frac{\partial \Lambda^I F_I}{\partial \phi^i} \quad \mbox{implies} \quad \Lambda^I=0 \quad \mbox{for all points on ${\cal M}_0$},
\label{eq:J_max_rank}
\end{align}
all the constraints $F_I=0$ are linearly independent
and hence $\cal M$ is a $(K-{\rm dim} \, G - n')$-dimensional smooth manifold.  
\subsection{Constraints and smoothness of vortex moduli space}
The constraints in Eq.(\ref{eq:constraint}) implies that the 
vortex moduli space for $L>1$ is the intersection of the zero loci of $F_I = \wt \mY_i \mY_{i+1}- Z_i\mW_i+\mW_i Z_{i+1}$. 
These constraints can be introduced by turning on the potential 
\begin{eqnarray}
\mathcal V = \sum_{i=1}^{L-1} \tr\left[ |\wt \mW_i|^2 + i \wt \mW_i  
\left(\wt \mY_i \mY_{i+1}- Z_i\mW_i+\mW_i Z_{i+1} \right) + (c.c.) \right],
\label{eq:super_potential}
\end{eqnarray}
where the $k_{i+1}$-by-$k_i$ matrix $\wt \mW_i$ are  
an auxiliary fields which give the on-shell potential
\beq
V = \sum_{i=1}^{L-1} \tr \left| \wt W_i \right|^2 = \sum_{i=1}^{L-1} \tr \left| \wt \mY_i \mY_{i+1}- Z_i\mW_i+\mW_i Z_{i+1} \right|^2.
\eeq 
The variations with respect to the other degrees of freedom give
\begin{alignat}{2}
0 &\ = \ \frac{\partial \mathcal V ~~~~ }{\partial \mY_{i+1}} &\ = \ & \wt W_i \wt \mY_i, \label{eq:WY}\\
0 &\ = \ \frac{\partial \mathcal V}{\partial \wt \mY_i} &\ = \ & \mY_{i+1}\wt W_i,  \label{eq:WYt}\\
0  &\ = \ \frac{\partial \mathcal V}{\partial \mW_i} &\ = \ & Z_{i+1}\wt W_i-\wt W_i Z_i, \label{eq:WW}
\end{alignat}
for $1 \le i\le L-1$ and 
\begin{eqnarray}
0=\frac{\partial \mathcal V}{\partial Z_i}=\wt W_{i-1} W_{i-1}- W_i \wt W_i \quad {\rm with} \quad \mW_{0,L}=\wt \mW_{0,L}=0,
\label{eq:WZ}
\end{eqnarray}
for $1 \le i\le L$.
As we have seen in Eq.\,\eqref{eq:J_max_rank}, 
the moduli space has no singularity if and only if $\wt W_i$ always vanishes when Eqs.\,\eqref{eq:WY}-\eqref{eq:WZ} are satisfied.
For any solution, 
we can show 
\begin{eqnarray}
\left(\mY_j \mW_j\mW_{j+1}\cdots W_{i} Z_{i+1}^{p-1}\right) \wt \mW_i
&\stackrel{{\rm Eq.}(\ref{eq:WW})}{=}&
\mY_j \mW_j\mW_{j+1}\cdots W_{i} \wt \mW_i Z_{i}^{p-1}\nn
&\stackrel{{\rm Eq.}(\ref{eq:WZ})}{=}&
\mY_j \wt \mW_{j-1}\mW_{j-1}W_{j}\cdots W_{i-1} Z_{i}^{p-1}\nn
&\stackrel{{\rm Eq.}(\ref{eq:WYt})}{=}&0,
\end{eqnarray} 
for $1\leq j \leq i$ and $1 \leq p \leq k$.
Under the $\prod_{i=1}^L GL(k_i,\C)$ free condition (\ref{eq:GLkifree}), this equation implies that 
$\wt \mW_i=0$ for all $i$.
Therefore, the vortex moduli space is smooth and 
all the elements of the constraint (\ref{eq:constraint}) are independent, 
that is, the number of degrees of freedom suppressed by the constraints is the same as that of the components of $\{ \wt W_i \}$.
%%%%%%%%%%%%%%%%%%%%%%%

\section{The torus action on the moduli spaces and on the \kahler quotient}\label{sec:TorusActions}
In this appendix, we summarize the BPS vortex solutions 
in the presence of the omega background and the mass deformation. 
In such a case, BPS configurations have 
to minimize the deformation terms 
induced by the omega background $\epsilon$ and 
mass parameters $M={\rm diag}(m^1, \cdots, m^N)$
\beq
\delta \mathcal L ~=~ \sum_{i=1}^L \left| i \epsilon (z \D_z - \bar z \D_{\bar z}) q_i + \Sigma_i q_i - q_i \Sigma_{i+1} \right|^2 \hs{3}
\mbox{with} \hs{3} \Sigma_{L+1} = - M,
\eeq
where $\Sigma_i~(i=1,\cdots,L)$ are $SU(N_i)$ adjoint scalar fields.\footnote{In 2d $\mathcal N = (2,2)$ models, $\Sigma_i$ can be interpreted as the adjoint scalar fileds in the vector multiplets and become auxiliary fields in the nonlinear sigma model limit.} 
Since $\delta L$ is positive semi-definite, it is minimized when $\delta \mathcal L =0$, that is 
\beq
i \epsilon (z \D_z - \bar z \D_{\bar z}) q_i + \Sigma_i q_i - q_i \Sigma_{i+1} = 0, \hs{5}
(i=1,\cdots,L,~\Sigma_{L+1} = - M).
\label{eq:fixed_infinitesimal}
\eeq
This condition implies that the vortex configuration must be 
invariant under the (infinitesimal) spatial rotation and 
the flavor rotation up to gauge transformations $\Sigma_i$. 
Such fixed points are classified by a set of $N$ Young tableaux 
$Y^{(j,\alpha)}$ where $\alpha = 1, \cdots, n_i$ for each $j=1,\cdots,L$. 
The height of $Y^{(j,\alpha)}$ is $L-j+1$ and 
we denote the length of $i$-th row as $l_{i+j-1}^{(j,\alpha)}$, i.e. 
\beq
Y^{(j,\alpha)} = \left( l_j^{(j,\alpha)},l_{j+1}^{(j,\alpha)},\cdots,l_L^{(j,\alpha)} \right), \hs{10} 
l_j^{(j,\alpha)} > l_{j+1}^{(j,\alpha)} > \cdots > l_L^{(j,\alpha)} > 0.
\label{eq:Young_tab}
\eeq
The integers $l_i^{(j,\alpha)}$ are related to 
the magnetic flux at the fixed point
\beq
\frac{1}{2\pi} \int F_i ~=~ \mbox{block-diag}(\boldsymbol l_i^1,\cdots,\boldsymbol l_i^i)
\hs{5} \mbox{with} \hs{5}
\boldsymbol l_i^j = 
{\rm diag}\left( l_i^{(j,1)},\cdots, l_i^{(j,n_j)} \right),
\eeq
%\beq
%\frac{1}{2\pi} \int F_i ~=~ 
%{\renewcommand{\arraystretch}{0.8}
%{\setlength{\arraycolsep}{1.2mm}
%\ba{ccc} \boldsymbol l_i^1 & & \\ & \ddots & \\ & & %\boldsymbol l_i^i \ea}} 
%\hs{5} \mbox{with} \hs{5}
%\boldsymbol l_i^j = 
%{\renewcommand{\arraystretch}{0.6}
%{\setlength{\arraycolsep}{0mm}
%\ba{ccc} l_i^{(j,1)} & & \\ & \ddots & \\ & &~ l_i^{(j,n_j)} %\ea}},
%\eeq
where $\boldsymbol l_i^j$ is the $n_j$-by-$n_j$ diagonal block of 
the $SU(N_i)$ magnetic flux of the $i$-th gauge group. 
They are also related to the winding numbers of the scalar fields 
\beq
q_i = 
{\renewcommand{\arraystretch}{0.5}
{\setlength{\arraycolsep}{1.0mm}
\ba{ccc|ccc} \boldsymbol q_i^1 & & & 0 & \cdots & 0 \\ & \ddots & & \vdots & \ddots & \vdots \\ & & \boldsymbol q_i^i \phantom{\Big|} & \phantom{\Big|} 0 & \cdots & 0 \ea}}, 
\hs{5} \mbox{with} \hs{5}
\boldsymbol q_i^j = \boldsymbol f_i^j(r) \, \exp \left( i \boldsymbol \nu_i^j \theta \right), \hs{5} \boldsymbol \nu_i^j \equiv \boldsymbol l_i^j - \boldsymbol l_{i+1}^j, 
\eeq
where $\boldsymbol f_i^j(r)$ and $\boldsymbol \nu_i^j$ 
are diagonal matrices of profile functions and winding numbers, respectively. 
We can confirm that $q_i(z)$ is invariant under 
the torus action (the combination of the spatial rotation and the Cartan part of the flavor rotation) up to $V$-transformations
\begin{eqnarray}
q_i(z) = V_i \, q_i(e^{i\epsilon }z) \, V_{i+1}^{-1}, 
\hs{10} 
V_i = \exp (i \Sigma_i), \hs{5} 
V_{L+1}(z) = \exp (-i M). 
\end{eqnarray}
Note that the left hand side of the fixed point condition \eqref{eq:fixed_infinitesimal} is 
the infinitesimal version of this transformation. 
The element of the $V$-transformations are 
specified by the fixed point values of the adjoint scalar 
$\Sigma_i$, which take the forms 
\beq
\Sigma_i = 
\mbox{block-diag} \,
(\boldsymbol \sigma_i^1, \cdots, \boldsymbol \sigma_i^i)
\hs{10}
\boldsymbol \sigma_i^j = {\rm diag} \,
(\sigma_i^{(j,1)}, \cdots, \sigma_i^{(j,n_j)}),
\eeq
with the eigenvalues
\beq
\sigma_i^{(j,\alpha)} = - m^{(j,\alpha)} - l_i^{(j,\alpha)} \epsilon,
\eeq
where we have labeled the eigenvalues of the mass matrix as
\beq
M = 
\mbox{block-diag} \, 
(\boldsymbol m^1, \cdots, \boldsymbol m^{L+1}) , 
\hs{10} 
\boldsymbol m^j = 
{\rm diag} \,
( m^{(j,1)}, \cdots m^{(j,n_i)}).
\end{eqnarray}

\subsection{Half-ADHM data at fixed points}
We can show that the vortex data 
corresponding to the fixed point specified by the Young tableaux $Y^{(j,\alpha)}$ take the form
\beq
\mD_i = 
\mbox{block-diag} \, 
( \boldsymbol \mD_i^1 , \cdots , \boldsymbol \mD_i^i) 
\hs{5} \mbox{with} \hs{5} 
\boldsymbol \mD_i^j = 
{\rm diag} \, (z^{l_i^{(j,1)}} , \cdots , z^{l_i^{(j,n_j)}}) 
\hs{5} \mbox{and} \hs{5} \wt \mD_i = 0.
\eeq
This implies that each diagonal component represents axially symmetric Abelian vortices with flux $l_i^{(j,\alpha)}$
and hence all the matrix data can be obtained 
by embedding those of Abelian vortices. 
For an axially symmetric Abelian vortex configuration $\mD = z^l$, 
the vortex data satisfying $\mD \Psi = \mJ (z \mathbf 1_{l} - Z)$ 
are given by (see Sec. \ref{subsec:patch_L=1})
\beq
\mJ(l) = ( z^{l-1} \,,\, z^{l-2} \,,\, \cdots \,,\, 1 ), \hs{10}
\Psi(l) = ( 1 \,,\, 0 \,,\, \cdots \,,\, 0 ), \hs{10}
Z(l) = \left. 
{\renewcommand{\arraystretch}{0.7}
{\setlength{\arraycolsep}{0.8mm}
\ba{c|ccc} 0 \phantom{|} & 1 & & \\ \vdots \phantom{|} & & \ddots & \\ 0 \phantom{|} & & & 1 \\ \hline 0 \phantom{|} & 0 & \cdots & 0 
\ea}} \ \right\} l.
\eeq
By embedding these matrices, 
we can construct the matrices satisfying 
$\mD_i \Psi_i = \mJ_i (z \mathbf 1_{k_i} - Z_i)$
as 
\beq
\mJ_i = 
\mbox{block-diag} \, ( \boldsymbol \mJ_i^1 , \cdots , \boldsymbol \mJ_i^{i}), \hs{5}
\Psi_i = 
\mbox{block-diag} \, (\boldsymbol \Psi_i^1 , \cdots , \boldsymbol \Psi_i^{i}), \hs{5}
Z_i = 
\mbox{block-diag} \, (\boldsymbol Z_i^1, \cdots , \boldsymbol Z_i^{i}) ,
\eeq
with
\beq
\boldsymbol \mJ_i^j \! = \!
{\rm diag}
(\mJ(l_i^{(j,1)}), \cdots, \mJ(l_i^{(j,n_j)})), \,
\boldsymbol \Psi_i^j \! = \! {\rm diag} (\Psi(l_i^{(j,1)}), \cdots, \Psi(l_i^{(j,n_j)})), \,
\boldsymbol Z_i^j \! = \!
{\rm diag} (Z(l_i^{(j,1)}), \cdots, Z(l_i^{(j,n_j)})). 
\eeq
Note that $\wt \Psi_i =0$ since $\wt \mD_i = 0$
for the fixed point configurations. 
The matrices $\mY_i$ and $\wt \mY_i$ defined in \eqref{eq:def_upsilon} can be extracted from 
$\Psi_i$ and $\wt \Psi_i$ as
\beq
\mY_i = \ba{c|c} \mathbf 0_{n_i,k_{i-1}} & \boldsymbol \Psi_i^i  \ea, \hs{10}
\wt \mY_i = 0. 
\eeq
The matrix $W_i$ can be determined 
by solving the constraint 
$Z_i W_i - W_i Z_{i+1} = \wt \mY_i \mY_{i+1}$ as 
\beq
W_i = 
{\renewcommand{\arraystretch}{0.8}
{\setlength{\arraycolsep}{0.7mm}
\ba{ccc|ccc} \boldsymbol W_i^1 & & & ~ \mathbf 0 & ~ \cdots & ~ \mathbf 0 \\ & \ddots & & ~ \vdots & ~ \ddots  & ~ \vdots \\ & & \boldsymbol W_i^{i} & ~ \mathbf 0 & ~ \cdots & ~ \mathbf 0 \ea
}}, \hs{10}
\boldsymbol W_i^j = 
{\renewcommand{\arraystretch}{0.6}
{\setlength{\arraycolsep}{0.1mm}
\ba{ccc} W(l_i^{(j,1)}, l_{i+1}^{(j,1)}) & & \\ & \ddots & \\ & & W(l_i^{(j,n_j)},l_{i+1}^{(j,n_j)})
\ea}}
\eeq
where $W(l,l')$ is the matrix satisfying 
$Z(l) W(l,l') - W(l,l') Z(l') = 0$, which takes the form
\beq
W(l,l') = \ba{c} \mathbf 1_{l'} \\ \mathbf 0_{l-l',l'} \ea.
\eeq
Note that $\wt W_i = 0$ as shown above. 

\subsection{Torus action on half-ADHM data}
\label{Appendix:Torus action on Half-ADHM data}
The above set of matrices $\{Z_i,\mY_i,\tilde \mY_i,W_i,\wt W_i\}$
corresponds to the BPS configuration in the presence of the deformations. 
This satisfies the fixed point condition of the 
torus action 
\begin{eqnarray}
\{Z_i,\mY_i,\wt \mY_i,\mW_i,\wt \mW_i\} \to \{ e^{-i\Phi_i-i\epsilon} Z_i e^{i\Phi_i}, e^{-iY_i}\mY_i e^{i\Phi_i}, e^{-i\Phi_i-i\epsilon} \wt \mY_i e^{i Y_{i+1}},
e^{-i\Phi_{i}} \mW_ie^{i\Phi_{i+1}}, e^{-i\Phi_{i+1}} \wt \mW_i e^{i\Phi_i+i\epsilon} \}, 
\label{eq:hADHM_torus}
\end{eqnarray}
where $\Phi_i$ are the elements of ${\mathfrak gl}(k_i)$ given by
\beq
\Phi_i = 
\mbox{block-diag}
(\boldsymbol \Phi_i^1, \cdots,\boldsymbol \Phi_i^{i}), \hs{3}
\boldsymbol \Phi_i^j = 
\mbox{block-diag} (\Phi_i^{(j,1)},\cdots, \Phi_i^{(j,n_j)}), \hs{3}
\Phi_i^{(j,\alpha)} = 
{\rm diag}
( \phi_i^{(j,\alpha,1)}, \cdots, \phi_i^{(j,\alpha,p)}),
\eeq
with the eigenvalues\footnote{
The matrices $\Phi_i^{(j,\alpha)}$ can be determined by solving  the equations
\beq
[\Phi_i^{(j,\alpha)} , Z_i^{(j,\alpha)}] + \epsilon Z_i^{(j,\alpha)} = 0, 
\hs{5}
m^{(i,\alpha)} \Psi_i^{(i,\alpha)} - \Psi_i^{(i,\alpha)} \Phi_i^{(i,\alpha)} = 0, \hs{5} \Phi_i^{(j,\alpha)} - \Phi_{i+1}^{(j,\alpha)} = 0 ~~~ (\mbox{for $j=1,\cdots,i$}). \notag
\eeq}
\beq
\phi_i^{(j,\alpha,p)} =  m^{(j,\alpha)} + (p-1) \epsilon.
\eeq
These eigenvalues at the fixed point 
correspond to the poles of the integrand 
for the vortex partition function \eqref{eq:topopartition}, 
whose residue give the contribution of the fixed point configuration.  
The fixed point condition can also be rewritten 
by using the infinitesimal form of the torus action 
as 
\beq
&[\Phi_i, Z_i] - \epsilon Z_i = 0, \hs{5}
M_i \mY_i - \mY_i \Phi_i = 0, \hs{5}
\Phi_i \wt \mY_i - \wt \mY_i M_{i+1} + \epsilon \wt \mY_i = 0,& \\
&\Phi_i W_i - W_i \Phi_{i+1} = 0, \hs{5} 
\Phi_{i+1} \wt W_i - \wt W_i \Phi_i + \epsilon \wt W_i = 0.&
\eeq

One can explicitly check that 
the torus action on the half-ADHM data
is consistent with that on $(\mD_i(z),\wt \mD_i(z))$ as follows.
With $\hat M_j$ and $\wt M_j$ defined by 
\begin{eqnarray}
\hat M_j={\rm diag} (\boldsymbol m_1,\boldsymbol m_2,\cdots,\boldsymbol m_j), 
\hs{10} \wt M_j={\rm diag} (\boldsymbol m_{j+1},\cdots,\boldsymbol m_L, \boldsymbol m_{L+1}),
\end{eqnarray}
the torus action on $(\mD_j(z),\wt \mD_j(z))$ 
can be read off from that on $q_i$ as
\begin{eqnarray}
(\mD_j(z),\wt \mD_j(z)) \quad \to \quad (\mD_j'(z),\wt \mD_j'(z))
= V_j(z)\left(\mD_j(e^{i\epsilon}z)e^{i\hat M_j},
\wt \mD_j(e^{i\epsilon}z) e^{i\wt M_j}\right).
\end{eqnarray}
Since $\mJ'_j(z)$ must satsify
\begin{eqnarray}
{\cal O}(z^{-1}) = \mD'_j(z)^{-1} \mJ'_j(z) = e^{-i\hat M_j} \mD_j(e^{i\epsilon}z)^{-1}(V_j(z)^{-1}\mJ'_j(z)),
\end{eqnarray}
we find that $\mJ'_j(z)$ is given by
\begin{eqnarray}
\mJ'_j(z) = V_j(z) \mJ_j(e^{i\epsilon}z) e^{i\Phi'_j}, \hs{10}
\left( \because \mD_j(z)^{-1} \mJ_j(z) = \mathcal O({z}^{-1}) \right),
\end{eqnarray}
where $\Phi'_j \in {\mathfrak gl}(k_j,\C)$ is a certain constant square matrix.
Since the torus action on $(Z_j, \Psi_j, \wt \Psi_j) \rightarrow (Z_j', \Psi_j', \wt \Psi_j') $ must be consistent with the half-ADHM mapping relation
\begin{eqnarray}
%&&\mD_j(e^{i\epsilon}z)\Psi_j e^{-i\epsilon}=\mJ_j(e^{i\epsilon}z)(z{\bf 1}-Z_je^{-i\epsilon}),\quad \wt \mD_j(e^{i\epsilon}z)=\mJ_j(e^{i\epsilon}z)\wt \Psi_j, \\
&& \mD'_j(z)\Psi'_j=\mJ'_j(z)(z{\bf 1}-Z'_j),\quad \wt \mD'_j(z) = \mJ'_j(z)\wt \Psi'_j,
\end{eqnarray}
it follows that 
\begin{eqnarray}
\mD_j(e^{i\epsilon} z)e^{i\hat M_j}\Psi'_je^{-i\Phi'_j}=\mJ_j(e^{i\epsilon}z)  (z{\bf 1}-e^{i\Phi'_j}Z'_je^{-i\Phi'_j}),
\hs{5} 
\wt \mD_j(e^{i\epsilon} z)
=\mJ_j(e^{i\epsilon}z) e^{i\Phi'_j}\wt \Psi'_j e^{-i\wt M_j}.
\end{eqnarray}
Comparing with the original half-ADHM mapping relation
$\mD_j(z) \Psi_j =\mJ_j (z)  (z {\bf 1} - Z_j)$ and 
$\wt \mD_j(z) =\mJ_j(z) \wt \Psi_j$, 
we obtain the torus action on $\{Z_j,\,\Psi_j,\wt \Psi_j\}$ as
\begin{eqnarray}
(Z_j,\,\Psi_j,\wt \Psi_j) \quad \rightarrow \quad (Z'_j,\,\Psi'_j,\wt \Psi'_j)
= \{ e^{-i\Phi_j-i\epsilon}Ze^{i\Phi_j},\, e^{-i\hat M_j}\Psi_j e^{i\Phi_j},\, e^{-i\Phi_j-i\epsilon}\wt \Psi_j e^{i \wt M_j} \},
\end{eqnarray}
where we have defined $\Phi_j=\Phi'_j-\epsilon \bf 1$.
The torus action on $(\mY_i,\,\wt \mY_i)$ can be read off 
from that on $(\Psi_j,\wt \Psi_j)$
\begin{eqnarray}
( \mY_i,\,\wt \mY_i ) \quad \to \quad 
( Z_i',\, \mY_i',\,\wt \mY_i' )
= ( e^{-i\Phi_i-i\epsilon} Z_ie^{i\Phi_i},\, e^{-iM_i }\mY_i e^{i\Phi_i},\, e^{-i\Phi_i-i\epsilon} \wt \mY_i e^{i  M_{i+1}} ).
\end{eqnarray}
The torus action on $(W_i, \wt W_i)$ can be obtained from 
$q_i'(z) \mJ_{i+1}'(z)=\mJ_i'(z) \mW_i'$, 
which can be rewritten as
\begin{eqnarray}
q_i(e^{i\epsilon}z) \mJ_{i+1}(e^{i\epsilon}z)e^{i\Phi_{i+1}+i\epsilon }=\mJ_i(e^{i\epsilon}z) e^{i\Phi_{i}+i\epsilon } \mW_i', 
\end{eqnarray}
Comparing with $q_i(z) \mJ_{i+1}(z)=\mJ_i(z) \mW_i$, 
we find that
\begin{eqnarray}
( \mW_i,\, \wt \mW_i ) \quad \to \quad (\mW_i',\, \wt \mW_i' )  
= ( e^{-i\Phi_{i}} \mW_ie^{i\Phi_{i+1}} ,\, e^{-i\Phi_{i+1}} \wt \mW_i e^{i\Phi_i+i\epsilon} ),
\end{eqnarray}
where we have determined the torus action on $\wt W_i$
so that $\mathcal W$ in Eq.\,\eqref{eq:super_potential} is invariant.

\subsection{Fluctuation around the fixed points}
Next, let us consider the fluctuation around the fixed point configuration discussed in the previous subsection. 
Let us label the fluctuations of $q_i$ around the fixed point as
\beq
\delta q_i = 
\ba{ccc|c} 
\delta \boldsymbol q_i^{11} & \cdots & \delta \boldsymbol q_i^{1j} & \delta \boldsymbol q_i^{1,i+1} \\
\vdots  & \ddots & \vdots & \vdots \\
\delta \boldsymbol q_i^{i1} & \cdots & \delta \boldsymbol q_i^{ii} & \delta \boldsymbol q_i^{i,i+1} 
\ea
\hs{5} \mbox{with} \hs{5}
 \delta \boldsymbol q_i^{jk} = 
\ba{ccc} 
\delta q_i^{(j,1),(k,1)} & \cdots & \delta q_i^{(j,1),(k,n_k)} \\
\vdots & \ddots  & \vdots \\
\delta q_i^{(j,n_j),(k,1)} & \cdots & \delta q_i^{(j,n_j),(k,n_k)} 
\ea,
\eeq
where $\delta q_i^{(j,\alpha),(k,\beta)}$ are polynomials of $z$. 
Similarly, we label the fluctuations of $\xi_i$ as
\beq
\delta \xi_i = \ba{ccc|ccc} 
\delta \boldsymbol \xi_i^{11} & \hs{-2} \cdots & \hs{-2} \delta \boldsymbol \xi_i^{1j} & \delta \boldsymbol \xi_i^{1,i+1} & \hs{-2} \cdots & \hs{-2} \delta \boldsymbol \xi_i^{1,L+1} \\
\vdots  & \hs{-2} \ddots & \hs{-2} \vdots & \vdots & \hs{-2} \ddots & \hs{-2} \vdots \\
\delta \boldsymbol \xi_i^{i1} & \hs{-2} \cdots & \hs{-2} \delta \boldsymbol \xi_i^{ii} & \delta \boldsymbol \xi_i^{i,i+1} & \hs{-2} \cdots & \hs{-2} \delta \boldsymbol \xi_i^{i,L+1}
\ea
\hs{3} \mbox{with} \hs{3}
 \delta \boldsymbol \xi_i^{jk} = 
\ba{ccc} 
\delta \xi_i^{(j,1),(k,1)} & \hs{-2} \cdots & \hs{-2} \delta \xi_i^{(j,1),(k,n_k)} \\
\vdots & \ddots  & \vdots \\
\delta \xi_i^{(j,n_j),(k,1)} & \hs{-2} \cdots & \hs{-2} \delta \xi_i^{(j,n_j),(k,n_k)}
\ea,
\eeq
where $\delta \xi_i^{(j,\alpha),(k,\beta)}$ are polynomials of $z$, which we denote
\beq
\delta \xi_i^{(j,\alpha),(k,\beta)} = \left\{ \begin{array}{ll} \displaystyle \delta \mD_i^{(j,\alpha),(k,\beta)} & \mbox{for $k \leq i$} \\
\delta \wt \mD_i^{(j,\alpha),(k,\beta)} & \mbox{for $k \geq i+1$} \end{array} \right..
\eeq
For a fixed point spacified by the Young tableaux $Y^{(j,\alpha)}=(l_j^{(j,\alpha)},l_{j+1}^{(j,\alpha)},\cdots,l_L^{(j,\alpha)})$, 
$\delta \mD_i^{(j,\alpha),(k,\beta)}$ and $\wt \delta \mD_i^{(j,\alpha),(k,\beta)}$ are polynomial of degree $l_i^{(k,\beta)}-1$ and $l_i^{(j,\alpha)}-1$, respectively. 

Since $\xi_i = q_i \xi_{i+1}$, the fluctuations of $\xi_i$ 
must satisfy the recursive relations
\beq
\delta \xi_i^{(j,\alpha),(k,\beta)} = z^{l_i^{(j,\alpha)} - l_{i+1}^{(j,\alpha)}} \delta \xi_{i+1}^{(j,\alpha),(k,\beta)} + \delta q_i^{(j,\alpha),(k,\beta)} z^{l_{i+1}^{(k,\beta)}}.
\label{eq:recursive relations}
\eeq
This condition gives a constraint to the fluctuations $\delta q_i$. 
To find such constraints, let us write
\beq
\delta \xi_i^{(j,\alpha),(k,\beta)} &=& \sum_{n} c_i^{(j,\alpha),(k,\beta),n} z^n, \\
\delta q_i^{(j,\alpha),(k,\beta)} &=& \sum_{n} a_i^{(j,\alpha),(k,\beta),n} z^n. 
\eeq
Then the recursive relation \eqref{eq:recursive relations}
can be written as
\beq
c_i^{(j,\alpha),(k,\beta),n} = c_{i+1}^{(j,\alpha),(k,\beta), n - l_i^{(j,\alpha)} + l_{i+1}^{(j,\alpha)}} + a_i^{(j,\alpha),(k,\beta), n - l_{i+1}^{(k,\beta)}}. 
\eeq
Solving these equation, we find that
\begin{alignat}{3}
c_i^{(j,\alpha),(k,\beta),n} &= 0, \hs{10} &
&(\mbox{for $k \leq i$ and $n \geq l_i^{(k,\beta)}$})  \\
a_i^{(j,\alpha),(k,\beta), n} &= 0, \hs{10} &&(\mbox{for $k \geq i+2,$ and $n \geq l_i^{(j,\alpha)} - l_{i+1}^{(k,\beta)}$}) \\
a_{l}^{(j,\alpha),(k,\beta),n} &= - c_{l+1}^{(j,\alpha),(k,\beta), n +  l_{l+1}^{(k,\beta)} - l_{l}^{(j,\alpha)} + l_{l+1}^{(j,\alpha)}}, \hs{3} 
& &(\mbox{for $l_l^{(j,\alpha)} - l_{l+1}^{(k,\beta)} \leq n < l_l^{(k,\beta)} - l_{l-1}^{(j,\alpha)} + l_l^{(j,\alpha)} - l_{l+1}^{(k,\beta)}$}).
\end{alignat}
The coefficients $a_i^{(j,\alpha),(k,\beta),p}$ satisfying these conditions can be regarded as the coordinates around the fixed point. 
We can check that the number of the degrees of freedom agrees with the dimension of the moduli space. They transform under the torus action as
\beq
a_i^{(j,\alpha),(k,\beta), p} &\rightarrow&a_i^{(j,\alpha),(k,\beta),p} \, \exp i \left[ \sigma_i^{(j,\alpha)} - \sigma_{i+1}^{(k,\beta)} + p \epsilon \right] \notag \\
&=&  a_i^{(j,\alpha),(k,\beta),p} \, \exp i \left[ m^{(j,\alpha)} - m^{(k,\beta)} - (l_i^{(j,\alpha)} - l_i^{(k,\beta)} + p ) \epsilon \right] , 
\label{eq:torus_action_fluc}
\eeq
It is worth noting that the vortex partition function \eqref{eq:2d_VP} can be obtained from these transformation properties.
Having solved the constraints for the fluctuations of $(\mD,\wt \mD)$, 
we can determine those of the half ADHM data satisfying the constraints through the linearized version of the half-ADHM mapping relation.

\section{Vortex partition function}\label{appendix:VP}
%According to the torus action 
%\begin{eqnarray}
%\mY_j \mW_j\mW_{j+1}\cdots \mW_{i-1} Z_i^{p-1} \quad \to \quad e^{-i Y_j}(\mY_j \mW_j\mW_{j+1}\cdots \mW_{i-1} Z_i^{p-1})e^{i \Phi_i-i(p-1)\epsilon}
%\end{eqnarray}
%around the fixed point given by Eq.(\ref{eq:Kahlerfp}), 
%%\begin{equation}
%%(\mY_j \mW_j\mW_{j+1}\cdots \mW_{i-1} Z_i^{p-1})^{(j;a)_j}{}_{(j';b,p')_i}=\delta^j_{j'} \delta^a_b\delta^p_{p'}
%%\end{equation}
%we easily read $\Phi_i$ as, with $\alpha=(j;a,p)_i$
%\begin{eqnarray}
%(\Phi_i)^\alpha{}_\beta=\delta^\alpha_\beta \sigma_{i,\alpha},\qquad  \sigma_{i,\alpha}=\sigma_{i,(j;a,p)_i}=y_{j,a}+(p-1)\epsilon. 
%\label{eq:fpsigma}
%\end{eqnarray}
%\begin{eqnarray}
%(\delta \mY_i)^{(a)_i}{}_{(j;b,p)_i}\quad {\rm for~} a \in \{ a| l_{ii}^a=0\}:&& y_{j,b}+(p-1)\epsilon-y_{i,a} \nn
% (\delta \wt  \mY_i)^{(j;a,p)_i}{}_{(b)_{i+1}}:&& y_{i+1,b}-y_{j,a} -p \epsilon \nn
% (\delta Z_i)^{(j;a,p)_i}{}_{(j';b,p')_i}\quad {\rm with~} p=l_{ij}^a:&& y_{j',b}-y_{j,a}+(p'-i_{ij}^a-1)\epsilon \nn
%  (\delta W_i)^{(j;a,p)_i}{}_{(j';b,p')_{i+1}}\quad {\rm for~}   2\le  p \le  l_{ij}^a,:&& y_{j',b}-y_{j,a}+(p'-p)\epsilon 
%\end{eqnarray}
%\begin{eqnarray}
%(\wt W_i^\dagger )^{(j;a,p)_{i}}{}_{(j';b,p')_{i+1}}: y_{j',b}-y_{j,a}+(p'-p-1)\epsilon
%\end{eqnarray}
In this appendix, we derive the integration formula for the vortex partition function \eqref{eq:topopartition}. 
The vortex partition function is given by the determinant of 
the torus action on the moduli space \eqref{eq:PF_wegiht},
which can also be obtained from the torus action on the fluctuation \eqref{eq:torus_action_fluc} or that on the half-ADHM matrices. 

Let us first consider the character of the torus action on the fluctuations around the fixed point specified by each Young tableaux \eqref{eq:Young_tab}. 
The contributions of $(\Upsilon_i,\wt \Upsilon_i,Z_i,W_i)$ to the character can be read off from the torus action \eqref{eq:hADHM_torus} as 
\begin{eqnarray}
\chi(\delta \mY_i) &=& \sum_{\alpha \in \lambda_i} \sum_{j=1}^i\sum_{\beta=1}^{n_j}\sum_{p=1}^{l_i^{(j,\beta)}} \exp \left[ i m^{(j,\beta)} + i (p-1) \epsilon - i m^{(i,\alpha)} \right], \nn
\chi(\delta \wt \mY_i) &=& \sum_{j=1}^i\sum_{\alpha=1}^{n_i} \sum_{p=1}^{l_{i}^{(j,\alpha)}} \sum_{\beta=1}^{n_{i+1}} \exp \left[ im^{(i+1,\beta)} - i m^{(j,\alpha)} - i p \epsilon  \right], \nn
\chi(\delta Z_i) &=& \sum_{j=1}^i \sum_{\alpha \in \bar \lambda_i} \sum_{k=1}^i \sum_{b=1}^{n_{k}} \sum_{q=1}^{l_{i}{(k,\beta)}} 
\exp \left[ im^{(k,\beta)}-im^{(j,\alpha)}+i(q-l_{i}^{(j,\alpha)}-1) \epsilon \right], \nn
\chi(\delta \mW_i)&=& \sum_{j=1}^i \sum_{\alpha=1}^{n_j} \sum_{p=2}^{l_i^{(j,\alpha)}} \sum_{k=1}^{i+1} \sum_{\beta=1}^{n_{k}} \sum_{q=1}^{l_{i+1}^{(j,\beta)}}
\exp \left[ im^{(k,\beta)}-im^{(j,\alpha)}+i(q-p) \epsilon \right].
\end{eqnarray}
where $\lambda_i = \{\alpha \, | \, 1\le \alpha \le n_i, l_i^{(i,\alpha)}=0 \}$ and $\bar \lambda_i = \{ \alpha \, | \, 1\le \alpha \le n_i, l_i^{(i,\alpha)} \not= 0 \}$. 
Note that contributions eliminated by the constraints (\ref{eq:constraint}) must be removed from these characters. 
We can see from the ``superpotential term" in Eq.\,\eqref{eq:super_potential} that the contributions eliminated by the constraints can be identified with that of $\wt W_i^\dagger$
\begin{eqnarray}
\chi(\wt \mW_i^\dagger )= \sum_{j=1}^i \sum_{\alpha=1}^{n_j} \sum_{p=1}^{l_i^{(j,\alpha)}} \sum_{k=1}^{i+1} \sum_{\beta=1}^{n_{k}} \sum_{q=1}^{l_{i+1}^{(k,\beta)}}
\exp \left[ im^{(k,\beta)}-im^{(j,\alpha)}+i (q-p-1)\epsilon \right].
\end{eqnarray}
In total, the character is given by
\begin{eqnarray}
\chi_{\rm fp} &=&
\left[\sum_{i=1}^L \left(\chi(\delta \mY_i)+\chi(\delta \wt \mY_i)+\chi(\delta Z_i)\right)+\sum_{i=1}^{L-1} 
\left(\chi(\delta W_i)\right)\right]_{\rm constrained}\nn
&=&\sum_{i=1}^L \left(\chi(\delta \mY_i)+\chi(\delta \wt \mY_i)+\chi(\delta Z_i)\right)+\sum_{i=1}^{L-1} 
\left(\chi(\delta W_i)-\chi(\wt W_i^\dagger)\right)\nn
&=&\sum_{i=1}^L \left(\chi(\mY_i)+\chi( \wt \mY_i)+\chi( Z_i)-\chi(\mU_i)\right)+\sum_{i=1}^{L-1} 
\left(\chi(W_i)-\chi(\wt W_i^\dagger)\right)\nn
&=& \sum_{i=1}^L \left(
\tr[e^{i\Phi_i}] \tr[e^{-iM_i}]+e^{-i\epsilon}\tr[e^{-i\Phi_i}]\tr[e^{iM_{i+1}}]+(e^{-i\epsilon}-1)\tr[e^{i\Phi_i}]\tr[e^{-i\Phi_{i}}]\right)\nn
&&\qquad {}+\sum_{i=1}^{L-1}(1-e^{-i\epsilon}) \tr[e^{i\Phi_{i+1}}]\tr[e^{-i\Phi_i}],
\end{eqnarray}
where we have rewritten the characters of the fluctuations ($\chi(\delta \Upsilon_i), \cdots$) into those of the matrices ($\chi(\Upsilon_i), \cdots$) by subtracting the contributions eliminated by the gauge $GL(k_i,\mathbb C)$ action $\chi(\mU_i)=\tr[e^{i \Phi_i}]\tr[e^{-i\Phi_i}]$.
From the character, we can read off the determinant as
\beq
\chi_\sigma = \sum_{a=1}^d e^{i \omega_a}  
~ \Longrightarrow ~ \frac{1}{\det \mathcal M_\sigma} = \prod_{a=1}^d \frac{1}{\omega_a},
\eeq
where $d$ is the dimension of the moduli space. 
Furthermore, by using the relation  
\begin{eqnarray}
f(\phi_\sigma) = \oint \frac{d \phi}{2\pi i} \frac{1}{\phi-\phi_\sigma} f(\phi),
\end{eqnarray}
the determinant can be rewritten into a contour integral as 
\beq
\frac{1}{\det \mathcal M_\sigma} = \oint_{C_\sigma} 
\prod_{i=1}^L \prod_{r=1}^{k_i} \frac{d \phi_i^r}{2\pi i \epsilon} 
 \left[\prod_{i=1}^L {\cal Z}_i^{\mY \wt \mY} {\cal Z}_i^{Z \Phi} \prod_{i=1}^{L-1} {\cal Z}_i^{W \wt W} \right].
\end{eqnarray}
where $ {\cal Z}_i^{\mY\wt \mY}$, ${\cal Z}_i^{Z \Phi}$ and $ {\cal Z}_i^{W \wt W}$ are given by
\begin{eqnarray}
{\cal Z}_i^{\mY \wt \mY} &\equiv&
\prod_{r=1}^{k_i} \left[\prod_{\alpha=1}^{n_i} \frac{1}{\phi_i^r-m^{(i,\alpha)}}\prod_{\beta=1}^{n_{i+1}}\frac{1}{m^{(i+1,\beta)}-\phi_i^r-\epsilon} \right], \\
{\cal Z}_i^{Z \Phi} &\equiv& \prod_{r=1}^{k_i} \prod_{s=1}^{k_i} \hs{-1} {\phantom{\bigg|}}' \frac{\phi_i^r-\phi_i^s}{\phi_i^r-\phi_i^s-\epsilon}, \\
{\cal Z}_i^{W \wt W} &\equiv& 
\prod_{r=1}^{k_i} \prod_{s=1}^{k_{i+1}}\frac{\phi_{i+1}^s-\phi_i^r-\epsilon}{\phi_{i+1}^s-\phi_i^r}.
\end{eqnarray}
where $\prod'$ indicates that the factors with $\alpha=\beta$ are omitted. The integration contour $C_\sigma$ is the path surrounding the poles corresponding to the fixed point values of $\phi$. 
Since the integrand is common for all the fixed points, 
the vortex partition function can be obtained 
by integrating the same integrand along the contour surrounding all the poles corresponding to the fixed points. 
We can check that such contour is given by 
$C_{i}^{\pm}$ (Fig.\,\ref{fig:contours_duality}) as follows. 

\begin{figure}[h!]
\centering
\fbox{
\includegraphics[width=82mm, bb = 40 0 800 585]{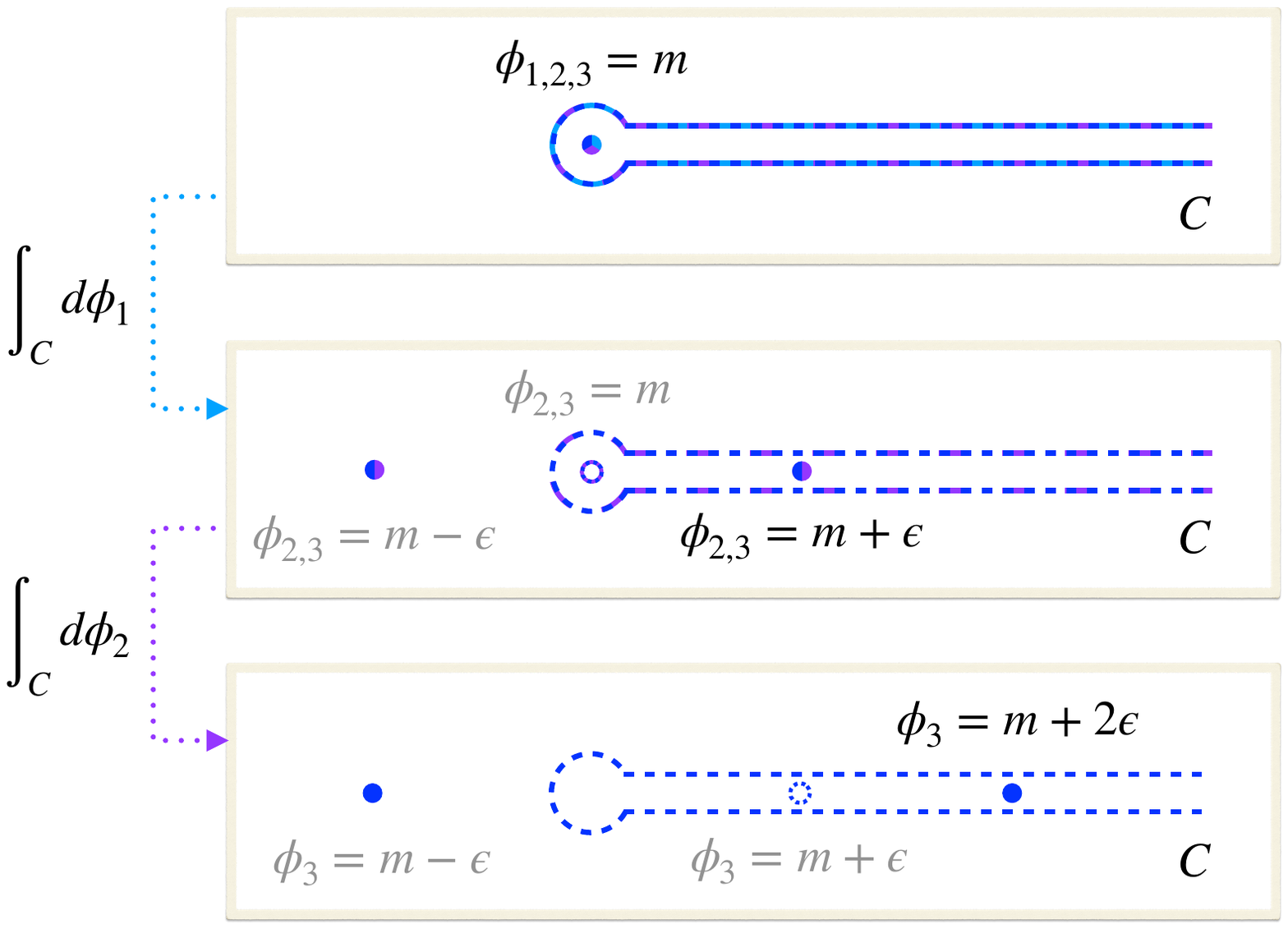}}
\caption{Integration contours for the vortex partition function in the Abelian gauge theory with a single charged scalar field ($L=1, N_1=1, N_2=0, k=3$).}
\label{fig:contours_duality1}
\end{figure}
\begin{figure}[h!]
\centering
\fbox{
\includegraphics[width=100mm, bb = 30 0 820 370]{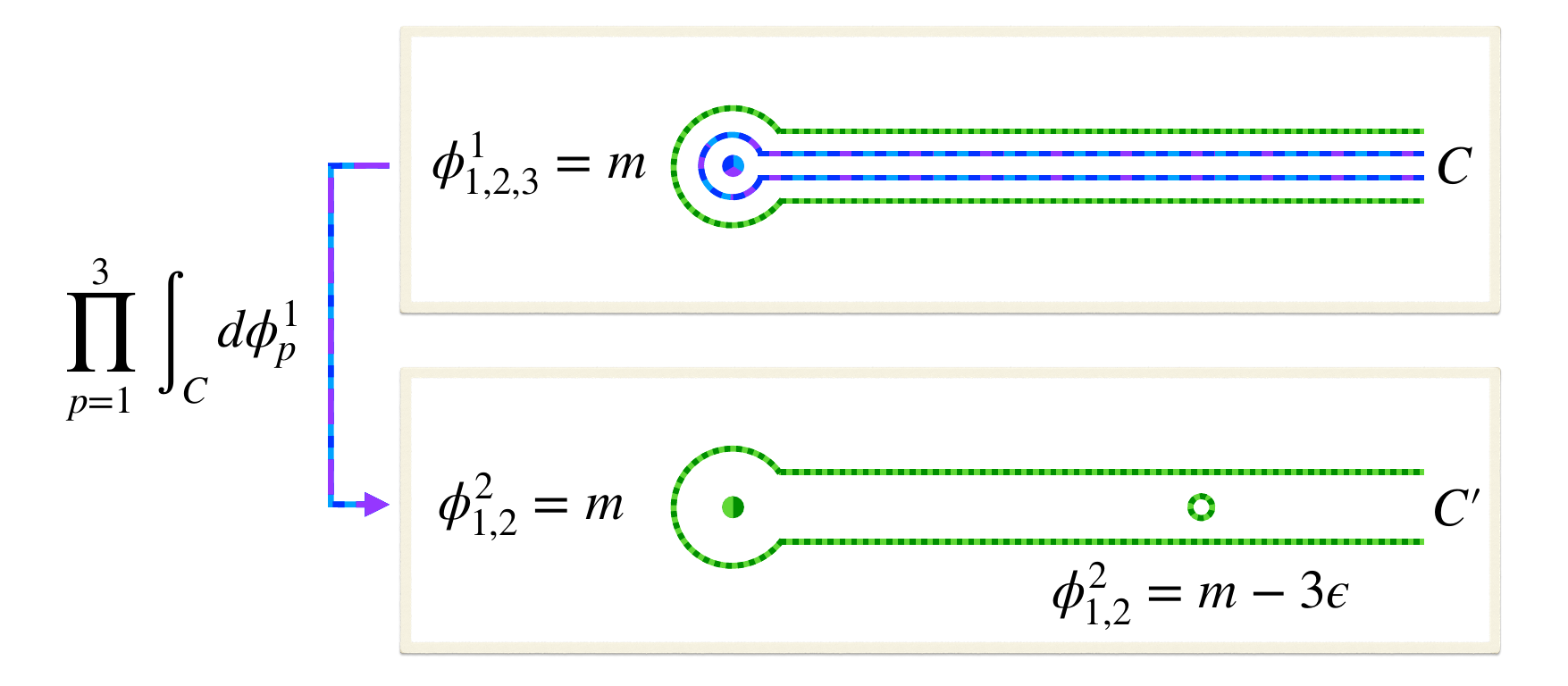}}
\caption{Integration contours in the case of $L=2, N_1=1, N_2=1, k_1=3, k_2=2$).}
\label{fig:contours_duality2}
\end{figure}
 
First, let us consider the case of $L=1, N_1=1, N_2=0, k=3$
(Fig.\,\ref{fig:contours_duality1}). 
In this case, the contour $C_1^{+}$ is the path surrounding the pole of $\mathcal Z_i^{\mY \wt \mY}$ located at $\phi = m~(=m^{(1,1)})$. 
If we first integrate $\phi_1$, 
the residue at the pole $\phi = m$ gives the poles at $\phi = m \pm \epsilon$ and the pole at $\phi=m$ is eliminated 
due to the factor $\mathcal Z_1^{Z\Phi}$.
Then, the integration of $\phi_2$ is 
given by the residue at the pole $\phi = m + \epsilon$.
which has a pole at $\phi = m + 2 \epsilon$
whose residue gives the final result of the integration. 
In this way, we can show that $C_1^+$ is the contour
surrounding all the poles corresponding to the fixed points. 

We can generalize the discussion to the case of $L > 1$. 
the contour $C_1^+$ can be decomposed into the paths 
surrounding the poles at $\phi = m^{(1,\alpha)}$. 
Then, we can repeat the same discussion as in the case of $L=1$
to show that the integration of $\phi_1^{(1,\alpha,p)}~(\alpha=1,\cdots,n_1,~p=1,\cdots,l_1^{(1,\alpha)})$ is given by the residues at the poles $\phi_1^{(1,\alpha,p)} = m^{(1,\alpha)} + (p-1) \epsilon$.
The only new ingredient for $L>1$ is the factors $\mathcal Z_i^{W \wt W}$, which have zeros at $\phi = m^{(1,\alpha)} + l_1^{(1,\alpha)} \epsilon$ (see Fig.\,\ref{fig:contours_duality2}). 
Due to these zeros the integrations of $\phi_i^{(1,\alpha,p)}~(i>1)$, 
which are again given by the residues at $\phi_i^{(1,\alpha,p)} = m^{(1,\alpha)} + (p-1) \epsilon$, 
vanish if $l_{i}^{(1,\alpha)} > l_{1}^{(1,\alpha)}$.
Repeating this argument, we can show that 
the contributions are nonzero only when $l_i^{(j,\alpha)} < l_{i'}^{(j,\alpha)}$ for $i > i'$. 
In other words, the nonzero contributions can be 
classified by the same Young tableaux 
corresponding to the fixed points. 
In this way, we can show that $C_i^+$ are the contours
surrounding all the poles corresponding to the fixed points.

\bibliographystyle{ieeetr}
\bibliography{flag_lump.bib}

\begin{thebibliography}{10}

\bibitem{DAdda:1978vbw}
A.~D'Adda, M.~Luscher, and P.~Di~Vecchia, ``{A 1/n Expandable Series of
  Nonlinear Sigma Models with Instantons},'' {\em Nucl. Phys. B}, vol.~146,
  pp.~63--76, 1978.

\bibitem{Witten:1978bc}
E.~Witten, ``{Instantons, the Quark Model, and the 1/n Expansion},'' {\em Nucl.
  Phys. B}, vol.~149, pp.~285--320, 1979.

\bibitem{Polyakov:1975yp}
A.~M. Polyakov and A.~A. Belavin, ``{Metastable States of Two-Dimensional
  Isotropic Ferromagnets},'' {\em JETP Lett.}, vol.~22, pp.~245--248, 1975.

\bibitem{Eichenherr:1978qa}
H.~Eichenherr, ``{SU(N) Invariant Nonlinear Sigma Models},'' {\em Nucl. Phys.
  B}, vol.~146, pp.~215--223, 1978.
\newblock [Erratum: Nucl.Phys.B 155, 544 (1979)].

\bibitem{Golo:1978de}
V.~L. Golo and A.~M. Perelomov, ``{Solution of the Duality Equations for the
  Two-Dimensional SU(N) Invariant Chiral Model},'' {\em Phys. Lett. B},
  vol.~79, pp.~112--113, 1978.

\bibitem{Cremmer:1978bh}
E.~Cremmer and J.~Scherk, ``{The Supersymmetric Nonlinear Sigma Model in
  Four-Dimensions and Its Coupling to Supergravity},'' {\em Phys. Lett. B},
  vol.~74, pp.~341--343, 1978.

\bibitem{Hanany:2003hp}
A.~Hanany and D.~Tong, ``{Vortices, instantons and branes},'' {\em JHEP},
  vol.~07, p.~037, 2003.

\bibitem{Auzzi:2003fs}
R.~Auzzi, S.~Bolognesi, J.~Evslin, K.~Konishi, and A.~Yung, ``{NonAbelian
  superconductors: Vortices and confinement in N=2 SQCD},'' {\em Nucl. Phys.
  B}, vol.~673, pp.~187--216, 2003.

\bibitem{Eto:2005yh}
M.~Eto, Y.~Isozumi, M.~Nitta, K.~Ohashi, and N.~Sakai, ``{Moduli space of
  non-Abelian vortices},'' {\em Phys. Rev. Lett.}, vol.~96, p.~161601, 2006.

\bibitem{Eto:2006cx}
M.~Eto, K.~Konishi, G.~Marmorini, M.~Nitta, K.~Ohashi, W.~Vinci, and N.~Yokoi,
  ``{Non-Abelian Vortices of Higher Winding Numbers},'' {\em Phys. Rev. D},
  vol.~74, p.~065021, 2006.

\bibitem{Tong:2005un}
D.~Tong, ``{TASI lectures on solitons: Instantons, monopoles, vortices and
  kinks},'' in {\em {Theoretical Advanced Study Institute in Elementary
  Particle Physics}: {Many Dimensions of String Theory}}, 6 2005.

\bibitem{Eto:2006pg}
M.~Eto, Y.~Isozumi, M.~Nitta, K.~Ohashi, and N.~Sakai, ``{Solitons in the Higgs
  phase: The Moduli matrix approach},'' {\em J. Phys. A}, vol.~39,
  pp.~R315--R392, 2006.

\bibitem{Shifman:2007ce}
M.~Shifman and A.~Yung, ``{Supersymmetric Solitons and How They Help Us
  Understand Non-Abelian Gauge Theories},'' {\em Rev. Mod. Phys.}, vol.~79,
  p.~1139, 2007.

\bibitem{Shifman:2009zz}
M.~Shifman and A.~Yung, {\em {Supersymmetric solitons}}.
\newblock Cambridge Monographs on Mathematical Physics, Cambridge University
  Press, 5 2009.

\bibitem{Balachandran:2005ev}
A.~P. Balachandran, S.~Digal, and T.~Matsuura, ``{Semi-superfluid strings in
  high density QCD},'' {\em Phys. Rev. D}, vol.~73, p.~074009, 2006.

\bibitem{Nakano:2007dr}
E.~Nakano, M.~Nitta, and T.~Matsuura, ``{Non-Abelian strings in high density
  QCD: Zero modes and interactions},'' {\em Phys. Rev. D}, vol.~78, p.~045002,
  2008.

\bibitem{Eto:2009kg}
M.~Eto and M.~Nitta, ``{Color Magnetic Flux Tubes in Dense QCD},'' {\em Phys.
  Rev. D}, vol.~80, p.~125007, 2009.

\bibitem{Eto:2009bh}
M.~Eto, E.~Nakano, and M.~Nitta, ``{Effective world-sheet theory of color
  magnetic flux tubes in dense QCD},'' {\em Phys. Rev. D}, vol.~80, p.~125011,
  2009.

\bibitem{Eto:2009tr}
M.~Eto, M.~Nitta, and N.~Yamamoto, ``{Instabilities of Non-Abelian Vortices in
  Dense QCD},'' {\em Phys. Rev. Lett.}, vol.~104, p.~161601, 2010.

\bibitem{Eto:2013hoa}
M.~Eto, Y.~Hirono, M.~Nitta, and S.~Yasui, ``{Vortices and Other Topological
  Solitons in Dense Quark Matter},'' {\em PTEP}, vol.~2014, no.~1, p.~012D01,
  2014.

\bibitem{Dvali:1993sg}
G.~R. Dvali and G.~Senjanovic, ``{Topologically stable electroweak flux
  tubes},'' {\em Phys. Rev. Lett.}, vol.~71, pp.~2376--2379, 1993.

\bibitem{Eto:2018hhg}
M.~Eto, M.~Kurachi, and M.~Nitta, ``{Constraints on two Higgs doublet models
  from domain walls},'' {\em Phys. Lett. B}, vol.~785, pp.~447--453, 2018.

\bibitem{Eto:2018tnk}
M.~Eto, M.~Kurachi, and M.~Nitta, ``{Non-Abelian strings and domain walls in
  two Higgs doublet models},'' {\em JHEP}, vol.~08, p.~195, 2018.

\bibitem{Eto:2004rz}
M.~Eto, Y.~Isozumi, M.~Nitta, K.~Ohashi, and N.~Sakai, ``{Instantons in the
  Higgs phase},'' {\em Phys. Rev. D}, vol.~72, p.~025011, 2005.

\bibitem{Eto:2006mz}
M.~Eto, T.~Fujimori, Y.~Isozumi, M.~Nitta, K.~Ohashi, K.~Ohta, and N.~Sakai,
  ``{Non-Abelian vortices on cylinder: Duality between vortices and walls},''
  {\em Phys. Rev. D}, vol.~73, p.~085008, 2006.

\bibitem{Dunne:2012ae}
G.~V. Dunne and M.~Unsal, ``{Resurgence and Trans-series in Quantum Field
  Theory: The CP(N-1) Model},'' {\em JHEP}, vol.~11, p.~170, 2012.

\bibitem{Dunne:2012zk}
G.~V. Dunne and M.~\"Unsal, ``{Continuity and Resurgence: towards a continuum
  definition of the $\mathbb{CP}$(N-1) model},'' {\em Phys. Rev. D}, vol.~87,
  p.~025015, 2013.

\bibitem{Misumi:2014jua}
T.~Misumi, M.~Nitta, and N.~Sakai, ``{Neutral bions in the ${\mathbb C}P^{N-1}$
  model},'' {\em JHEP}, vol.~06, p.~164, 2014.

\bibitem{Misumi:2014bsa}
T.~Misumi, M.~Nitta, and N.~Sakai, ``{Classifying bions in Grassmann sigma
  models and non-Abelian gauge theories by D-branes},'' {\em PTEP}, vol.~2015,
  p.~033B02, 2015.

\bibitem{Fujimori:2016ljw}
T.~Fujimori, S.~Kamata, T.~Misumi, M.~Nitta, and N.~Sakai, ``{Nonperturbative
  contributions from complexified solutions in $\mathbb{C}P^{N-1}$models},''
  {\em Phys. Rev. D}, vol.~94, no.~10, p.~105002, 2016.

\bibitem{Fujimori:2017oab}
T.~Fujimori, S.~Kamata, T.~Misumi, M.~Nitta, and N.~Sakai, ``{Exact resurgent
  trans-series and multibion contributions to all orders},'' {\em Phys. Rev.
  D}, vol.~95, no.~10, p.~105001, 2017.

\bibitem{Fujimori:2018kqp}
T.~Fujimori, S.~Kamata, T.~Misumi, M.~Nitta, and N.~Sakai, ``{Bion
  non-perturbative contributions versus infrared renormalons in two-dimensional
  $\mathbb C P^{N-1}$ models},'' {\em JHEP}, vol.~02, p.~190, 2019.

\bibitem{Misumi:2019upg}
T.~Misumi, T.~Fujimori, E.~Itou, M.~Nitta, and N.~Sakai, ``{Lattice study on
  the twisted ${\mathbb C} P^{N-1}$ models on ${\mathbb R} \times S^1$},'' {\em
  PoS}, vol.~LATTICE2019, p.~015, 2019.

\bibitem{Fujimori:2020zka}
T.~Fujimori, E.~Itou, T.~Misumi, M.~Nitta, and N.~Sakai, ``{Lattice ${\mathbb
  C} P^{N-1}$ model with ${\mathbb Z}_{N}$ twisted boundary condition: bions,
  adiabatic continuity and pseudo-entropy},'' {\em JHEP}, vol.~08, no.~08,
  p.~011, 2020.

\bibitem{Nitta:2017uog}
M.~Nitta and R.~Yoshii, ``{Self-consistent large-N analytical solutions of
  inhomogeneous condensates in quantum \ensuremath{\mathbb{C}}P$^{N-1}$
  model},'' {\em JHEP}, vol.~12, p.~145, 2017.

\bibitem{Bolognesi:2016zjp}
S.~Bolognesi, K.~Konishi, and K.~Ohashi, ``{Large-$N {\mathbb C}P^{N - 1}$
  sigma model on a finite interval},'' {\em JHEP}, vol.~10, p.~073, 2016.

\bibitem{Betti:2017zcm}
A.~Betti, S.~Bolognesi, S.~B. Gudnason, K.~Konishi, and K.~Ohashi, ``{Large-N
  $\mathbb{C}{\mathrm{\mathbb{P}}}^{\mathrm{N}-1}$ sigma model on a finite
  interval and the renormalized string energy},'' {\em JHEP}, vol.~01, p.~106,
  2018.

\bibitem{Bolognesi:2018njt}
S.~Bolognesi, S.~B. Gudnason, K.~Konishi, and K.~Ohashi, ``{Large-$N$
  $\mathbb{CP}^{N-1}$ sigma model on a finite interval: general Dirichlet
  boundary conditions},'' {\em JHEP}, vol.~06, p.~064, 2018.

\bibitem{Flachi:2017xat}
A.~Flachi, M.~Nitta, S.~Takada, and R.~Yoshii, ``{Casimir force for the
  ${\mathbb C}P^{N-1}$ model},'' {\em Phys. Lett. B}, vol.~798, p.~134999,
  2019.

\bibitem{Flachi:2019jus}
A.~Flachi, G.~Fucci, M.~Nitta, S.~Takada, and R.~Yoshii, ``{Ground state
  modulations in the ${\mathbb C}P^{N-1}$ model},'' {\em Phys. Rev. D},
  vol.~100, no.~8, p.~085006, 2019.

\bibitem{Haldane:1982rj}
F.~D.~M. Haldane, ``{Continuum dynamics of the 1-D Heisenberg antiferromagnetic
  identification with the O(3) nonlinear sigma model},'' {\em Phys. Lett. A},
  vol.~93, pp.~464--468, 1983.

\bibitem{Affleck:1988nt}
I.~Affleck, ``{Quantum Spin Chains and the Haldane Gap},'' {\em J. Phys. C},
  vol.~1, p.~3047, 1989.

\bibitem{Senthil:2003eed}
T.~Senthil, A.~Vishwanath, L.~Balents, S.~Sachdev, and M.~P.~A. Fisher,
  ``{Deconfined Quantum Critical Points},'' {\em Science}, vol.~303, no.~5663,
  pp.~1490--1494, 2004.

\bibitem{PhysRevB.70.144407}
T.~Senthil, L.~Balents, S.~Sachdev, A.~Vishwanath, and M.~P.~A. Fisher,
  ``Quantum criticality beyond the landau-ginzburg-wilson paradigm,'' {\em
  Phys. Rev. B}, vol.~70, p.~144407, Oct 2004.

\bibitem{Nogueira:2013oza}
F.~S. Nogueira and A.~Sudb\o{}, ``{Deconfined Quantum Criticality and Conformal
  Phase Transition in Two-Dimensional Antiferromagnets},'' {\em EPL}, vol.~104,
  no.~5, p.~56004, 2013.

\bibitem{Beard:2004jr}
B.~B. Beard, M.~Pepe, S.~Riederer, and U.~J. Wiese, ``{Study of CP(N-1)
  theta-vacua by cluster-simulation of SU(N) quantum spin ladders},'' {\em
  Phys. Rev. Lett.}, vol.~94, p.~010603, 2005.

\bibitem{Zohar:2015hwa}
E.~Zohar, J.~I. Cirac, and B.~Reznik, ``{Quantum Simulations of Lattice Gauge
  Theories using Ultracold Atoms in Optical Lattices},'' {\em Rept. Prog.
  Phys.}, vol.~79, no.~1, p.~014401, 2016.

\bibitem{Laflamme:2015wma}
C.~Laflamme, W.~Evans, M.~Dalmonte, U.~Gerber, H.~Mej\'\i{}a-D\'\i{}az,
  W.~Bietenholz, U.~J. Wiese, and P.~Zoller, ``{$\mathbb{C}$P(N\ensuremath{-}1)
  quantum field theories with alkaline-earth atoms in optical lattices},'' {\em
  Annals Phys.}, vol.~370, pp.~117--127, 2016.

\bibitem{Affleck:2021ypq}
I.~Affleck, D.~Bykov, and K.~Wamer, ``{Flag manifold sigma models:}: {Spin
  chains and integrable theories},'' {\em Phys. Rept.}, vol.~953, pp.~1--93,
  2022.

\bibitem{Bykov:2014efa}
D.~Bykov, ``{Integrable properties of sigma-models with non-symmetric target
  spaces},'' {\em Nucl. Phys. B}, vol.~894, pp.~254--267, 2015.

\bibitem{Bykov:2015pka}
D.~Bykov, ``{Classical solutions of a flag manifold $\sigma$-model},'' {\em
  Nucl. Phys. B}, vol.~902, pp.~292--301, 2016.

\bibitem{Bykov:2019jbz}
D.~Bykov, ``{Flag manifold $\sigma$-models: The $\frac1{N}$-expansion and the
  anomaly two-form},'' {\em Nucl. Phys. B}, vol.~941, pp.~316--360, 2019.

\bibitem{Bykov:2019vkf}
D.~Bykov, ``{Flag manifold sigma-models and nilpotent orbits},'' {\em Proc.
  Steklov Inst. Math.}, vol.~309, pp.~78--86, 2020.

\bibitem{Hongo:2018rpy}
M.~Hongo, T.~Misumi, and Y.~Tanizaki, ``{Phase structure of the twisted
  $SU(3)/U(1)^2$ flag sigma model on $\mathbb{R}\times S^1$},'' {\em JHEP},
  vol.~02, p.~070, 2019.

\bibitem{Tanizaki:2018xto}
Y.~Tanizaki and T.~Sulejmanpasic, ``{Anomaly and global inconsistency matching:
  $\theta$-angles, $SU(3)/U(1)^2$ nonlinear sigma model, $SU(3)$ chains and its
  generalizations},'' {\em Phys. Rev. B}, vol.~98, no.~11, p.~115126, 2018.

\bibitem{Ohmori:2018qza}
K.~Ohmori, N.~Seiberg, and S.-H. Shao, ``{Sigma Models on Flags},'' {\em
  SciPost Phys.}, vol.~6, no.~2, p.~017, 2019.

\bibitem{PhysRevA.93.021606}
H.~T. Ueda, Y.~Akagi, and N.~Shannon, ``Quantum solitons with emergent
  interactions in a model of cold atoms on the triangular lattice,'' {\em Phys.
  Rev. A}, vol.~93, p.~021606, Feb 2016.

\bibitem{Amari:2017qnb}
Y.~Amari and N.~Sawado, ``{BPS sphalerons in the $F_2$ nonlinear sigma
  model},'' {\em Phys. Rev. D}, vol.~97, no.~6, p.~065012, 2018.

\bibitem{Amari:2018gbq}
Y.~Amari and N.~Sawado, ``{$SU(3)$ Knot Solitons: Hopfions in the $F_2$
  Skyrme-Faddeev-Niemi model},'' {\em Phys. Lett. B}, vol.~784, pp.~294--300,
  2018.

\bibitem{Wamer:2020inf}
K.~Wamer and I.~Affleck, ``{Flag manifold sigma models from SU($n$) chains},''
  {\em Nucl. Phys. B}, vol.~959, p.~115156, 2020.

\bibitem{Kobayashi:2021qfj}
R.~Kobayashi, Y.~Lee, K.~Shiozaki, and Y.~Tanizaki, ``{Topological terms of
  (2+1)d flag-manifold sigma models},'' {\em JHEP}, vol.~08, p.~075, 2021.

\bibitem{Eto:2010aj}
M.~Eto, T.~Fujimori, S.~Bjarke~Gudnason, Y.~Jiang, K.~Konishi, M.~Nitta, and
  K.~Ohashi, ``{Group Theory of Non-Abelian Vortices},'' {\em JHEP}, vol.~11,
  p.~042, 2010.

\bibitem{Ireson:2019gtn}
E.~Ireson, ``{General Composite Non-Abelian Strings and Flag Manifold Sigma
  Models},'' {\em Phys. Rev. Res.}, vol.~2, no.~1, p.~013038, 2020.

\bibitem{Kobayashi:2013axa}
M.~Kobayashi, E.~Nakano, and M.~Nitta, ``{Color Magnetism in Non-Abelian Vortex
  Matter},'' {\em JHEP}, vol.~06, p.~130, 2014.

\bibitem{Shifman:2006kd}
M.~Shifman and A.~Yung, ``{Non-Abelian semilocal strings in N=2 supersymmetric
  QCD},'' {\em Phys. Rev. D}, vol.~73, p.~125012, 2006.

\bibitem{Eto:2007yv}
M.~Eto, J.~Evslin, K.~Konishi, G.~Marmorini, M.~Nitta, K.~Ohashi, W.~Vinci, and
  N.~Yokoi, ``{On the moduli space of semilocal strings and lumps},'' {\em
  Phys. Rev. D}, vol.~76, p.~105002, 2007.

\bibitem{Bando:1983ab}
M.~Bando, T.~Kuramoto, T.~Maskawa, and S.~Uehara, ``{Structure of Nonlinear
  Realization in Supersymmetric Theories},'' {\em Phys. Lett. B}, vol.~138,
  p.~94, 1984.

\bibitem{Bando:1984cc}
M.~Bando, T.~Kuramoto, T.~Maskawa, and S.~Uehara, ``{Nonlinear Realization in
  Supersymmetric Theories},'' {\em Prog. Theor. Phys.}, vol.~72, p.~313, 1984.

\bibitem{Bando:1984fn}
M.~Bando, T.~Kuramoto, T.~Maskawa, and S.~Uehara, ``{Nonlinear Realization in
  Supersymmetric Theories. 2.},'' {\em Prog. Theor. Phys.}, vol.~72, p.~1207,
  1984.

\bibitem{Itoh:1985ha}
K.~Itoh, T.~Kugo, and H.~Kunitomo, ``{Supersymmetric Nonlinear Realization for
  Arbitrary Kahlerian Coset Space $G/H$},'' {\em Nucl. Phys. B}, vol.~263,
  pp.~295--308, 1986.

\bibitem{Itoh:1985jz}
K.~Itoh, T.~Kugo, and H.~Kunitomo, ``{Supersymmetric Nonlinear Lagrangians of
  Kahlerian Coset Spaces $G/H$: $G$ = E6, E7 and E8},'' {\em Prog. Theor.
  Phys.}, vol.~75, p.~386, 1986.

\bibitem{Nitta:2003dv}
M.~Nitta, ``{Auxiliary field methods in supersymmetric nonlinear sigma
  models},'' {\em Nucl. Phys. B}, vol.~711, pp.~133--162, 2005.

\bibitem{Zumino:1979et}
B.~Zumino, ``{Supersymmetry and Kahler Manifolds},'' {\em Phys. Lett. B},
  vol.~87, p.~203, 1979.

\bibitem{Donagi:2007hi}
R.~Donagi and E.~Sharpe, ``{GLSM's for partial flag manifolds},'' {\em J. Geom.
  Phys.}, vol.~58, pp.~1662--1692, 2008.

\bibitem{Vachaspati:1991dz}
T.~Vachaspati and A.~Achucarro, ``{Semilocal cosmic strings},'' {\em Phys. Rev.
  D}, vol.~44, pp.~3067--3071, 1991.

\bibitem{Achucarro:1999it}
A.~Achucarro and T.~Vachaspati, ``{Semilocal and electroweak strings},'' {\em
  Phys. Rept.}, vol.~327, pp.~347--426, 2000.

\bibitem{Isozumi:2004vg}
Y.~Isozumi, M.~Nitta, K.~Ohashi, and N.~Sakai, ``{All exact solutions of a 1/4
  Bogomol'nyi-Prasad-Sommerfield equation},'' {\em Phys. Rev. D}, vol.~71,
  p.~065018, 2005.

\bibitem{Hanany:2014hia}
A.~Hanany and R.-K. Seong, ``{Hilbert series and moduli spaces of $k$ U(N )
  vortices},'' {\em JHEP}, vol.~02, p.~012, 2015.

\bibitem{Taubes:1979tm}
C.~H. Taubes, ``{Arbitrary N: Vortex Solutions to the First Order
  Landau-Ginzburg Equations},'' {\em Commun. Math. Phys.}, vol.~72,
  pp.~277--292, 1980.

\bibitem{Eto_2009}
M.~Eto, T.~Fujimori, S.~B. Gudnason, K.~Konishi, T.~Nagashima, M.~Nitta,
  K.~Ohashi, and W.~Vinci, ``Non-abelian vortices in {SO}(n) and {USp}(n) gauge
  theories,'' {\em Journal of High Energy Physics}, vol.~2009, pp.~004--004,
  jun 2009.

\bibitem{Yoshida:2011au}
Y.~Yoshida, ``{Localization of Vortex Partition Functions in $\mathcal{N}=(2,2)
  $ Super Yang-Mills theory},'' 1 2011.

\bibitem{Bonelli:2011fq}
G.~Bonelli, A.~Tanzini, and J.~Zhao, ``{Vertices, Vortices and Interacting
  Surface Operators},'' {\em JHEP}, vol.~06, p.~178, 2012.

\bibitem{Hori:2006dk}
K.~Hori and D.~Tong, ``{Aspects of Non-Abelian Gauge Dynamics in
  Two-Dimensional N=(2,2) Theories},'' {\em JHEP}, vol.~05, p.~079, 2007.

\bibitem{Benini:2012ui}
F.~Benini and S.~Cremonesi, ``{Partition Functions of ${\mathcal{N}=(2,2)}$
  Gauge Theories on S$^{2}$ and Vortices},'' {\em Commun. Math. Phys.},
  vol.~334, no.~3, pp.~1483--1527, 2015.

\bibitem{Fujimori:2010fk}
T.~Fujimori, G.~Marmorini, M.~Nitta, K.~Ohashi, and N.~Sakai, ``{The Moduli
  Space Metric for Well-Separated Non-Abelian Vortices},'' {\em Phys. Rev. D},
  vol.~82, p.~065005, 2010.

\bibitem{Eto:2011pj}
M.~Eto, T.~Fujimori, M.~Nitta, K.~Ohashi, and N.~Sakai, ``{Dynamics of
  Non-Abelian Vortices},'' {\em Phys. Rev. D}, vol.~84, p.~125030, 2011.

\bibitem{Higashijima:1999ki}
K.~Higashijima and M.~Nitta, ``{Supersymmetric nonlinear sigma models as gauge
  theories},'' {\em Prog. Theor. Phys.}, vol.~103, pp.~635--663, 2000.

\bibitem{Benini:2014mia}
F.~Benini, D.~S. Park, and P.~Zhao, ``{Cluster Algebras from Dualities of 2d
  ${\mathcal{N}}$ = (2, 2) Quiver Gauge Theories},'' {\em Commun. Math. Phys.},
  vol.~340, pp.~47--104, 2015.

\bibitem{Tarantello2011NonabelianVE}
G.~Tarantello, ``Non-abelian vortices: Existence, uniqueness and asymptotics,''
  {\em Milan Journal of Mathematics}, vol.~79, pp.~343--356, 2011.

\bibitem{Chen2012ExistenceOM}
S.~Chen and Y.~Yang, ``Existence of multiple vortices in supersymmetric gauge
  field theory,'' {\em Proceedings of the Royal Society A: Mathematical,
  Physical and Engineering Sciences}, vol.~468, pp.~3923 -- 3946, 2012.

\bibitem{HAN2014117}
X.~Han and C.-S. Lin, ``Existence of non-abelian vortices with product gauge
  groups,'' {\em Nuclear Physics B}, vol.~878, pp.~117--149, 2014.

\bibitem{Miyake:2011yr}
A.~Miyake, K.~Ohta, and N.~Sakai, ``{Volume of Moduli Space of Vortex Equations
  and Localization},'' {\em Prog. Theor. Phys.}, vol.~126, pp.~637--680, 2011.

\bibitem{Miyake:2011fq}
A.~Miyake, K.~Ohta, and N.~Sakai, ``{Moduli space volume of vortex and
  localization},'' {\em J. Phys. Conf. Ser.}, vol.~343, p.~012107, 2012.

\bibitem{Ohta:2019odi}
K.~Ohta and N.~Sakai, ``{Higgs and Coulomb Branch Descriptions of the Volume of
  the Vortex Moduli Space},'' {\em PTEP}, vol.~2019, no.~4, p.~043B01, 2019.

\bibitem{Ohta_2021}
K.~Ohta and N.~Sakai, ``The volume of the quiver vortex moduli space,'' {\em
  Progress of Theoretical and Experimental Physics}, vol.~2021, feb 2021.

\bibitem{Manton:1981mp}
N.~S. Manton, ``{A Remark on the Scattering of BPS Monopoles},'' {\em Phys.
  Lett. B}, vol.~110, pp.~54--56, 1982.

\bibitem{Eto:2006uw}
M.~Eto, Y.~Isozumi, M.~Nitta, K.~Ohashi, and N.~Sakai, ``{Manifestly
  supersymmetric effective Lagrangians on BPS solitons},'' {\em Phys. Rev. D},
  vol.~73, p.~125008, 2006.

\bibitem{Chen:2011sj}
H.-Y. Chen, N.~Dorey, T.~J. Hollowood, and S.~Lee, ``{A New 2d/4d Duality via
  Integrability},'' {\em JHEP}, vol.~09, p.~040, 2011.

\end{thebibliography}

\end{document}